\documentclass{algareport}

\input{preamble}

\usepackage{booktabs}

\graphicspath{{./figures/}{./}}

\begin{document}

\setCluster{Algorithms, Geometry \& Applications}
\setTitle{Kernelization for Treewidth-2 Vertex Deletion}
\setName{Jeroen L.G. Schols}
\setSupervisors{Dr. Bart M.P. Jansen\par Huib T. Donkers MSc.}
\setCommittee{Dr. Bart M.P. Jansen\par Huib T. Donkers MSc.\par Dr. Bas S.P. Luttik}
\setDefenseDate{March 17th 2022}
\setCredits{30}

\beginfront

\defcitealias{Preprocessing_for_Outerplanar_Vertex_Deletion_An_Elementary_Kernel_of_Quartic_Size}{IPEC 2021}
\defcitealias{Approximation-and-Tidying-A-Problem-Kernel-for-s-Plex-Cluster-Vertex-Deletion}{Algorithmica 2012}
~
\vspace{1.1cm}
\begin{abstract}
\centering
\noindent
\begin{minipage}{0.75\linewidth}
    \bigskip
    The \textsc{Treewidth-2 Vertex Deletion} problem asks whether a set of at most $t$ vertices can be removed from a graph, such that the resulting graph has treewidth at most two. A graph has treewidth at most two if and only if it does not contain a $K_4$ minor. Hence, this problem corresponds to the NP-hard $\mathcal{F}$-\textsc{Minor Cover} problem with $\mathcal{F} = \{K_4\}$. For any variant of the $\mathcal{F}$-\textsc{Minor Cover} problem where $\mathcal{F}$ contains a planar graph, it is known that a polynomial kernel exists. I.e., a preprocessing routine that in polynomial time outputs an equivalent instance of size $t^{O(1)}$. However, this proof is non-constructive, meaning that this proof does not yield an explicit bound on the kernel size. The $\{K_4\}$-\textsc{Minor Cover} problem is the simplest variant of the $\mathcal{F}$-\textsc{Minor Cover} problem with an unknown kernel size.
    
    \hspace{13.5pt}
    To develop a constructive kernelization algorithm, we present a new method to decompose graphs into near-protrusions, such that near-protrusions in this new decomposition can be reduced using elementary reduction rules. Our method extends the ``approximation and tidying'' framework by van Bevern et al. \citepalias{Approximation-and-Tidying-A-Problem-Kernel-for-s-Plex-Cluster-Vertex-Deletion} to provide guarantees stronger than those provided by both this framework and a regular protrusion decomposition. Furthermore, we provide extensions of the elementary reduction rules used by the $\{K_4, K_{2,3}\}$-\textsc{Minor Cover} kernelization algorithm introduced by Donkers et al. \citepalias{Preprocessing_for_Outerplanar_Vertex_Deletion_An_Elementary_Kernel_of_Quartic_Size}.
    
    \hspace{13.5pt}
    Using the new decomposition method and reduction rules, we obtain a kernel consisting of $O(t^{41})$ vertices, which is the first constructive kernel. This kernel is a step towards more concrete kernelization bounds for the $\mathcal{F}$-\textsc{Minor Cover} problem where $\mathcal{F}$ contains a planar graph, and our decomposition provides a potential direction to achieve these new bounds.
\end{minipage}
\end{abstract}

\tableofcontents

\stopfront

\chapter{Introduction}\label{ch:intro}
Many important computer science problems are difficult to solve efficiently. The class of NP-hard problems describes a family of problems that are expected to not be solvable using a polynomial time algorithm. In this area of `classical' complexity theory, the efficiency of algorithms is measured solely as a function over the size of the input (e.g. measured in bits). Parameterized algorithmics instead analyses the running time of an algorithm not solely on the size of the input, but instead also on one or more additional parameters of the input instances, which enables a way to analyse complexity in more detail. For example, unless NP = P, determining whether a graph $G$ with $n$ vertices has a \textsc{Vertex Cover} of size at most $t$ can not be done in $\poly{n}$ time. However, when we also measure the impact of parameter $t$ on the running time, then we can find algorithms that solve this problem in $\phi(t) \cdot \poly{n}$ time \cite[Chapter 1]{Cygan_Parameterized_Algorithms} (for some computable function $\phi$). A problem is called fixed parameter tractable (FPT) when it can be solved with an algorithm of which the running time can be expressed as $\phi(t) \cdot \poly{n}$. For such algorithms it holds that for any constant (fixed) parameter value $t$ the problem can be solved in polynomial (tractable) time. Hence, the \textsc{Vertex Cover} problem parameterized by $t$ is FPT, because for any fixed value $t = O(1)$ the problem complexity can be expressed as $\phi(O(1)) \cdot \poly{n} = O(1) \cdot \poly{n}$. The main benefit FPT algorithms offer over non-FPT algorithms is that such algorithms can compute a solution to any problem instance, but do so much more efficiently when these parameters are relatively small.

The most common parameterizations take the cost (target size) of the solution as the parameter \cite{reflections}, however, there are many alternatives. A multitude of graph problems can be efficiently solved when the input graph is a tree. On these type of problems one candidate parameter could be a measure of ``tree-likeness'', i.e., how similar the input graph is to a tree. For example, a parameter choice could be the number of cycles in the graph or the size of a \textsc{Feedback Vertex}\footnote{A set of vertices whose removal leaves a graph without cycles.} set. As stated by Cygan et al. \cite{Cygan_Parameterized_Algorithms}, the approach most useful from algorithmic and graph theoretical perspectives, is to view tree-likeness of a graph $G$ as the existence of a structural decomposition of $G$ into pieces of bounded size that are connected in a tree-like fashion (see Section \ref{sec:prelim:tree-decomposition-and-treewidth} for a formal definition). The parameter $\eta$ that measures this tree-likeness is called the treewidth. For example, a graph without an edge has treewidth $\eta = 0$ and a graph without cycles (i.e. a tree) has treewidth $\eta = 1$. The usage of treewidth as a measure of complexity has found widespread usage in parameterized algorithmics, some of the results we will cover in Sections \ref{sec:intro:background} and
\ref{sec:intro:related-works}.

One rather fundamental NP-hard graph problem is the $\mathcal{F}$-\textsc{Minor Cover} problem. A minor relationship between graphs is similar to the subgraph relationship. A graph $H$ is a minor of graph $G$ when $H$ can be obtained from $G$ by deleting edges and vertices and by contracting edges. Contracting an edge correspond to merging the endpoints of the edge. We show an example graph $G$ that has $H$ as a minor in Figure \ref{fig:example-minor}. We note that graph $H$ corresponds to the graph $K_4$, which is a clique consisting of four nodes.
\clearpage

\begin{figure}
\centering
\begin{subfigure}{.4\textwidth}
  \centering
  \includegraphics[page=1, scale=.61]{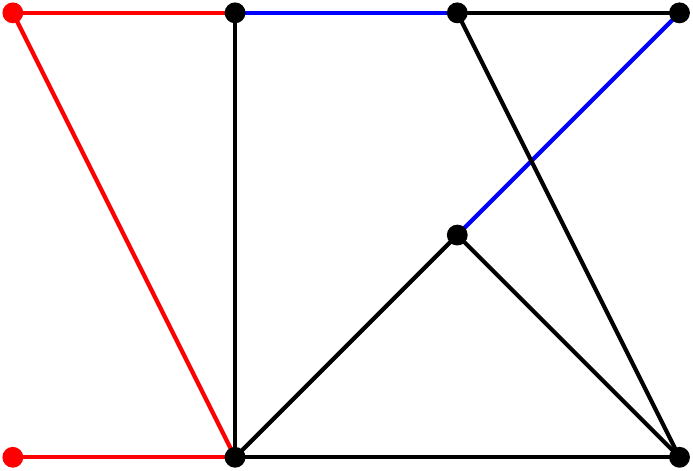}
  \caption{Graph $G$}
  \label{fig:example-minor:G}
\end{subfigure}%
\begin{subfigure}{.4\textwidth}
  \centering
  \includegraphics[page=2, scale=.61]{example-kernelization.pdf}
  \caption{Minor $H$}
  \label{fig:example-minor:H}
\end{subfigure}%
\caption{An example graph $H$ that is a minor of $G$ obtained by contracting blue edges and deleting red vertices and edges}
\label{fig:example-minor}
\end{figure}

\noindent
The $\mathcal{F}$-\textsc{Minor Cover} problem asks, given a graph $G$ and a target size $t$, does there exist a set $S$ of at most $t$ vertices such that the removal of all vertices $S$ from $G$ results in a graph that has no graph $F \in \mathcal{F}$ as a minor. Well-known variants of the $\mathcal{F}$-\textsc{Minor Cover} problem are the \textsc{Vertex Cover} problem ($\mathcal{F} = \{K_2\}$), the \textsc{Feedback Vertex Set} problem ($\mathcal{F} = \{K_3\}$), and the \textsc{Planarization} problem ($\mathcal{F} = \{K_5, K_{3,3}\}$). Interestingly, a graph without a $K_2$ minor corresponds to a graph with treewidth $\eta = 0$, and a graph without a $K_3$ minor corresponds to a graph with treewidth $\eta = 1$. Hence, these problems correspond respectively to the \textsc{Treewidth-0 Vertex Deletion} problem and the \textsc{Treewidth-1 Vertex Deletion} problem. The next problem in line is the \textsc{Treewidth-2 Vertex Deletion} problem, which corresponds to the $\{K_4\}$-\textsc{Minor Cover} problem, which will be the problem tackled in this thesis.

A commonplace technique for designing FPT-algorithms and speeding up computation is that of preprocessing. In `classical' complexity theory, unless NP = P, it is impossible to prove guarantees on the effectiveness of polynomial time preprocessing routines for NP-hard problems. This, because polynomial time preprocessing algorithms for NP-hard (decision) problems that are proven to make progress (i.e. strictly reduce the size of the problem) and do not change the answer would lead to an algorithm that solves a NP-hard problem in polynomial time. By using parameterized complexity theory we are able to prove guarantees on the effectiveness of polynomial time preprocessing routines. Kernelization algorithms are those polynomial time preprocessing algorithms that are proven to make progress until the problem instance has a size bounded by some function $f$ over the parameter $t$. I.e., a kernelization algorithm is a preprocessing algorithm that can be proven to make progress (without changing the decision-answer) when the problem instances with parameter $t$ are larger than some known bound $f(t)$.

We note that these kernelization algorithms themselves do not find a solution to the problem. I.e., these algorithms are preprocessing routines and not `solvers'. Nevertheless, after a problem instance is kernelized, applying a brute force search (assuming such an algorithm exists) will result in a FPT-algorithm. This, because the brute force algorithm will only take $\phi(|G'|) = \phi(f(t))$ time, which in case $t$ is a constant will be $\phi(f(t)) = O(1)$.

We illustrate the application of such a kernelization algorithm for the \textsc{Vertex Cover} problem. In our example we apply the kernelization method described by Fomin et al. \cite[Section 2.2]{Kernelization_Theory_of_Parameterized_Preprocessing}, although we will not introduce this algorithm. A $f(t)$ kernelization algorithm for the \textsc{Vertex Cover} problem is an algorithm that given a graph $G$ with $n$ vertices and a parameter $t$ returns in $\poly{n}$ time a graph $G'$ with $n' \leq f(t)$ vertices and a parameter $t'$. The algorithm needs to provide the guarantee that: graph $G'$ has a vertex cover of size $t'$ if and only if graph $G$ has a vertex cover of size $t$. For example, this kernelization algorithm could obtain as input the graph $G$ shown in Figure \ref{fig:example-kernelization:G} and a target size $t = 10$. In this case, the algorithm should be able to find in $\poly{n}$ time a graph $G'$ with $|V(G')| \leq f(t)$ and a target size $t'$, such as the graph $G'$ with $t' = 3$ shown in Figure \ref{fig:example-kernelization:G'}. As shown in Figure \ref{fig:example-kernelization} by red vertices, both $G'$ and $G$ have vertex covers of size $t$ and $t'$ respectively. Hence, this obtained problem instance would be a valid output for the kernelization algorithm. Nevertheless, as indicated before, the kernelization algorithm itself will not determine whether such vertex covers exist.\clearpage

\begin{figure}[ht]
\centering
\begin{subfigure}{.6\linewidth}
  \centering
  \includegraphics[page=5, scale=.75]{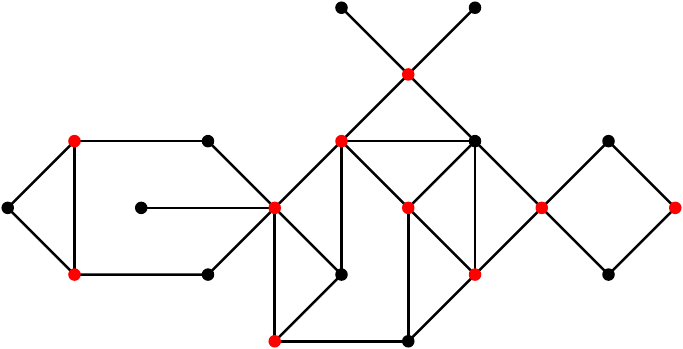}
  \caption{Graph $G$ with target size $t = 10$}
  \label{fig:example-kernelization:G}
\end{subfigure}%
\begin{subfigure}{.4\textwidth}
  \centering
  \includegraphics[page=6, scale=.75]{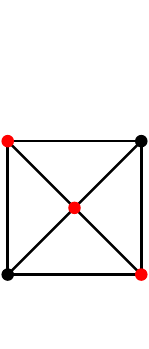}
  \caption{Graph $G'$ with target size $t' = 3$}
  \label{fig:example-kernelization:G'}
\end{subfigure}%
\caption{Example \textsc{Vertex Cover} problem $(G,t)$ and its reduced problem $(G',t')$}
\label{fig:example-kernelization}
\end{figure}

\noindent
An important consequence of kernelizations algorithms is that they prove that large NP-hard problems can in polynomial time be shrunken to their computationally difficult ``core'' (also called the kernel), and the size of the core is independent of the sizes of the original problems. Besides their practical application, these kernels therefore also provide a key to better understanding the ``nature'' of the computational intractability underlying the problem \cite{reflections}.

An important goal in the design of kernelization algorithms is to obtain a kernel of a size bounded by a small value $f(t)$. This, because the value of $f(t)$ dictates the super-polynomial time complexity $\phi(f(t))$ in solving the NP-hard problem. In case a problem has a kernelization algorithm where function $f$ is polynomial or linear, then this problem is said to admit a polynomial or linear kernel respectively.

\section{Background}\label{sec:intro:background}
The $\mathcal{F}$-\textsc{Minor Cover} problem (also called the $\mathcal{F}$-\textsc{Minor-Free Deletion} problem) has received much attention since the advent of parameterized algorithmics. Lewis and Yannakakis \cite{The_node_deletion_problem_for_hereditary_properties_is_NP_complete} have proven that any variant of the  $\mathcal{F}$-\textsc{Minor Cover} problem, where each graph in $\mathcal{F}$ contains at least one edge, is NP-hard.

As indicated, well-known variants of the $\mathcal{F}$-\textsc{Minor Cover} problem include the \textsc{Vertex Cover} problem ($\mathcal{F} = \{K_2\}$), the \textsc{Feedback Vertex Set} problem ($\mathcal{F} = \{K_3\}$), and the \textsc{Planarization} problem ($\mathcal{F} = \{K_5, K_{3,3}\}$). Each of these problems have garnered a lot of attention and have been extensively studied separately \cite{Fellows2018, Festa99feedbackset, planarizing-graphs}. The $\mathcal{F}$-\textsc{Minor Cover} generalises these problems into one fundamental problem. As a consequence, (meta-)results on the $\mathcal{F}$-\textsc{Minor Cover} can provide algorithmic insights relevant to all problems from the entire class of $\mathcal{F}$-\textsc{Minor Cover} problems.

For the \textsc{Vertex Cover} problem and the \textsc{Feedback Vertex Set} problem explicit kernelization algorithms are known. The \textsc{Vertex Cover} problem parameterized by target size $t$ admits a linear kernel of size $2t$, which can be obtained by either a crown reduction \cite{Kernelization_Theory_of_Parameterized_Preprocessing} or by using the linear programming approach introduced by Chen et al. \cite{CHEN2001280}. The \textsc{Feedback Vertex Set} problem parameterized by target size $t$ admits a kernel of size $2t^2 + t$ as introduced by Iwata \cite{iwata:LIPIcs:2017:7430}.

Robertson and Seymour \cite{ROBERTSON198692} have proven that if $H$ is a planar graph, that every graph $G$ that does not contain $H$ as a minor must have a bounded treewidth. Opposed to that, for every $\eta > 0$ there exists a planar graph $H$ such that no graph with treewidth at most $\eta$ contains $H$ as a minor \cite{graph-minor-theory}. Hence, a family of $\mathcal{F}$-minor free graphs has bounded treewidth if and only if $\mathcal{F}$ includes a planar graph. Fomin et al. \cite{Planar_F-Deletion_Approximation_and_Optimal_FPT_Algorithms} used this property to obtain constant factor approximation algorithms, single exponential FPT algorithms, and polynomial kernels for the variant of the $\mathcal{F}$-\textsc{Minor Cover} problem where $\mathcal{F}$ contains a planar graph. Their method for obtaining these algorithms, however, rely on protrusion replacement techniques. A protrusion is a graph with a constant treewidth and a constant boundary size, and protrusion replacement is a technique where large protrusions are replaced by smaller ones of constant size. These generic protrusion replacement techniques, however, rely on a non-constructive argument that every protrusion can be replaced by a representative from a finite set of graphs \cite{DBLP:journals/corr/JansenW16}. Consequently, the results by Fomin et al. \cite{Planar_F-Deletion_Approximation_and_Optimal_FPT_Algorithms} are non-constructive, meaning that they only prove that such algorithms must exist, but their proof does not provide these algorithms\footnote{We note that the kernelization algorithm introduced by Fomin et al. \cite{Planar_F-Deletion_Approximation_and_Optimal_FPT_Algorithms} does not apply a protrusion replacement technique. Instead, their kernelization algorithm finds a single edge in a protrusion to remove. Nevertheless, also this method relies on a non-constructive argument that bounds the size of the largest graph in certain antichains
in a well-quasi-order. Meaning that also their kernelization algorithm is non-constructive.}. Therefore, because $K_4$ is a planar graph, we know that a polynomial kernel for the \textsc{Treewidth-2 Vertex Deletion} problem exists, however, the size of this kernel is unknown.

Kim et al. \cite{A_single-exponential_FPT_algorithm_for_the_K4-minor_cover_problem} introduced the first single-exponential FPT algorithm for the $\{K_4\}$-\textsc{Minor Cover} problem parameterized by target size. Their algorithm, however, also made use of generic protrusion replacement rules. Hence, similar to the work by Fomin et al. \cite{Planar_F-Deletion_Approximation_and_Optimal_FPT_Algorithms}, the work by Kim et al. \cite{A_single-exponential_FPT_algorithm_for_the_K4-minor_cover_problem} only yields existential results.

Fomin et al. \cite{Planar_F-Deletion_Approximation_and_Optimal_FPT_Algorithms} indicate that it is tempting to conjecture that the line of tractability is determined by whether $\mathcal{F}$ contains a planar graph or not. They base this argument on the observation that no constant factor approximation algorithm, no single exponential FPT algorithm, and no polynomial kernel is known for any variant of $\mathcal{F}$-\textsc{Minor Cover} problem where $\mathcal{F}$ does not contain a planar graph.

To better understand the tractability of the $\mathcal{F}$-\textsc{Minor Cover}, recent work by Donkers et al. \cite{Preprocessing_for_Outerplanar_Vertex_Deletion_An_Elementary_Kernel_of_Quartic_Size} focused on the simplest planar variant of the $\mathcal{F}$-\textsc{Minor Cover} problem for which no explicit kernel was known. This problem was the \textsc{Outerplanar Vertex Deletion} problem ($\mathcal{F} = \{K_4, K_{2,3}\}$). Their algorithm has a design similar to the non-constructive kernelization algorithm for the general $\mathcal{F}$-\textsc{Minor Cover} problem by Fomin et al. \cite{doi:10.1137/140997889}. The algorithm by Fomin et al. \cite{doi:10.1137/140997889} finds an approximate solution $S$ to the $\mathcal{F}$-\textsc{Minor Cover} problem and uses this to decompose the graph into so called near-protrusion. The algorithm by Donkers et al. \cite{Preprocessing_for_Outerplanar_Vertex_Deletion_An_Elementary_Kernel_of_Quartic_Size}, however, avoids using non-constructive arguments by using the procedure of ``tidying the modulator'' introduced by van Bevern et al. \cite{Approximation-and-Tidying-A-Problem-Kernel-for-s-Plex-Cluster-Vertex-Deletion}. A modulator corresponds to a vertex deletion set $S$ such that the removal of $S$ from the input graph $G$ results in a graph with the desired property (e.g., outerplanarity). Tidying the modulator corresponds to the process of adding additional vertices to modulator $S$ to ensure that any vertex from $S$ can be removed, and the resulting set will remain being a modulator for $G$. Donkers et al. \cite{Preprocessing_for_Outerplanar_Vertex_Deletion_An_Elementary_Kernel_of_Quartic_Size} use this `tidiness property' to reduce protrusions without relying on non-constructive arguments. Using this method they obtained a $O(t^4)$ kernel for the \textsc{Outerplanar Vertex Deletion} problem parameterized by target size $t$.

\section{Our Contributions}\label{sec:intro:our-contributions}
Within this thesis we develop a kernelization algorithm for the currently simplest planar variant of the $\mathcal{F}$-\textsc{Minor Cover} problem for which no explicit kernel is known, which is the \textsc{Treewidth-2 Vertex Deletion} problem ($\mathcal{F} = \{K_4\}$). We provide an $O(t^{41})$ kernel to the \textsc{Treewidth-2 Vertex Deletion} problem parameterized on the target size $t$. This yields the first explicit kernel size for the \textsc{Treewidth-2 Vertex Deletion} problem.

Our work builds upon the recent results on \textsc{Outerplanar Vertex Deletion} ($\mathcal{F} = \{K_4, K_{2,3}\}$) by Donkers et al. \cite{Preprocessing_for_Outerplanar_Vertex_Deletion_An_Elementary_Kernel_of_Quartic_Size}. Their algorithm exploits the fact that no $K_{2,3}$ minors are present in the graph $G - S$, where $S$ is an outerplanarity modulator. The absence of $K_{2,3}$ minors is a powerful property that limits the number of internally vertex-disjoint paths between any two vertices. Donkers et al. use the absence of $K_{2,3}$ minors to decompose graphs into near-protrusions and to reduce those near-protrusions without relying on non-constructive arguments. Our algorithm uses a novel approach for decomposing graphs, without relying on the absence of $K_{2,3}$ minors.

To decompose our graph, we extend the ``approximation and tidying'' framework introduced by van Bevern et al. \cite{Approximation-and-Tidying-A-Problem-Kernel-for-s-Plex-Cluster-Vertex-Deletion}. We briefly introduce this framework and our modifications.

The first step is an approximation step where a (treewidth-2) modulator $X$ is obtained using a constant factor approximation algorithm. The tidying step aims to find a tidy modulator $X' \subseteq V(G)$ of size polynomial in $t$ with $X \subseteq X'$. A tidy modulator is a vertex set for which any $x \in X'$ has that $X' \setminus \{x\}$ is a modulator. To obtain this set $X'$ the tidying step executes a routine for every vertex $x \in X$. It either determines that every modulator of size at most $t$ must include vertex $x$, or it finds a vertex set $X'_x$ of size polynomial in $t$, that when added to $X$ ensures that $(X \cup X'_x) \setminus \{x\}$ is a modulator. In case the former holds, then we know that there exists a treewidth-2 modulator of size $t$ for graph $G$ if and only if there exists a treewidth-2 modulator of size $t-1$ for graph $G - x$, in which case our kernelization algorithm recurses on $G - x$ with target size $t - 1$. In case the latter holds for each $x \in X$ we simply let $X'$ be the union over $X$ and all sets $X'_x$. The resulting set $X'$ then is a tidy modulator.

Where our algorithm deviates is that we want to run this tidying routine twice. I.e., we wish to find a tidy tidy modulator $Y \subseteq V(G)$ of size polynomial in $t$ with $X' \subseteq Y$, which is a vertex set for which any $x \in Y$ has that $Y \setminus \{x\}$ is a tidy modulator. In other words, we can remove any two vertices from $Y$ and the resulting set will be a modulator. It is straightforward to modify the tidying procedure by van Bevern et al. \cite{Approximation-and-Tidying-A-Problem-Kernel-for-s-Plex-Cluster-Vertex-Deletion} to execute the same routine for pairs of vertices $\{x,y\} \subseteq X'$ opposed to individual vertices. However, the tidying procedure does not provide a method to handle the case where it is determined that every modulator of size at most $t$ must include either vertex $x$ or $y$. I.e., in this case we know that there exists a treewidth-2 modulator of size $t$ for graph $G$ if and only if there exists a treewidth-2 modulator of size $t-1$ for either graph $G - x$ or graph $G - y$, however, we do not know whether this will be graph $G - x$ or $G - y$. We note that trying both options will lead to an algorithm that takes time exponential in $t$, whereas we search for an algorithm that takes $\poly{n}$ time. Our modified version, in this case, instead also finds a vertex set $Y_{x,y}$. Opposed to ensuring that the union $Y$ over all these sets is a tidy tidy modulator, we ensure that for every component $C$ in $G - (X' \cup Y)$ and for every modulator $S$ of size at most $t$; that all but at most one neighbour of $C$ in $X'$ must be included in $S$. We name this set $Y$ the component separator.

The property that all components in $G - (X' \cup Y)$ and all modulators $S$ of size at most $t$ must include all but at most one neighbour of $C$ in $X'$ stems from the method of obtaining an outerplanar decomposition by Donkers et al. \cite{Preprocessing_for_Outerplanar_Vertex_Deletion_An_Elementary_Kernel_of_Quartic_Size}. However, their usage of an outerplanar modulator $X$ ensures that $G - X$ contains no $K_{2,3}$ minor, which gives them a straightforward approach to obtaining a component separator $Y$. We note that $Y$ corresponds to vertex set $Z$ in Definition 3.13 in \cite{Preprocessing_for_Outerplanar_Vertex_Deletion_An_Elementary_Kernel_of_Quartic_Size}. Our extension of the ``approximation and tidying'' framework does not make use of the absence of $K_{2,3}$ minors, and has the potential to be more generally applicable for other variants of the $\mathcal{F}$-\textsc{Minor Cover} problem as well.

After having obtained our component separator we have decomposed our graph into near-protrusions, similar to the outerplanar decomposition introduced by Donkers et al. \cite{Preprocessing_for_Outerplanar_Vertex_Deletion_An_Elementary_Kernel_of_Quartic_Size}. Most reduction rules used for reducing components in $G - (X' \cup Y)$ are similar to those introduced by Donkers et al. Broadly speaking, these rules can be categorised into rules that target reducing large maximal biconnected subgraphs in $G - (X' \cup Y)$ and rules that target components in $G - (X' \cup Y)$ that consist of many maximal biconnected subgraphs. Although our reduction rules are similar to those by Donkers et al., significantly different methods for finding applications of these rules are introduced, due to the potential presence of $K_{2,3}$ minors in $G - (X' \cup Y)$.

In terms of graph theoretical relevance, we derive the first explicit kernel and first constructive fixed-parameter tractable algorithm for the \textsc{Treewidth-2 Vertex Deletion} problem. Furthermore, our extension of the ``approximation and tidying'' framework has the potential to provide a method for decomposing graphs for various variants of the $\mathcal{F}$-\textsc{Minor Cover} problem. However, the generalisation of this extended framework beyond its application to the \textsc{Treewidth-2 Vertex Deletion} problem is beyond the scope of this thesis.

\section{Related Works}\label{sec:intro:related-works}
Robertson and Seymour's graph minor theorem states that the class of all graphs is a well-quasi-order under the minor relation \cite{ROBERTSON2004325}. A consequence of this theorem is that a class of graphs closed under minors can be characterised by a finite set of minor obstructions \cite{BIENSTOCK1995481}. A minor obstruction for a graph property is a characterisation of a graph property by means of minor-minimal excluded graphs. For example, the property of having a treewidth bounded by $\eta$ is closed under minors. Hence, for each $\eta$ there exists a finite set of graphs $\mathcal{F}$, such that any graph $G$ has treewidth at most $\eta$ if and only if $G$ contains no graph from $\mathcal{F}$ as minor.

An important consequence of the graph minor theorem, as proven by Robertson and Seymour \cite{BIENSTOCK1995481, ROBERTSON199565}, is that for any fixed graph $F$ there exists
a polynomial-time algorithm that tests whether an input graph $G$ has $F$ as minor. As indicated by Bienstock and Langston \cite{BIENSTOCK1995481}, this yields that for any set $\mathcal{F}$ there exists an FPT algorithm that solves the $\mathcal{F}$-\textsc{Minor Cover} problem parameterized on target size $t$. As stated by Kim et al. \cite{A_single-exponential_FPT_algorithm_for_the_K4-minor_cover_problem}, an algorithm resulting from this meta-theorem will generally involve a huge exponential function $f$, that will be impractical even for small values of $t$. Regardless, the existence of such an algorithm also yields that the $\mathcal{F}$-\textsc{Minor Cover} problem has a kernel \cite[Chapter 2.1]{Cygan_Parameterized_Algorithms}.

We note that these results do not provide an answer to the conjecture posed by Fomin et al. \cite{Planar_F-Deletion_Approximation_and_Optimal_FPT_Algorithms}, that there does not exist a polynomial kernel for any instance of the $\mathcal{F}$-\textsc{Minor Cover} problem where $\mathcal{F}$ does not contain any planar graph. Even though all $\mathcal{F}$-\textsc{Minor Cover} problems have a kernel, these kernels need not be polynomial kernels. Bodlaender et al. \cite{BODLAENDER2009423} have proven that, under reasonable complexity-theoretic assumptions, many fixed-parameter tractable problems do not admit polynomial kernels.

The \textsc{Treewidth}-$\eta$ \textsc{Vertex Deletion} problem is the problem of finding a vertex deletion set such that the removal of this set of vertices results in a graph of treewidth at most $\eta$. As indicated, the \textsc{Vertex Cover} problem corresponds to the \textsc{Treewidth-0 Vertex Deletion} problem and the $\{K_2\}$-\textsc{Minor Cover} problem, the \textsc{Feedback Vertex Set} problem corresponds to the \textsc{Treewidth-1 Vertex Deletion} problem and the $\{K_3\}$-\textsc{Minor Cover} problem, and the \textsc{Treewidth-2 Vertex Deletion} problem corresponds to the $\{K_4\}$-\textsc{Minor Cover} problem. For any higher treewidth bound $\eta$ the obstruction set $\mathcal{F}$ will no longer consist of a single graph. For $\eta = 3$ this obstruction set, as proven by Arnborg et al. \cite{ARNBORG19901}, consists out of four graphs. For $\eta = 4$, Sanders \cite{10.5555/193587} gives an incomplete\footnote{This obstruction is incomplete, because it does not prove that the absence of all these minors necessarily results in a graph with treewidth at most 4.} obstruction set consisting out of 75 graphs. Ramachandramurthi \cite{doi:10.1137/S0895480195280010} has proven that the set of minimal forbidden minors consists of at least $\Omega(2^{\sqrt{\eta}})$ graphs.

The treewidth of a graph has in many cases been proven be a good measure of the intrinsic difficulty of various NP-hard graph problems \cite{Solving_connectivity_problems_parameterized_by_treewidth_in_single_exponential_time}. Several researchers observed that many graph problems on graphs of bounded treewidth can be solved efficiently using dynamic programming on tree decompositions \cite{ARNBORG198911, BERN1987216, 10.1007/3-540-19488-6_110}. For example, using the dynamic programming approach described by Niedermeier \cite{niedermeier2006invitation}, the \textsc{Vertex Cover} problem can be solved in $O(2^{\treewidth{G}} \cdot |G|)$ time, where $\treewidth{G}$ denotes the treewidth of graph $G$. As stated by Cygan et al. \cite{Solving_connectivity_problems_parameterized_by_treewidth_in_single_exponential_time}, if the problem is local\footnote{A property is local, if the property of the object to be found can be verified by separately checking the neighbourhoods of every vertex.}, these dynamic programming approaches often led to algorithms of which the running time matches known upper and lower bounds. E.g., for the \textsc{Vertex Cover} problem, unless the Strong Exponential Time Hypothesis fails, there will not exist an $(2 - \varepsilon)^{O(\treewidth{G})} \cdot |G|^{O(1)}$ algorithm for any $\varepsilon > 0$ \cite{Solving_connectivity_problems_parameterized_by_treewidth_in_single_exponential_time}.

Another application of treewidth follows from Courcelle's (meta-)theorem \cite{COURCELLE199012}, which states that every graph property definable in the monadic second-order logic of graphs\footnote{In literature sometimes referred to as the extended monadic second-order language.} can be decided in linear time on graphs of bounded treewidth. Even though these algorithms take a linear amount of time, the hidden constants can become extremely large \cite{FRICK20043}. Several techniques have been developed to address this problem \cite{Biehl1996AlgorithmsFG, mona_implementation_secrets, practical_approach_to_courcelles_theorem}, which drastically decrease the hidden constants. Nevertheless, as stated by Kneis and Langer \cite{practical_approach_to_courcelles_theorem}: ``whenever we
actually need Courcelle’s Theorem in practice because the formula is too complex to give direct algorithms, the corresponding automaton becomes large'', which results in these faster techniques still not being practically applicable. As such, Courcelle's theorem provides a powerful tool for deriving fixed parameter tractability results, however, these results seldom lead to practical applications.

An important note to make is that algorithms that efficiently solve the \textsc{Treewidth-$\eta$ Vertex Deletion} problem do not necessarily enable the usage of algorithms that are efficient on graphs of treewidth bounded by $\eta$. A solution to the \textsc{Treewidth-$\eta$ Vertex Deletion} problem is a vertex set of which the removal results in a graph of treewidth at most $\eta$. The algorithm that is run on this graph of bounded treewidth will not necessarily yield a solution that corresponds to a solution for the original graph. Hence, even though treewidth is a powerful parameter to measure difficulty of graph problems, the same does not necessarily hold for the deletion distance towards graphs of bounded treewidth. 

There exist some instances where taking the deletion distance towards forests (feedback vertex set number) as the parameter resulted in polynomial kernelizations. For example, Jansen en Bodlaender \cite{2012VCKR} derived an $O(k^3)$ kernel for the \textsc{Vertex Cover} problem parameterized on the minimum size $k$ of a feedback vertex set. Although this kernel is non-linear, opposed to the $2t$ kernels for the \textsc{Vertex Cover} problem parameterized on target size $t$ \cite{ CHEN2001280, Kernelization_Theory_of_Parameterized_Preprocessing}, in many graphs the size $t$ of a minimum vertex cover is significantly larger than the size $k$ of a minimum feedback vertex set. Hence, these kernelization techniques can often lead to smaller reduced graphs. Donkers and Jansen \cite{DONKERS2021164} have proven for the $\mathcal{F}$-\textsc{Minor Cover} problem where $\mathcal{F}$ consists only out of graphs with some components of at least three vertices, that unless $\text{NP} \subseteq \text{coNP}/\text{poly}$, there do not exist polynomial kernels when parameterized by the deletion distance to a graph of treewidth bounded by the treewidth of all graphs in $\mathcal{F}$. Hence, the deletion distance towards forests (or treewidth-2 graphs) will have limited applicability in terms of finding common approaches to solving all variants of the $\mathcal{F}$-\textsc{Minor Cover} problem.

\section{Organisation}\label{sec:intro:organisation}
In Chapter \ref{ch:prelim} we provide the basic definitions and notations used throughout this thesis. This also includes a formal problem definition and some well-known results from related works. In Chapter \ref{ch:general-reduction-rules} we introduce three rather simple reduction rules, and in all of the chapters following this chapter we will assume these three rules have been exhaustively applied. In Chapter \ref{ch:limit-1-subsets-and-cliques} we introduce the concept of limit-\textit{m} subsets, which provides a compact and consistent terminology to discuss graph properties relating to the interaction between candidate solutions and vertex sets. In this chapter we also provide an algorithm by which some limit-\textit{m} subsets can be found. In Chapter \ref{ch:graph-decompositions} we introduce our extended ``tidying the modulator framework'', which relies on the algorithm introduced in Chapter \ref{ch:limit-1-subsets-and-cliques}. We note that the application of the results stemming from this extended framework are deferred to Chapter \ref{ch:combining-results}. In Chapter \ref{ch:reducing-biconnected-subgraphs} we provide a set of reduction rules, alongside a method for finding applications of these reduction rules, on large biconnected subgraphs that are disjoint from a tidy modulator. In Chapter \ref{ch:reducing-block-cut-trees} we introduce reduction rules, and methods for applying these rules, on components that contain many maximal biconnected subgraphs. In Chapter \ref{ch:combining-results} we combine the methods from Chapters \ref{ch:graph-decompositions}, \ref{ch:reducing-biconnected-subgraphs}, and \ref{ch:reducing-block-cut-trees} to obtain our kernelization algorithm. Finally, in Chapter \ref{ch:conclusion}, we conclude by listing and discussing our main findings and we note some extensions as future work.

\chapter{Preliminaries}\label{ch:prelim}
We start by introducing the notation we will use all throughout this paper alongside references to some well-known results. Furthermore, we provide a formal definition to the \textsc{Treewidth-2 Vertex Deletion} problem, alongside some short proofs about problem-specific properties that follow from the known results.

\section{Graph Theory}\label{sec:prelim:graph-theory}
\paragraph{Graphs}
A graph $G = (V,E)$ is a set of nodes/vertices $V$ and a set of edges $E \subseteq \binom{V}{2}$. We note that this definition of a graph corresponds to that of an undirected unweighted graph without parallel edges and self-loops. We use $V(G)$ and $E(G)$ to denote the set of vertices and edges in graph $G$ respectively. For a set of edges $F$ we let $V(F)$ denote the set of endpoints of all edges in $F$. Two graphs $G$ and $H$ are vertex-disjoint when $V(G) \cap V(H) = \emptyset$. A graph $G$ contains a set of vertices $U$ when $U \subseteq V(G)$.

\paragraph{Subgraphs}
A graph $H$ is a subgraph of a graph $G$ when $V(H) \subseteq V(G) \wedge E(H) \subseteq E(G)$. A graph $G$ contains a graph $H$ when $H$ is a subgraph of $G$. For a set $F \subseteq \binom{V(G)}{2}$ let $G - F$ denote the graph with $V(G - F) = V(G) \wedge E(G - F) = E(G) \setminus F$ and let $G \cup F$ denote the graph with $V(G \cup F) = V(G) \wedge E(G \cup F) = E(G) \cup \parenthesised{F \cap \binom{V(G)}{2}}$. For an edge $e$ let $G - e = G - \{e\}$ and $G \cup e = G \cup \{e\}$.

The induced subgraph $G[U]$ of graph $G$ and vertex set $U$ is the graph with $V(G[U]) = V(G) \cap U$ and $E(G[U]) = E(G) \cap \binom{U}{2}$. We call $G[U]$ the graph induced by $U$ in $G$. For a set of vertices $U$ let $G - U = G[V(G) \setminus U]$. For a vertex $u$ let $G - u = G - \{u\}$. For a graph $H$ let $G - H = G - V(H)$.

\paragraph{Neighbourhoods}
The open neighbourhood $N_G(v)$ of a vertex $v$ is the set of vertices $u \in V(G)$ with $uv \in E(G)$. The closed neighbourhood $N_G[v]$ of a vertex $v$ is $N_G(v) \cup \{v\}$. For a set of vertices $U$ let $N_G(U) = \bigcup_{u \in U} N_G(u) \setminus U$ and $N_G[U] = N_G(U) \cup U$. The boundary $\partial_G(U)$ of a vertex set $U$ is $N_G(G - U)$. For a graph $H$ let $N_G(H) = N_G(V(H))$ and $N_G[H] = N_G[V(H)]$ and $\partial_G(H) = \partial_G(V(H))$. The degree of a vertex $v$ equals $|N_G(v)|$. The minimum/maximum degree of a graph $G$ is the minimum/maximum degree of any vertex in $G$.

\paragraph{Paths}
A path $P$ is a non-empty graph $(V,E)$ of the form $V(P) = \langle u_1, u_2, \dots, u_m \rangle$ and $E(P) = \langle u_1u_2, u_2u_3, \dots, u_{m-1}u_m \rangle$. In this definition we consider $V(P)$ and $E(P)$ to be sequences that might contain duplicates. We often refer to a path only by its sequence of vertices $P = V(P) = \langle u_1, u_2, \dots, u_m \rangle$, because it uniquely defines $E(P)$. A path $P$ contained in a path $Q$ is a subpath. A path $P = \langle u_1, u_2, \dots, u_m \rangle$ is a path between $u_1$ and $u_m$, which are called the endpoints of $P$. Similarly, we say that $P$ is a path from $u_1$ to $u_m$ and a path from $u_m$ to $u_1$. I.e., we do not associate a direction to paths. Vertices $u$ and $v$ are connected in $G$ when there exists a path in $G$ between $u$ and $v$. A path $P$ visits a set of vertices $U$ when $U \subseteq V(P)$. A path $P$ is simple if each vertex in sequence $V(P)$ is distinct. For a simple path $P$ and vertices $u,v \in V(P)$ we define subpath $P[u,v]$ as the path defined by the subsequence of $P$ from $u$ to $v$. Furthermore, we let $P[u,v) = P[u,v] - \{v\}$, we let $P(u,v] = P[u,v] - \{u\}$, and we let $P(u,v) = P[u,v] - \{u,v\}$. A simple path $P$ is an induced path in $G$ if it is an induced subgraph of $G$. Two paths $P = \langle p_1, p_2, \dots, p_m \rangle$ and $Q = \langle q_1, q_2, \dots, q_n \rangle$ are internally vertex-disjoint when $P - \{p_1, p_m\}$ and $Q - \{q_1, q_n\}$ are vertex-disjoint.

\paragraph{Cycles}
A cycle $C$ is a non-empty graph $(V,E)$ of the form $V(C) = \langle u_1, u_2, \dots u_m \rangle$ and $E(C) = \langle u_1u_2, u_2u_3, \dots, u_{m-1}u_m, u_mu_1 \rangle$ with $m \geq 3$. Similar to the definition of paths, $V(C)$ and $E(C)$ can contain duplicates. A cycle $C$ visits a set of vertices $U$ when $U \subseteq V(C)$. A cycle $C$ is simple if each vertex in sequence $V(C)$ is distinct. A graph that contains no simple cycles is called acyclic. A simple cycle $C$ is an induced cycle in $G$ if it is an induced subgraph of $G$. 

\paragraph{Separators}
A vertex set $S$ is a separator for vertex sets $X$ and $Y$ with $X \cap Y = X \cap S = Y \cap S = \emptyset$ in a graph $G$ if there exist no paths in $G - S$ between any vertex $x \in X$ and any vertex $y \in Y$. We say that $S$ separates $X$ from $Y$ in $G$ if $S$ is a separator for $X$ and $Y$ in $G$. Vertex set $S$ is a separator for vertices $x$ and $y$ in $G$ if it is a separator for $\{x\}$ and $\{y\}$ in $G$.

From Menger's theorem (Theorem 4.2.17 in \cite{Introduction-to-Graph-Theory}) it follows that each $\{x,y\} \subseteq V(G)$ with $xy \notin E(G)$ has a minimum-cardinality separator for $x$ and $y$ in $G$ of size equal to the maximum number of pairwise internally vertex-disjoint paths between $x$ and $y$ in $G$. A maximum set of internally vertex-disjoint paths between $x$ and $y$ can be found in $\polyG$ time using a maximum flow algorithm (Remark 4.3.15 in \cite{Introduction-to-Graph-Theory}). From these paths a minimum-cardinality separator can be directly obtained (Remark 4.3.13 and 4.3.15 in \cite{Introduction-to-Graph-Theory}). Hence, we make the following observations:

\begin{observation}\label{obs:get-max-internally-vertex-disjoint-paths}
    For any $\{x,y\} \subseteq V(G)$ with $xy \notin E(G)$ we can in $\polyG$ time obtain a maximum set of internally vertex-disjoint paths between $x$ and $y$ in $G$.
\end{observation}

\begin{observation}\label{obs:get-min-separator}
    For any $\{x,y\} \subseteq V(G)$ with $xy \notin E(G)$ we can in $\polyG$ time obtain a minimum-cardinality separator for $x$ and $y$ in $G$.
\end{observation}

\paragraph{Connectivity}
A graph $G$ is connected when every pair of vertices $u,v \in V(G)$ are connected. A component $C$ is a maximal connected subgraph of $G$. Set $\mathcal{C}(G)$ is a partitioning of $G$ into components. We note that set $\mathcal{C}(G)$ can be constructed in $\polyG$ time (e.g. by using DFS).

A vertex $u \in V(G)$ is an articulation vertex (also called a cut vertex) when $|\mathcal{C}(G - u)| > |\mathcal{C}(G)|$. A graph $G$ is biconnected when it does not contain articulation vertices. A biconnected component (also called a block) is a maximal biconnected subgraph of $G$.

\begin{lemma}[Theorem 4.2.4 in \cite{Introduction-to-Graph-Theory}]\label{lemma:biconnected-has-cycle}
    For every $\{u,v\} \subseteq V(B)$ in a biconnected graph $B$ with $|V(B)| \geq 3$ there exists a simple cycle in $B$ that visits $u$ and $v$.
\end{lemma}

\paragraph{Trees}
A tree $T$ is a connected acyclic graph. A forest is a graph in which each component is a tree. A tree/forest is linear when it has maximum degree 2. A vertex $\ell \in V(T)$ is a leaf when it has degree at most one. A vertex $t \in V(T)$ is called a branching vertex when it has degree at least three. A tree $T'$ contained in a tree $T$ is a subtree.

A rooted tree $T$ is a tree in which one vertex $r \in V(T)$ has been designated as the root. A vertex $u \in V(T)$ is a descendant of a vertex $v \in V(T)$ in a tree $T$ with root $r$ if $v$ is contained on the unique simple path from $u$ to $r$. A vertex $u \in N_T(v)$ is a child of $v \in V(T)$ if $u$ is a descendant of $v$. The subtree rooted at $r' \in V(T)$ of a tree rooted at $r$ is the tree induced by vertex $r'$ and all descendants of $r'$. 

A spanning tree $T$ of a graph $G$ is a tree with $V(G) = V(T)$. We note that a spanning tree of a connected graph $G$ can be found in $\polyG$ time (e.g. using depth-first search or breath-first search).

\begin{observation}\label{obs:get-spanning-tree}
    For any connected graph $G$ we can in $\polyG$ time obtain a spanning tree.
\end{observation}

\paragraph{Matching}
A matching $M$ in a graph $G$ is a subset of edges from $E(G)$ with no shared endpoints. A maximal matching can be found in a graph $G$ in $\polyG$ time (e.g. by greedily adding edges until the formed set is maximal).

\begin{observation}\label{obs:get-maximal-matching}
    For any graph $G$ we can in $\polyG$ time obtain a maximal matching.
\end{observation}

\paragraph{Other}
A graph $K$ is a clique if $E(K) = \binom{V(K)}{2}$. We let $K_m$ for $m \in \nat$ denote a clique graph with $|V(K_m)| = m$. A bipartite graph $B$ is a graph with two disjoint vertex sets $L \cup R = V(B)$ and $E(B) \subseteq L \times R$. A ladder graph $L_n$ is a graph with $V(L_n) = \{p_1, p_2, \dots, p_n, q_1, q_2, \dots, q_n\}$ and $E(L_n) = \bigcup_{1 \leq i < n} \{p_ip_{i+1}, q_iq_{i+1}\} \cup \bigcup_{1 \leq i \leq n} \{p_iq_i\}$. I.e., $L_n$ can be seen as a $2 \times n$ grid.

A partitioning $\mathcal{U}$ of a graph $G$ is a set of pairwise disjoint vertex sets, such that $\bigcup_{U \in \mathcal{U}} U = V(G)$. Furthermore, we let a set $\mathcal{D}$ consisting of subgraphs of $G$ be a partitioning of $G$ when $\setdef{V(D)}{D \in \mathcal{D}}$ is a partitioning of $G$.

Two graphs $G$ and $H$ are isomorphic when there exists a bijection $f : V(G) \rightarrow V(H)$ such that for all $\{u,v\} \subseteq V(G)$ we have $uv \in E(G) \Longleftrightarrow f(u)f(v) \in E(H)$.

\section{Tree Decompositions and Treewidth}\label{sec:prelim:tree-decomposition-and-treewidth}
We introduce a notion of treewidth, which is a well-known graph parameter that measures the similarity between a graph and a tree. A relatively intuitive explanation behind treewidth is given by Fomin et al. \cite{Kernelization_Theory_of_Parameterized_Preprocessing}: Suppose we want to ``draw'' graph $G$ inside a tree $T$. If $G$ is not a tree then the drawing has to self-intersect. The treewidth parameter captures how much such a drawing will need to intersect itself.

We will first introduce the notion of a tree decomposition (based on the definition by Bodlaender et al. \cite{Preprocessing_for_Treewidth_A_Combinatorial_Analysis_through_Kernelization}), which is necessary to formally define the treewidth parameter.

\begin{definition}[tree decomposition]\label{def:tree-decomposition}
    A tree decomposition of a graph $G$ is a pair $(T,\chi)$, where $T$ is a tree and $\chi : V(T) \rightarrow 2^{V(G)}$ assigns to every node of $T$ a subset of $V(G)$ called a bag, such that:
    \begin{enumerate}
        \item $\bigcup_{t \in V(T)} \chi(t) = V(G)$
        \item for each $uv \in E(G)$ there exists some $t \in V(T)$ with $\{u,v\} \subseteq \chi(t)$
        \item for each $v \in V(G)$ nodes $\setdef{t \in V(T)}{v \in \chi(t)}$ induce a connected subtree of $T$
    \end{enumerate}
\end{definition}
\bigskip

\noindent
We show an example tree decomposition in Figure \ref{fig:example-tree-decomposition} (example by Chatterjee et al. \cite{Chatterjee2014OptimalTB}). The nodes contained in each bag $\chi(t)$ are drawn within their corresponding node $t$.

\begin{figure}[ht]
\centering
\begin{subfigure}{.5\textwidth}
  \centering
  \includegraphics[page=2, scale=.6]{example-tree-decomposition.pdf}
  \caption{Graph $G$}
  \label{fig:example-tree-decomposition:graph}
\end{subfigure}%
\begin{subfigure}{.5\textwidth}
  \centering
  \includegraphics[page=1, scale=.6]{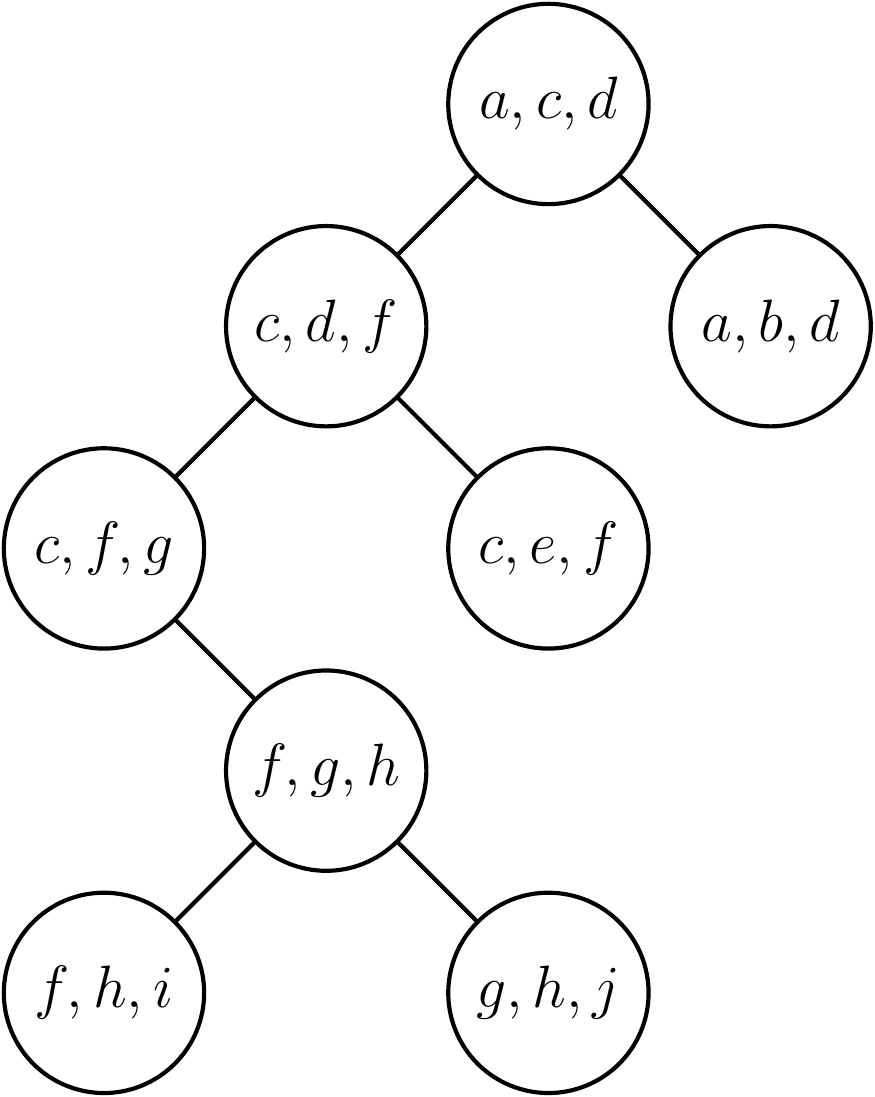}
  \caption{A tree decomposition $(T,\chi)$ of $G$}
  \label{fig:example-tree-decomposition:tree}
\end{subfigure}%
\caption{An example tree decomposition of width 2}
\label{fig:example-tree-decomposition}
\end{figure}

We also define a model function $\bchi : V(G) \rightarrow 2^{V(T)} \setminus \emptyset$ for tree decompositions $(T,\chi)$. A model function $\bchi$ is, in essence, the inverse of function $\chi : V(T) \rightarrow 2^{V(G)}$, although it maps individual vertices from $V(G)$ opposed to subsets of $V(G)$. We define $\bchi(v) = \setdef{t \in V(T)}{v \in \chi(t)}$. This definition yields that each $v \in V(G)$ and $t \in V(T)$ has $v \in \chi(t) \Longleftrightarrow t \in \bchi(v)$.

For any set $U \subseteq V(T)$ we let $\chi(U) = \bigcup_{t \in U} \chi(t)$ and for any subgraph $F$ of $T$ we let $\chi(F) = \chi(V(F))$. Similarly, for any set $U \subseteq V(G)$ we let $\bchi(U) = \bigcup_{v \in U} \bchi(v)$ and for any subgraph $H$ of $G$ we let $\bchi(H) = \bchi(V(H))$.

\noindent
The width of a tree decomposition $(T,\chi)$ equals $\max_{t \in V(T)} |\chi(t)| - 1$. A width-$\eta$ tree decomposition is a tree decomposition of width at most $\eta$. The treewidth $\treewidth{G}$ of a graph $G$ equals the minimum width of a tree decomposition of $G$.

Determining the treewidth of a graph $G$ is NP-complete \cite{Complexity-of-Finding-Embeddings-in-a-k-Tree}. Nevertheless, throughout this thesis we will only be interested in obtaining width-2 tree decompositions (assuming these exist). Hence, the correctness of the following lemma will suffice for our purpose.

\begin{lemma}\label{lemma:find-tree-decompositions}
	For any graph $G$ we can in $\polyG$ time obtain a width-2 tree decomposition of $G$ or derive $\treewidth{G} > 2$.
\end{lemma}

\begin{proof}
	As follows from the works by Bodlaender \cite{Bodlaender_linear-time_algorithm_for_finding_tree-decomposition_of_small_treewidth}, if $\eta$ is a constant then there is a linear-time algorithm that given a graph $G$ outputs a width-$\eta$ tree decomposition or correctly concludes that $\treewidth{G} > \eta$. For $\eta = 2$ we can therefore in $\polyG$ time obtain a width-2 tree decomposition of $G$ or derive $\treewidth{G} > 2$.
\end{proof}

\noindent
It should, however, be noted that Bodlaender \cite{Bodlaender_linear-time_algorithm_for_finding_tree-decomposition_of_small_treewidth} states that the constant factor of this algorithm is too large for practical purposes, and no practically feasible algorithm is known \cite{Bodlaender2012TheMA}.

Throughout this thesis we will frequently make use of tree decompositions. We list some known properties about tree decompositions and derive some properties that frequently reoccur in proofs within later sections.

\begin{lemma}[Lemma 14.10 in \cite{Kernelization_Theory_of_Parameterized_Preprocessing}]
    Let $(T,\chi)$ be a tree decomposition of $G$ and let $H$ be a connected subgraph of $G$. Then $T[\bchi(H)]$ is connected.
\end{lemma}

\begin{corollary}[Corollary 14.11 in \cite{Kernelization_Theory_of_Parameterized_Preprocessing}]\label{corollary:treedecomp-node-separator}
    Let $(T,\chi)$ be a tree decomposition of $G$, let $t \in V(T)$, and let $C \in \mathcal{C}(G - \chi(t))$. There exists some component $T' \in \mathcal{C}(T - t)$ with $\bchi(C) \subseteq V(T')$.
\end{corollary}

\begin{corollary}[Lemma 12.3.1 in \cite{diestel_graph_theory}]\label{corollary:treedecomp-edge-separator}
    Let $(T,\chi)$ be a tree decomposition of $G$, let $ab \in E(T)$, and let $\{T_a, T_b\} = \mathcal{C}(T - ab)$ be such that $a \in V(T_a) \wedge b \in V(T_b)$. Then $\chi(a) \cap \chi(b)$ separates $\chi(T_a) \setminus \chi(a)$ from $\chi(T_b) \setminus \chi(b)$ in $G$.
\end{corollary}

\begin{lemma}\label{lemma:bchi-connected-in-subtree}
    Let $G$ be a graph, let $(T,\chi)$ be a tree decomposition of $G$, let $T'$ be a subtree of $T$, and let $v \in V(G)$. Then $T'[\bchi(v)]$ is connected. 
\end{lemma}

\begin{proof}
    Let $a,b \in V(T') \cap \bchi(v)$. Between any two nodes in a tree there exists exactly one simple path (Theorem 2.1.4 in \cite{Introduction-to-Graph-Theory}). We let $P$ be the path that connects $a$ and $b$ in $T'$. Because $T'$ is a subgraph of $T$ we have that $P$ is the only path in $T$ that connects $a$ and $b$. By Definition \ref{def:tree-decomposition} we know that $a$ and $b$ are connected in $T[\bchi(v)]$, which therefore implies $V(P) \subseteq \bchi(v)$. Hence, $P$ is a path that connects $a$ and $b$ in $T'[\bchi(v)]$, proving that $T'[\bchi(v)]$ is connected.
\end{proof}

\begin{lemma}[Proposition 3 in \cite{Preprocessing_for_Treewidth_A_Combinatorial_Analysis_through_Kernelization}]\label{lemma:clique-shares-bag}
    If $K$ is a clique subgraph of $G$, then any tree decomposition of $G$ has a bag that contains all vertices $V(K)$.
\end{lemma}

\begin{lemma}[Lemma 1 in \cite{Preprocessing_for_Treewidth_A_Combinatorial_Analysis_through_Kernelization}]\label{lemma:connect-treewidth-components-around-clique}
    Let $K$ be a clique subgraph in graph $G$. Then $\treewidth{G} = \max_{C \in \mathcal{C}(G - K)} \treewidth{G[V(C) \cup V(K)]}$.
\end{lemma}

\noindent
Because we will generalise some consequences of Lemma \ref{lemma:connect-treewidth-components-around-clique} we will first provide an informal proof why Lemma \ref{lemma:connect-treewidth-components-around-clique} is correct. Consider a clique subgraph $K$ of $G$ such that each $C \in \mathcal{C}(G - K)$ has $\treewidth{G[V(C) \cup V(K)]} \leq \eta$. For each component $C \in \mathcal{C}(G - K)$ there exists a width-$\eta$ tree decomposition for $G[V(C) \cup V(K)]$. From Lemma \ref{lemma:clique-shares-bag} it follows each of these tree decompositions contain a bag with all vertices $V(K)$. We can create a new tree decomposition with a bag $V(K)$ that connects all tree decompositions $G[V(C) \cup V(K)]$. The tree decomposition formed by this construction then is a tree decomposition of width $\eta$.

An example of this construction is shown in Figure \ref{fig:connect-treewidth-components-around-clique} where width-2 tree decompositions of $G[\{a,b,c,d,f\}]$, $G[\{c,e,f\}]$, and $G[\{c,f,g,h,i,j\}]$ can be connected to a bag $V(K)$ to form a width-2 tree decomposition of $G$.

\bigskip
\begin{figure}[ht]
\centering
\begin{subfigure}{.29\textwidth}
  \centering
  \includegraphics[page=4, scale=.6]{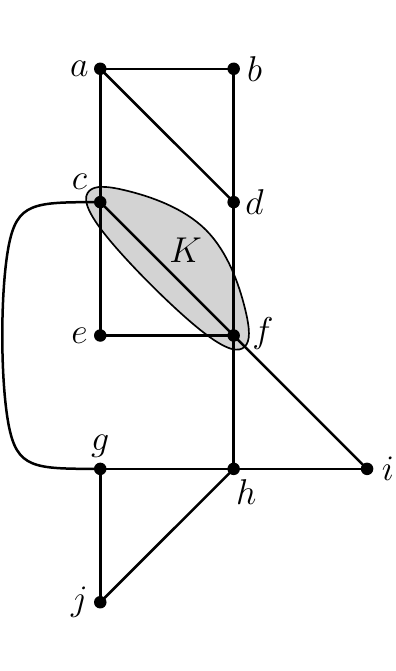}
  \caption{Graph $G$ with clique $K$}
  \label{fig:connect-treewidth-components-around-clique:graph}
\end{subfigure}%
\hspace{.05\textwidth}
\begin{subfigure}{.62\textwidth}
  \centering
  \includegraphics[page=3, scale=.6]{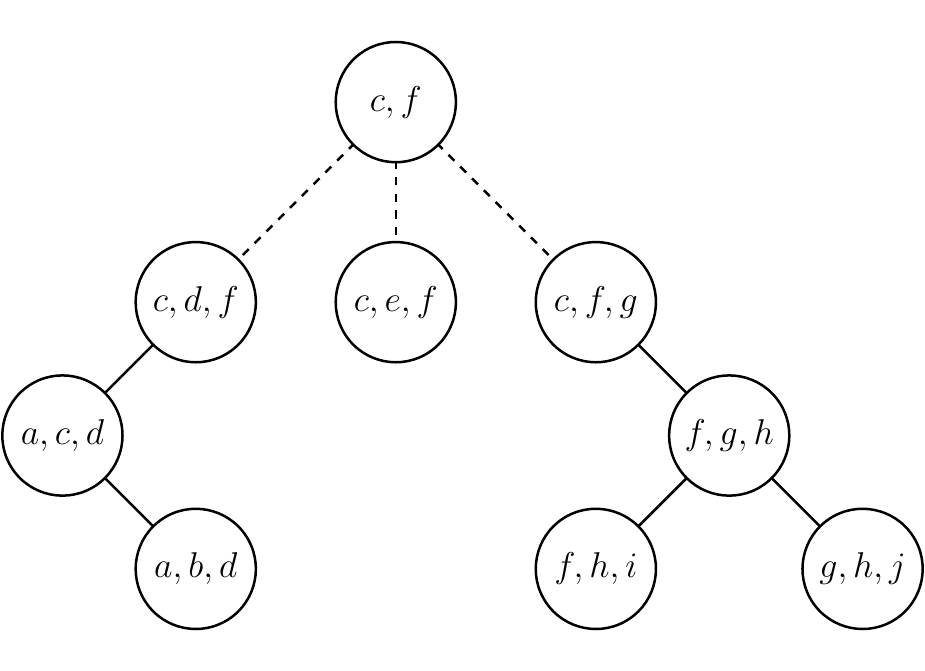}
  \caption{Tree decomposition formed by combining width-2 tree decompositions of $G[\{a,b,c,d,f\}]$, $G[\{c,e,f\}]$, and $G[\{c,f,g,h,i,j\}]$}
  \label{fig:connect-treewidth-components-around-clique:tree}
\end{subfigure}%
\caption{Example application of Lemma \ref{lemma:connect-treewidth-components-around-clique} with $\treewidth{G} \leq 2$}
\label{fig:connect-treewidth-components-around-clique}
\end{figure}
\bigskip

\noindent
We will often want to apply Lemma \ref{lemma:connect-treewidth-components-around-clique} to show that two (or more) induced subgraphs $H_1$ and $H_2$ of $G$ can be combined into a larger induced subgraph in $G$ of given treewidth $\eta$. The condition that in this case needs to be satisfied is that $G[V(H_1) \cap V(H_2)]$ is a clique, that $V(H_1) \cap V(H_2)$ is a separator for $V(H_1) \setminus V(H_2)$ and $V(H_2) \setminus V(H_1)$ in $G[V(H_1) \cup V(H_2)]$, and that $\treewidth{H_1} \leq 2 \wedge \treewidth{H_2} \leq 2$. This, because when these conditions are satisfied, we have $\mathcal{C}(G - (V(H_1) \cap V(H_2))) = \{H_1 - H_2, H_2 - H_1\}$, in which case $\treewidth{G[V(H_1) \cup V(H_2)]}$ follows from \ref{lemma:connect-treewidth-components-around-clique}. This lemma, however, is rather cumbersome to apply, as it requires proving these properties many times. For this reason, we prove the following (easier to apply) theorem:

\begin{theorem}\label{theorem:connect-treewidth-graphs}
    Let $G$ be a graph and let $U \subseteq V(G)$ with $\treewidth{G[U]} \leq \eta$. If each $C \in \mathcal{C}(G - U)$ has $\treewidth{G[N_G[C]]} \leq \eta$ and $G[N_G(C)]$ is a clique, then $\treewidth{G} \leq \eta$.
\end{theorem}

\begin{proof}
    We prove using induction on $V(G) \setminus U$.
    
    \begin{basecase}[$U = V(G)$]
        Trivial.
    \end{basecase}
    
    \begin{inductivestep}[$U \subset V(G)$]
        Let $C \in \mathcal{C}(G - U)$ and $U' = U \cup V(C)$. Because $G[N_G(C)]$ is a clique Lemma \ref{lemma:connect-treewidth-components-around-clique} yields $\treewidth{G[U']} = \max_{D \in \mathcal{C}(G[U'] - N_G(C))} \treewidth{G[V(D) \cup N_G(C)]}$. We have that $N_G(C)$ separates $V(C)$ from $U \setminus N_G(C)$ in $G[U']$. Hence, each component $D \in \mathcal{C}(G[U'] - N_G(C))$ is a subgraph of either $C$ or $G[U]$. If $D$ is a subgraph of $C$ we derive, by Theorem 14.15 in \cite{Kernelization_Theory_of_Parameterized_Preprocessing}, that $\treewidth{G[V(D) \cup N_G(C)]} \leq \treewidth{G[N_G[C]]} \leq \eta$. If $D$ is a subgraph of $G[U]$ we derive, by Theorem 14.15 in \cite{Kernelization_Theory_of_Parameterized_Preprocessing}, that $\treewidth{G[V(D) \cup N_G(C)]} \leq \treewidth{G[U]} \leq \eta$. Hence, we have $\treewidth{G[U']} = \max_{D \in \mathcal{C}(G[U'] - N_G(C))} \treewidth{G[V(D) \cup N_G(C)]} \leq \eta$.
    
        We have $\mathcal{C}(G - U') = \mathcal{C}(G - U) \setminus \{C\}$, which implies that each $D \in \mathcal{C}(G - U')$ has $\treewidth{G[N_G[D]]} \leq \eta$ and $G[N_G(D)]$ is a clique. Hence, from the induction hypothesis it follows that $\treewidth{G} \leq \eta$.
    \end{inductivestep}
    
    \noindent
    Because these cases cover all cases, correctness trivially follows.
\end{proof}

\noindent
We will frequently apply Theorem \ref{theorem:connect-treewidth-graphs}. Hence, for brevity, we say that we can apply Theorem \ref{theorem:connect-treewidth-graphs} on $(G,U)$ to derive $\treewidth{G} \leq \eta$ when Theorem \ref{theorem:connect-treewidth-graphs} on graph $G$ and vertex set $U$ allows us to derive $\treewidth{G} \leq \eta$. Similarly, for a graph $G$ and a subgraph $H$ of $G$ we say that we can apply Theorem \ref{theorem:connect-treewidth-graphs} on $(G,H)$ to derive $\treewidth{G} \leq \eta$ when we can apply Theorem \ref{theorem:connect-treewidth-graphs} on $(G, V(H))$ to derive $\treewidth{G} \leq \eta$.

\section{Treewidth-2 Vertex Deletion}\label{sec:treewidth-2 vertex deletion}
We will next formally define the problem central to our thesis:

\bigskip
\indent\textsc{Treewidth-2 Vertex Deletion}\\
\indent\textbf{input:}~a graph $G$ and an integer $t \in \nat$\\
\indent\textbf{question:}~does there exist a set $S \subseteq V(G)$ with $|S| \leq t$ and $\treewidth{G - S} \leq 2$?
\bigskip

\noindent
We provide some additional definitions to more concisely refer to concepts present within the definition of the \textsc{Treewidth-2 Vertex Deletion} problem.

\begin{definition}
    A problem instance is a pair $(G,t)$ with $G$ being a graph and $t \in \nat$.
\end{definition}

\begin{definition}
    A modulator $X$ of graph $G$ is a vertex set $X$ such that $\treewidth{G - X} \leq 2$.
\end{definition}

\begin{definition}
    $\textrm{\textsc{TW2D}}\parenthesised{G}$ denotes the size of a minimum modulator of $G$.
\end{definition}

\noindent
We note that these definitions provide an alternative formulation to the \textsc{Treewidth-2 Vertex Deletion} problem: given a problem instance $(G,t)$, does $\twtwodeletion{G}{t}$ hold? We will also introduce the notion of a solution to a problem instance, which essentially is set $S$ that proves $\twtwodeletion{G}{t}$.

\begin{definition}
    A solution $S$ to a problem instance $(G,t)$ is a vertex set $S$ such that $S$ is a modulator of $G$ with $|S| \leq t$.
\end{definition}

\section{Minors}\label{sec:prelim:minors}
To better understand the \textsc{Treewidth-2 Vertex Deletion} problem we will extensively make use of graph minors. Before defining graph minors, we first define the edge-contraction operation.

An edge-contraction of an edge $uv$ in $G$ is the operation of removing vertices $u$ and $v$ from $G$ and introducing a new vertex that is adjacent to $N_G(\{u,v\})$. We write $G \contract e$ for an edge $e \in E(G)$ to denote the graph resulting from the contraction of edge $e$ in $G$. For a non-empty connected subgraph $H$ of $G$ we let $G \contract H$ denote the graph resulting from contracting $H$ in $G$, which corresponds to replacing $H$ by a single vertex that is adjacent to $N_G(H)$. We note that graph $G \contract H$ can be obtained via a sequence of edge contractions. For a connected subgraph $H$ of $G$ and a vertex $u \in V(H)$ we define the operation of contracting $H$ into $u$ in $G$ as contracting $H$ in $G$ and having vertex $u$ correspond to the vertex resulting from this contraction.

A graph $H$ is a minor of $G$ when it can be formed from $G$ by using a (possibly empty) sequence of vertex-deletions, edge-deletions, and edge-contractions. The notion of graph minors will be essential in providing a kernel for the \textsc{Treewidth-2 Vertex Deletion} problem, in part due to the following properties:

\begin{lemma}[lemma 2 in \cite{A_single-exponential_FPT_algorithm_for_the_K4-minor_cover_problem}]\label{lemma:k_4_equals_treewidth_2}
    Graph $G$ contains a $K_4$ as a minor if and only if $\treewidth{G} > 2$.
\end{lemma}

\begin{lemma}[Theorem 14.15 in \cite{Kernelization_Theory_of_Parameterized_Preprocessing}]\label{lemma:minor-treewidth}
    Let $H$ be a minor of $G$. Then $\treewidth{H} \leq \treewidth{G}$.
\end{lemma}

\begin{lemma}\label{lemma:minor-of-G-implication}
	Let $H$ be a minor of $G$. Then $\twtwodeletionG{H} \leq \twtwodeletionG{G}$.
\end{lemma}

\begin{proof}
    By Corollary 1.7.2 in \cite{diestel_graph_theory} we have that $H$ can be constructed from $G$ by performing a sequence of vertex-deletions, edge-deletions, and edge-contractions. Let $\langle G_0, G_1, \dots, G_m \rangle$ be a sequence of graphs with $G_0 = G$ and $G_m = H$, such that for any $0 < i \leq m$ it holds that $G_i$ can be obtained by performing a single vertex-deletion, edge-deletion, or edge-contraction on $G_{i-1}$. We prove by induction on $i$ that $\twtwodeletionG{G_i} \leq \twtwodeletionG{G}$.
    
    \begin{basecase}[$i = 0$]
        Trivial because $G_0 = G$.
    \end{basecase}

    \begin{inductivestep}[$i > 0$]
        Let $X_{i-1}$ be a minimum modulator of $G_{i-1}$. By the induction hypothesis we have $|X_{i-1}| = \twtwodeletionG{G_{i-1}} \leq \twtwodeletionG{G}$. We prove by case distinction on the operation performed on $G_{i-1}$ that $\twtwodeletionG{G_i} \leq \twtwodeletionG{G}$.
        
        \begin{indsubcase}[$G_i = G_{i-1} - v \vee G_i = G_{i-1} - uv$]
            Because $G_i - X_{i-1}$ is a minor of $G_{i-1} - X_{i-1}$ it follows from Lemma \ref{lemma:minor-treewidth} that $X_{i-1}$ is a modulator for $G_i$. Hence, we have $\twtwodeletionG{G_i} \leq |X_{i-1}| \leq \twtwodeletionG{G}$.
        \end{indsubcase}
    
        \begin{indsubcase}[$G_i = G_{i-1} \contract uv$]
            Let $w \in V(G_i)$ be the vertex resulting from the contraction $uv$ in $G_{i-1}$. If $X_{i-1} \cap \{u,v\} = \emptyset$ let $X_i = X_{i-1}$ and otherwise let $X_i = (X_{i-1} \setminus \{u,v\}) \cup \{w\}$. Regardless of whether $X_{i-1} \cap \{u,v\} = \emptyset$ holds, graph $G_i - X_i$ is a minor of $G_{i-1} - X_{i-1}$. It follows from Lemma \ref{lemma:minor-treewidth} that $X_i$ is a modulator for $G_i$, which implies $\twtwodeletionG{G_i} \leq |X_i| \leq |X_{i-1}| \leq \twtwodeletionG{G}$.
        \end{indsubcase}
    \end{inductivestep}
    
    \noindent
    Because these cases cover all cases, we conclude $\twtwodeletionG{H} = \twtwodeletionG{G_m} \leq \twtwodeletionG{G}$.
\end{proof}

\noindent
Throughout this thesis we will often identify $K_4$ minors in graphs $G$ to prove that $\treewidth{G} > 2$. To identify these $K_4$ minors we find four vertex-disjoint connected subgraphs $(H_1,H_2,H_3,H_4)$ in $G$, such that any pair of subgraphs $H_i$ and $H_j$ has an edge between a vertex from $H_i$ and a vertex from $H_j$. For brevity we say that such a tuple $(H_1,H_2,H_3,H_4)$ describes a $K_4$ minor in $G$.

\begin{lemma}\label{lemma:describes-k4-minor}
    Let $(H_1,H_2,H_3,H_4)$ be a tuple of vertex-disjoint connected subgraphs of graph $G$ such that for each $\{H_i,H_j\} \subseteq (H_1,H_2,H_3,H_4)$ we have $N_G(H_i) \cap V(H_j) \neq \emptyset$. Then $\treewidth{G} > 2$.
\end{lemma}

\begin{proof}
    Directly follows Lemma \ref{lemma:k_4_equals_treewidth_2} and Remark 6.2.13 in \cite{Introduction-to-Graph-Theory}.
\end{proof}

\noindent
We say that we can apply Lemma \ref{lemma:describes-k4-minor} on a tuple $(H_1,H_2,H_3,H_4)$ to derive $\treewidth{G} > 2$ when this tuple together with $G$ satisfies the precondition of Lemma \ref{lemma:describes-k4-minor}. For brevity we will also say that we can apply Lemma \ref{lemma:describes-k4-minor} on a tuple $(X_1,X_2,X_3,X_4)$ where we allow each $X_i$ to be a vertex set $U_i \subseteq V(G)$ in which case $H_i = G[U_i]$ is used. Similarly, we also allow $X_i$ to be a single vertex $u_i \in V(G)$ in which case $H_i = G[\{u_i\}]$ is used. E.g. we can apply Lemma \ref{lemma:describes-k4-minor} on $(H, U, x, y)$ to derive $\treewidth{G} > 2$ when $H$ is a subgraph of $G$, $U \subseteq V(G)$, and $\{x,y\} \subseteq V(G)$; and $(H, G[U], G[\{x\}], G[\{y\}])$ together with $G$ satisfies the precondition of Lemma \ref{lemma:describes-k4-minor}.

\section{Parameterized Algorithms}\label{sec:prelim:parameterized-algorithms}
We introduce the formal definitions used in the field of parameterized algorithmics.

\begin{definition}[parameterized problem]
    A parameterized problem is a language $L \subseteq \Sigma^* \times \nat$, where $\Sigma$ is a fixed, finite alphabet. For an instance $(x,k) \in \Sigma^* \times \nat$, value $k \in \nat$ is called the parameter.
\end{definition}

\noindent
We can define the \textsc{Treewidth-2 Vertex Deletion} problem as a parameterized problem. We let $\Sigma^* = \mathcal{G}$ corresponds to the family of all undirected unweighted graphs without parallel edges and self-loops. We note that this is possible by encoding any graph as a string consisting of characters from $\Sigma$. We define language $L \subseteq \mathcal{G} \times \nat$ as the set of all problem instances $(G,t)$ with $\twtwodeletion{G}{t}$. Using this parameterization, we have that solving the \textsc{Treewidth-2 Vertex Deletion} problem corresponds to determining for a problem instance $(G,t)$ whether $(G,t) \in L$ holds.

A tractable problem can be solved using an algorithm that takes time polynomial in the size of the input. An intractable problem is a problem that can not be solved by such an algorithm. Because \textsc{Treewidth-2 Vertex Deletion} is NP-hard \cite{The_node_deletion_problem_for_hereditary_properties_is_NP_complete} we expect this problem to be intractable. The parameterized definition of this problem, however, also allows the efficiency of an algorithm to be measured in terms of the additional parameter $t$. An important notion of tractability in the field of parameterized algorithmics is that of a fixed-parameter tractability.

\begin{definition}[fixed-parameter tractable]
    A parameterized problem $L \subseteq \Sigma^* \times \nat$ is called fixed-parameter tractable (FPT) if there exists an algorithm $\mathcal{A}$ and a computable function $f : \nat \rightarrow \nat$, such that given $(x, k) \in \Sigma^* \times \nat$, algorithm $\mathcal{A}$ correctly decides whether $(x, k) \in L$ in time bounded by $f(k) \cdot \poly{|x| + k}$.
\end{definition}

\noindent
The main strength of a fixed-parameter algorithm is that they are able to `efficiently' solve intractable problems when the problem instances at hand have a small parameter value. I.e., for any parameter $k = O(1)$, such an algorithm $\mathcal{A}$ takes $f(O(1)) \cdot \poly{|x| + O(1)} = \poly{|x|}$ time. In other words, when the value of parameter $k$ is fixed to some constant $O(1)$, then a fixed-parameter tractable algorithm is tractable (takes time polynomial in the size of the input). In our case, we will provide a method to obtain a fixed-parameter tractable algorithm that determines whether $\twtwodeletion{G}{t}$.

One commonplace tool for designing FPT algorithms is that of kernelization. A kernelization algorithm is, in essence, a preprocessing algorithm that in polynomial time reduces a problem instance into its computationally difficult to solve `core'. A kernelization algorithm does not necessarily solve the problem, but rather shrink any problem instance into a smaller problem instance that then can be solved using exact algorithms. Before formally defining a kernelization algorithm, we first introduce its main building block.

\begin{definition}[reduction rule]
    A reduction rule for a parameterized problem $L \subseteq \Sigma^* \times \nat$ is a function $\phi : \Sigma^* \times \nat \rightarrow \Sigma^* \times \nat$ that maps instances $(x,k) \in \Sigma^* \times \nat$ into instances $(x',k') \in \Sigma^* \times \nat$ such that $(x,k) \in L \Longleftrightarrow (x',k') \in L$ and for which $\phi$ is computable in $\poly{|x| + k}$ time.
\end{definition}

\noindent
The property that $(x,k) \in L \Longleftrightarrow (x',k') \in L$ holds when a reduction rule is applied is referred to as the safeness of the reduction rule. In our case, this corresponds to proving that $\twtwodeletion{G}{t} \Longleftrightarrow \twtwodeletion{G'}{t'}$ holds.

A kernelization algorithm is a polynomial-time algorithm that repeatedly applies safe reduction rules on a problem instance until the resulting problem instance has a size bounded by a function over the parameter $k \in \nat$.

\begin{definition}[kernelization]\label{def:kernelization}
    A kernelization algorithm for a parameterzied problem $L \subseteq \Sigma^* \times \nat$ is an algorithm that, given an instance $(x,k) \in \Sigma^* \times \nat$, in $\poly{|x| + k}$ time returns a problem instance $(x',k') \in \Sigma^* \times \nat$ such that $(x,k) \in L \Longleftrightarrow (x',k') \in L$ and $|x'| + k' \leq g(k)$ for some computable function $g : \nat \rightarrow \nat$.
\end{definition}

\noindent
We note that in case a kernelization algorithm $\mathcal{A}$ exists for some parameterized problem and there exists some algorithm $\mathcal{B}$ that solves this parameterized problem (e.g. a brute-force algorithm), that then this problem is fixed-parameter tractable. This, because we can apply algorithm $\mathcal{A}$ on any problem instance $(x,k)$ in $\poly{|x| + k}$ time to obtain a problem instance $(x',k')$ with $|x'| + k' \leq g(k)$, which then can be solved using algorithm $\mathcal{B}$ in $f(|x'| + k') \leq f(g(k))$ time\footnote{We note here that we assume that $f$ is a non-decreasing function that expresses the running time of algorithm $\mathcal{B}$. Although almost all such functions would be non-decreasing by nature, if this were not the case one could still obtain a function describing the maximum running time using the non-decreasing function $\bar{f}(k) = \max_{0 \leq i \leq k} f(k)$.}, yielding an algorithm that solves any instance in $f(g(k)) + \poly{|x| + k}$ time.

The efficiency of a kernelization algorithm can be measured along two axes. One being the time the algorithm takes and the other being the size of the resulting problem instance. We will throughout this thesis mainly focus on latter. Although in practice the running time of the kernelization algorithm is important, from a theoretical standpoint the size of the reduced instances asymptotically dominates the total running time of the combined FPT-algorithm. We will only prove that the running time of our kernelization algorithm is polynomial in $|G| + t$, however, not the degree of this polynomial.

\section{Approximation}\label{sec:prelim:approximation}
Another method to efficiently find solutions to NP-hard problems is that of approximation. An approximation algorithm is an algorithm that finds near-optimal solutions. We use the definition of an approximation algorithm by Cormen et al. \cite{CLRS}. An algorithm has an approximation ratio $\varepsilon$ if for any problem instance $x$ it finds a solution with cost $c(x)$, such that $\max\parenthesised{\frac{c(x)}{\textrm{\textsc{opt}}(x)}, \frac{\textrm{\textsc{opt}}(x)}{c(x)}} \leq \varepsilon$, where $\textrm{\textsc{opt}}(x)$ is the optimal cost of a solution for $x$. An algorithm with an approximation ratio $\varepsilon$ is called an $\varepsilon$-approximation algorithm.

In our case, we can interpret the \textsc{Treewidth-2 Vertex Deletion} problem as finding a modulator for $G$ of minimum size. An $\varepsilon$-approximation algorithm for this optimisation problem then is an algorithm with $\varepsilon \geq 1$ that finds for a graph $G$ a modulator $X$ with $\twtwodeletionG{G} \leq |X| \leq \varepsilon \cdot \twtwodeletionG{G}$.

\chapter{General Reduction Rules}\label{ch:general-reduction-rules}
Within this section we start by introducing three reduction rules. The first two of these rules are rules that can be applied when we already know whether or not $\twtwodeletion{G}{t}$ holds. These rules each transform the input graph into a small constant size widget of which the solution is known. We return a small widget, opposed to returning the solution to the \textsc{Treewidth-2 Vertex Deletion} problem, to adhere to the definition of a kernelization algorithm (Definition \ref{def:kernelization}). The third rule is a rule that will ensure that the minimum degree of the graph is at least three.

\section{Reduce if Solution is known}
\begin{reduction}[solution is known]\label{red:solution-is-known}
    Given a problem instance $(G,t)$ and a solution $S$ for $(G,t)$. Replace $G$ by the empty graph.
\end{reduction}

\begin{lemma}[safeness]\label{lemma:safeness:solution-is-known}
    Let $(G',t)$ be the problem instance obtained by applying Reduction \ref{red:solution-is-known} on $(G,t)$. Then $\twtwodeletion{G}{t} \Longleftrightarrow \twtwodeletion{G'}{t}$ holds.
\end{lemma}

\begin{proof}
    Trivially holds because $\emptyset$ is a solution for $(G',t)$.
\end{proof}

\begin{lemma}\label{lemma:time:solution-is-known}
	We can detect whether Reduction \ref{red:solution-is-known} can be applied on a problem instance $(G,t)$ and a vertex set $S \subseteq V(G)$, and apply this rule in $\polyG$ time.
\end{lemma}

\begin{proof}
    Directly follows from Lemma \ref{lemma:find-tree-decompositions}.
\end{proof}

\noindent
We note that Lemma \ref{lemma:time:solution-is-known} only implies that verifying whether a given set $S$ is a solution for $(G,t)$ can be done in $\polyG$ time. It does not imply that such a candidate set $S$ can be found in $\polyG$ time.

\section{Reduce if no Solution exists}
\begin{reduction}[no existing solution]\label{red:no-existing-solution}
    Given a problem instance $(G,t)$ with $\nottwtwodeletion{G}{t}$. Replace $(G,t)$ by $(K_4,0)$.
\end{reduction}

\begin{lemma}[safeness]\label{lemma:safeness:no-existing-solution}
    Let $(G',t')$ be the problem instance obtained by applying Reduction \ref{red:no-existing-solution} on $(G,t)$. Then $\twtwodeletion{G}{t} \Longleftrightarrow \twtwodeletion{G'}{t}$ holds.
\end{lemma}

\begin{proof}
    Trivially holds because there does not exist a solution for $(G',0)$.
\end{proof}

\begin{lemma}\label{lemma:time:no-existing-solution}
	We can detect whether Reduction \ref{red:no-existing-solution} can be applied on a problem instance $(G,0)$ and apply this rule in $\polyG$ time.
\end{lemma}

\begin{proof}
    Directly follows from Lemma \ref{lemma:find-tree-decompositions}.
\end{proof}

\noindent
We note that Lemma \ref{lemma:time:no-existing-solution} can only be applied on problem instances $(G,t)$ with $t = 0$.

\section{Contract Component with Small Neighbourhood}
We introduce a reduction rule that will ensure that after it has been exhaustively applied, that the remaining graph $G$ will have a minimum degree of at least three, which will be an essential property in our kernelization algorithm.

The reduction rule we introduce is a generalisation of Reduction 1a, 1b, and 2 in the work of Bodlaender et al. \cite{Preprocessing-rules-for-triangulation-of-probabilistic-networks}, which are respectively rules that delete degree one, delete degree zero, or contract degree two vertices. Our definition of this reduction rule generalises these rules by considering connected subgraphs of $G$ opposed to individual vertices.

\begin{reduction}[contract component]\label{red:contract-component-with-small-neighborhood}
	Given a problem instance $(G,t)$ and a connected subgraph $H$ of $G$ with $\treewidth{G[N_G[H]] \cup \binom{N_G(H)}{2}} \leq 2$. Remove $H$ from $G$ and add edges $\binom{N_G(H)}{2}$ to $G$.
\end{reduction}

\noindent
In case $H$ has zero or one neighbours this reduction rule boils down to: `if $\treewidth{H} \leq 2$, then remove $H$ from $G$'. In case $H$ has two neighbours this reduction rule corresponds to: `if the graph $H$ together with its neighbours will have treewidth at most two, irregardless of whether there exists a path between those two neighbours outside of $H$, then contract $H$ into one of the two neighbours'.

\begin{lemma}[safeness]\label{lemma:safeness:contract-component-with-small-neighborhood}
	Let $(G',t)$ be the problem instance obtained by applying Reduction \ref{red:contract-component-with-small-neighborhood} on $(G,t)$. Then $\twtwodeletion{G}{t} \Longleftrightarrow \twtwodeletion{G'}{t}$ holds.
\end{lemma}

\begin{proof}
    We have $|N_G(H)| \leq 2$, because otherwise contracting $H$ into a single vertex would yield a $K_4$ minor in $G[N_G[H]] \cup \binom{N_G(H)}{2}$. In case $|N_G(H)| = 2$ we can obtain $G'$ from $G$ by contracting $H$ into a vertex in $N_G(H)$. In case $|N_G(H)| \leq 1$ we obtain $G'$ from $G$ by removing $V(H)$. Hence $G'$ is a minor of $G$. Therefore, by Lemma \ref{lemma:minor-of-G-implication} we have that $\twtwodeletion{G}{t}$ implies $\twtwodeletion{G'}{t}$. It remains to prove that $\twtwodeletion{G'}{t}$ implies $\twtwodeletion{G}{t}$.
    
    Let $S'$ be a solution for $(G',t)$. We have $\parenthesised{G \cup \binom{N_G(H)}{2}} - (V(H) \cup S') = G' - S'$ and that $N_G(H) \setminus S'$ induces a clique in $\parenthesised{G \cup \binom{N_G(H)}{2}} - S'$. Hence, we can apply Theorem \ref{theorem:connect-treewidth-graphs} on $\parenthesised{\parenthesised{G \cup \binom{N_G(H)}{2}} - S', G' - S'}$ to derive $\treewidth{\parenthesised{G \cup \binom{N_G(H)}{2}} - S'} \leq 2$. By Lemma \ref{lemma:minor-treewidth} this yields $\treewidth{G - S'} \leq 2$, which implies $\twtwodeletion{G}{t}$.
\end{proof}

\begin{lemma}\label{lemma:time:contract-component-with-small-neighborhood}
	We can detect whether Reduction \ref{red:contract-component-with-small-neighborhood} can be applied on a problem instance $(G,t)$ and apply this rule in $\polyG$ time.
\end{lemma}

\begin{proof}
    The only connected subgraphs $H$ on which Reduction \ref{red:contract-component-with-small-neighborhood} can be applied have $|N_G(H)| \leq 2$. For each $X \subseteq V(G)$ with $|X| \leq 2$ and each component $H \in \mathcal{C}(G - X)$ we can verify in $\polyG$ time whether $\treewidth{G[N_G[H]] \cup \binom{N_G(H)}{2}} \leq 2$ holds (Lemma \ref{lemma:find-tree-decompositions}). Because there exist at most $\polyG$ such subsets $X$ and components $H$, we can test all potential applications of Reduction \ref{red:contract-component-with-small-neighborhood} in $\polyG$ time.
\end{proof}

\begin{lemma}\label{lemma:min-degree-3}
    Given a problem instance $(G,t)$ on which Reduction \ref{red:contract-component-with-small-neighborhood} can not be applied. Graph $G$ has a minimum degree of at least three.
\end{lemma}

\begin{proof}
    Assume per contradiction that there exists a vertex $u \in V(G)$ with $|N_G(u)| \leq 2$. We have that $G[N_G[u]] \cup \binom{N_G(u)}{2}$ is a clique of size $|N_G[u]| \leq 3$, which has treewidth $|N_G[u]| - 1 \leq 2$. This contradicts Reduction \ref{red:contract-component-with-small-neighborhood} not being applicable. Hence, no such vertex $u$ can exist.
\end{proof}

\section{Trivial Problem Instances}
The property that a graph has minimum degree three is essential to many of the reduction rules introduced in this thesis. Because we can apply this rule, and in specific cases the previous rules, in polynomial time, we will always want to apply these reduction rules whenever possible. I.e., whenever an application of these `trivial' reduction rules is possible, then we will want to immediately apply them. To avoid repeatedly mentioning the requirement that these `trivial' reduction rules have been applied on a problem instance $(G,t)$ we will introduce the notation of a trivial problem instance.

\begin{definition}[trivial problem instance]\label{def:trivial-problem-instance}
    A problem instance $(G,t)$ is a trivial problem instance when $V(G) \leq 4$ or $\treewidth{G} \leq 2$ or $t=0$ or Reduction \ref{red:contract-component-with-small-neighborhood} can be applied on $(G,t)$.
\end{definition}

\noindent
We note that a problem instance $(G,t)$ being trivial does not imply that it is trivial to determine whether $\twtwodeletion{G}{t}$ holds. Rather, one of the `trivial' reduction rules can be applied.

\begin{lemma}\label{lemma:get-non-trivial-problem-instance}
    Given a trivial problem instance $(G,t)$, we can obtain in $\polyG$ time obtain a problem instance $(G',t)$ with $\twtwodeletion{G'}{t} \Longleftrightarrow \twtwodeletion{G}{t}$ such that either $|V(G')| \leq 4$ or $(G', t)$ is a non-trivial problem instance with $|V(G')| < |V(G)|$.
\end{lemma}

\begin{proof}
    We prove using induction on $|V(G)|$.
    \begin{basecase}[$|V(G)| \leq 4$]
        Then returning $(G,t)$ suffices.
    \end{basecase}
    \begin{inductivestep}[$|V(G)| > 4$]~
        \begin{indsubcase}[$\treewidth{G} \leq 2$]
            We can apply Reduction \ref{red:solution-is-known} on $(G,t)$ and $S = \emptyset$ to obtain problem instance $(G',t)$. As follows from the definition of Reduction \ref{red:solution-is-known} and Lemmas \ref{lemma:safeness:solution-is-known} and \ref{lemma:time:solution-is-known}, returning $(G',t)$ suffices.
        \end{indsubcase}
    
        \begin{indsubcase}[$t = 0 \wedge \treewidth{G} > 2$]
            We can apply Reduction \ref{red:no-existing-solution} on $(G,t)$ to obtain problem instance $(G',0)$. As follows from the definition of Reduction \ref{red:no-existing-solution} and Lemmas \ref{lemma:safeness:no-existing-solution} and \ref{lemma:time:no-existing-solution}, returning $(G',t)$ suffices.
        \end{indsubcase}
        
        \begin{indsubcase}[$t > 0 \wedge \treewidth{G} > 0$]
            From Definition \ref{def:trivial-problem-instance} it follows that we can apply Reduction \ref{red:contract-component-with-small-neighborhood} on $(G,t)$. As follows from the definition of Reduction \ref{red:contract-component-with-small-neighborhood} and Lemmas \ref{lemma:safeness:contract-component-with-small-neighborhood} and \ref{lemma:time:contract-component-with-small-neighborhood}, we can in $\polyG$ time obtain a problem instance $(G',t)$ with $|V(G')| < |V(G)|$ and $\twtwodeletion{G'}{t} \Longleftrightarrow \twtwodeletion{G}{t}$. From the induction hypothesis it follows that we can in $\poly{|G'|} \leq \polyG$ time obtain a problem instance $(G'',t)$ with $\twtwodeletion{G''}{t} \Longleftrightarrow \twtwodeletion{G'}{t}$ such that $|V(G'')| \leq 4$ or $(G'', t)$ is a non-trivial problem instance with $|V(G'')| < |V(G')|$. Hence, returning $(G'',t)$ suffices.
        \end{indsubcase}
    \end{inductivestep}
    
    \noindent
    Because these cases cover all cases, correctness trivially follows.
\end{proof}

\noindent
We note that after obtaining a problem instance $(G',t)$ with $|V(G')| \leq 4$ our kernelization algorithm will terminate. This, because we can define the function $f$ that bounds the size of the resulting problem instance $|V(G')| \leq f(t)$ such that $4 \leq f(t)$ holds for any $t \in \nat$. In the remainder of this thesis we only focus on the remaining case where a non-trivial problem instance $(G',t)$ is obtained.

\chapter{Limit-\textit{m} Subsets and Cliques}\label{ch:limit-1-subsets-and-cliques}
In this section we introduce a new notion of a limit-$m$ subset. The idea behind a limit-$m$ subset $X$ is that it conveys whether there exist solutions for $(G,t)$ that remove few vertices from $X$, or whether it is essential that any solution for $(G,t)$ contains at least a certain number of vertices from $X$.

\begin{definition}[limit-$m$ subset\footnote{We note that the definition of a limit-$m$ subset corresponds to neighbourhood $X$ of a near-protrusion $C$, as introduced by Fomin et al. \cite{Planar_F-Deletion_Approximation_and_Optimal_FPT_Algorithms}. However, because we wish to refer to such sets $X$ independently of whether they are adjacent to a near-protrusion $C$, and our algorithm does not make use of protrusion replacement, we introduce new terminology that addresses the interaction between an arbitrary vertex set $X$ and a solution $S$.}]\label{def:limit-m-subset}
	A vertex subset $X \subseteq V(G)$ is a limit-$m$ subset for $(G,t)$ if every solution $S$ for $(G,t)$ has $|X \setminus S| \leq m$.
\end{definition}

\noindent
For example, if $\{x,y\}$ is a limit-1 subset of $(G,t)$ it will hold for any solution $S$ for $(G,t)$ that $|\{x,y\} \setminus S| \leq 1$. I.e., $S$ will need to include either $x$ or $y$ (or both) because there does not exist a solution that removes neither. Another interpretation of a limit-$m$ subset is that from this subset we can choose at most $m$ vertices to not include in a solution set $S$ (or rather leave in $G - S$).

An important aspect to point out is that for a vertex subset $X$ to be a limit-$m$ subset there does not need to exist a solution with $|X \setminus S| \leq m$. Rather, there should not exist one with $|X \setminus S| > m$. Furthermore, this definition only considers solutions for $(G,t)$ and not modulators of size larger than $t$. I.e., there might exist modulators $X$ of $G$ with $|X \setminus S| > m$, but those modulators will have $|X| > t$.

Throughout this thesis we will only work with limit-0, limit-1, and limit-2 subsets. Limit-0 subsets $X$ simply state that no solution can exist for $(G,t)$ with $X \not\subseteq S$, in which case we will end up removing vertex set $X$ from $G$ and decreasing $t$ by $|X|$. Limit-1 and limit-2 subsets will be used to denote properties of the neighbourhood $N_G(H)$ of some subgraphs $H$ in $G$. The statement that $N_G(H)$ is a limit-1 or limit-2 subset implies that $H-S$ can be separated from $(G - H) - S$ in graph $G - S$ by at most one or two neighbours $N(H) \setminus S$, which will be an essential property in proving the safeness of reduction rules.

Identifying whether any set $X$ is a limit-$m$ subset for $(G,t)$ is, in essence, the same problem as determining whether there exists a solution for $(G,t)$. For example, $V(G)$ is a limit-$0$ subset for $(G,t)$ with $t < |V(G)|$ if and only if each solution $S$ for $(G,t)$ has $|V(G) \setminus S| \leq 0$. Because any solution $S$ for $(G,t)$ has $|S| \leq t < |V(G)|$ we would derive for any such solution that $|V(G) \setminus S| > 0$, which means that $V(G)$ is a limit-$0$ subset for $(G,t)$ if and only if $\twtwodeletion{G}{t}$. Hence, we do not expect to find polynomial time algorithms that is able to always correctly determine whether a vertex set is a limit-$m$ subset. In some instances, however, we will be able to determine whether a set $X$ is a limit-$m$ subset for $(G,t)$. Within this section we will explore some of these methods, which will be used multiple times throughout this thesis.

\begin{lemma}\label{lemma:limit-m-subset-identification}
    Let $(G,t)$ be a problem instance and let $X \subseteq V(G)$. If there exist $t+1$ vertex-disjoint subgraphs $G_i$ in $G - X$ with $\treewidth{G[V(G_i) \cup X]} > 2$, then $X$ is a limit-{$|X|{{-}}1$} subset for $(G,t)$.
\end{lemma}

\begin{proof}
    Assume per contradiction that $S$ is a solution for $(G,t)$ with $|X \setminus S| > |X|-1$, i.e., $X \cap S = \emptyset$. Because $|S| \leq t$ there exists a subgraph $G_i$ with $V(G_i) \cap S = \emptyset$. Then $G[V(G_i) \cup X]$ is a subgraph of $G - S$. Lemma \ref{lemma:minor-treewidth} yields $\treewidth{G - S} \geq \treewidth{G[V(G_i) \cup X]} > 2$, which contradicts $S$ being a solution of $(G,t)$. Hence, $X$ must be a limit-{$|X|{{-}}1$} subset for $(G,t)$.
\end{proof}

\begin{lemma}\label{lemma:limit-m-of-subgraph-to-supergraph}
    Let $(G,t)$ be a problem instance, let $H$ be a subgraph of $G$, and let $X$ be a limit-$m$ subset for $(H,t)$. Then $X$ is a limit-$m$ subset for $(G,t)$.
\end{lemma}

\begin{proof}
    Assume per contradiction that $S$ is a solution for $(G,t)$ with $|X \setminus S| > m$. By Lemma \ref{lemma:minor-treewidth} we have $\treewidth{H - S} \leq \treewidth{G - S} \leq 2$, which contradicts $X$ being a limit-$m$ subset for $(H,t)$. Hence, $X$ must be a limit-$m$ subset for $(G,t)$.
\end{proof}

\begin{lemma}\label{lemma:subset-of-limit-m-of-subset}
    Let $(G,t)$ be a problem instance, let $X$ be a limit-$m$ subset for $(G,t)$, and let $Y \subseteq X$. Then $Y$ is a limit-$m$ subset for $(G,t)$.
\end{lemma}

\begin{proof}
    Assume per contradiction that $S$ is a solution for $(G,t)$ with $|Y \setminus S| > m$. Because $|X \setminus S| \geq |Y \setminus S|$ this contradicts $X$ being a limit-$m$ subset for $(G,t)$. Hence, $Y$ must be a limit-$m$ subset for $(G,t)$.
\end{proof}

\begin{lemma}\label{lemma:limit-m-subset-to-superset}
    Let $(G,t)$ be a problem instance and let $X \subseteq V(G)$. If each $Y \subseteq X$ with $|Y| = m+1$ is a limit-$m$ subset for $(G,t)$, then $X$ is a limit-$m$ subset.
\end{lemma}

\begin{proof}
    Assume per contradiction that $S$ is a solution for $(G,t)$ with $|X \setminus S| > m$. Let $Y \subseteq X \setminus S$ with $|Y| = m+1$. We have $|Y \setminus S| > m$, which contradicts $Y$ being a limit-$m$ subset for $(G,t)$. Hence, $X$ must be a limit-$m$ subset for $(G,t)$.
\end{proof}

\noindent
In some instances we will desire the additional property that a limit-$m$ subset $X$ of $(G,t)$ also induces a clique in $G$. To avoid repeating this additional requirement that a limit-$m$ subset $X$ also induces a clique in $G$ we define a limit-$m$ clique, which simply is a combination of these two requirements.

\begin{definition}[limit-$m$ clique]\label{def:limit-m-clique}
    A vertex subset $X \subseteq V(G)$ is a limit-$m$ clique for $(G,t)$ if $G[X]$ is a clique and $X$ is a limit-$m$ subset for $(G,t)$.
\end{definition}

\section{Finding Limit-\textit{m} Cliques or Disjoint Modulators}
We introduce the algorithm that will be used to obtain from a modulator $X$ a modulator $X \cup Y$, such that components in $G - (X \cup Y)$ will have additional properties that are desirable for reducing those components. I.e., the tidiness property briefly introduced in Section \ref{sec:intro:our-contributions}. Informally, the first property is that we can remove any vertex $u \in X \cup Y$ from $X \cup Y$ and retain $(X \cup Y) \setminus \{u\}$ being a modulator for $G$. The second property is that any component $C \in \mathcal{C}(G - (X \cup Y))$ has $N_G(C) \cap X$ being a limit-1 clique for $(G - C,t)$. In Chapter \ref{ch:graph-decompositions} we will handle the formalisation of these properties and how they follow from Algorithm \ref{algo:find-disjoint-modulator}.

The algorithm we define operates on problem instances $(G,t)$ together with modulators $X'$ of $G$ for which $G[X']$ is a clique and $0 < |X'| \leq 3$. The limitation that $X'$ is a small clique will not be problematic, as our kernelization algorithm will invoke this algorithm once for each $x \in X$ on graphs $G - (X \setminus \{x\})$ with modulator $X' = \{x\}$ and once for each $xy \in E[G[X]]$ on graphs $G - (X \setminus \{x,y\})$ with modulator $X' = \{x,y\}$. In other words, we will run this algorithm not on the full graph $G$ with modulator $X$. Instead we will run this algorithm multiple times on different subgraphs of $G$ with corresponding modulators $X' \subseteq X$.

\begin{algorithm}
    \SetAlgoLined
	\caption{\textsc{Find-Disjoint-Modulator}($G, t, X$)}
	\label{algo:find-disjoint-modulator}
	\smallskip
	\KwIn{A graph $G$, a parameter $t$, and a modulator $X$ for $G$, for which $G[X]$ is a clique and $0 < |X| \leq 3$.}
	\KwOut{A set $Y \subseteq V(G) \setminus X$ with $|Y| \leq 3t+3$ such that either $\treewidth{G - Y} \leq 2$ or every $C \in \mathcal{C}(G - (X \cup Y))$ has that $X$ is a limit-{$|X|{{-}}1$} clique for $(G-C,t)$.}
	\medskip
	\everypar={\nl}
	let $(T, \chi)$ be a width-2 tree decomposition of $G - X$ with an arbitrary root\\
	for each $r \in V(T)$ let $T(r)$ denote the subtree of $T$ rooted at $r$\\
	$Y = \emptyset$\\
	\While{\upshape$\treewidth{G[\chi(T) \cup X]} > 2$ and up to $t{{+}}1$ times}{
	    let $r \in V(T)$ be a vertex furthest from the root with $\treewidth{G[\chi(T(r)) \cup X]} > 2$\\
	    add all vertices $\chi(r)$ to $Y$\\
	    remove\footnotemark{} $\chi(T(r))$ from $(T,\chi)$
	}
	\Return $Y$
\end{algorithm}
\footnotetext{The obtained tree decomposition $(T,\chi')$ has for each $t \in V(T)$ that $\chi'(t) = \chi(t) \setminus \chi(T(r))$. We note that the resulting tree decomposition can contain empty bags. I.e., we do not remove or contract empty bags.}

\noindent
We will prove that Algorithm \ref{algo:find-disjoint-modulator} is correct and takes a polynomial amount of time. In these proofs we often need to make a distinction between the state of the variables throughout multiple iterations of the while-loop. We will use $(T,\chi_i)$ and $Y_i$ to respectively denote the tree decomposition $(T,\chi)$ and vertex set $Y$ after $i$ iterations of the while-loop. We let $r_i$ be the vertex chosen to be vertex $r$ during the $i{{+}}1$'th iteration, which means it is a vertex in $T$ furthest from the root with $\treewidth{G[\chi_i(T(r_i)) \cup X]} > 2$. In case the $i$'th iteration was the last iteration, then we let $r_i$ be undefined.

From the definitions of $(T,\chi_i)$, $Y_i$, and $r_i$ it follows that initially $Y_0 = \emptyset$ and $\chi_0(T) = V(G) \setminus X$, and after $i > 0$ iterations we have $Y_i = Y_{i-1} \cup \chi_{i-1}(r_{i-1})$ and $\chi_i(T) = \chi_{i-1}(T) \setminus \chi_{i-1}(T(r_{i-1}))$.

Before we formally prove that Algorithm \ref{algo:find-disjoint-modulator} is correct we first provide the intuition behind this algorithm. We display an example tree decomposition $(T,\chi)$ of $G - X$ in Figure \ref{fig:finding-limit-m-cliques-or-disjoint-modulators}.

\begin{figure}[ht]
    \centering
    \includegraphics[scale=.62]{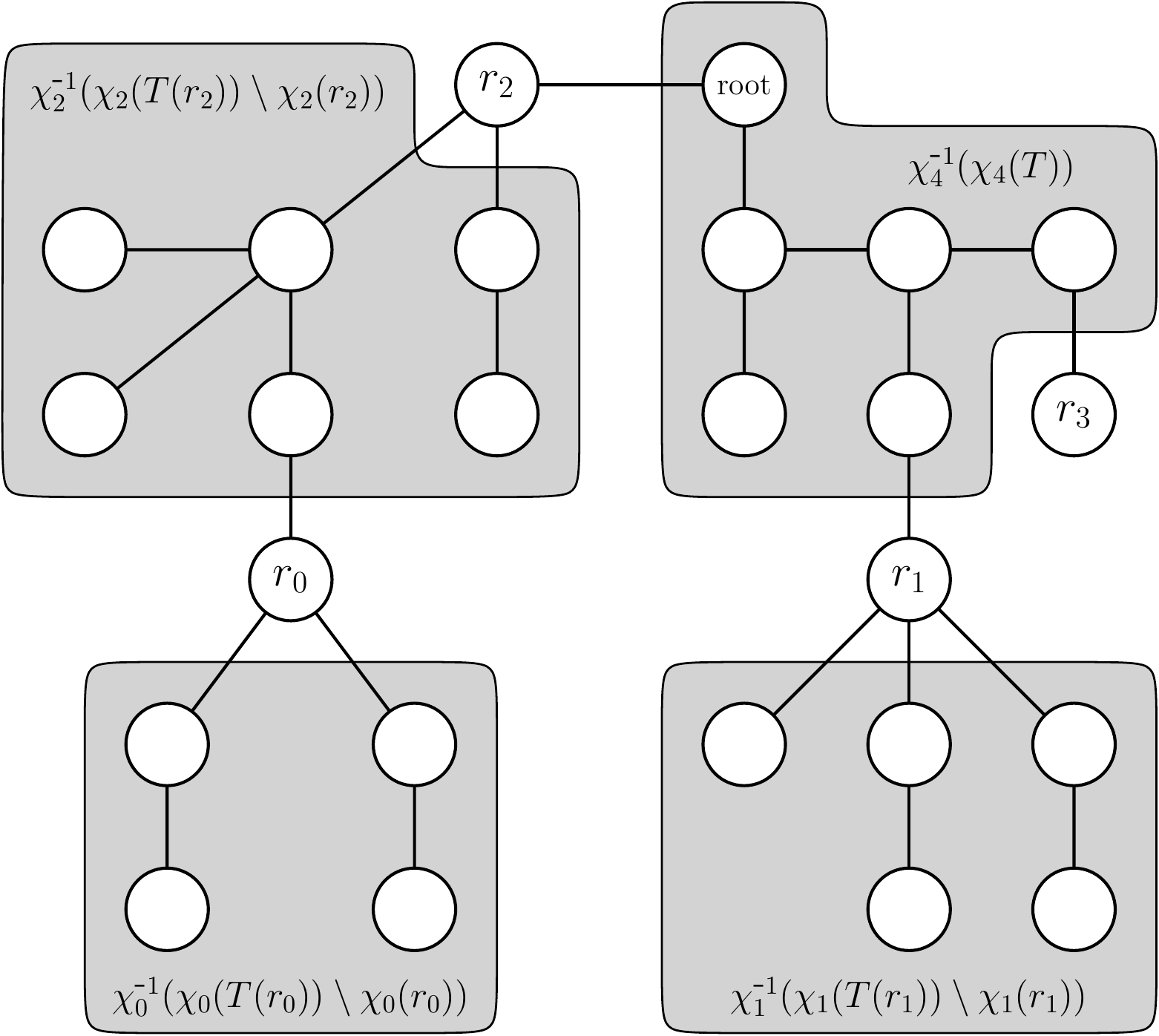}
    \caption{Example tree decomposition and subgraphs $\bchi_i(\chi_i(T(r_i)) \setminus \chi_i(r_i))$}
    \label{fig:finding-limit-m-cliques-or-disjoint-modulators}
\end{figure}
\clearpage

\noindent
Algorithm \ref{algo:find-disjoint-modulator} first finds a vertex $r_0$ furthest from the root with $\treewidth{G[\chi_0(T(r_0)) \cup X]} > 2$. We will be able to prove that $\treewidth{G[(\chi_0(T(r_0)) \setminus \chi_0(r_0)) \cup X]} \leq 2$. Hence, in the first iteration we identify a set of vertices $\chi_0(T(r_0))$ and add up to three vertices $\chi_0(r_0)$ to $Y_0$ such that $\treewidth{G[\chi_0(T(r_0)) \cup X]} > 2$ and $\treewidth{G[(\chi_0(T(r_0)) \setminus Y_1) \cup X]} \leq 2$.

The next iteration we `ignore' all vertices in $\chi_0(T(r_0))$. Using the same procedure as before we find a vertex $r_1$ and add three vertices $\chi_1(r_1)$ to $Y_1$ such that $\treewidth{G[\chi_1(T(r_1)) \cup X]} > 2$ and $\treewidth{G[(\chi_1(T(r_1)) \setminus Y_2) \cup X]} \leq 2$. Because $\chi_0(r_0) \subseteq Y_2$ separates $\chi_1(T(r_1)) \setminus Y_2$ from $\chi_0(T(r_0)) \setminus Y_2$ in $G - X$ we can prove that $\treewidth{G[((\chi_1(T(r_1)) \cup \chi_0(T(r_0))) \setminus Y_2) \cup X]} = \treewidth{G - (\chi_2(T) \cup Y_2)} \leq 2$.

We keep repeating this procedure, each time adding up to three vertices $\chi_i(r_i)$ to $Y_i$ such that $\treewidth{G[\chi_i(T(r_i)) \cup X]} > 2$ and $\treewidth{G - (\chi_{i+1}(T) \cup Y_{i+1})} \leq 2$. At some point a situation is reached where either $\treewidth{G[\chi_i(T) \cup X]} \leq 2$ or $i > t+1$. 

In case of the former holds, because $Y_i$ separates $\chi_i(T)$ from $V(G) \setminus X$ in $G - X$, we can prove $\treewidth{G - Y_i} \leq 2$. Hence, in this case returning $Y_i$ suffices.

In case the latter holds there exists a partitioning $\mathcal{D} = \{\chi_0(T(r_0)), \dots, \chi_t(T(r_t)), \chi_{t+1}(T)\}$ of $G - X$ into $t+2$ sets $D \in \mathcal{D}$ with $\treewidth{G[D \cup X]} > 2$. We can prove that each component $C \in \mathcal{C}(G - (X \cup Y_{t+1}))$ is a subgraph of some $G[D]$. When component $C$ is removed from $G$ there will still exist $t+1$ vertex sets $D' \in \mathcal{D} \setminus \{D\}$ in $G - C$ with $\treewidth{G[D \cup X]} > 2$. By Lemma \ref{lemma:limit-m-subset-identification} we then derive that $X$ is a limit-{$|X|{{-}}1$} clique for $(G-C,t)$. Hence, in this case returning $Y_i$ suffices as well.

With this we will then have proven that Algorithm \ref{algo:find-disjoint-modulator} is safe. The proof that Algorithm \ref{algo:find-disjoint-modulator} takes $\poly{G}$ time can easily be derived from Lemma \ref{lemma:find-tree-decompositions}. We will next formally prove the correctness of Algorithm \ref{algo:find-disjoint-modulator}.

\begin{lemma}\label{lemma:partial-safeness:find-disjoint-modulator:stays-tree-decomposition}
    After $i$ iterations $(T,\chi_i)$ is a tree decomposition of $G[\chi_i(T)]$.
\end{lemma}

\begin{proof}
    We prove using induction on $i$.
    
    \begin{basecase}[$i = 0$]
        Holds because $(T,\chi_0)$ is a tree decomposition of $G[\chi_0(T)] = G - X$.
    \end{basecase}
    
    \begin{inductivestep}[$i > 0$]
        We assume the induction hypothesis that $(T,\chi_{i-1})$ is a tree decomposition of $G[\chi_{i-1}(T)]$ and prove that $(T,\chi_i)$ is a tree decomposition of $G[\chi_i(T)]$.
        
        Let $uv \in E(G[\chi_i(T)])$. Because $E(G[\chi_i(T)]) \subseteq E(G[\chi_{i-1}(T)])$ there exists some node $t \in V(T)$ with $u,v \in \chi_{i-1}(t)$. Because $\chi_i(t) = \chi_{i-1}(t) \setminus \chi_{i-1}(T(r_{i-1}))$ and $u,v \in \chi_i(T) = \chi_{i-1}(T) \setminus \chi_{i-1}(T(r_{i-1}))$ we have $u,v \in \chi_i(t)$. Hence, for each $uv \in E(G[\chi_i(T)])$ there exists some $t \in V(T)$ with $u,v \in \chi_i(t)$, which proves that Property 2 of Definition \ref{def:tree-decomposition} holds.

        For each $u \in \chi_i(T)$ we have $T[\bchi_{i-1}(u)] = T[\bchi_i(u)]$, which is connected. Hence, Property 3 of Definition \ref{def:tree-decomposition} holds. Because $V(G[\chi_i(T)]) = \chi_i(T)$ Property 1 of Definition \ref{def:tree-decomposition} holds. Hence, $(T,\chi_{i-1})$ is a tree decomposition of $G[\chi_{i-1}(T)]$.
    \end{inductivestep}
    
    \noindent
    We conclude that after $i$ iterations $(T,\chi_i)$ is a tree decomposition of $G[\chi_i(T)]$.
\end{proof}

\begin{lemma}\label{lemma:partial-safeness:find-disjoint-modulator:0}
    After $i$ iterations $X \cup Y_i$ separates $\chi_i(T)$ from $V(G) \setminus (\chi_i(T) \cup X \cup Y_i)$ in $G$.
\end{lemma}

\begin{proof}
    We prove using induction on $i$.
    
    \begin{basecase}[$i = 0$]
        It vacuously holds that $Y_0 = \emptyset$ separates $\chi_0(T) = V(G) \setminus X$ from $V(G) \setminus (\chi(T) \cup X \cup Y_0) = \emptyset$ in $G$.
    \end{basecase}

    \begin{inductivestep}[$i > 0$]
        We assume the induction hypothesis that $X \cup Y_{i-1}$ separates $\chi_{i-1}(T)$ from $V(G) \setminus (\chi_{i-1}(T) \cup X \cup Y_{i-1})$ in $G$ and prove that $X \cup Y_i$ separates $\chi_i(T)$ from $V(G) \setminus (\chi_i(T) \cup X \cup Y_i)$ in $G$.
    
        Any two vertices $u \in \chi_i(T) \subseteq \chi_{i-1}(T)$ and $v \in V(G) \setminus (\chi_{i-1}(T) \cup X \cup Y_i)$ are separated by $X \cup Y_i$ in $G$ because they are separated by $X \cup Y_{i-1} \subseteq X \cup Y_i$ in $G$. Therefore, it remains to prove that any two vertices $u \in \chi_i(T)$ and $v \in (V(G) \setminus (\chi_i(T) \cup X \cup Y_i)) \setminus (V(G) \setminus (\chi_{i-1}(T) \cup X \cup Y_i)) = (\chi_{i-1}(T) \cup X \cup Y_i) \setminus (\chi_i(T) \cup X \cup Y_i) = \chi_{i-1}(T) \setminus (\chi_i(T) \cup X \cup Y_i) = \chi_{i-1}(T) \setminus ((\chi_{i-1}(T) \setminus \chi_{i-1}(T(r_{i-1}))) \cup X \cup Y_i) = \chi_{i-1}(T(r_{i-1})) \setminus (X \cup Y_i) = \chi_{i-1}(T(r_{i-1})) \setminus \chi_{i-1}(r_{i-1})$ are separated by $X \cup Y_i$ in $G$.

        Assume per contradiction $u \in \chi_i(T)$ and $v \in \chi_{i-1}(T(r_{i-1})) \setminus \chi_{i-1}(r_{i-1})$ are not separated by $X \cup Y_i$ in $G$. Let $C \in \mathcal{C}(G - (X \cup Y_i))$ be the component that contains both $u$ and $v$. Because $u$ is separated from any vertex $v' \in V(G) \setminus (\chi_{i-1}(T) \cup X \cup Y_i)$ by $X \cup Y_i$ we must have $V(C) \subseteq \chi_{i-1}(T) \setminus \chi_{i-1}(r_{i-1})$. 
    
        By Lemma \ref{lemma:partial-safeness:find-disjoint-modulator:stays-tree-decomposition} we have that $(T,\chi_{i-1})$ is a tree decomposition of $G[\chi_{i-1}(T)]$. From Corollary \ref{corollary:treedecomp-node-separator} it follows that there exists some $T' \in \mathcal{C}(T - r_{i-1})$ such that $V(C) \subseteq \chi_{i-1}(T')$. Because $v \in \chi_{i-1}(T(r_{i-1})) \setminus \chi_{i-1}(r_{i-1})$ we know $T' \in \mathcal{C}(T(r_{i-1}) - r_{i-1})$. This, however, yields $u \in \chi_{i-1}(T(r_{i-1}))$, which contradicts $u \in \chi_i(T)$. Hence, no such vertices $u$ and $v$ can exist, proving that $X \cup Y_i$ separates $\chi_i(T)$ from $V(G) \setminus (\chi_i(T) \cup X \cup Y_i)$ in $G$.
    \end{inductivestep}
    
    \noindent
    We conclude that after $i$ iterations $X{{\cup}}Y_i$ separates $\chi_i(T)$ from $V(G) \setminus (\chi_i(T){{\cup}}X{{\cup}}Y_i)$ in $G$.
\end{proof}

\begin{lemma}\label{lemma:partial-safeness:find-disjoint-modulator:1}
    After $i$ iterations $\treewidth{G - (\chi_i(T) \cup Y_i)} \leq 2$ holds.
\end{lemma}

\begin{proof}
    We prove using induction on $i$.
    
    \begin{basecase}[$i = 0$]
        Holds because $\treewidth{G - (\chi_0(T) \cup Y_0)} = \treewidth{G - ((V(G) \setminus X) \cup \emptyset)} = \treewidth{G[X]} = |X|-1 \leq 2$.
    \end{basecase}
    
    \begin{inductivestep}[$i > 0$]
        We assume the induction hypothesis $\treewidth{G - (\chi_{i-1}(T) \cup Y_{i-1})} \leq 2$ and we prove $\treewidth{G - (\chi_i(T) \cup Y_i)} \leq 2$. To prove this we will apply Theorem \ref{theorem:connect-treewidth-graphs} on $(G - (\chi_i(T) \cup Y_i), X)$. This requires proving that each $C \in \mathcal{C}(G - (\chi_i(T) \cup X \cup Y_i))$ has $\treewidth{G[V(C) \cup X]} \leq 2$.
        
        Let $C \in \mathcal{C}(G - (\chi_i(T) \cup X \cup Y_i))$. By Lemma \ref{lemma:partial-safeness:find-disjoint-modulator:0} we know that $X \cup Y_{i-1}$ separates $\chi_{i-1}(T)$ from $V(G) \setminus (\chi_{i-1}(T) \cup X \cup Y_{i-1})$ in $G$. Because $Y_{i-1} \subseteq Y_i$ this means that $X$ separates $G[\chi_{i-1}(T(r_{i-1}))] - Y_i$ from $G - (\chi_{i-1}(T) \cup X \cup Y_i)$ in $G - Y_i$. Because $V(C) \subseteq V(G) \setminus (\chi_i(T) \cup X \cup Y_i) = V(G) \setminus ((\chi_{i-1}(T) \setminus \chi_{i-1}(T(r_{i-1}))) \cup X \cup Y_i) = (V(G) \setminus (\chi_{i-1}(T) \cup X \cup Y_i)) \cup (\chi_{i-1}(T(r_{i-1})) \setminus (X \cup Y_i))$ we have that $C$ is a subgraph of either $G[\chi_{i-1}(T(r_{i-1}))] - Y_i$ or $G - (\chi_{i-1}(T) \cup X \cup Y_i)$.
        
        \begin{indsubcase}[$V(C) \subseteq \chi_{i-1}(T(r_{i-1})) \setminus Y_i$]
            We note $\chi_{i-1}(T(r_{i-1})) \setminus Y_i \subseteq \chi_{i-1}(T(r_{i-1})) \setminus \chi_{i-1}(r_{i-1})$. Hence, from Lemma \ref{lemma:partial-safeness:find-disjoint-modulator:stays-tree-decomposition} and Corollary \ref{corollary:treedecomp-node-separator} it follows that there exists some $T' \in \mathcal{C}(T - r_{i-1})$ with $V(C) \subseteq \chi_{i-1}(T')$. Because $V(C) \subseteq \chi_{i-1}(T(r_{i-1})) \setminus \chi_{i-1}(r_{i-1})$ we have $T' \in \mathcal{C}(T(r_{i-1}) - r_{i-1})$. By definition of $r_i$ we have $\treewidth{G[\chi_{i-1}(T') \cup X]} \leq 2$. Hence, Lemma \ref{lemma:minor-treewidth} yields $\treewidth{G[V(C) \cup X]} \leq \treewidth{G[\chi_{i-1}(T') \cup X]} \leq 2$.
        \end{indsubcase}
        
        \begin{indsubcase}[$V(C) \subseteq V(G) \setminus (\chi_{i-1}(T) \cup X \cup Y_i)$]
            By Lemma \ref{lemma:minor-treewidth} we have $\treewidth{G[V(C) \cup X]} \leq \treewidth{G - (\chi_{i-1}(T) \cup Y_{i-1})} \leq 2$.
        \end{indsubcase}
        
        \noindent
        Because each $C \in \mathcal{C}(G - (\chi_i(T) \cup X \cup Y_i))$ has $\treewidth{G[V(C) \cup X]} \leq 2$ and $G[X]$ is a clique, we can apply Theorem \ref{theorem:connect-treewidth-graphs} on $(G - (\chi_i(T) \cup Y_i), X)$ to derive $\treewidth{G - (\chi_i(T) \cup Y_i)} \leq 2$.
    \end{inductivestep}
    
    \noindent
    We conclude that after $i$ iterations $\treewidth{G - (\chi_i(T) \cup Y_i)} \leq 2$ holds.
\end{proof}

\begin{lemma}[safeness]\label{lemma:partial-safeness:find-disjoint-modulator:2}
    If Algorithm \ref{algo:find-disjoint-modulator} terminates with $\treewidth{G[\chi_i(T) \cup X]} \leq 2$, then the returned set $Y_i \subseteq V(G) \setminus X$ is such that $|Y_i| \leq 3t+3$ and $\treewidth{G - Y_i} \leq 2$.
\end{lemma}

\begin{proof}
    Because each vertex set $\chi_j(r_j)$ that was added to $Y$ to form $Y_i$ has $\chi_j(r_j) \subseteq \chi_0(r_j) \subseteq V(G) \setminus X$ and $|\chi_0(r_j)|\leq 3$ we know that $Y_i \subseteq V(G) \setminus X$ and $|Y_i| \leq 3i \leq 3t+3$ holds. It remains to prove that $\treewidth{G - Y_i} \leq 2$, for which we will apply Theorem \ref{theorem:connect-treewidth-graphs} on $(G - Y_i, X)$. This requires proving that each $C \in \mathcal{C}(G - (X \cup Y_i))$ has $\treewidth{G[V(C) \cup X]} \leq 2$.
    
    Let $C \in \mathcal{C}(G - (X \cup Y_i))$. By Lemma \ref{lemma:partial-safeness:find-disjoint-modulator:0} we know that $X \cup Y_i$ separates $\chi_i(T)$ from $V(G) \setminus (\chi_i(T) \cup X \cup Y_i)$ in $G$. Hence, we have that $C$ is a subgraph of either $G[\chi_i(T)]$ or $G - (\chi_i(T) \cup X \cup Y_i)$.
    
    \begin{case}[$V(C) \subseteq \chi_{i}(T)$]
        Lemma \ref{lemma:minor-treewidth} yields $\treewidth{G[V(C) \cup X]} \leq \treewidth{G[\chi_i(T) \cup X]} \leq 2$.
    \end{case}
        
    \begin{case}[$V(C) \subseteq V(G) \setminus (\chi_i(T) \cup X \cup Y_i)$]
        Lemmas \ref{lemma:minor-treewidth} and \ref{lemma:partial-safeness:find-disjoint-modulator:1} yields $\treewidth{G[V(C) \cup X]} \leq \treewidth{G - (\chi_i(T) \cup Y_i)} \leq 2$.
    \end{case}
    
    \noindent
    Because each $C \in \mathcal{C}(G - (X \cup Y_i))$ has $\treewidth{G[V(C) \cup X]} \leq 2$ and $G[X]$ is a clique we can apply Theorem \ref{theorem:connect-treewidth-graphs} on $(G - Y_i, X)$ to derive $\treewidth{G - Y_i} \leq 2$.
\end{proof}

\begin{lemma}\label{lemma:partial-safeness:find-disjoint-modulator:3}
    After $i$ iterations there exists a partitioning $\mathcal{D}_i$ of $G - (\chi_i(T) \cup X)$ into $i$ subgraphs $D$ with $\treewidth{G[V(D) \cup X]} > 2$, such that each component $C \in \mathcal{C}(G - (\chi_i(T) \cup X \cup Y_i))$ is a subgraph of some $D \in \mathcal{D}_i$.
\end{lemma}

\begin{proof}
    We prove using induction on $i$.
    
    \begin{basecase}[$i = 0$]
        Initially $G - (\chi_0(T) \cup X)$ is the empty graph of which the partitioning $\mathcal{D}_0$ into zero subgraphs trivially suffices.
    \end{basecase}
    
    \begin{inductivestep}[$i > 0$]
        We assume the induction hypothesis that $\mathcal{D}_{i-1}$ is a partitioning of $G - (\chi_{i-1}(T) \cup X)$ into $i-1$ subgraphs $D$ with $\treewidth{G[V(D) \cup X]} > 2$, such that each component $C \in \mathcal{C}(G - (\chi_{i-1}(T) \cup X \cup Y_{i-1}))$ is a subgraph of some $D \in \mathcal{D}_{i-1}$.

        Let $\mathcal{D}_i = \mathcal{D}_{i-1} \cup \{G[\chi_{i-1}(T(r_{i-1}))]\}$. From the definition of $r_{i-1}$ and $\mathcal{D}_{i-1}$ it immediately follows that each $D \in \mathcal{D}_i$ has $\treewidth{G[V(D) \cup X]} > 2$. It remains to prove that each component $C \in \mathcal{C}(G - (\chi_i(T) \cup X \cup Y_i))$ is a subgraph of some $D \in \mathcal{D}_i$.

        Let $C \in \mathcal{C}(G - (\chi_i(T) \cup X \cup Y_i))$. Using the same reasoning as used in the proof of Lemma \ref{lemma:partial-safeness:find-disjoint-modulator:1} we derive that $C$ is a subgraph of either $G[\chi_{i-1}(T(r_{i-1}))] - Y_i$ or $G - (\chi_{i-1}(T) \cup X \cup Y_i)$.
        
        \begin{indsubcase}[$V(C) \subseteq \chi_{i-1}(T(r_{i-1})) \setminus Y_i$]
            Then $C$ is a subgraph of $G[\chi_{i-1}(T(r_{i-1}))] \in \mathcal{D}_i$.
        \end{indsubcase}
        
        \begin{indsubcase}[$V(C) \subseteq V(G) \setminus (\chi_{i-1}(T) \cup X \cup Y_i)$]
            There then must be exist a component $C' \in \mathcal{C}(G - (\chi_{i-1}(T) \cup X \cup Y_{i-1}))$ of which $C$ is a subgraph. From the definition of $\mathcal{D}_{i-1}$ we therefore have that $C$ is a subgraph of some $D \in \mathcal{D}_{i-1} \subset \mathcal{D}_i$.
        \end{indsubcase}
    \end{inductivestep}
    
    \noindent
    Because these cases cover all cases, correctness trivially follows.
\end{proof}

\begin{lemma}[safeness]\label{lemma:partial-safeness:find-disjoint-modulator:4}
    If Algorithm \ref{algo:find-disjoint-modulator} terminates with $\treewidth{G[\chi_{t+1}(T) \cup X]} > 2$, then the returned set $Y_{t+1} \subseteq V(G) \setminus X$ is such that $|Y_{t+1}| \leq 3t+3$ and for each $C \in \mathcal{C}(G - (X \cup Y_{t+1}))$ it holds that $X$ is a limit-{$|X|{{-}}1$} clique for $(G-C,t)$.
\end{lemma}

\begin{proof}
    Because each vertex set $\chi_i(r_i)$ that was added to $Y$ to form $Y_{t+1}$ has $\chi_i(r_i) \subseteq \chi_0(r_i) \subseteq V(G) \setminus X$ and $|\chi_0(r_i)|\leq 3$ we know that $Y_{t+1} \subseteq V(G) \setminus X$ and $|Y_{t+1}| \leq 3t+3$.
    
    By Lemma \ref{lemma:partial-safeness:find-disjoint-modulator:3} we know that there exist a partitioning $\mathcal{D}_{t+1}$ of $G - (\chi_{t+1}(T) \cup X)$ into $t+1$ subgraphs $D$ with $\treewidth{G[V(D) \cup X]} > 2$, such that each component $C \in \mathcal{C}(G - (\chi_{t+1}(T) \cup X \cup Y_{t+1}))$ is a subgraph of some $D \in \mathcal{D}_{t+1}$. Let $\mathcal{D}_{t+2} = \mathcal{D}_{t+1} \cup \{G[\chi_{t+1}(T)]\}$. Using the same reasoning as used in the proof of Lemma \ref{lemma:partial-safeness:find-disjoint-modulator:3} we derive that each $D \in \mathcal{D}_{t+2}$ has $\treewidth{G[V(D) \cup X]} > 2$ and each $C \in \mathcal{C}(G - (X \cup Y_{t+1}))$ is a subgraph of some $D \in \mathcal{D}_{t+2}$.
    
    We let $C \in \mathcal{C}(G - (X \cup Y_{t+1}))$ and argue that $X$ is a limit-{$|X|{{-}}1$} clique for $(G-C,t)$. Let $D \in \mathcal{D}_{t+2}$ such that $C$ is a subgraph of $D$. There exist $t+1$ vertex-disjoint subgraphs $D' \in \mathcal{D}_{t+2} \setminus \{D\}$ of $G - D$ with $\treewidth{G[V(D') \cup X]} > 2$. By Lemma \ref{lemma:limit-m-subset-identification} this yields that $X$ is a limit-{$|X|{{-}}1$} subset for $(G-D,t)$, which by Lemma \ref{lemma:limit-m-of-subgraph-to-supergraph} yields that $X$ is a limit-{$|X|{{-}}1$} subset for $(G-C,t)$.
\end{proof}

\begin{lemma}\label{lemma:time:find-disjoint-modulator}
	We can apply Algorithm \ref{algo:find-disjoint-modulator} on a graph $G$ in $\polyG$ time.
\end{lemma}

\begin{proof}
    By Lemma \ref{lemma:find-tree-decompositions} we can obtain a tree decomposition of $G - X$ in $\polyG$ time. Again, as follows from Lemma \ref{lemma:find-tree-decompositions}, we can determine in $\polyG$ time for each $r \in V(T)$ whether $\treewidth{G[\chi(T(r)) \cup X]} > 2$ holds. All other operations can trivially be executed in $\polyG$ time. Because each iteration at least one vertex from $G$ is added to $Y$, the while-loop iterates at most $\min(t+1,|V(G)|)$ times. Hence, we know that Algorithm \ref{algo:find-disjoint-modulator} takes $\polyG$ time.
\end{proof}

\begin{theorem}\label{theorem:find-disjoint-modulator}
    Given a graph $G$, a parameter $t$, and a modulator $X$ of $G$ such that $G[X]$ is a clique and $0 < |X| \leq 3$. Algorithm \ref{algo:find-disjoint-modulator} finds in $\polyG$ time a set $Y \subseteq V(G) \setminus X$ with $|Y| \leq 3t+3$ such that either $\treewidth{G - Y} \leq 2$ holds or for every $C \in \mathcal{C}(G - (X \cup Y))$ it holds that $X$ is a limit-{$|X|{{-}}1$} clique for $(G-C,t)$.
\end{theorem}

\begin{proof}
    Directly follows from Lemmas \ref{lemma:partial-safeness:find-disjoint-modulator:2}, \ref{lemma:partial-safeness:find-disjoint-modulator:4}, and \ref{lemma:time:find-disjoint-modulator}.
\end{proof}

\chapter{Graph Decompositions}\label{ch:graph-decompositions}
Within this section we will introduce a method to decompose a graph into multiple different subgraphs that can each be reduced separately. Conceptually, our algorithm uses the concept of ``approximation and tidying'' introduced by van Bevern et al \cite{Approximation-and-Tidying-A-Problem-Kernel-for-s-Plex-Cluster-Vertex-Deletion}.

The first step in this framework is an approximation step. In this step we run a polynomial time constant factor approximation algorithm to obtain a modulator $X$ of $G$ of size bounded by a function over $t$. Having a modulator $X$ makes defining reduction rules on components $C$ in $G - X$ easier, as we know that $\treewidth{C} \leq 2$ will hold.

The second step is a tidying step in which we grow the modulator into a tidy modulator $X$. A tidy modulator $X$ is a modulator of $G$, such that for any $x \in X$ it holds that $X \setminus \{x\}$ is a modulator of $G$. A tidy modulator $X$ has the additional benefit that for each component $C \in G - X$ that $\treewidth{G[V(C) \cup \{x\}]} \leq 2$ holds for any $x \in X$.

The third step we introduce, which is not included in the framework by van Bevern et al. \cite{Approximation-and-Tidying-A-Problem-Kernel-for-s-Plex-Cluster-Vertex-Deletion}, is the step of finding a component separator of $(G,t,X)$. A component separator is a set $Y \subseteq V(G) \setminus X$ that ensures that any component $C \in \mathcal{C}(G - (X \cup Y))$ will have $N_G(C) \cap X$ being a limit-1 subset for $(G - C, t)$. The benefit of having a component separator $Y$ alongside a tidy modulator $X$ is that when we reduce $(G,t)$ by making a modification in $C$, that a solution $S'$ for the reduced problem instance $(G',t)$ has $\treewidth{N_G[C] \setminus S'} \leq 2$.

The final step is the framework is the shrinking step, which entails the process of applying reduction rules on components in $G - (X \cup Y)$ until those components are of a bounded size. This step will be discussed in Chapters \ref{ch:reducing-biconnected-subgraphs} and \ref{ch:reducing-block-cut-trees}.

\section{Approximate Modulators}
Gupta et al. \cite{Losing-Treewidth-by-Separating-Subsets} have defined a $O(\log \eta)$-approximation algorithm for the \textsc{Treewidth}-$\eta$ \textsc{Vertex Deletion} problem. We will use this algorithm to obtain a modulator $X$ of $G$ with size $O(t)$. An important aspect to note, however, is that we will not derive the exact approximation ratio $\varepsilon = O(\log 2) = O(1)$. Careful analysis of the algorithm by Gupta et al. \cite{Losing-Treewidth-by-Separating-Subsets} would provide this approximation ratio. Nevertheless, this analysis is beyond the scope of this thesis. Furthermore, any polynomial time constant factor approximation algorithm for the \textsc{Treewidth-2 Vertex Deletion} problem could be used as a subroutine. This means that a kernel bound that abstract over $\varepsilon$ directly translates to tighter kernel bounds when algorithms with better approximation ratios are discovered.

\begin{lemma}\label{lemma:apply-approximation-algorithm}
    Given a polynomial time $\varepsilon$-approximation algorithm $\mathcal{A}$ for the \textsc{Treewidth-2 Vertex Deletion} problem and a problem instance $(G,t)$. We can in $\polyG$ time obtain a modulator $X$ of $G$ with $t < |X| \leq \varepsilon t$ or a problem instance $(G',t')$ with $\twtwodeletion{G'}{t'} \Longleftrightarrow \twtwodeletion{G}{t}$ and $|V(G')| \leq 4 \wedge t' \leq t$.
\end{lemma}

\begin{proof}
    Let $k$ be the minimum size of any modulator of $G$. Algorithm $\mathcal{A}$ gives in $\polyG$ time a modulator $X$ of $G$ with $k \leq |X| \leq \varepsilon k$.
    
    \begin{case}[$|X| \leq t$]
        We can apply Reduction \ref{red:solution-is-known} on $(G,t)$ and $X$. Correctness then immediately follows from the definition of Reduction \ref{red:solution-is-known} and Lemma \ref{lemma:safeness:solution-is-known}.
    \end{case}
    
    \begin{case}[$|X| > \varepsilon t$]
        We have $\varepsilon t < |X| \leq \varepsilon k$, which implies $k > t$. Hence $\nottwtwodeletion{G}{t}$ holds. In this case we can apply Reduction \ref{red:no-existing-solution} on $(G,t)$. Correctness then immediately follows from the definition of Reduction \ref{red:no-existing-solution} and Lemma \ref{lemma:safeness:no-existing-solution}.
    \end{case}
    
    \begin{case}[$t < |X| \leq \varepsilon t$]
        Returning $X$ suffices.
    \end{case}
    
    \noindent
    Because these cases cover all cases, correctness trivially follows.
\end{proof}

\begin{lemma}\label{lemma:existance-approximation-algorithm}
    There exists a polynomial $O(1)$-approximation algorithm for the \textsc{Treewidth-2 Vertex Deletion} problem.
\end{lemma}

\begin{proof}
    By Corollary 3 in \cite{Losing-Treewidth-by-Separating-Subsets} we have that there exists a $O(\log \eta)$-approximation algorithm for the \textsc{Treewidth}-$\eta$ \textsc{Vertex Deletion} problem, which takes $2^{O(\eta^3 \log^2 \eta)} \cdot n \log n + n^{O(1)}$ time on a graph $G$ of size $n$. With $\eta = 2$ we obtain a polynomial time $O(1)$-approximation algorithm for the \textsc{Treewidth-2 Vertex Deletion} problem.
\end{proof}

\section{Tidy Modulators}\label{sec:graph-decompositions:tidy-modulators}
Once given a modulator $X$ for graph $G$ we wish to grow $X$ into a modulator $Y$ such that the removal of any vertex $y \in Y$ results in $Y \setminus \{y\}$ being a modulator for $G$. The property that $\treewidth{G - (Y \setminus \{y\})} \leq 2$ holds for any $y \in Y$ is referred to as $Y$ being tidy modulator for $G$.

\begin{definition}[tidy modulator]\label{def:tidy-modulator}
    A modulator $X$ of $G$ is a tidy modulator of $G$ when for each $x \in X$ we have $\treewidth{G - (X \setminus \{x\})} \leq 2$.
\end{definition}

\noindent
To find a tidy modulator we will run Algorithm \ref{algo:find-disjoint-modulator} once for each vertex $x \in X$. I.e., for each $x \in X$ we will run \textsc{Find-Disjoint-Modulator}($G - (X \setminus \{x\}), t, \{x\}$). This algorithm will provide for any vertex $x \in X$ either one of the following: a proof that $\{x\}$ is a limit-0 subset for $(G,t)$ or a set $Y_x \subseteq V(G) \setminus X$ with $|Y_x| \leq 3t+3$ and $\treewidth{G - ((X \cup Y_x) \setminus \{x\})} \leq 2$.

We first introduce a reduction rule to employ when $\{x\}$ is a limit-0 subset for $(G,t)$. Then afterwards we will provide a method to combine each set $Y_x$ into a tidy modulator $Y$ of $G$.

We note that when a set $U$ is a limit-0 subset for $(G,t)$ that any solution $S$ for $(G,t)$ must contain all vertices in $U$. Once we have identified such a limit-0 subset $U$ we can remove all vertices $U$ from $G$ and focus on solving the problem for $(G - U, t - |U|)$. We formalise this in the following reduction rule.

\begin{reduction}[remove limit-0 subset]\label{red:remove-limit-0-subset}
    Given a problem instance $(G,t)$ and a non-empty limit-0 subset $U$ of $(G,t)$. Remove $U$ from $G$ and decrease $t$ by $|U|$.
\end{reduction}

\begin{lemma}[safeness]\label{lemma:safeness:remove-limit-0-subset}
    Let $(G',t')$ be the problem instance obtained by applying Reduction \ref{red:remove-limit-0-subset} on $(G,t)$. Then $\twtwodeletion{G}{t} \Longleftrightarrow \twtwodeletion{G'}{t'}$ holds.
\end{lemma}

\begin{proof}
    $\twtwodeletion{G'}{t'}$ holds if and only if there exists a solution $S'$ for $(G',t')$. We define $S = S' \cup U$ and note $|S| = |S'| + |U|$. We have $\treewidth{G - S} = \treewidth{G' - S'}$, which means that $S$ is a solution for $(G,t)$ if and only if $S'$ is a solution for $(G',t')$. Because $\twtwodeletion{G}{t}$ holds if and only if there exists a solution $S$ for $(G,t)$ we conclude that $\twtwodeletion{G}{t} \Longleftrightarrow \twtwodeletion{G'}{t'}$ holds.
\end{proof}

\begin{lemma}\label{lemma:find-tidy-modulator}
    Given a non-trivial problem instance $(G,t)$ and a modulator $X$ of $G$. We can in $\polyG$ time apply Reduction \ref{red:remove-limit-0-subset} or obtain a tidy modulator $Y$ of $G$ with $|Y| \leq (3t+4)|X|$.
\end{lemma}

\begin{proof}
     In case $|V(G)| \leq (3t+4)|X|$ returning $V(G)$ trivially suffices. Hence, we assume $|V(G)| > (3t+4)|X|$. Because $(G,t)$ is non-trivial we have $\treewidth{G} > 2$, which implies $X \neq \emptyset$. Hence, $|V(G)| > (3t+4)|X| \geq 3t+3+|X|$ must hold.

    For each $x \in X$ we run \textsc{Find-Disjoint-Modulator}($G - (X \setminus \{x\}), t, \{x\}$). By Theorem \ref{theorem:find-disjoint-modulator} we obtain in $\polyG$ time a set $Y_x \subseteq V(G) \setminus X$ with $|Y_x| \leq 3t+3$, such that either $\treewidth{G - ((X \setminus \{x\}) \cup Y_x)} \leq 2$ or every $C \in \mathcal{C}(G - (X \cup Y_x))$ has that $\{x\}$ is a limit-0 clique for $(G - ((X \setminus \{x\}) \cup V(C)), t)$.
    
    In case for some $x \in X$ we obtain a set $Y_x$ such that every $C \in \mathcal{C}(G - (X \cup Y_x))$ has $\{x\}$ being a limit-0 clique for $(G - ((X \setminus \{x\}) \cup V(C)), t)$. Then, because $|V(G) \setminus (X \cup Y_x)| \geq |V(G)| - |X| - |Y_x| > 3t + 3 + |X| - |X| - |Y_x| \geq 0$ we have that there exists a component $C \in \mathcal{C}(G - (X \cup Y_x))$ for which $\{x\}$ is a limit-0 clique for $(G - ((X \setminus \{x\}) \cup V(C)), t)$. Then by Lemma \ref{lemma:limit-m-of-subgraph-to-supergraph} we know that $\{x\}$ is a limit-0 clique for $(G,t)$, which implies that we can apply Reduction \ref{red:remove-limit-0-subset} on $\{x\}$.
    
    In case for each $x \in X$ we obtain a set $Y_x$ for which $\treewidth{G - ((X \setminus \{x\}) \cup Y_x)} \leq 2$ holds, then let $Y = X \cup \bigcup_{x \in X} Y_x$. For any vertex $y \in Y \setminus X$ we have by Lemma \ref{lemma:minor-treewidth} that $\treewidth{G - (Y \setminus \{y\})} \leq \treewidth{G - X} \leq 2$. For any vertex $y \in Y \cap X$ we have by Lemma \ref{lemma:minor-treewidth} that $\treewidth{G - (Y \setminus \{y\})} \leq \treewidth{G - ((X \setminus \{x\}) \cup Y_x)} \leq 2$. Hence, $Y$ is a tidy modulator of $G$ with $|Y| \leq |X| + (3t+3)|X| = (3t+4)|X|$.
\end{proof}

\section{Pre-Separation Reductions}
After having obtained a tidy modulator $X$ of $G$ we wish to obtain a component separator. To obtain a component separator of size $\poly{t}$ we first need to introduce two reduction rules.

\subsection{Add Necessary Edge}\label{sec:graph-decompositions:pre-separation:add-necessary-edge}
We use the notion of a necessary pair as introduced by Bodlaender \cite{Bodlaender_Necessary_Edges_in_k-Chordalisations_of_Graphs}. A vertex pair $\{x,y\} \subseteq V(G)$ is a necessary pair for treewidth $\eta$ if every width-$\eta$ tree decomposition of $G$ will have some bag that contains both $x$ and $y$. We note that any two vertices which are connected via an edge are a necessary pair, as follows from Definition \ref{def:tree-decomposition}. The reduction rule we introduce adds an edge between two non-adjacent vertices in a necessary pair for treewidth $2$. The addition of an edge ensures that a necessary pair in a graph $G$ will remain being a necessary pair in a reduced graph $G'$ obtained by a reduction that does not involve the necessary pair itself.

In Lemma 5 in the paper by Bodlaender \cite{Bodlaender_Necessary_Edges_in_k-Chordalisations_of_Graphs} it is proven that the addition of an edge between two distinct and independent vertices $x$ and $y$ that have at least $\eta+1$ internally vertex-disjoint paths connecting them will result in a graph $G'$ with treewidth at most $\eta$ if and only if $G$ had treewidth at most $\eta$ as well. We generalise this for our problem by requiring there to be at least $t+3$ internally vertex-disjoint paths. This, because regardless of the choice of a solution $S$ for $(G,t)$, either $S \cap \{x,y\} \neq \emptyset$ holds or $G - S$ will contain at least three internally vertex-disjoint paths between $x$ and $y$.

\begin{reduction}[add necessary edge]\label{red:add-necessary-edge}
	Given a problem instance $(G,t)$ and vertices $\{x,y\} \subseteq V(G)$ with $xy \notin E(G)$. If there exist $t+3$ internally vertex-disjoint paths in $G$ between $u$ and $v$, then add edge $uv$ to $G$.
\end{reduction}

\begin{lemma}[safeness]\label{lemma:safeness:add-necessary-edge}
	Let $(G',t)$ be the problem instance obtained by applying Reduction \ref{red:add-necessary-edge} on $(G,t)$. Then $\twtwodeletion{G}{t} \Longleftrightarrow \twtwodeletion{G'}{t}$ holds.
\end{lemma}

\begin{proof}
    Because $G$ is a minor of $G'$ we obtain by Lemma \ref{lemma:minor-of-G-implication} that $\twtwodeletion{G'}{t}$ implies $\twtwodeletion{G}{t}$. It remains to prove that $\twtwodeletion{G}{t}$ implies $\twtwodeletion{G'}{t}$. Let $S$ be a solution for $(G,t)$.
    
    \begin{case}[$S \cap \{x,y\} \neq \emptyset$]
        Because $G' - S = G - S$ we have $\twtwodeletion{G'}{t}$.
    \end{case}

    \begin{case}[$S \cap \{x,y\} = \emptyset$]
        Because $|S| \leq t$ there exist at least three internally-vertex disjoint paths in $G - S$ between $u$ and $v$. Hence, by Lemma 3 in \cite{Bodlaender_Necessary_Edges_in_k-Chordalisations_of_Graphs}, we have that $\{x,y\}$ is a necessary pair for treewidth 2 in $G - S$. This means that a width-2 tree decomposition of $G - S$ will contain a bag that contains both $x$ and $y$. It trivially follows such a tree decomposition is a width-2 tree decomposition of $G' - S$ as well, which implies $\twtwodeletion{G'}{t}$.
    \end{case}
    
    \noindent
    Because these cases cover all cases, correctness trivially follows.
\end{proof}

\begin{lemma}\label{lemma:time:add-necessary-edge}
	We can detect whether Reduction \ref{red:add-necessary-edge} can be applied on a vertex pair $\{u,v\} \subseteq V(G)$ and apply this rule in $\polyG$ time.
\end{lemma}

\begin{proof}
	Directly follows from Observation \ref{obs:get-max-internally-vertex-disjoint-paths}.
\end{proof}

\noindent
Although this reduction rule can be applied on any vertex pair $\{u,v\} \subseteq V(G)$, we choose to only apply it on vertex pairs $\{x,y\} \subseteq X$ where $X$ is a tidy modulator of $G$. The reason for this is that we will provide multiple reduction rules on vertex pairs $\{u,v\} \not\subseteq X$ that remove edges between $u$ and $v$. If, however, we would always exhaustively apply Reduction \ref{red:add-necessary-edge} on any vertex pair $\{u,v\} \subseteq V(G)$, then this would undo the progress made by these other reduction rules.

To avoid repeatedly phrasing that Reduction \ref{red:add-necessary-edge} is exhaustively applied on all vertex pairs in a (tidy) modulator $X$ we introduce the notion of a linked (tidy) modulator.

\begin{definition}[linked set]\label{def:linked-set}
    A set $X \subseteq V(G)$ is linked for $(G,t)$ when for each $\{x,y\} \subseteq X$ with $xy \notin E(G)$ there are at most $t+2$ internally vertex-disjoint paths in $G$ between $x$ and $y$.
\end{definition}

\begin{lemma}\label{lemma:get-linked-modulator}
    Let $(G,t)$ be a problem instance and let $X$ be a (tidy) modulator of $G$. We can in $\polyG$ time obtain a problem instance $(G',t)$ with $V(G') = V(G) \wedge E(G' - X) = E(G - X)$ and $\twtwodeletion{G'}{t} \Longleftrightarrow \twtwodeletion{G}{t}$, such that $X$ is a linked (tidy) modulator of $(G',t)$.
\end{lemma}

\begin{proof}
    While $X$ is not linked for $(G,t)$, it follows from Lemma \ref{lemma:time:add-necessary-edge}, that we can in $\polyG$ time apply Reduction \ref{red:add-necessary-edge} on a pair $\{x,y\} \subseteq X$. We can add at most $\binom{|X|}{2}$ edges until $G[X]$ is a clique, for which $X$ is guaranteed to be linked. Hence, in $\polyG$ time we obtain a problem instance $(G',t)$ with $V(G') = V(G) \wedge E(G' - X) = E(G - X)$ such that $X$ is linked for $(G',t)$. By Lemma \ref{lemma:safeness:add-necessary-edge} we have $\twtwodeletion{G'}{t} \Longleftrightarrow \twtwodeletion{G}{t}$.
    
    Because we only add edges between pairs $\{x,y\} \subseteq X$ we have $G - X = G' - X$, which implies that $X$ is a linked modulator for $(G',t)$. Similarly, we have for each $x \in X$ that $G - (X \setminus \{x\}) = G' - (X \setminus \{x\})$, which implies that if $X$ is a tidy modulator for $G$ that then $X$ is a linked tidy modulator for $(G',t)$.
\end{proof}

\noindent
The main purpose of using linked modulators is that it will allow us to identify limit-$m$ cliques in modulators, opposed to only limit-$m$ subsets.

\begin{lemma}\label{lemma:identify-clique}
    Let $(G,t)$ be a problem instance and let $X \subseteq V(G)$ be a linked set for $(G,t)$. If for each $\{x,y\} \subseteq X$ there exist $t+3$ vertex-disjoint connected subgraphs $G_i$ of $G$ with $x,y \in N_G(G_i)$, then $G[X]$ is a clique.
\end{lemma}

\begin{proof}
    Each graph $G_i$ yields a path $P_i$ in $G[V(G_i) \cup \{x,y\}]$ from $x$ to $y$, such that all paths $P_1, P_2, \dots, P_{t+3}$ are internally vertex-disjoint. By Definition \ref{def:linked-set} this yields $xy \in E(G)$. Because this holds for all $\{x,y\} \subseteq X$ we have that $G[X]$ is a clique.
\end{proof}

\subsection{Reduce Number of Components}
The following reduction rule allows us to bound the number of components in $G - X$, where $X$ is a tidy modulator of $G$. After exhaustively applying this reduction rule on $(G,t)$ and a linked tidy modulator $X$ of $(G,t)$, we will be able to find a component separator $Y$ of size bounded by $\poly{t}$.

\begin{reduction}[reduce number of components]\label{red:reduce-number-of-components}
    Let $(G,t)$ be a problem instance, let $X$ be a tidy modulator of $G$, and let $C \in \mathcal{C}(G - X)$. Remove $C$ from $G$ if all of the following hold:
    \begin{enumerate}
        \item $N_G(C)$ is a limit-2 clique for $(G-C,t)$
        \item for each $\{x,y\} \subseteq N_G(C)$ at least one of the following holds: 
        \begin{enumerate}
            \item $\treewidth{G[V(C) \cup \{x,y\}]} \leq 2$
            \item $\{x,y\}$ is a limit-1 subset for $(G-C,t)$
        \end{enumerate}
    \end{enumerate}
\end{reduction}
\medskip

\begin{lemma}[safeness]\label{lemma:safeness:reduce-number-of-components}
	Let $(G',t)$ be the problem instance obtained by applying Reduction \ref{red:reduce-number-of-components} on $(G,t)$. Then $\twtwodeletion{G}{t} \Longleftrightarrow \twtwodeletion{G'}{t}$ holds.
\end{lemma}

\begin{proof}
    Because $G'$ is a minor of $G$ we have by Lemma \ref{lemma:minor-of-G-implication} that $\twtwodeletion{G}{t}$ implies $\twtwodeletion{G'}{t}$. It remains to prove that $\twtwodeletion{G'}{t}$ implies $\twtwodeletion{G}{t}$. Let $S'$ be a solution for $(G',t)$.
    
    We first prove $\treewidth{G[N_G[C] \setminus S']} \leq 2$. Assume per contradiction $\treewidth{G[N_G[C] \setminus S']} > 2$. Because $N_G(C)$ is a limit-2 clique for $(G',t)$ we know $|N_G(C) \setminus S'| \leq 2$. We must have $|N_G(C) \setminus S'| = 2$, because $|N_G(C) \setminus S'| \leq 1$ yields by Lemma \ref{lemma:minor-treewidth} and Definition \ref{def:tidy-modulator} that $\treewidth{G[N_G[C] \setminus S']} \leq \treewidth{G - (X \setminus (N_G(C) \setminus S'))} \leq 2$. Let $\{x,y\} = N_G(C) \setminus S'$. We have $\treewidth{G[N_G[C] \setminus S']} = \treewidth{G[V(C) \cup \{x,y\}]} > 2$, which by definition of Reduction \ref{red:reduce-number-of-components} yields that $\{x,y\}$ is a limit-1 subset for $(G',t)$. This, however, contradicts $\{x,y\} \cap S' = \emptyset$. Hence, we have $\treewidth{G[N_G[C] \setminus S']} \leq 2$.
    
    Because $G[N_G(C) \setminus S']$ is a clique we can apply Theorem \ref{theorem:connect-treewidth-graphs} on $(G - S', G - (V(C) \cup S'))$ to derive $\treewidth{G - S'} \leq 2$. Hence, we conclude $\twtwodeletion{G}{t}$.
\end{proof}

\noindent
We next prove that when $G - X$ contains many components, that then we are able to remove one using Reduction \ref{red:reduce-number-of-components}. To prove this we will construct sets of components $\mathcal{P}$, $\mathcal{Q}$, and $\mathcal{R}$ that together ensure that any component in $G - X$ not in any of these sets will satisfy all conditions for Reduction \ref{red:reduce-number-of-components}.

For notation convenience we define function $\nbound{C} : \nat^2 \rightarrow \nat$ such that $|\mathcal{C}(G - X)| > \Cbound{t}{|X|}$ implies that we can apply Reduction \ref{red:reduce-number-of-components}.

\begin{definition}\label{def:Cbound}
    \[\Cbound{t}{x} = \binom{x}{3}(t+1) + \binom{x}{2}(2t+3)\]
\end{definition}

\begin{lemma}\label{lemma:partial-application:reduce-number-of-components}
    Let $(G,t)$ be a problem instance and let $X$ be a linked tidy modulator of $(G,t)$. If $|\mathcal{C}(G - X)| > \Cbound{t}{|X|}$, then in $\polyG$ time we can apply Reduction \ref{red:reduce-number-of-components}.
\end{lemma}

\begin{proof}
    For each $\{x,y\} \subseteq X$ we let $\mathcal{P}_{x,y} = \setdef{C \in \mathcal{C}(G - X)}{x,y \in N_G(C)}$ and we let $\mathcal{P} = \setdef{C \in \mathcal{P}_{x,y}}{|\mathcal{P}_{x,y}| \leq t+2}$. For each $C \in \mathcal{C}(G - X) \setminus \mathcal{P}$ and $\{x,y\} \subseteq N_G(C)$ there exist at least $t+3$ components $C' \in \mathcal{C}(G - X)$ with $\{x,y\} \subseteq N_G(C')$. By Lemma \ref{lemma:identify-clique} we therefore have that $G[N_G(C)]$ is a clique.
    
    For each $\{x,y,z\} \subseteq X$ we let $\mathcal{Q}_{x,y,z} = \setdef{C \in \mathcal{C}(G - X) \setminus \mathcal{P}}{x,y,z \in N_G(C)}$ and we let $\mathcal{Q} = \setdef{C \in \mathcal{Q}_{x,y,z}}{|\mathcal{Q}_{x,y,z}| \leq t+1}$. For each $C \in \mathcal{C}(G - X) \setminus (\mathcal{P} \cup \mathcal{Q})$ and $\{x,y,z\} \subseteq N_G(C)$ there exist at least $t+1$ components $C' \in \mathcal{C}((G - C) - X) \setminus \mathcal{P}$ with $x,y,z \in N_G(C')$. For each such $C'$ we have that $G[N_{G-C}(C')]$ is a clique, which implies $\treewidth{(G-C)[V(C') \cup \{x,y,z\}]} > 2$. By Lemma \ref{lemma:limit-m-subset-identification} we have that $\{x,y,z\}$ is a limit-2 subset for $(G-C,t)$. By Lemma \ref{lemma:limit-m-subset-to-superset} we have that $N_G(C)$ is a limit-2 subset for $(G-C, t)$. Hence, for each $C \in \mathcal{C}(G - X) \setminus (\mathcal{P} \cup \mathcal{Q})$ property 1 of Reduction \ref{red:reduce-number-of-components} holds.
    
    For each $\{x,y\} \subseteq X$ we let $\mathcal{R}_{x,y} = \setdef{C \in \mathcal{C}(G - X)}{\treewidth{G[V(C) \cup \{x,y\}]} > 2}$ and we let $\mathcal{R} = \setdef{C \in \mathcal{R}_{x,y}}{|\mathcal{R}_{x,y}| \leq t+1}$. For each $C \in \mathcal{C}(G - X) \setminus \mathcal{R}$ and $\{x,y\} \subseteq N_G(C)$ we have $\treewidth{G[V(C) \cup \{x,y\}]} \leq 2$ or there exist at least $t+1$ components $C' \in \mathcal{C}((G - C) - X)$ with $\treewidth{G[V(C') \cup \{x,y\}]} > 2$. In case of the former property 2a of Reduction \ref{red:reduce-number-of-components} holds. In case of the latter we have by Lemma \ref{lemma:limit-m-subset-identification} that property 2b of Reduction \ref{red:reduce-number-of-components} holds. Hence, on each component $C \in \mathcal{C}(G - X) \setminus (\mathcal{P} \cup \mathcal{Q} \cup \mathcal{R})$ we can apply Reduction \ref{red:reduce-number-of-components}.
    
    We note that we can trivially construct sets $\mathcal{P}$ and $\mathcal{Q}$ in $\polyG$ time. From Lemma \ref{lemma:find-tree-decompositions} it also follows that we can construct $\mathcal{R}$ in $\polyG$ time. From the definitions of $\mathcal{P}$, $\mathcal{Q}$, and $\mathcal{R}$ it trivially follows that $|\mathcal{P} \cup \mathcal{Q} \cup \mathcal{R}| \leq \binom{|X|}{2}(t+2) + \binom{|x|}{3}(t+1) + \binom{|X|}{2}(t+1) = \Cbound{t}{|X|}$. Because $|\mathcal{C}(G - X)| > \Cbound{t}{|X|}$ there exists at least one component $C \in \mathcal{C}(G - X) \setminus (\mathcal{P} \cup \mathcal{Q} \cup \mathcal{R})$ on which we can apply Reduction \ref{red:reduce-number-of-components}.
\end{proof}

\noindent
We will also want to apply Lemma \ref{lemma:partial-application:reduce-number-of-components} on non-linked tidy modulators. Although we do not wish to add edges that undo progress made by other reduction rules, we do permit this behaviour if the resulting graph is guaranteed to have strictly fewer vertices. This, because we will not introduce any reduction rules that increase the number of vertices.

\begin{lemma}\label{lemma:application:reduce-number-of-components}
    Let $(G,t)$ be a problem instance and let $X$ be a tidy modulator of $G$. If $|\mathcal{C}(G - X)| > \Cbound{t}{|X|}$, then in $\polyG$ time we can obtain a problem instance $(G',t)$ with $\twtwodeletion{G}{t} \Longleftrightarrow \twtwodeletion{G'}{t}$ and $|V(G')| < |V(G)|$.
\end{lemma}

\begin{proof}
    By Lemma \ref{lemma:get-linked-modulator} we can in $\polyG$ time obtain a problem instance $(G',t)$ with $V(G') = V(G) \wedge E(G' - X) = E(G - X)$ and $\twtwodeletion{G'}{t} \Longleftrightarrow \twtwodeletion{G}{t}$, such that $X$ is a linked tidy modulator of $(G',t)$. Because $V(G') = V(G) \wedge E(G' - X) = E(G - X)$ we have $|\mathcal{C}(G' - X)| = |\mathcal{C}(G - X)| > \Cbound{t}{|X|}$. By Lemma \ref{lemma:partial-application:reduce-number-of-components} we can therefore apply in $\poly{|G'|} = \polyG$ time Reduction \ref{red:reduce-number-of-components} on $(G',t)$. From the definition of Reduction \ref{red:reduce-number-of-components} and Lemma \ref{lemma:safeness:reduce-number-of-components} it follows that we obtain a problem instance $(G'',t)$ with $\twtwodeletion{G'}{t} \Longleftrightarrow \twtwodeletion{G''}{t}$ and $|V(G'')| < |V(G')|$. Correctness trivially follows.
\end{proof}

\section{Component Separators}\label{sec:graph-decompositions:component-separators}
Once given a linked tidy modulator $X$ of $(G,t)$ on which Lemma \ref{lemma:partial-application:reduce-number-of-components} does not provide an application of Reduction \ref{red:reduce-number-of-components} we will want to obtain a set $Y \subseteq V(G) \setminus X$ of $\poly{t}$ size, such that any component in $G - (X \cup Y)$ will have few neighbours in $X$. Ideally we would obtain a situation similar to the outerplanar decomposition introduced by Donkers et al. \cite{Preprocessing_for_Outerplanar_Vertex_Deletion_An_Elementary_Kernel_of_Quartic_Size}, where each component in $G - (X \cup Y)$ only has a single neighbour in $X$. This would be desirable, because it results in limited interaction between a tidy modulator and the components outside of this tidy modulator, making it significantly easier to reduce components outside this tidy modulator.

Although we are unable to construct such a set $Y$, we are able to construct a set $Y$ that guarantees that any solution $S$ for $(G,t)$ will leave each component in $C \in \mathcal{C}(G - (X \cup Y))$ with at most one neighbour $x \in (N_G(C) \cap X) \setminus S$. I.e., the neighbourhood of $C$ in $X$ will be a limit-1 subset in $(G,t)$.

We will construct a set $Y$ with a stronger guarantee, being that $N_G(C) \cap X$ will remain being a limit-1 subset after the removal of $C$ from $G$. The reason for strengthening this requirement is that when we want to prove safeness of reduction rules (that modify $C$), that then we will need to prove that the existence of a solution $S'$ for $(G',t)$ implies the existence of a solution $S$ for $(G,t)$. By only providing the guarantee that $N_G(C) \cap X$ is a limit-1 subset for $(G,t)$ we would not guarantee that $N_G(C) \cap X$ is a limit-1 subset for $(G',t)$, which is an essential property in proving the correctness of our reduction rules. The requirement that $N_G(C) \cap X$ is a limit-1 subset for $(G - C,t)$ will allow us to prove that $N_G(C) \cap X$ is a limit-1 subset for $(G',t)$ as well.

\begin{definition}[component separator]\label{def:component-separator}
    A vertex set $Y \subseteq V(G) \setminus X$ is a component separator of $(G,t,X)$ when each $C \in \mathcal{C}(G - (X \cup Y))$ has that $N_G(C) \cap X$ is a limit-1 subset for $(G-C,t)$.
\end{definition}

\noindent
To construct a component separator we will want to find for each $\{x,y\} \subseteq X$ either a set that separates $x$ and $y$ in $G - (X \setminus \{x,y\})$ or a set that separates $G - X$ into multiple components of which any single one being removed will still yield $\{x,y\}$ being a limit-1 subset. We will find these sets for a pair $\{x,y\} \subseteq X$ using different methods dependent on whether $xy \in E(G)$ holds.

In case $xy \notin E(G)$ holds we have, due to having a linked modulator $X$, that there will be at most $t+2$ internally vertex-disjoint paths between $x$ and $y$. By Menger's theorem (see Observation \ref{obs:get-min-separator}) we then know that we can in $\polyG$ time find a separator of $x$ and $y$ of size $t+2$.

In case $xy \in E(G)$ we will run \textsc{Find-Disjoint-Modulator}($G - (X \setminus \{x,y\}), t, \{x,y\}$). This yields a set $Y \subseteq V(G) \setminus X$ with $|Y| \leq 3t+3$ such that either $\treewidth{(G - (X \setminus \{x,y\})) - Y} \leq 2$ or each $C \in \mathcal{C}((G - (X \setminus \{x,y\})) - Y)$ has $\{x,y\}$ being a limit-1 clique for $((G - (X \setminus \{x,y\})) - C,t)$. In case the latter holds we are done. In case the former holds we will prove that we can find a separator for $x$ and $y$ of size $\poly{t}$.

\begin{lemma}\label{lemma:component-separator:without-edge}
    Let $(G,t)$ be a problem instance, let $X$ be a linked modulator of $(G,t)$, and let $\{x,y\} \subseteq X$ with $xy \notin E(G)$. In $\polyG$ time we can construct a set $Y \subseteq V(G) \setminus X$ with $|Y| \leq t+2$ such that each component $C \in \mathcal{C}(G - (X \cup Y))$ has $N_G(C) \cap \{x,y\}$ being a limit-1 subset for $(G-C,t)$.
\end{lemma}

\begin{proof}
    Because $xy \notin E(G[X])$ and $X$ is linked for $(G,t)$ there exist at most $t+2$ internally vertex-disjoint paths in $G$ between $x$ and $y$. Hence, there exist at most $t+2$ such paths in $G - (X \setminus \{x,y\})$. By Observation \ref{obs:get-min-separator} we can in $\polyG$ time obtain a separator $Y \subseteq V(G) \setminus X$ for $x$ and $y$ in $G - (X \setminus \{x,y\})$ with $|Y| \leq t+2$. It then holds for each $C \in \mathcal{C}(G - (X \cup Y))$ that $|N_G(C) \cap \{x,y\}| \leq 1$, which trivially yields $N_G(C) \cap \{x,y\}$ being a limit-1 subset for $(G-C,t)$.
\end{proof}

\begin{lemma}\label{lemma:component-separator:with-edge}
    Let $(G,t)$ be a problem instance, let $X$ be a linked modulator of $(G,t)$, and let $\{x,y\} \subseteq X$ with $xy \in E(G)$. In $\polyG$ time we can obtain a problem instance $(G',t)$ with $\twtwodeletion{G}{t} \Longleftrightarrow \twtwodeletion{G'}{t}$ and $|V(G')| < |V(G)|$ or construct a set $Y \subseteq V(G) \setminus X$ with $|Y| \leq 3t + 3 + \Cbound{t}{|X|+3t+3}$ such that each component $C \in \mathcal{C}(G - (X \cup Y))$ has $N_G(C) \cap \{x,y\}$ being a limit-1 subset for $(G-C,t)$.
\end{lemma}

\begin{proof}
    We run \textsc{Find-Disjoint-Modulator}($G - (X \setminus \{x,y\}), t, \{x,y\}$). By Theorem \ref{theorem:find-disjoint-modulator} we obtain in $\polyG$ time a set $Y' \subseteq V(G) \setminus X$ with $|Y'| \leq 3t+3$ such that either $\treewidth{(G - (X \setminus \{x,y\})) - Y'} \leq 2$ or every $C \in \mathcal{C}(G - (X \cup Y'))$ has $\{x,y\}$ being a limit-1 clique for $((G - (X \setminus \{x,y\}))-C,t)$.
    
    \begin{case}[$\treewidth{(G - (X \setminus \{x,y\})) - Y'} > 2$]
        Lemmas \ref{lemma:limit-m-of-subgraph-to-supergraph} and \ref{lemma:subset-of-limit-m-of-subset} prove that $Y = Y'$ suffices.
    \end{case}
    
    \begin{case}[$|\mathcal{C}(G - (X \cup Y'))| > \Cbound{t}{|X|+3t+3}$]
        Directly follows from Lemma \ref{lemma:application:reduce-number-of-components} and the observation that $X \cup Y'$ is a tidy modulator for $G$ of size at most $|X| + 3t + 3$.
    \end{case}
    
    \begin{case}[$\treewidth{(G - (X \setminus \{x,y\})) - Y'} \leq 2 \wedge |\mathcal{C}(G - (X \cup Y'))| \leq \Cbound{t}{|X|+3t+3}$]
        By Lemma \ref{lemma:minor-treewidth} we have that each $D \in \mathcal{C}(G - (X \cup Y'))$ has $\treewidth{G[V(D) \cup \{x,y\}]} \leq \treewidth{(G - (X \setminus \{x,y\})) - Y'} \leq 2$. We prove that there can not exist two internally vertex-disjoint paths between $x$ and $y$ in $G[V(D) \cup \{x,y\}] - xy$. Assume per contradiction that $P$ and $Q$ are two such paths. Because $D$ is connected there exists a path in $D$ between $V(P)$ and $V(Q)$. Let $R$ be a shortest path between $V(P)$ and $V(Q)$ in $D$. As shown in Figure \ref{fig:component-separator:with-edge}, we can apply Lemma \ref{lemma:describes-k4-minor} on $(P, (V(R) \setminus V(P)) \cup V(Q), x, y)$ to derive $\treewidth{G[V(D) \cup \{x,y\}]} > 2$, which is a contradiction. Hence, there can not exist two internally vertex-disjoint paths between $x$ and $y$ in $G[V(D) \cup \{x,y\}] - xy$.
        
        By Observation \ref{obs:get-min-separator} we can in $\polyG$ time obtain a set $Y_D \subseteq V(D)$ with $|Y_D| \leq 1$ that separates $x$ and $y$ in $G[V(D) \cup \{x,y\}] - xy$. Such a separator $Y_D$ has for each $C \in \mathcal{C}(D - Y_D)$ that $|N_G(C) \cap \{x,y\}| \leq 1$, which implies that $N_G(C) \cap \{x,y\}$ is a limit-1 subset for $(G-D,t)$.
    
        Let $Y = Y' \cup \bigcup_{D \in \mathcal{C}(G - (X \cup Y'))} Y_D$. We have $|Y| \leq 3t + 3 + \Cbound{t}{|X|+3t+3}$. Each $C \in \mathcal{C}(G - (X \cup Y))$ is a subgraph of some $D \in \mathcal{C}(G - (X \cup Y'))$, and for this component $D$ we then have $C \in \mathcal{C}(D - Y_D)$, which implies that $N_G(C) \cap \{x,y\}$ is a limit-1 subset for $(G-D,t)$. By Lemma \ref{lemma:limit-m-of-subgraph-to-supergraph} we then derive that $N_G(C) \cap \{x,y\}$ is a limit-1 subset for $(G-C,t)$. Hence, returning $Y$ suffices.
    \end{case}
    
    \noindent
    Because these cases cover all cases, correctness trivially follows.
\end{proof}

\begin{figure}[ht]
    \centering
    \includegraphics[scale=.9]{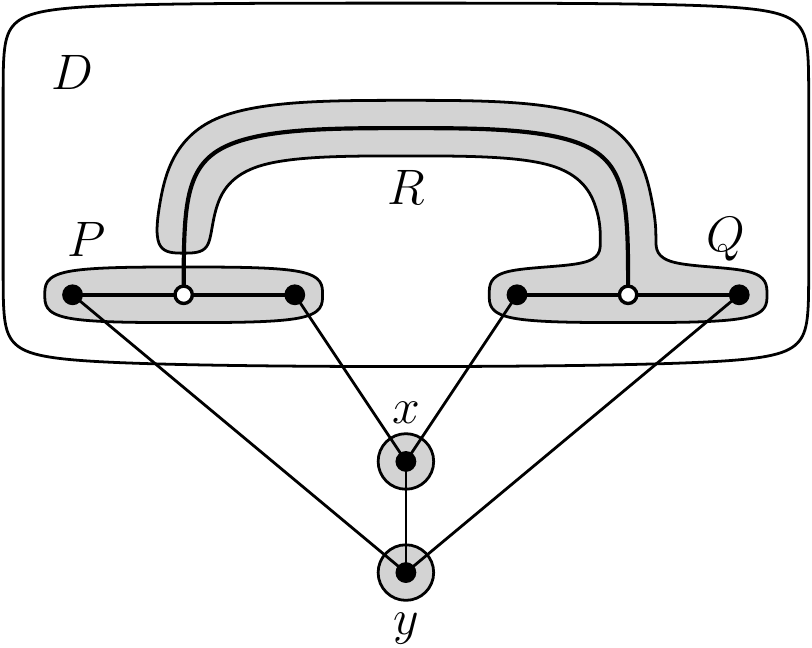}
    \caption{Two internally vertex-disjoint paths $P$ and $Q$ between $x$ and $y$ yields a $K_4$ minor}
    \label{fig:component-separator:with-edge}
\end{figure}

\begin{lemma}\label{lemma:get-component-separator}
    Let $(G,t)$ be a problem instance and let $X$ be a linked tidy modulator of $(G,t)$. In $\polyG$ time we can obtain a problem instance $(G',t)$ with $\twtwodeletion{G}{t} \Longleftrightarrow \twtwodeletion{G'}{t}$ and $|V(G')| < |V(G)|$ or construct a component separator $Y$ of $(G,t,X)$ with $|Y| \leq \binom{|X|}{2}(3t+3 + \Cbound{t}{|X|+3t+3})$.
\end{lemma}

\begin{proof}
    From Lemmas \ref{lemma:component-separator:without-edge} and \ref{lemma:component-separator:with-edge} it follows that in $\polyG$ time we can obtain either a problem instance $(G',t)$ with $\twtwodeletion{G}{t} \Longleftrightarrow \twtwodeletion{G'}{t}$ and $|V(G')| < |V(G)|$, or for each $\{x,y\} \subseteq X$ we obtain a set $Y_{x,y} \subseteq V(G) \setminus X$ with $|Y_{x,y}| \leq 3t + 3 + \Cbound{t}{|X|+3t+3}$ such that each component $C \in \mathcal{C}(G - (X \cup Y_{x,y}))$ has $N_G(C) \cap \{x,y\}$ being a limit-1 subset for $(G-C,t)$. In case the former holds correctness trivially follows. Hence, we assume the latter.

    Let $Y = \bigcup_{\{x,y\} \subseteq X} Y_{x,y}$. We have $|Y| \leq \binom{|X|}{2}(3t+3+\Cbound{t}{|X|+3t+3})$. It remains to prove that each component $C \in \mathcal{C}(G - (X \cup Y))$ has $N_G(C) \cap X$ being a limit-1 subset for $(G-C,t)$. Let $C \in \mathcal{C}(G - (X \cup Y))$, let $\{x,y\} \subseteq N_G(C) \cap X$, and let $D \in \mathcal{C}(G - (X \cup Y_{x,y}))$ be the component of which $C$ is a subgraph. By Lemma \ref{lemma:limit-m-of-subgraph-to-supergraph} we have that $N_G(D) \cap \{x,y\}$ is a limit-1 subset for $(G-C,t)$. By Lemma \ref{lemma:limit-m-subset-to-superset} we then have that $N_G(D) \cap X$ is a limit-1 subset for $(G-C,t)$. Hence, we conclude that $Y$ is a component separator for $(G,t,X)$.
\end{proof}

\section{Small \textit{Y}-Neighbourhood Component Separator}
Similar to the work by Donkers et al. \cite{Preprocessing_for_Outerplanar_Vertex_Deletion_An_Elementary_Kernel_of_Quartic_Size}, we wish to simplify the interface between a connected components $C \in \mathcal{C}(G - (X \cup Y))$ and component separator $Y$. To this end we will add additional vertices to $Y$, such that every component $C \in \mathcal{C}(G - (X \cup Y))$ will have at most four neighbours in $Y$.

We will apply Lemma 3.11 from \cite{Preprocessing_for_Outerplanar_Vertex_Deletion_An_Elementary_Kernel_of_Quartic_Size} to obtain this extended set $Y'$. We note that Lemma 3.11 in \cite{Preprocessing_for_Outerplanar_Vertex_Deletion_An_Elementary_Kernel_of_Quartic_Size} states that $G$ must be an outerplanar graph. The only property of outerplanarity that their proof employs is that $G$ has treewidth at most two. Hence, the weaker assumption that $G$ is a graph with treewidth at most two still leads to a valid proof.

\begin{lemma}[Lemma 3.11 in \cite{Preprocessing_for_Outerplanar_Vertex_Deletion_An_Elementary_Kernel_of_Quartic_Size}]\label{lemma:take-LCA-closure}
    Let $G$ be a graph with $\treewidth{G} \leq 2$ and let $Z \subseteq V(G)$. We can in $\polyG$ time find a set $Z' \subseteq V(G)$ with $|Z'| \leq 6|Z|$ and $Z \subseteq Z'$, such that each component $C \in \mathcal{C}(G - Z')$ has $|N_G(C)| \leq 4$.
\end{lemma}

\begin{lemma}\label{lemma:get-small-Y-neighbourhood-component-separator}
    Let $(G,t)$ be a problem instance, let $X$ be a tidy modulator of $G$, and let $Y$ be a component separator of $(G,t,X)$. We can in $\polyG$ time obtain a component separator $Y'$ of $(G,t,X)$ with $|Y'| \leq 6|Y|$, such that each $C \in \mathcal{C}(G - (X \cup Y'))$ has $|N_G(C) \cap Y'| \leq 4$.
\end{lemma}

\begin{proof}
    By Lemma \ref{lemma:take-LCA-closure} we can in $\polyG$ time obtain a set $Y' \subseteq V(G) \setminus X$ with $|Y'| \leq 6|Y|$ and $Y \subseteq Y'$, such that each component $C \in \mathcal{C}(G - (X \cup Y'))$ has $|N_G(C) \cap Y'| \leq 4$. It remains to prove that $Y'$ is a component separator of $(G,t,X)$. Let $C' \in \mathcal{C}(G - (X \cup Y'))$ and let $C \in \mathcal{C}(G - (X \cup Y))$ such that $C'$ is a subgraph of $C$. Because $V(C') \subseteq V(C)$ and $V(C) \cap X = \emptyset$ we have $N_G(C') \cap X \subseteq N_G(C) \cap X$. From Lemmas \ref{lemma:limit-m-of-subgraph-to-supergraph} and \ref{lemma:subset-of-limit-m-of-subset} it follows that $N_G(C') \cap X$ is a limit-1 subset for $(G - C',t)$. Hence, $Y'$ is a component separator of $(G,t,X)$.
\end{proof}

\section{Combining Results}
We finally combine all results from this section into two theorems.

\begin{theorem}\label{theorem:get-linked-tidy-modulator}
    Let $(G,t)$ be a non-trivial problem instance and let $\mathcal{A}$ be a polynomial time $\varepsilon$-approximation algorithm for the \textsc{Treewidth-2 Vertex Deletion} problem. In $\polyG$ time we can obtain one of the following:
    \begin{enumerate}
        \item A problem instance $(G',t')$ with $\twtwodeletion{G}{t} \Longleftrightarrow \twtwodeletion{G'}{t'}$ and $|V(G')| < |V(G)|$ and $t' \leq t$
        \item A problem instance $(G',t)$ and a linked tidy modulator $X$ for $(G',t)$ with $\twtwodeletion{G}{t} \Longleftrightarrow \twtwodeletion{G'}{t}$ and $V(G') = V(G)$ and $|X| \leq \varepsilon t (3t+4)$
    \end{enumerate}
\end{theorem}

\begin{proof}
    By Lemma \ref{lemma:apply-approximation-algorithm} we can in $\polyG$ time obtain a modulator $X$ of $G$ with $t < |X| \leq \varepsilon t$ or a problem instance $(G',t')$ with $\twtwodeletion{G'}{t'} \Longleftrightarrow \twtwodeletion{G}{t}$ and $|V(G')| \leq 4 \wedge t' \leq t$. In case the former holds correctness immediately follows from Definition \ref{def:trivial-problem-instance}. Hence we assume the latter.
    
    By Lemma \ref{lemma:find-tidy-modulator} we can in $\polyG$ time apply Reduction \ref{red:remove-limit-0-subset} or obtain a tidy modulator $X'$ of $G$ with $|X'| \leq (3t+4)|X| \leq \varepsilon t (3t+4)$. In case the former holds correctness immediately follows from the definition of Reduction \ref{red:remove-limit-0-subset} and Lemma \ref{lemma:safeness:remove-limit-0-subset}. In case the latter holds correctness immediately follows from Lemma \ref{lemma:get-linked-modulator}.
\end{proof}

\noindent
Similar to the definition of $\nbound{C}$ we define a function $\nbound{Y} : \nat^2 \rightarrow \nat$ such that any component separator $Y$ of $(G,t,X)$ obtained by our method will have $|Y| \leq \Ybound{t}{|X|}$.

\begin{definition}\label{def:Ybound}
    \[\Ybound{t}{x} = 6\binom{x}{2}(3t+3 + \Cbound{t}{x+3t+3})\]
\end{definition}

\begin{theorem}\label{theorem:get-component-separator}
    Let $(G,t)$ be a non-trivial problem instance and let $X$ be a linked tidy modulator for $(G,t)$. In $\polyG$ time we can obtain one of the following:
    \begin{enumerate}
        \item A problem instance $(G',t)$ with $\twtwodeletion{G}{t} \Longleftrightarrow \twtwodeletion{G'}{t}$ and $|V(G')| < |V(G)|$
        \item A component separator $Y$ for $(G,t,X)$ with $|Y| \leq \Ybound{t}{|X|}$, such that $|\mathcal{C}(G - (X \cup Y))| \leq \Cbound{t}{|X| + |Y|}$ and each $C \in \mathcal{C}(G - (X \cup Y))$ has $|N_G(C) \cap Y| \leq 4$
    \end{enumerate}
\end{theorem}

\begin{proof}
    By Lemma \ref{lemma:get-component-separator} we can in $\polyG$ time obtain a problem instance $(G',t)$ with $\twtwodeletion{G}{t} \Longleftrightarrow \twtwodeletion{G'}{t}$ and $|V(G')| < |V(G)|$ or construct a component separator $Y$ of $(G,t,X)$ with $|Y| \leq \binom{|X|}{2}(3t+3 + \Cbound{t}{|X|+3t+3})$. In case the former holds correctness trivially follows. Hence we assume the latter.
    
    By Lemma \ref{lemma:get-small-Y-neighbourhood-component-separator} we can in $\polyG$ time obtain a component separator $Y'$ with $|Y'| \leq 6|Y| \leq \Ybound{t}{|X|}$ and for which each $C \in \mathcal{C}(G - (X \cup Y'))$ has $|N_G(C) \cap Y'| \leq 4$. In case $|\mathcal{C}(G - (X \cup Y'))| > \Cbound{t}{|X| + |Y'|}$ correctness directly follows from Lemma \ref{lemma:application:reduce-number-of-components} and the observation that $X \cup Y'$ is a tidy modulator of $G$. Otherwise, returning $Y'$ suffices.
\end{proof}

\noindent
Theorems \ref{theorem:get-linked-tidy-modulator} and \ref{theorem:get-component-separator} now provide us with a method to decompose our graph $G$ into a linked tidy modulator $X$ for $(G,t)$, a component separator $Y$ for $(G,t,X)$, and a bounded set of components $\mathcal{C}(G - (X \cup Y))$ in which each component $C$ has $|N_G(C) \cap Y| \leq 4$. Because the number of components in $G - (X \cup Y)$ and the sizes $X$ and $Y$ are bounded by a function over $t$, it remains to reduce components $C \in \mathcal{C}(G - (X \cup Y))$ until these are of size bounded by a function over $t$ as well. Chapters \ref{ch:reducing-biconnected-subgraphs} and \ref{ch:reducing-block-cut-trees} will introduce methods to achieve this.

As indicated in Section \ref{sec:graph-decompositions:pre-separation:add-necessary-edge}, we only wish to apply Reduction \ref{red:add-necessary-edge} on vertex pairs in a tidy modulator $X$ for $(G,t)$, to avoid undoing progress made by reduction rules introduced in Chapters \ref{ch:reducing-biconnected-subgraphs} and \ref{ch:reducing-block-cut-trees}. Because our method of obtaining a component separator $Y$ for $(G,t,X)$ is dependent on $X$ being a linked tidy modulator for $(G,t)$, we have that reductions that do not decrease the number of vertices in $G$ must ensure that $X$ remains being a linked tidy modulator for reduced problem instance $(G',t)$. This, because otherwise we would need acquire a new linked tidy modulator for $(G',t)$, during which process previously made progress could be undone. As such, all reductions introduced in Chapters \ref{ch:reducing-biconnected-subgraphs} and \ref{ch:reducing-block-cut-trees} that remove only edges from $G$ will only be allowed to be applied on edges not in $G[X]$. Reductions that decrease the number of vertices in $G$ do not need to adhere to this constraint, because the process of obtaining a tidy modulator will never increase the number of vertices, and hence can not undo progress made by these reduction rules.

In Chapter \ref{ch:reducing-biconnected-subgraphs} we will focus on reducing biconnected subgraphs of $G - X$ that are large with respect to $|X|$. During this chapter we will not make use of component separators, nor the property that $X$ is a linked for $(G,t)$. Nevertheless, we still need to adhere to previously posed constraint, because the reductions introduced in Chapter \ref{ch:reducing-biconnected-subgraphs} will be applied in Chapters \ref{ch:reducing-block-cut-trees} and \ref{ch:combining-results} as well, in which case $X$ should remain being linked.

\chapter{Reducing Biconnected Subgraphs}\label{ch:reducing-biconnected-subgraphs}
Within this section we will focus on reducing large biconnected induced subgraphs $B$ in $G - X$, where $X$ will be a tidy modulator of $G$. The essence of our method is that we consider a tree decomposition of $B$, and in this tree decomposition find a large linear subtree (i.e. a subtree with two leaves). In this large linear subtree we will find a subtree, such that the vertices contained in all bags of this subtree do not have neighbours outside of $B$. The `structure' present within the graph induced by vertices from such bags is rather limited, which allows us to define reduction rules.

We note that a linear subtree simply corresponds to a path in a tree decomposition. The main reason why we will refer to them as linear subtrees, opposed to paths, is that in our search for these linear subtrees we will recurse on (non-linear) subtrees. By considering paths to be a type of tree it will allow us to express the same results using more compact and fewer definitions.

The tree decompositions that we will work with will have an additional smoothness property, as defined by Bodlaender \cite{Bodlaender_linear-time_algorithm_for_finding_tree-decomposition_of_small_treewidth}.

\begin{definition}[smooth tree decomposition]\label{def:smooth-tree-decomposition}
    A tree decomposition $(T,\chi)$ of width $\eta$ is smooth if for all $b \in V(T)$ it holds that $|\chi(b)| = \eta+1$ and each $ab \in E(T)$ has $|\chi(a) \cap \chi(b)| = \eta$.
\end{definition}

\begin{lemma}\label{lemma:get-smooth-tree-decomposition}
    Let $G$ be a graph with $\treewidth{G} \leq 2$ and $|V(G)| \geq 3$. We can in $\polyG$ time obtain a smooth tree decomposition of $G$ of width 2.
\end{lemma}

\begin{proof}
    Directly follows from Lemma \ref{lemma:find-tree-decompositions} and the observation by Bodlaender \cite{Bodlaender_linear-time_algorithm_for_finding_tree-decomposition_of_small_treewidth} that any tree decomposition can in polynomial time be transformed into a smooth tree decomposition of equal width.
\end{proof}

\noindent
We state a slightly modified version of Lemma 2.2 in \cite{Bodlaender_linear-time_algorithm_for_finding_tree-decomposition_of_small_treewidth}, which also measures the number of vertices contained in bags of subtrees of $T$ opposed to only the complete tree $T$.

\begin{lemma}\label{lemma:size=of-smooth-tree-decomposition}
    Let $(T,\chi)$ be a smooth tree decomposition of width $\eta$ of a graph $G$ and let $T'$ be a non-empty subtree of $T$. Then $|V(T')| = |\chi(T')| - \eta$.
\end{lemma}

\begin{proof}
    Trivially follows from the proof of Lemma 2.2 in \cite{Bodlaender_linear-time_algorithm_for_finding_tree-decomposition_of_small_treewidth}.
\end{proof}

\section{Properties of Biconnected Graphs}
We will first prove two essential properties about biconnected induced subgraphs. The biconnected subgraphs $B$ we consider are subgraphs of $G - X$, where $(G,t)$ is a problem instance on which Reduction \ref{red:contract-component-with-small-neighborhood} has been exhaustively applied and where $X$ is a tidy modulator of $G$. The first property we will prove for such subgraphs is that the number of vertices in $B$ with neighbours outside of $B$ is bounded. The second property is that a smooth tree decomposition of $(T,\chi)$ will yield that each $u \in V(B)$ with $|\bchi(u)| = 1$ has a neighbour outside of $B$.

We start by proving the first property. To prove this we first prove a property that bounds the number of vertices in $B$ that have neighbours in a connected subgraph $H$ in $G - B$. Then afterwards we will prove that we can find a bounded number of connected subgraphs in $G$ such that all neighbours of $B$ are located in these subgraphs.

\begin{lemma}\label{lemma:biconnected-td:three-connected-to-biconnected}
    Let $G$ be a graph with $\treewidth{G} \leq 2$, let $B$ be a biconnected subgraph of $G$, and let $H$ be a connected subgraph of $G - B$. We have $|N_G(H) \cap V(B)| \leq 2$.
\end{lemma}

\begin{proof}
    Assume per contradiction $\{u,v,w\} \subseteq N_G(H) \cap V(B)$. Because $B$ is biconnected and $|V(B)| \geq |N_G(H) \cap V(B)| \geq 3$ there exists a simple cycle $C$ in $B$ that contains $u$ and $w$ (see Lemma \ref{lemma:biconnected-has-cycle}).
    
    \begin{case}[$v \in V(C)$]
        The construction shown in Figure \ref{fig:three-connected-to-biconnected:1} describes a $K_4$ minor in $G$, which by Lemma \ref{lemma:describes-k4-minor} contradicts $\treewidth{G} \leq 2$.
    \end{case}

    \begin{case}[$v \notin V(C)$]
        Because $B - u$ is connected there exists a path $P$ from $v$ to $C - u$ in $B - u$. We let $P$ be a shortest path between $v$ and $C - u$ in $B - u$, which means $|V(P) \cap V(C)| = 1$ holds. Similarly, we let $Q$ be a shortest path between $v$ and $C - w$ in $B - w$.
        
        \begin{subcase}[$w \notin V(P) \vee u \notin V(Q)$]
            Irregardless of whether $w \notin V(P)$ or $u \notin V(Q)$ holds, there exists a path $R$ from $v$ to $C - \{u,w\}$ in $B - \{u,w\}$. The construction shown in Figure \ref{fig:three-connected-to-biconnected:2.1} describes a $K_4$ minor in $G$, which by Lemma \ref{lemma:describes-k4-minor} contradicts $\treewidth{G} \leq 2$.
        \end{subcase}
        
        \begin{subcase}[$w \in V(P) \wedge u \in V(Q)$]
            Because $V(P) \cap V(C) = \{w\} \wedge V(Q) \cap V(C) = \{u\}$ it follows from the definition of $P$ and $Q$ that concatenating paths $Q(u,v]$ and $P[v,w)$ yields a (potentially non-simple) path $R$ in $B - C$ from a vertex in $N_G(u)$ to a vertex in $N_G(w)$, and this path visits vertex $v$. The construction shown in Figure \ref{fig:three-connected-to-biconnected:2.2} describes a $K_4$ minor in $G$, which by Lemma \ref{lemma:describes-k4-minor} contradicts $\treewidth{G} \leq 2$.
        \end{subcase}
    \end{case}
    
    \noindent
    Because these cases\footnote{We note that we could prove the same without distinguishing Case 1 and Case 2.1. Nevertheless, we chose for the more explicit approach to avoid potential confusion.} cover all cases and each lead to a contradiction, we can conclude that $|N_G(H) \cap V(B)| \leq 2$.
\end{proof}

\begin{figure}[ht]
\centering
\begin{subfigure}{.33\textwidth}
  \centering
  \includegraphics[page=1, width=.9\columnwidth]{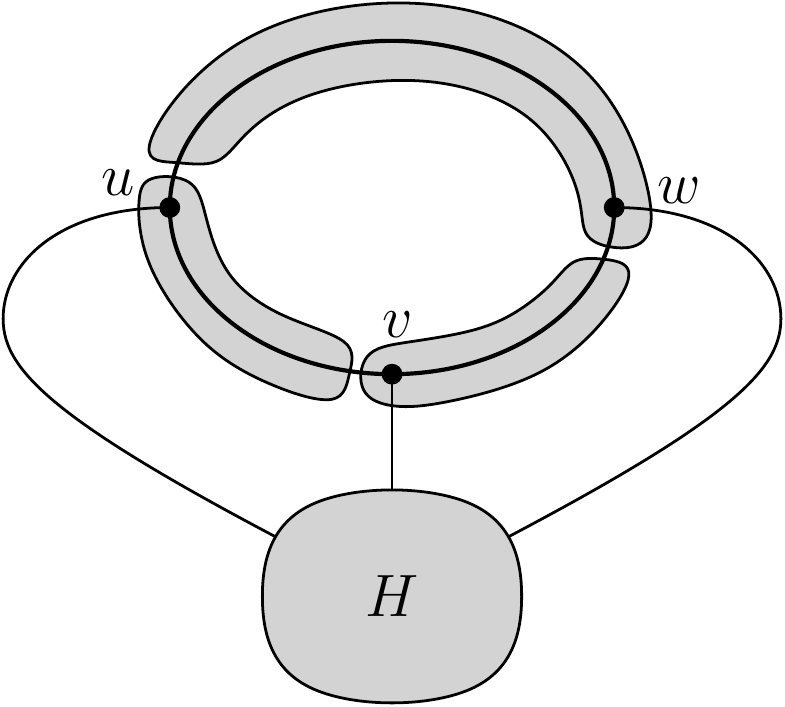}
  \caption{Case $v \in V(C)$}
  \label{fig:three-connected-to-biconnected:1}
\end{subfigure}%
\begin{subfigure}{.33\textwidth}
  \centering
  \includegraphics[page=2, width=.9\columnwidth]{max-3-shared-leafs-biconnected.pdf}
  \caption{Case $w \notin V(P) \vee u \notin V(Q)$}
  \label{fig:three-connected-to-biconnected:2.1}
\end{subfigure}%
\begin{subfigure}{.33\textwidth}
  \centering
  \includegraphics[page=3, width=.9\columnwidth]{max-3-shared-leafs-biconnected.pdf}
  \caption{Case $w \in V(P) \wedge u \in V(Q)$}
  \label{fig:three-connected-to-biconnected:2.2}
\end{subfigure}
\caption{Connected graph $H$ being triple connected to biconnected graph $B$ yields a $K_4$ minor}
\label{fig:three-connected-to-biconnected}
\end{figure}

\begin{lemma}\label{lemma:biconnected-td:number-of-vertices-with-neighbours-outside-of-biconnected}
    Let $(G,t)$ be a non-trivial problem instance, let $X$ be a tidy modulator of $G$, and let $B$ be a biconnected induced subgraph of $G - X$. Then $|\partial_G(B)| \leq 2|X|$.
\end{lemma}

\begin{proof}
    For each $x \in X$ let $C_x \in \mathcal{C}((G - B) - (X \setminus \{x\}))$ such that $x \in V(C_x)$. We first prove that $N_G(B) \subseteq \bigcup_{x \in X} V(C_x)$. Assume per contradiction that $v \in N_G(B) \setminus \bigcup_{x \in X} V(C_x)$.
    
    We have $v \notin X$ because otherwise $v \in V(C_v)$ would hold. Let $C \in \mathcal{C}((G - B) - X)$ be the component that contains $v$. If $N_G(C) \cap X \neq \emptyset$ let $x \in N_G(C) \cap X$. By definition of $C_x$ we have $v \in V(C_x)$. Hence $N_G(C) \cap X = \emptyset$ holds, which implies $N_G(C) \subseteq V(B)$.
    
    By Lemma \ref{lemma:biconnected-td:three-connected-to-biconnected} we have $|N_G(C)| \leq 2$. Because $N_G(C) \subseteq V(B)$ and $B$ is connected we have that $G[N_G[C]] \cup \binom{N_G(C)}{2}$ is a minor of $G[V(B) \cup V(C)]$. We can obtain this minor by contracting $B$ into the vertices $N_G(C)$. By Lemma \ref{lemma:minor-treewidth} we have $\treewidthsmallp{G[N_G[C]] \cup \binom{N_G(C)}{2}} \leq \treewidth{G[V(B) \cup V(C)]} \leq \treewidth{G - X} \leq 2$. Therefore we can apply Reduction \ref{red:contract-component-with-small-neighborhood} on $C$, which by Definition \ref{def:trivial-problem-instance} contradicts $(G,t)$ being non-trivial. Hence, $N_G(B) \subseteq \bigcup_{x \in X} V(C_x)$ must hold.
    
    We next prove that $|\partial_G(B)| \leq 2|X|$ holds. Assume per contradiction that $|\partial_G(B)| > 2|X|$. Then, due to the pigeonhole principle, there exists some $C_x$ with $|N_G(C_x) \cap V(B)| \geq 3$, which by Lemma \ref{lemma:biconnected-td:three-connected-to-biconnected} contradicts $\treewidth{G - (X \setminus \{x\})} \leq 2$. Therefore, we conclude $|\partial_G(B)| \leq 2|X|$.
\end{proof}

\noindent
We have proven that a biconnected induced subgraph $B$ has a bounded number of vertices with neighbours outside of $B$, however, we have not yet identified which vertices in $B$ will have neighbours outside of $B$. We will next show that in a smooth tree decomposition $(T,\chi)$ of $B$ each vertex $u \in V(B)$ with $|\bchi(u)| = 1$ will have a neighbour outside of $B$. We will also show that each leaf $\ell \in V(T)$ contains a vertex $u \in \chi(\ell)$ with $\bchi(u) = \{\ell\}$, thus indirectly bounding the number of leaves in a smooth tree decomposition of $B$.

\begin{lemma}\label{lemma:biconnected-td:beta-1-neighbours}
    Let $(G,t)$ be a non-trivial problem instance, let $B$ be a biconnected induced subgraph of $G$, and let $(T, \chi)$ be a smooth tree decomposition of $B$ of width $2$. For each vertex $u \in V(B)$ with $|\bchi(u)| = 1$ we have $u \in \partial_G(B)$.
\end{lemma}

\begin{proof}
    Let $u \in V(B)$ with $\bchi(u) = \{b\}$. By Property 2 of Definition \ref{def:tree-decomposition} we have for each $v \in N_B(u)$ that $v \in \chi(b) \setminus \{u\}$, which implies $|N_B(u)| \leq 2$. By Definition \ref{def:trivial-problem-instance} and Lemma \ref{lemma:min-degree-3} we have $|N_G(u)| \geq 3$. This yields $N_G(u) \setminus V(B) \neq \emptyset$ and hence $u \in \partial_G(B)$.
\end{proof}

\begin{lemma}\label{lemma:biconnected-td:leaf-has-beta-1}
    Let $(G,t)$ be a non-trivial problem instance, let $B$ be a biconnected induced subgraph of $G$, and let $(T, \chi)$ be a smooth tree decomposition of $B$ of width 2. For each leaf $\ell \in V(T)$ there exists some $u \in \chi(\ell)$ with $\bchi(u) = \{\ell\}$.
\end{lemma}

\begin{proof}
    In case $|V(T)| \leq 1$ this trivially holds. Otherwise, let $\ell \in V(T)$ with $N_T(\ell) = \{b\}$. By Definition \ref{def:smooth-tree-decomposition} we have $|\chi(\ell)| = 3 \wedge |\chi(\ell) \cap \chi(b)| = 2$. Hence, there exists some $u \in \chi(\ell) \setminus \chi(b)$. Because $T[\bchi(u)]$ is connected (Property 3 of Definition \ref{def:tree-decomposition}) and $b$ separates $\ell$ from $T - \{\ell, b\}$ this yields $\bchi(u) = \{\ell\}$.
\end{proof}

\section{Frequent Path Neighbour Reduction}
We will next introduce a reduction rule that, when given a path $P$ in $G - X$ that visits many neighbours of some vertex $u$, will allow us to find a neighbour $v \in N_G(u) \cap V(P)$ for which we can safely remove edge $uv$ from $G$. This reduction will be used to remove an edge incident to a vertex $u$ that appears in many bags of a linear subtree of a smooth tree decomposition of a biconnected graph.

Because we will reuse this reduction in Chapter \ref{ch:reducing-block-cut-trees} we will, however, prove that this reduction can be applied more generally. Within this section we will search for a simple path $P$ in $G - X$ that contains four neighbours $\{v_1,v_2,v_3,v_4\}$ of a vertex $u \in V(G)$ (not necessarily in $X$), and the component in $G - (X \cup \{u,v_1,v_4\})$ that contains $P(v_1,v_4)$ has no neighbours in $X$. Within Chapter \ref{ch:reducing-block-cut-trees} we will instead search for a simple path $P$ in $G - X$ that contains four neighbours $\{v_1,v_2,v_3,v_4\}$ of a vertex $x \in X$ such that the component $C \in \mathcal{C}(G - (X \cup \{v_1,v_4\}))$ that contains $P(v_1,v_4)$ has $N_G(C) \cap X$ being a limit-1 subset. Because $\{u\}$ will trivially be a limit-1 subset, the second definition (with modulator $X \cup \{u\}$) will suffice for the first use case as well. Hence, we will only introduce and apply the more general reduction rule.

\begin{reduction}[frequent path neighbour reduction]\label{red:remove-middle-edge-to-isolated-path}
    Let $(G,t)$ be a problem instance, let $X$ be a tidy modulator of $G$, let $x \in X$, and let $P$ be a simple path in $G - X$ visiting $\{v_1, v_2, v_3, v_4\} \subseteq N_G(x)$ in order of increasing index. Let $C \in \mathcal{C}(G - (X \cup \{v_1, v_4\}))$ such that $P(v_1, v_4)$ is contained in $C$. If $N_G(C) \cap X$ is a limit-1 subset for $(G - v_2x, t)$, then remove edge $v_2x$ from $G$.
\end{reduction}

\noindent
To prove that this reduction is safe we will first prove that any simple path $P$ in $G-X$ that visits three vertices $\{v_1, v_2, v_3\} \subseteq N_G(x)$ has that $\{v_2,x\}$ separates $v_1$ from $v_3$ in $G - X$.

\begin{lemma}\label{lemma:x-path-separator}
    Let $G$ be a graph with $\treewidth{G} \leq 2$, let $x \in V(G)$, and let $P$ be a simple path in $G - x$ visiting $\{v_1, v_2, v_3\} \subseteq N_G(x)$ in order of increasing index. Then $\{v_2,x\}$ separates $v_1$ from $v_3$ in $G$.
\end{lemma}

\begin{proof}
    Assume per contradiction that there exists a path $Q$ in $G - \{v_2,x\}$ from $v_1$ to $v_3$. Then $Q$ starts in $P[v_1,v_2)$ and ends in $P(v_2,v_3]$. We let $Q'$ be a shortest subpath of $Q$ between a vertex $a \in V(P[v_1,v_2))$ and a vertex vertex $b \in V(P(v_2,v_3])$. We display this situation in Figure \ref{fig:x-path-separator:construction}. We note $a \neq b$ because $P$ is a simple path.
    
    Because $Q'$ is a shortest subpath we have $V(P) \cap V(Q') = \{a,b\}$. We contract $P[v_1,a]$ together with $Q' - b$ into a vertex $c$, we contract $P[b,v_3]$ into a vertex $d$, and we contract $P(a,b)$ into a vertex $e$ (see Figure \ref{fig:x-path-separator:k4}). Vertices $c$, $d$, $e$, and $x$ together form a $K_4$ minor in $G$. By Lemma \ref{lemma:minor-treewidth} this contradicts $\treewidth{G} \leq 2$. Hence $\{v_2,x\}$ must separate $v_1$ from $v_3$.
\end{proof}

\begin{figure}[ht]
\centering
\begin{subfigure}{.5\textwidth}
  \centering
  \includegraphics[page=2, width=.75\columnwidth]{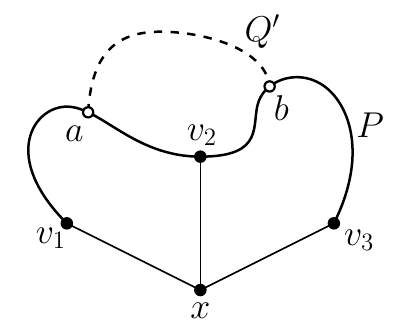}
  \caption{definition of path $Q'$}
  \label{fig:x-path-separator:construction}
\end{subfigure}%
\begin{subfigure}{.5\textwidth}
  \centering
  \includegraphics[page=3, width=.75\columnwidth]{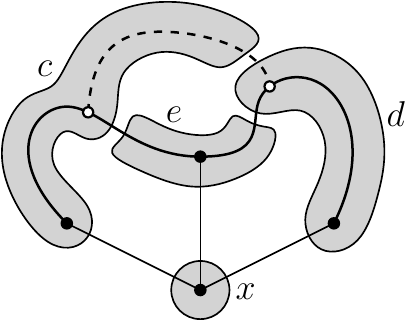}
  \caption{$K_4$-minor in $G$}
  \label{fig:x-path-separator:k4}
\end{subfigure}
\caption{Proof of Lemma \ref{lemma:x-path-separator} that $\{v_2,x\}$ not being a separator results in a $K_4$ minor}
\label{fig:x-path-separator}
\end{figure}

\begin{lemma}[safeness]\label{lemma:safeness:remove-middle-edge-to-isolated-path}
    Let $(G',t)$ be the problem instance obtained by applying Reduction \ref{red:remove-middle-edge-to-isolated-path} on $(G,t)$. Then $\twtwodeletion{G}{t} \Longleftrightarrow \twtwodeletion{G'}{t}$ holds.
\end{lemma}

\begin{proof}
    Because $G'$ is a minor of $G$ we obtain by Lemma \ref{lemma:minor-of-G-implication} that $\twtwodeletion{G'}{t}$ implies $\twtwodeletion{G}{t}$. It remains to prove that $\twtwodeletion{G'}{t}$ implies $\twtwodeletion{G}{t}$. We let $S'$ be a solution for $(G',t)$.

    \begin{case}[$\{v_2, x\} \cap S' \neq \emptyset$]
        We have $G' - S' = G - S$, which yields $\twtwodeletion{G}{t}$.
    \end{case}
    
    \noindent
    For all other cases we will assume that $v_2, x \notin S'$ holds. Due to $N_G(C) \cap X$ being a limit-1 subset for $(G',t)$ this yields $(N_G(C) \cap X) \setminus S' = \{x\}$. Hence, $N_G(C) \setminus S' \subseteq \{v_1, v_4, x\}$ holds for all other cases. Because $G[N_G[C]] - S'$ is a subgraph of $G - (X \setminus \{x\})$ we have by Lemma \ref{lemma:minor-treewidth} that $\treewidth{G[N_G[C]] - S'} \leq \treewidth{G - (X \setminus \{x\})} \leq 2$.
    
    \begin{case}[$\{v_1,v_4\} \cap S' \neq \emptyset$]
        W.l.o.g. assume $v_4 \in S'$. From Lemma \ref{lemma:minor-treewidth} it follows that $\treewidth{G - (V(C) \cup S')} \leq \treewidth{G' - S'} \leq 2$. Because $N_G(C) \setminus S' \subseteq \{v_1,x\}$ and $v_1x \in E(G)$ we have that $G[N_G(C) \setminus S']$ is a clique. Therefore we can apply Theorem \ref{theorem:connect-treewidth-graphs} on $(G - S', G - (V(C) \cup S'))$ to derive $\treewidth{G - S'} \leq 2$, which yields $\twtwodeletion{G}{t}$.
    \end{case}
    
    \begin{case}[$V(C) \cap S' \neq \emptyset$]
        Let $S = (S' \setminus V(C)) \cup \{v_4\}$. We prove $\treewidth{G' - S} \leq 2$. By Lemma \ref{lemma:minor-treewidth} we have $\treewidth{G' - (V(C) \cup S)} \leq \treewidth{G' - S'} \leq 2$ and $\treewidth{G'[N_G[C]] - S} \leq \treewidth{G - (X \setminus \{x\})} \leq 2$. Because $N_G(C) \setminus S \subseteq \{v_1,x\}$ and $v_1x \in E(G')$ we have that $G'[N_G(C) \setminus S]$ is a clique. Therefore we can apply Theorem \ref{theorem:connect-treewidth-graphs} on $(G' - S, G' - (V(C) \cup S))$ to derive $\treewidth{G' - S} \leq 2$. Hence, in case $V(C) \cap S' \neq \emptyset$ there exists a solution $S$ for $(G',t)$ with $\{v_1,v_4\} \cap S \neq \emptyset$, which as proven in Case 2 implies $\twtwodeletion{G}{t}$.
    \end{case}
    
    \begin{case}[$(V(C) \cup \{v_1,v_4,x\}) \cap S' = \emptyset$]
        We first prove for each $D \in \mathcal{C}(G - (N_G[C] \cup S'))$ that $G[N_G(D)] - S'$ is a clique. We have $N_G(D) \setminus S' \subseteq N_G(C) \setminus S' = \{v_1,v_4,x\}$ and $v_1x, v_4x \in E(G)$. Hence, for $G[N_G(D)] - S'$ to not be a clique we must have $v_1,v_4 \in N_G(D)$. By Lemma \ref{lemma:x-path-separator} we have that $\{v_3,x\}$ separates $v_1$ from $v_4$ in $G' - S'$. A component $D \in \mathcal{C}(G - (N_G[C] \cup S'))$ with $v_1,v_4 \in N_G(D)$ contradicts $\{v_3,x\}$ being a separator. Hence $G[N_G(D)] - S'$ must be a clique. This corresponds to the situation shown in Figure \ref{fig:safeness:remove-middle-edge-to-isolated-path}.
        
        We have by Lemma \ref{lemma:minor-treewidth} that $\treewidth{G[N_G[C]] - S'} = \treewidth{G[V(C) \cup \{v_1,v_4,x\}]} \leq \treewidth{G - (X \setminus \{x\})} \leq 2$ and that each $D \in \mathcal{C}(G - (N_G[C] \cup S'))$ has $\treewidth{G[N_G[D]] - S'} \leq \treewidth{G' - S'} \leq 2$. Furthermore, for each such $D$ we have proven that $G[N_G(D)] - S'$ is a clique. Therefore we can apply Theorem \ref{theorem:connect-treewidth-graphs} on $(G - S', N_G[C] \setminus S')$ to derive $\treewidth{G' - S} \leq 2$, which implies $\twtwodeletion{G}{t}$.
    \end{case}
    
    \noindent
    Because these cases cover all cases, correctness trivially follows.
\end{proof}

\begin{figure}[ht]
    \centering
    \includegraphics[width=.8\columnwidth]{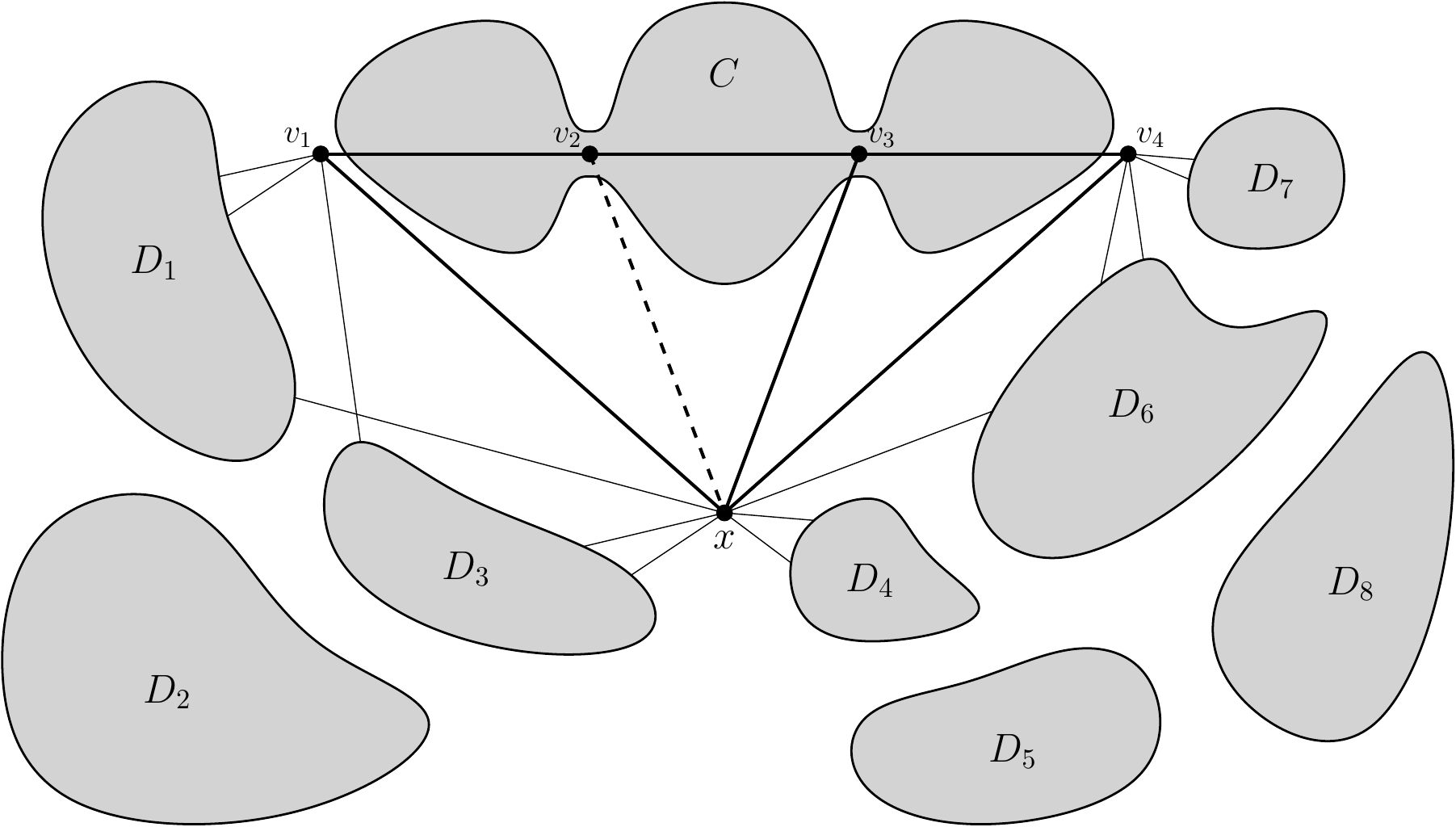}
    \caption{Example application of Reduction \ref{red:remove-middle-edge-to-isolated-path} and the connectivity of components $D \in \mathcal{C}(G - (N_G[C] \cup S'))$ in case $V(C) \cup \{v_1,v_4,x\} \cap S' = \emptyset$}
    \label{fig:safeness:remove-middle-edge-to-isolated-path}
\end{figure}

\section{Ladder Reduction}
We will next introduce another reduction. This reduction is defined to reduce the situation displayed in Figure \ref{fig:ladder-reduction}. There are two vertex-disjoint induced paths $P$ and $Q$ in $B$ alongside a matching $M = \{p_1q_1, p_2q_2, \dots, p_9q_9\}$ between nodes in $P$ and $Q$, such that $\partial_G(V(P) \cup V(Q)) \subseteq \{p_1,p_9,q_1,q_9\}$. This matching $M$ will be such that the endpoints of edges in $M$ are `ordered' along $P$ and $Q$. I.e., the situation in Figure \ref{fig:non-cross-double-path:1} will never occur. We note that Figure \ref{fig:ladder-reduction} is a simplification of the described situation, because it omits edges not contained in $M$. We furthermore note that this graph $G[V(P) \cup V(Q)]$ has a ladder graph $L_9$ as minor. As we will prove later, large linear subtrees in a smooth tree decomposition of $B$ will contain a pair of paths $P$ and $Q$ that together induce a graph in $G$ with a ladder minors $L_9$.\clearpage

\begin{figure}[ht]
    \centering
    \includegraphics[scale=0.7]{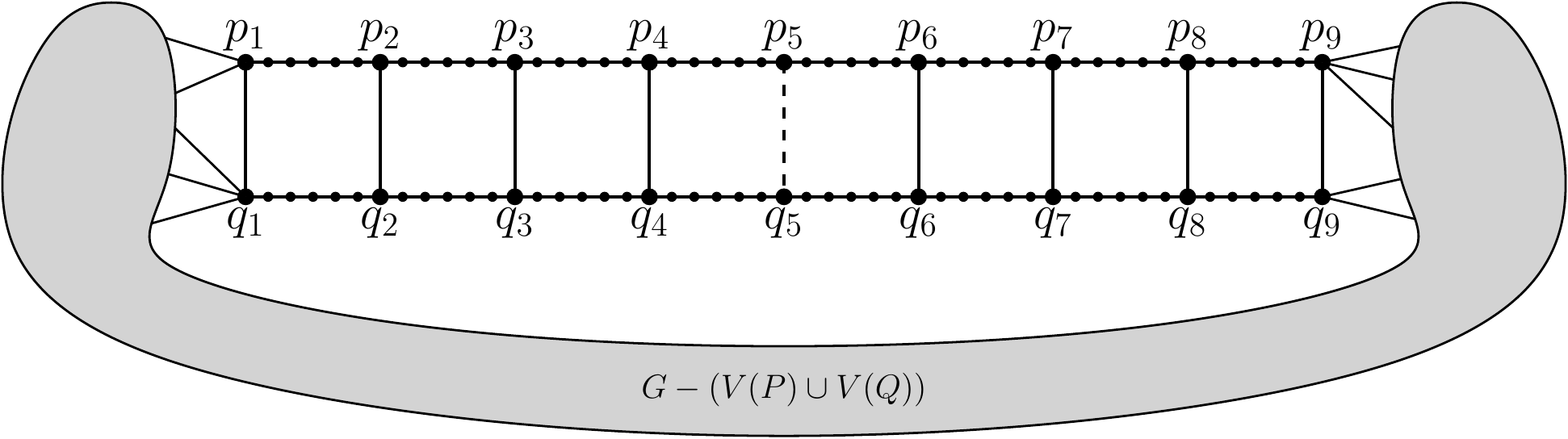}
    \caption{Situation reduced by Reduction \ref{red:ladder-reduction}}
    \label{fig:ladder-reduction}
\end{figure}
\medskip
\begin{figure}[ht]
\begin{subfigure}{.49\textwidth}
  \centering
  \includegraphics[width=.7\linewidth, page=2]{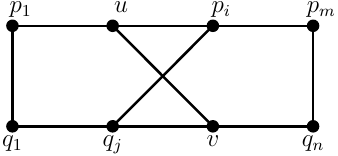}  
  \caption{Case $u \in V(P) \wedge v \notin V(P)$}
  \label{fig:non-cross-double-path:1}
\end{subfigure}
\begin{subfigure}{.49\textwidth}
  \centering
  \includegraphics[width=.7\linewidth, page=3]{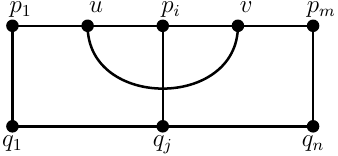}
  \caption{Case $u, v \in V(P)$}
  \label{fig:non-cross-double-path:2}
\end{subfigure}
\caption{Illegal edge pairs when $\treewidth{G[V(P) \cup V(Q)] \cup \{p_1q_1,p_mq_n\}} \leq 2$}
\label{fig:non-cross-double-path}
\end{figure}

\noindent
Within this subsection we will only focus on defining the reduction and proving that it is safe. In later subsections we will show how one can find these reducible structure.

\begin{reduction}[ladder reduction]\label{red:ladder-reduction}
    Let $(G,t)$ be a problem instance. Let $P$ and $Q$ be two vertex-disjoint induced paths in $G$ visiting in order the vertices $\{p_1, p_2, \dots, p_9\}$ and $\{q_1, q_2, \dots, q_9\}$ respectively, such that $p_1q_1, p_2q_2, \dots, p_9q_9 \in E(G)$. If $\treewidth{G[V(P) \cup V(Q)]} \leq 2$ and $\partial_G(V(P) \cup V(Q)) \subseteq \{p_1,p_9,q_1,q_9\}$, then remove $p_5q_5$ from $G$. 
\end{reduction}

\noindent
Before we can prove the safeness of this reduction, we will first prove a lemma that states that edges in $G[V(P) \cup V(Q)]$ can not `cross' each other. I.e. the situations shown in Figure \ref{fig:non-cross-double-path} can not occur. We purposefully assume in this lemma that $\treewidth{G[V(P) \cup V(Q)] \cup \{p_1q_1, p_mq_n\}} \leq 2$, opposed to the assumption $\treewidth{G[V(P) \cup V(Q)]} \leq 2$ with $p_1q_1, p_9q_9 \in M \subset E(G)$. This, because we will reuse this lemma in the proof of Lemma \ref{lemma:find-ladder-reduction}, which uses paths $P = \langle p_1, p_2, \dots, p_m \rangle$ and $Q = \langle q_1, q_2, \dots, q_n \rangle$ for which not necessarily $p_1q_1, p_mq_n \in E(G)$ holds.

\begin{lemma}\label{lemma:non-cross-double-path}
    Let $G$ be a graph, let $P = \langle p_1, p_2, \dots, p_m \rangle$ and $Q = \langle q_1, q_2, \dots, q_n \rangle$ be two vertex-disjoint induced paths in $G$ for which $\treewidth{G[V(P) \cup V(Q)] \cup \{p_1q_1, p_mq_n\}} \leq 2$ holds. For each $p_i \in V(P)$ and $q_i \in V(Q)$ with $p_iq_j \in E(G)$ we have that $\{p_i,q_j\}$ separates $V(P[p_1,p_i)) \cup V(Q[q_1,q_j))$ from $V(P(p_i,p_m]) \cup V(Q(q_j,q_n])$ in $G[V(P) \cup V(Q)]$.
\end{lemma}

\begin{proof}
    Assume per contradiction that for $\{p_i,q_j\}$ this does not hold. Let $R$ be a shortest path in $G[V(P) \cup V(Q)] - \{p_i,q_j\}$ from a vertex $u \in V(P[p_1,p_i)) \cup V(Q[q_1,q_j))$ to a vertex $v \in V(P(p_i,p_m]) \cup V(Q(q_j,q_n])$. Because $R$ is a shortest path in $G[V(P) \cup V(Q)]$ we have $V(R) = \{u,v\}$, which implies $uv \in E(G)$. W.l.o.g. we assume $u \in V(P)$.
    
    \begin{case}[$v \in V(Q)$]
        This case is displayed in Figure \ref{fig:non-cross-double-path:1}. We can apply Lemma \ref{lemma:describes-k4-minor} on $(P[p_1,u], P(u,p_m], Q[q_1,v), Q[v,q_n])$ to derive $\treewidth{G[V(P) \cup V(Q)] \cup \{p_1q_1, p_mq_n\}} > 2$, which is a contradiction.
    \end{case}
    
    \begin{case}[$v \in V(P)$]
        This case is displayed in Figure \ref{fig:non-cross-double-path:2}. Because $P$ is an induced path we have that $p_i$ separates $P[p_1,p_i)$ from $P(p_i,p_m]$ in $G[P]$, which contradicts $uv \in E(G)$.
    \end{case}
    
    \noindent
    Because both cases lead to a contradiction, correctness trivially follows.
\end{proof}

\noindent
Next we prove that Reduction \ref{red:ladder-reduction} is safe. We first provide the intuition on why this reduction rule is safe, before diving into the formal proof. This proof deals with two potential situations for a solution $S'$ of $(G',t)$:

The first situation is the situation where $S'$ contains two or more vertices from $(V(P) \cup V(Q)) \setminus \{p_1,p_9,q_1,q_9\}$. In this case we let $S$ be the set $S'$ where the vertices in $(V(P) \cup V(Q)) \setminus \{p_1,q_1,p_9,q_9\}$ are replaced by $p_9$ and $q_9$. This yields the situation shown in Figure \ref{fig:ladder-intuition:1}. We then prove that $\treewidth{G - S} \leq 2$ using Theorem \ref{theorem:connect-treewidth-graphs} on $(G - S, (V(P) \cup V(Q)) \setminus S)$.

In the second situation $S'$ contains at most one vertex from $(V(P) \cup V(Q)) \setminus \{p_1,q_1,p_9,q_9\}$. In this case we will show that $G[(V(P) \cup V(Q))  \setminus \{p_1,q_1,p_9,q_9\}]$ will have a minor that shows that $\{p_1,q_1\} \setminus S'$ and $\{p_9,q_9\} \setminus S'$ are separated by $(V(P) \cup V(Q)) \setminus S'$ in $G - S'$. This situation is shown in Figure \ref{fig:ladder-intuition:2}. A path between $\{p_1,q_1\} \setminus S'$ and $\{p_9,q_9\} \setminus S'$ that is disjoint from the displayed minor trivially yields a $K_4$ minor. With this observation we can apply Theorem \ref{theorem:connect-treewidth-graphs} on $(G - S', (V(P) \cup V(Q)) \setminus S')$ to derive $\treewidth{G - S'} \leq 2$.

The second situation is, however, a bit more complex that described here. The reason for this is caused by $p_1$, $p_9$, $q_1$, or $q_9$ potentially being included in $S'$. Conceptually, the way our proof works will be the same as we described, however, we will not use $\{p_1,p_9,q_1,q_9\}$ to separate $P - \{p_1,p_9\}$ and $Q - \{q_1,q_9\}$ from the remaining graph. Instead we will use some set $\{p_i,p_j,q_i,q_j\}$ with $1 < i < j < 9$ to separate $P(i,j)$ and $Q(i,j)$ from the remaining graph.

\begin{figure}[ht]
\begin{subfigure}{\textwidth}
  \centering
  \includegraphics[scale=0.7,page=5]{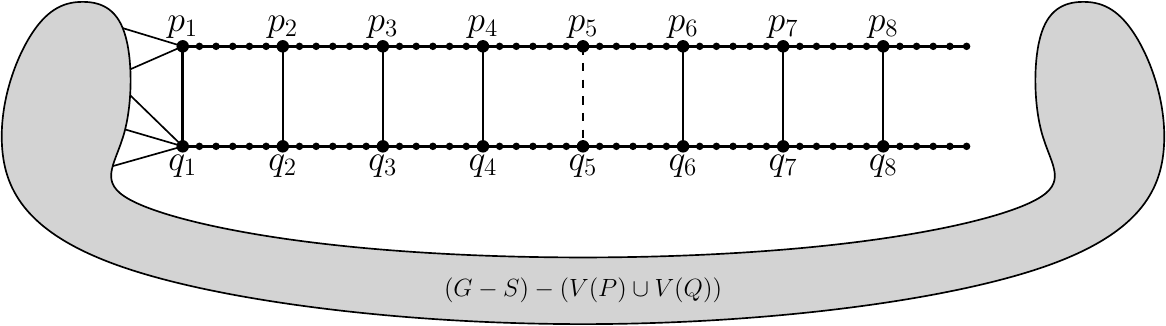}  
  \caption{Case $|S' \cap ((V(P) \cup V(Q)) \setminus \{p_1,p_9,q_1,q_9\})| \geq 2$}
  \label{fig:ladder-intuition:1}
\end{subfigure}
\begin{subfigure}{\textwidth}
  \centering
  \bigskip\bigskip
  \includegraphics[scale=0.7,page =7]{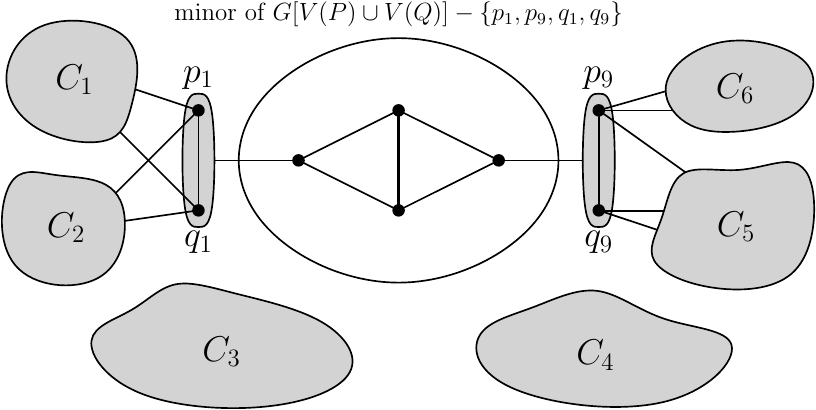}  
  \caption{Case $|S' \cap ((V(P) \cup V(Q)) \setminus \{p_1,p_9,q_1,q_9\})| \leq 1$}
  \label{fig:ladder-intuition:2}
\end{subfigure}
\caption{Intuition why Reduction \ref{red:ladder-reduction} is safe}
\label{fig:ladder-intuition}
\end{figure}

\begin{lemma}[safeness]\label{lemma:safeness:ladder-reduction}
    Let $(G',t)$ be the problem instance obtained by applying Reduction \ref{red:ladder-reduction} on $(G,t)$. Then $\twtwodeletion{G}{t} \Longleftrightarrow \twtwodeletion{G'}{t}$ holds.
\end{lemma}

\begin{proof}
    Because $G'$ is a minor of $G$ we obtain by Lemma \ref{lemma:minor-of-G-implication} that $\twtwodeletion{G'}{t}$ implies $\twtwodeletion{G}{t}$. It remains to prove that $\twtwodeletion{G'}{t}$ implies $\twtwodeletion{G}{t}$. We let $S'$ be a solution for $(G',t)$ and we define for each $1 \leq i,j \leq 9$ graphs $H[i,j] = G[V(P[p_i,p_j]) \cup V(Q[q_i,q_j])]$, graphs $H[i,j) = H[i,j] - \{p_j,q_j\}$, graphs $H(i,j] =  H[i,j] - \{p_i,q_i\}$, and graphs $H(i,j) = H[i,j] - \{p_i,p_j,q_i,q_j\}$.
    
    \begin{case}[$|S' \cap V(H(1,9))| \geq 2$]
        Let $S = (S' \setminus V(H(1,9))) \cup \{p_9,q_9\}$. By Lemma \ref{lemma:minor-treewidth} we have $\treewidth{H[1,9] - S} \leq \treewidth{G[V(P) \cup V(Q)]} \leq 2$ and $\treewidth{G - (V(H(1,9)) \cup S)} \leq \treewidth{G' - S'} \leq 2$. Furthermore, $N_G(G - (V(H(1,9)) \cup S)) \setminus S \subseteq \partial_G(H[1,9]) \setminus S \subseteq \{p_1,q_1\}$, which induces a clique in $G$. Hence, we can apply Theorem \ref{theorem:connect-treewidth-graphs} on $(G - S, H[1,9] - S)$ to derive $\treewidth{G - S} \leq 2$. Because $|S| \leq |S'| \leq t$ this yields $\twtwodeletion{G}{t}$.
    \end{case}

    \begin{case}[$|S' \cap V(H(1,9))| \leq 1$]
        W.l.o.g. we assume $S' \cap V(H(1,5)) = \emptyset$. We first argue that $G'[V(H(3,6]) \setminus S']$ is a connected graph. We have that $G'[V(H[4,6]) \setminus S']$ will be connected because $G'[V(H[4,6])]$ trivially is biconnected and $|S' \cap V(H[4,6])| \leq |S' \cap V(H(1,9))| \leq 1$. Because $S' \cap V(H(3,4]) = \emptyset$ we also have that $G'[V(H(3,4]) \setminus S']$ is connected. Hence, $G'[V(H(3,6]) \setminus S']$ is connected.
    
        We argue that $V(H(2,6))$ separates $\{p_2,q_2\}$ from $\{p_6,q_6\} \setminus S'$ in $G' - S'$. Assume per contradiction that this does not hold and let $R$ be a shortest path in $G' - (V(H(2,6)) \cup S')$ between $\{p_2,q_2\}$ and $\{p_6,q_6\} \setminus S'$. When we contract $G'[V(H[2,3)) \cup V(R)] - \{p_6,q_6\}$ into a single vertex $a$ and we contract $G'[V(H(3,6]) \setminus S']$ into a single vertex $b$, then we have that $a$, $b$, $p_3$, and $q_3$ together form a $K_4$ minor in $G' - S'$ (see Figure \ref{fig:safeness:ladder-reduction}). Lemma \ref{lemma:minor-of-G-implication} therefore yields the contradiction $\treewidth{G' - S'} > 2$. Hence, $V(H(2,6))$ must separate $\{p_2,q_2\}$ from $\{p_6,q_6\} \setminus S'$ in $G' - S'$.
        
        We argue that $\partial_G(H[2,6]) \subseteq \{p_2,p_6,q_2,q_6\}$. Assume per contradiction that there exists a vertex $u \in \partial_G(H[2,6]) \setminus \{p_2,p_6,q_2,q_6\} \subseteq V(H(2,6))$ and let $v \in N_G(u) \setminus V(H[2,6])$. Because $\partial_G(H[1,9]) \subseteq \{p_1,p_9,q_1,q_9\}$ and $u \in V(H(1,9))$ we must have $N_G(u) \subset V(H[1,9])$, which implies $v \in V(H[1,2)) \cup V(H(6,9])$. By Lemma \ref{lemma:non-cross-double-path} we have that $\{p_2,p_6,q_2,q_6\}$ separates $u$ from $v$ in $H[1,9]$, which contradicts $v \in N_G(u)$. Therefore $\partial_G(H[2,6]) \subseteq \{p_2,p_6,q_2,q_6\}$ holds.
        
        By Lemma \ref{lemma:minor-treewidth} we have $\treewidth{H[2,6] - S'} \leq \treewidth{G[V(P) \cup V(Q)]} \leq 2$. Because $\partial_G(H[2,6]) \subseteq \{p_2,p_6,q_2,q_6\}$ and $V(H(2,6))$ separates $\{p_2,q_2\}$ from $\{p_6,q_6\} \setminus S'$ in $G' - S'$ we have for each $C \in \mathcal{C}(G - (V(H[2,6]) \cup S'))$ that $G[N_G(C)] \setminus S' \subseteq \{p_2,q_2\} \vee G[N_G(C)] \setminus S' \subseteq \{p_6,q_6\}$, which in either case is a clique. Furthermore, by Lemma \ref{lemma:minor-treewidth} we have $\treewidth{G[N_G[C] \setminus S']} \leq \treewidth{G - (V(H(2,6)) \cup S')} \leq \treewidth{G' - S'} \leq 2$. Therefore we can apply Theorem \ref{theorem:connect-treewidth-graphs} on $(G - S, H[2,6] - S')$ to derive $\treewidth{G - S'} \leq 2$, which yields $\twtwodeletion{G}{t}$.
    \end{case}

    \noindent
    Because these cases cover all cases, correctness trivially follows.
\end{proof}
\medskip
\begin{figure}[ht]
    \centering
    \includegraphics[page=4,width=.55\linewidth]{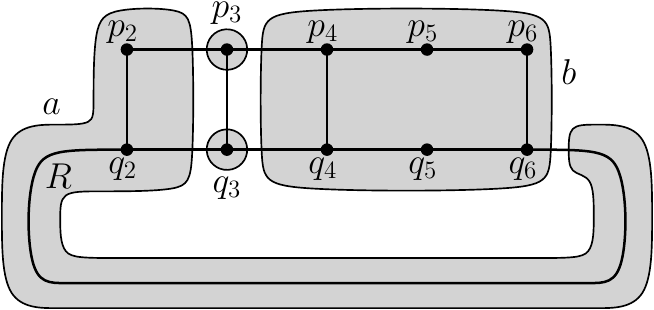}
    \caption{$K_4$ minor when $H(2,6)$ does not separate $\{p_2,q_2\}$ from $\{p_6,q_6\} \setminus S'$ in $G' - S'$}
    \label{fig:safeness:ladder-reduction}
\end{figure}
\medskip
\noindent
With this we have proven that Reduction \ref{red:ladder-reduction} is safe. However, we have not yet provided a method to find such paths $P$ and $Q$. The remaining part of this section will focus on finding applications of Reductions \ref{red:remove-middle-edge-to-isolated-path} and \ref{red:ladder-reduction}.

\section{Finding Reducible Structures}
Within this section we will let $(G,t)$ be a non-trivial problem instance, we let $X$ be a tidy modulator of $G$, and we let $B$ be a biconnected induced subgraph of $G - X$. In case $B$ is large we will want to be able to find a subgraph of $B$ on which Reduction \ref{red:remove-middle-edge-to-isolated-path} or \ref{red:ladder-reduction} can be applied.

To this end we will use the following algorithm. We first compute a smooth tree decomposition $(T^*,\chi)$ of $B$. We then recursively explore subtrees $T$ of $T^*$. In these subtrees $T$ we pick an arbitrary leaf to be its root. The subtrees $T$ we explore will have for each non-leaf node $b \in V(T)$ that $N_T(b) = N_{T^*}(b)$. This is the same as stating that $\partial_{T^*}(T)$ consists of only leaves of $T$. For example, the tree shown in Figure \ref{fig:valid-subtrees:valid} satisfies this condition and the tree shown in Figure \ref{fig:valid-subtrees:invalid} does not.

\begin{figure}[ht]
\centering
\begin{subfigure}{.45\textwidth}
  \centering
  \includegraphics[scale=.5, page=1]{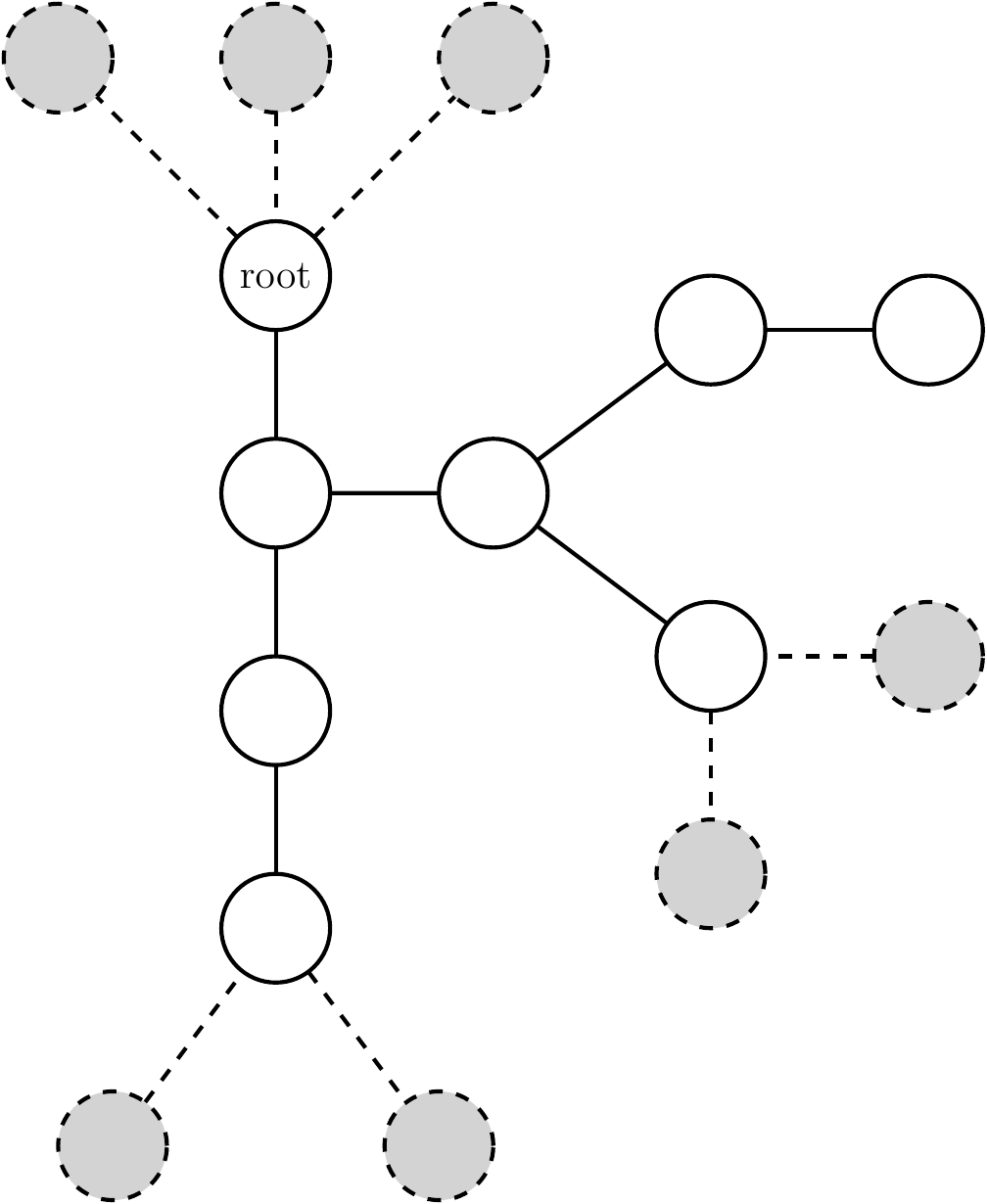}  
  \caption{A `valid' subtree with $\partial_{T^*}(T)$ consisting of only leaves}
  \label{fig:valid-subtrees:valid}
\end{subfigure}
\hspace{.04\textwidth}
\begin{subfigure}{.45\textwidth}
  \centering
  \includegraphics[scale=.5, page=2]{biconnected-reduction-problem-instance.pdf}  
  \caption{An `invalid' subtree with $\partial_{T^*}(T)$ not consisting of only leaves}
  \label{fig:valid-subtrees:invalid}
\end{subfigure}
\caption{Subtrees $T$ of $T^*$ with vertices $V(T^*) \setminus V(T)$ shown in grey}
\label{fig:valid-subtrees}
\end{figure}

\noindent
The base case of this recursive algorithm is the case where $T$ is a large linear tree (see Figure \ref{fig:finding-reducible-structures:base-case}). For such a tree we will prove that we can find a structure on which Reduction \ref{red:remove-middle-edge-to-isolated-path} or \ref{red:ladder-reduction} can be applied.

\begin{figure}[ht]
    \centering
    \includegraphics[scale=.5, page=3]{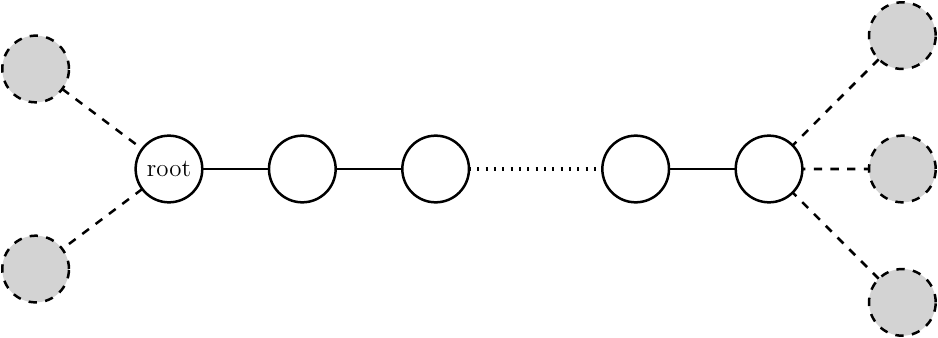}  
    \caption{Base case handled by recursive algorithm}
    \label{fig:finding-reducible-structures:base-case}
\end{figure}

\noindent
In case $T$ is not a linear tree we will find a branching node $b$ (i.e. a node with degree at least three) furthest from the root of $T$. By our choice of $b$ there are $|N_{T^*}(b)|-1$ linear subtrees and one potentially non-linear subtree in $\mathcal{C}(T - b)$. This potentially non-linear subtree will be the subtree containing the root of $T$. We will check for each subtree $T' \in \mathcal{C}(T - b)$ whether subtree $T[V(T') \cup \{b\}]$ is large and leads to the application of a reduction rule. Figure \ref{fig:finding-reducible-structures:recursive-step} displays these subtrees.

In case one of the linear subtrees is large we reach a base case. In case all linear subtrees are small, then we can prove that the subtree containing the root will be large. On this subtree we will then recurse. Because each time we recurse the number of leaves in $T$ strictly decreases, we will eventually reach a base case.

\begin{figure}[ht]
    \centering
    \includegraphics[scale=.5, page=4]{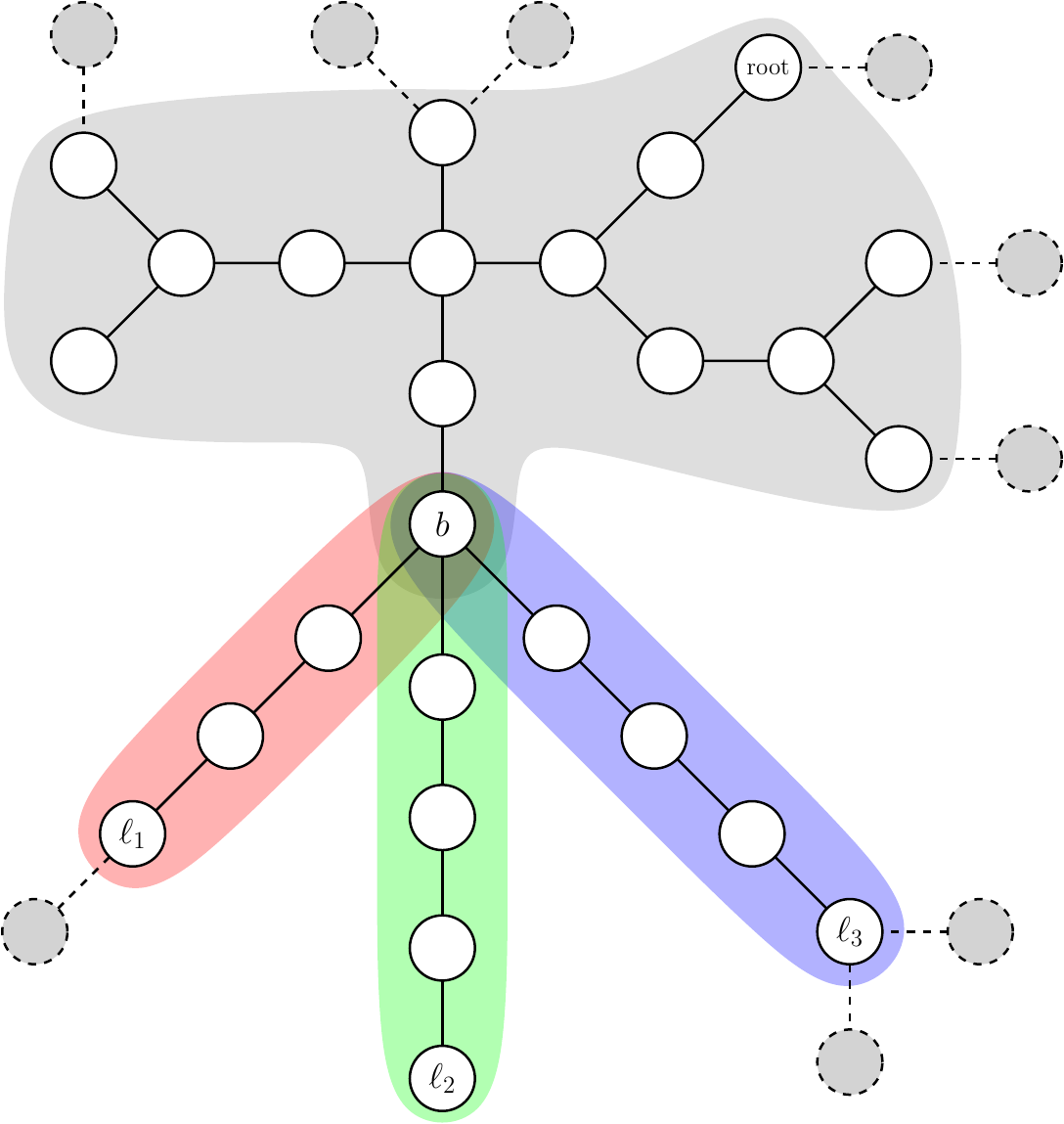}  
    \caption{Recursive step where each candidate subproblem has a different colour}
    \label{fig:finding-reducible-structures:recursive-step}
\end{figure}

\noindent
For notation convenience we introduce the definition of a biconnected reduction problem instance.

\begin{definition}\label{def:biconnected-reduction-problem-instance}
     We call $(G,t,X,B,T^*,\chi,T,L,U)$ a biconnected reduction problem instance when all of the following hold:
     \begin{itemize}
         \item $(G,t)$ is a non-trivial problem instance
         \item $X$ is a tidy modulator for $G$
         \item $B$ is a biconnected induced subgraph of $G - X$
         \item $(T^*,\chi)$ is a smooth tree decomposition of $B$ of width 2
         \item $T$ is a subtree of $T^*$
         \item $L = \setdef{\ell \in V(T)}{|N_T(\ell)| \leq 1}$
         \item $U = \chi(T) \cap \partial_G(B)$
         \item for each $b \in V(T) \setminus L$ it holds that $N_T(b) = N_{T^*}(b)$.
     \end{itemize}
\end{definition}
\bigskip
\noindent
We note that a given tuple $(G,t,X,B,T^*,\chi,T)$ uniquely defines sets $L$ (the set of leaves of $T$) and $U$ (the set of vertices in $\chi(T)$ that have neighbours outside of $B$). Hence, we can always find sets $L$ and $U$ for a given $(G,t,X,B,T^*,\chi,T)$ tuple.

We start describing this algorithm by first handling its base case, where $T$ is a large linear subtree of $T^*$. We will initially focus on two types of linear subtrees. The first type being a tree $T$ where some vertex $u \in \chi(T)$ is contained in many nodes $\bchi(u) \cap V(T)$. The second type will be a tree for which $U = \emptyset$ holds. Afterwards we will combine these two types into a general reduction for large linear subtrees. We then finish by describing the recursive search.

\subsection{Reducing Long Paths with Frequent Vertex}
We first focus on the situation where $T$ is a linear tree with a vertex $u \in V(T)$ that is contained in many nodes $\bchi(u) \cap V(T)$. We start with a simplified problem where $U \subseteq \{u\}$ holds. I.e. there are no vertices $v \in \chi(T) \setminus \{u\}$ with a neighbour outside of $B$. Then afterwards we will generalise this to any set $U$.

\begin{lemma}\label{lemma:reducing-long-paths-with-frequent-vertex-without-U}
    Let $(G,t,X,B,T^*,\chi,T,L,U)$ be a biconnected reduction problem instance with $|L| = 2$ and let $u \in V(B)$. If $U \subseteq \{u\}$ and $V(T) \subseteq \bchi(u)$ and $|V(T)| = 7$, then we can in $\polyG$ time apply Reduction \ref{red:remove-middle-edge-to-isolated-path} on an edge in $B$.
\end{lemma}

\begin{proof}
    Because $|L| = 2$ we have that $T$ is a linear tree. We let $V(T) = \{b_1, b_2, \dots, b_7\}$ be such that each $1 < i < 7$ has $N_T(b_i) = \{b_{i-1}, b_{i+1}\}$. We first prove that we can order the vertices in $\chi(T) = \{u, v_1, v_2, \dots, v_8\}$ such that each $b_i \in V(T)$ has $\chi(b_i) = \{u, v_i, v_{i+1}\}$. I.e. the situation displayed in Figure \ref{fig:lemma:reducing-long-paths-with-frequent-vertex-without-U:tree}.
    
    Because $V(T) \subseteq \bchi(u)$ it follows from Definition \ref{def:smooth-tree-decomposition} that for each $1 \leq i < 7$ there exists a single vertex $v_{i+1} \in (\chi(b_i) \cap \chi(b_{i+1})) \setminus \{u\}$. We argue that for any $1 \leq i,j < 7$ with $i \neq j$ vertices $v_{i+1}$ and $v_{j+1}$ are distinct. W.l.o.g. we assume per contradiction that $v_{i+1} = v_{j+1}$ and $i < j$. By Lemma \ref{lemma:bchi-connected-in-subtree} we have that $T[\bchi(v_{i+1})]$ is connected, which yields $\{b_i, b_{i+1}, b_{i+2}\} \subseteq \{b_i, b_{i+1}, \dots, b_{j+1}\} \subseteq \bchi(v_{i+1})$. This yields $\chi(b_i) \cap \chi(b_{i+1}) = \chi(b_{i+1}) \cap \chi(b_{i+2}) = \{u,v_{i+1}\}$. From Definition \ref{def:smooth-tree-decomposition} it follows that there exists a vertex $w \in \chi(b_{i+1}) \setminus \{u,v_{i+1}\}$ with $w \notin \chi(b_i) \wedge w \notin \chi(b_{i+2})$. Because $T[\bchi(w)]$ is connected (Lemma \ref{lemma:bchi-connected-in-subtree}) we have $\bchi(w) = \{b_{i+1}\}$, which by Lemma \ref{lemma:biconnected-td:beta-1-neighbours} yields the contradiction $w \in U \subseteq \{u\}$. Hence $v_{i+1}$ and $v_{j+1}$ must be distinct.
    
    We let $v_1 \in \chi(b_1) \setminus \chi(b_2)$ and $v_8 \in \chi(b_7) \setminus \chi(b_6)$. From Lemma \ref{lemma:bchi-connected-in-subtree} it follows that $T[\bchi(v_1)]$ is connected and $T[\bchi(v_8)]$ is connected, which implies that all vertices $v_i$ must be distinct. From the definitions of $v_i$ it follows that each $b_i \in V(T)$ has $\chi(b_i) = \{u, v_i, v_{i+1}\}$.
    
    We next prove that $N_G(v_i) = \{u, v_{i-1}, v_{i+1}\}$ holds for each $3 \leq i \leq 5$. Because $v_i \notin U$ it follows from Definition \ref{def:trivial-problem-instance} and Lemma \ref{lemma:min-degree-3} that $|N_B(v_i)| = |N_G(v_i)| \geq 3$. By Definition \ref{def:tree-decomposition} we have $N_B(v_i) \subseteq \chi(\bchi(v_i)) \setminus \{v_i\} = \chi(\{b_{i-1}, b_i\}) \setminus \{v_i\} = \{u, v_{i-1}, v_{i+1}\}$. Hence $N_G(v_i) = \{u, v_{i-1}, v_{i+1}\}$ holds. We display $G[\chi(T[b_2,b_6])]$ in Figure \ref{fig:lemma:reducing-long-paths-with-frequent-vertex-without-U:graph}.
    
     Because $X \cup \{u\}$ is a tidy modulator of $G$ and $\langle v_3, v_4, v_5, v_6 \rangle$ is a simple path in $G - (X \cup \{u\})$, and the component $C \in \mathcal{C}(G - (X \cup \{u, v_3, v_6\})$ that contains $\{v_4,v_5\}$ has $N_G(C) \cap (X \cup \{u\}) = N_G(\{v_4,v_5\}) \cap (X \cup \{u\}) = \{u\}$, which trivially is a limit-1 subset for $(G - uv_4,t)$, we can apply Reduction \ref{red:remove-middle-edge-to-isolated-path} on $uv_4 \in E(B)$. Furthermore, it is trivial that all these steps can be executed in $\polyG$ time.
\end{proof}

\begin{figure}[ht]
\centering
\begin{subfigure}{.95\linewidth}
    \centering
    \includegraphics[width=.9\linewidth,page=5]{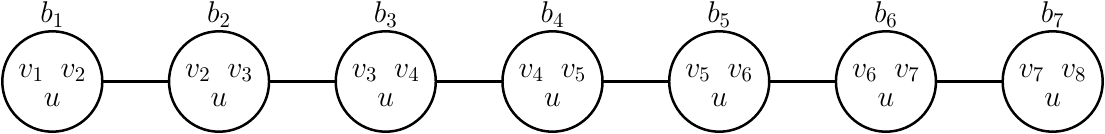}
    \caption{Tree $T$}
    \label{fig:lemma:reducing-long-paths-with-frequent-vertex-without-U:tree}
\end{subfigure}
~\\~\\~\\
\begin{subfigure}{.95\linewidth}
    \centering
    \includegraphics[width=.5\linewidth,page=6]{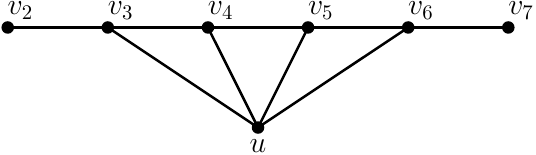}
    \caption{Graph $G[\chi(T[b_2,b_6])]$}
    \label{fig:lemma:reducing-long-paths-with-frequent-vertex-without-U:graph}
\end{subfigure}
\caption{Biconnected reduction instance with $|L| = 2 \wedge U \subseteq \{u\} \wedge V(T) \subseteq \bchi(u) \wedge |V(T)| = 7$}
\label{fig:lemma:reducing-long-paths-with-frequent-vertex-without-U:}
\end{figure}

\noindent
To generalise Lemma \ref{lemma:reducing-long-paths-with-frequent-vertex-without-U} we will look at the earliest node $b_v$ along $T$ where a vertex $v \in U \setminus \{u\}$ appears `for the last time'. If the path leading up to $b_v$ is small, then the suffix (path after $b_v$) will be long compared to the number of vertices in $U \setminus \{u,v\}$, in which case we recurse on this suffix. If the path leading up to $b_v$ is large, then on this prefix we will be able to apply Lemma \ref{lemma:reducing-long-paths-with-frequent-vertex-without-U}.

\begin{lemma}\label{lemma:reducing-long-paths-with-frequent-vertex-with-U}
    Let $(G,t,X,B,T^*,\chi,T,L,U)$ be a biconnected reduction problem instance with $|L| = 2$ and let $u \in V(B)$. If $V(T) \subseteq \bchi(u)$ and $|V(T)| \geq 9|U| + 7$, then we can in $\polyG$ time apply Reduction \ref{red:remove-middle-edge-to-isolated-path} on an edge in $B$.
\end{lemma}

\begin{proof}
    Because $|L| = 2$ we have that $T$ is a linear tree. We let $V(T) = \{b_1, b_2, \dots, b_m\}$ with $m = |V(T)|$ be such that each $1 < i < m$ has $N_T(b_i) = \{b_{i-1}, b_{i+1}\}$. We prove that we can apply Reduction \ref{red:remove-middle-edge-to-isolated-path} using induction on $U$.
    
    \begin{basecase}[$U \subseteq \{u\}$]
        Let $T' = T[b_1, b_7]$, let $L' = \{b_1,b_7\}$, and let $U' = U$. Then $(G,t,X,B,T^*,\chi,T',L',U')$ is a biconnected reduction problem instance. Correctness immediately follows from Lemma \ref{lemma:reducing-long-paths-with-frequent-vertex-without-U}.
    \end{basecase}
    
    \begin{inductivestep}[$U \not\subseteq \{u\}$]
        For each $v \in U \setminus \{u\}$ let $b_v$ be the node furthest from $b_1$ with $v \in \chi(b_v)$. Let $v \in U \setminus \{u\}$ be a vertex for which $b_v$ is closest to $b_1$. Let $1 \leq i \leq j \leq m$ such that $b_i$ is the vertex closest to $b_1$ with $v \in \chi(b_i)$ and $b_j = b_v$. Let $W = U \cap \chi(T[b_1,b_j])$. Due to our choice of $v$ we have $W \subseteq \chi(b_j)$.
        
        \begin{indsubcase}[$j \leq 9$]
            Let $T' = T(b_j, b_m]$, let $L' = \{b_{j+1}, b_m\}$, and let $U' = U \cap \chi(T') \subseteq U \setminus \{v\}$. Then $(G,t,X,B,T^*,\chi,T',L',U')$ is a biconnected reduction problem instance. We have $|V(T')| = |V(T)|-j \geq 9|U| + 7 - 9 \geq 9|U'| + 7$. Hence, when recursing on $(G,t,X,B,T^*,\chi,T',L',U')$, the removal of an edge in $B$ follows from the induction hypothesis.
        \end{indsubcase}
        
        \begin{indsubcase}[$U \cap {\chi(T[b_1,b_7])} \subseteq \{u\}$]
            Let $T' = T[b_1, b_7]$, let $L' = \{b_1, b_7\}$, and let $U' = U \cap \chi(T')$. Then $(G,t,X,B,T^*,\chi,T',L',U')$ is a biconnected reduction problem instance. Hence, when recursing on $(G,t,X,B,T^*,\chi,T',L',U')$, the removal of an edge in $B$ follows from the Lemma \ref{lemma:reducing-long-paths-with-frequent-vertex-without-U}.
        \end{indsubcase}
        
        \begin{indsubcase}[$j \geq 10 \wedge U \cap {\chi(T[b_1,b_7])} \not\subseteq \{u\}$]
            Let $w \in (U \cap \chi(T[b_1,b_7])) \setminus \{u\}$. By Lemma \ref{lemma:bchi-connected-in-subtree} and $w \in W \subseteq \chi(b_j)$ we have $V(T[b_7,b_{10}]) \subseteq \bchi(w)$. From Definition \ref{def:smooth-tree-decomposition} it follows that each $7 \leq k \leq 10$ has $\chi(b_k) = \{u,w,x_k\}$ for some vertex $x_k$, and these vertices $x_k$ are distinct. Lemma \ref{lemma:bchi-connected-in-subtree} therefore yields $\bchi(x_8) = \{b_8\}$ and $\bchi(x_9) = \{b_9\}$. By Lemma \ref{lemma:biconnected-td:beta-1-neighbours} this yields $\{x_9,x_{10}\} \subseteq W \setminus \{u,w\}$. Because $u,w \in W$ and $|W| \leq |\chi(j)| = 3$ we have $|W \setminus \{u,w\}| \leq 1$, which contradicts $u \neq w$.
        \end{indsubcase}
    \end{inductivestep}
    
    \noindent
    Each case either leads to a contradiction or describes a straightforward recursive algorithm that applies Reduction \ref{red:remove-middle-edge-to-isolated-path} on an edge in $B$. Because $|U| \leq |V(G)|$ the algorithm described by this recursion recurses at most $\polyG$ times. By Lemma \ref{lemma:size=of-smooth-tree-decomposition} we have $|V(T)| \leq \polyG$ from which it trivially follows that this algorithm takes at most $\polyG$ time.
\end{proof}

\subsection{Reducing Long Paths without Neighbours}
We next focus on the situation where $\chi(T)$ contains no vertices with neighbours outside of $B$, i.e., $U = \emptyset$. For this case we will prove that we can either apply Reduction \ref{red:remove-middle-edge-to-isolated-path} or \ref{red:ladder-reduction}.

We will first prove that we can find two vertex-disjoint induced paths $P$ and $Q$ that visit all and only these vertices $\chi(T - L)$. Once these paths are identified we will compute a maximal matching between their vertices. When this matching is of size at least 9 we will have obtained a graph with a $L_9$ ladder as minor on which Reduction \ref{red:ladder-reduction} can be applied. When this matching is of size strictly less than 9, then we will find a vertex in $P$ with many neighbours in $Q$ (or a vertex in $Q$ with many neighbours in $P$) on which we will be able to apply Reduction \ref{red:remove-middle-edge-to-isolated-path}.

We prove that we can find these two vertex-disjoint induced paths $P$ and $Q$ in two steps. We first prove this for trees $T$ with three vertices. Afterwards we inductively prove that we can construct such paths for trees $T$ with $|V(T)| \geq 3$.

\begin{lemma}\label{lemma:single-step-double-path}
    Let $(G,t,X,B,T^*,\chi,T,L,U)$ be a biconnected reduction problem instance with $|L| = 2$ and $U = \emptyset$. Let $b_2 \in V(T) \setminus L$ with $N_T(b_2) = \{b_1,b_3\}$. We can order the vertices in $\chi(b_2) = \{p,q,r\}$ such that $\chi(b_1) \cap \chi(b_2) = \{p,q\} \wedge \chi(b_2) \cap \chi(b_3) = \{p,r\} \wedge qr \in E(B)$.
\end{lemma}

\begin{proof}
    Definition \ref{def:smooth-tree-decomposition} yields $|\chi(b_1)| = |\chi(b_2)| = |\chi(b_3)| = 3 \wedge |\chi(b_1) \cap \chi(b_2)| = |\chi(b_2) \cap \chi(b_3)| = 2$. We order vertices $\chi(b_2) = \{p,q,r\}$ such that $\chi(b_1) \cap \chi(b_2) = \{p,q\} \wedge p \in \chi(b_2) \cap \chi(b_3)$.
    
    \begin{case}[$\chi(b_2) \cap \chi(b_3) = \{p,q\}$]
        From Lemma \ref{lemma:bchi-connected-in-subtree} and $r \notin \chi(b_1) \wedge r \notin \chi(b_3)$ we derive $\bchi(r) = \{b_2\}$, which by Lemma \ref{lemma:biconnected-td:beta-1-neighbours} yields the contradiction $r \in U = \emptyset$.
    \end{case}
    
    \begin{case}[$\chi(b_2) \cap \chi(b_3) = \{p,r\}$]
        Assume per contradiction $qr \notin E(B)$. Because $b_2 \notin L$ we have by Definition \ref{def:biconnected-reduction-problem-instance} that $N_{T^*}(b_2) = \{b_1, b_3\}$. Because $q \notin \chi(b_3)$ we have that $b_2$ is a leaf in $T^*[\bchi(q)]$. From Definition \ref{def:tree-decomposition} it trivially follows that removing $q$ from bag $b_2$ in $(T^*, \chi)$ results in a (non-smooth) tree decomposition $(T', \chi')$ of $B$. By Definition \ref{def:smooth-tree-decomposition} and Corollary \ref{corollary:treedecomp-edge-separator} we have that $\chi'(b_1) \cap \chi'(b_2) = \{p\}$ separates $\chi'(b_1) \setminus \chi'(b_2) \supset \{q\}$ from $\chi'(b_3) \setminus \chi'(b_2) = \chi(b_3) \setminus \chi(b_2) \neq \emptyset$ in $B$. This, however, contradicts $B$ being biconnected. Hence $qr \in E(B)$ must hold.
    \end{case}
    
    \noindent
    Because these cases cover all cases, correctness trivially follows.
\end{proof}

\begin{lemma}\label{lemma:many-steps-double-path}
   Let $(G,t,X,B,T^*,\chi,T,L,U)$ be a biconnected reduction problem instance with $L = \{\ell_1, \ell_2\}$ and $U = \emptyset$ and $|V(T)| \geq 3$. We can in $\polyG$ time obtain two vertex-disjoint induced paths $P = \langle p_1, p_2, \dots, p_m \rangle$ and $Q = \langle q_1, q_2, \dots, q_n \rangle$ in $B$ with $V(P) \cup V(Q) = \chi(T - L) \wedge p_1,q_1 \in \chi(\ell_1) \wedge p_m,q_n \in \chi(\ell_2)$.
\end{lemma}

\begin{proof}
    We prove using induction on $|V(T)|$.

    \begin{basecase}[$|V(T)| = 3$]
        We let $b$ be the single vertex in $V(T) \setminus L$. Lemma \ref{lemma:single-step-double-path} yields $\chi(b) = \{p,q,r\} \wedge \chi(\ell_1) \cap \chi(b) = \{p,q\} \wedge \chi(b) \cap \chi(\ell_2) = \{p,r\} \wedge qr \in E(B)$. We let $P = \langle p \rangle$ and $Q = \langle q, r \rangle$ which trivially satisfies the requirements for $P$ and $Q$.
    \end{basecase}

    \begin{inductivestep}[$|V(T)| > 3$]
        Let $b \in N_T(\ell_2)$, let $T' = T - \ell_2$, let $L' = \{\ell_1, b\}$, and let $U' = \emptyset$. Then $(G,t,X,B,T^*,\chi,T',L',U')$ is a biconnected reduction problem instance. By the induction hypothesis we obtain vertex-disjoint induced paths $P = \langle p_1, p_2, \dots, p_m \rangle$ and $Q = \langle q_1, q_2, \dots, q_n \rangle$ in $B$ with $V(P) \cup V(Q) = \chi(T' - L') \wedge p_1,q_1 \in \chi(\ell_1) \wedge p_m,q_n \in \chi(b)$.
    
        We let $b'$ such that $N_T(b) = \{b',\ell_2\}$. Because $b' \in V(T') \setminus L'$ we have by definition of $P$ and $Q$ that $\chi(b') \cap \chi(b) = \{p_m,q_n\}$. By Lemma \ref{lemma:single-step-double-path} we have $\chi(b) = \{p_m,q_n,r\}$ and either $\chi(b) \cap \chi(\ell_2) = \{p_m,r\} \wedge q_nr \in E(B)$ or $\chi(b) \cap \chi(\ell_2) = \{q_n,r\} \wedge p_mr \in E(B)$. W.l.o.g. we assume the former.
        
        Let $Q' = \langle q_1, q_2, \dots, q_n, r \rangle$. By Lemma \ref{lemma:bchi-connected-in-subtree} and $r \notin \chi(b')$ it follows that $r \notin \chi(T' - L')$. Hence, $P$ and $Q'$ must be simple and vertex-disjoint. We derive $V(P) \cup V(Q') = V(P) \cup V(Q) \cup \{p_m,q_n,r\} = \chi(T' - L') \cup \chi(b) = \chi(T - L)$. We argue that $Q'$ is an induced path. Assume per contradiction $q \in N_B(r) \cap (V(Q) \setminus \{q_n\})$. By Corollary \ref{corollary:treedecomp-edge-separator} we have that $\{p_m,q_n\}$ separates $q$ from $r$ in $B$, which contradicts $qr \in E(B)$. Therefore $Q'$ must be an induced path. We conclude that $P$ and $Q'$ are vertex-disjoint induced paths in $B$ with $V(P) \cup V(Q) = \chi(T - L) \wedge p_1,q_1 \in \chi(\ell_1) \wedge p_m,q_n \in \chi(\ell_2)$.
    \end{inductivestep}
    
    \noindent
    We note that our base case and inductive step directly describe how to obtain paths $P$ and $Q$. Because Lemma \ref{lemma:size=of-smooth-tree-decomposition} yields $|V(T)| \leq \polyG$ it trivially follows that this method allows us to obtain $P$ and $Q$ in $\polyG$ time.
\end{proof}

\noindent
By Lemma \ref{lemma:many-steps-double-path} we can now find two vertex-disjoint induced paths $P$ and $Q$ and we can prove (as we will do later) that these paths have $\partial_G(V(P) \cup V(Q)) = \{p_1,p_m,q_1,q_n\}$. I.e., the situation shown in Figure \ref{fig:many-steps-double-path}. Because these paths are induced in $G$ and $G$ has minimum degree three there will be edges between vertices in $P$ and $Q$. In the next step we will compute a maximal matching over these edges. In case the maximal matching is small we will find a structure that can be reduced using Reduction \ref{red:remove-middle-edge-to-isolated-path} (see Figure \ref{fig:lemma:reducing-long-paths-with-frequent-vertex-without-U:graph}). In case the maximal matching is large we will instead find a structure that can be reduced using Reduction \ref{red:ladder-reduction} (see Figure \ref{fig:ladder-reduction}).

\begin{figure}[ht]
    \centering
    \includegraphics[scale=.75,page=8]{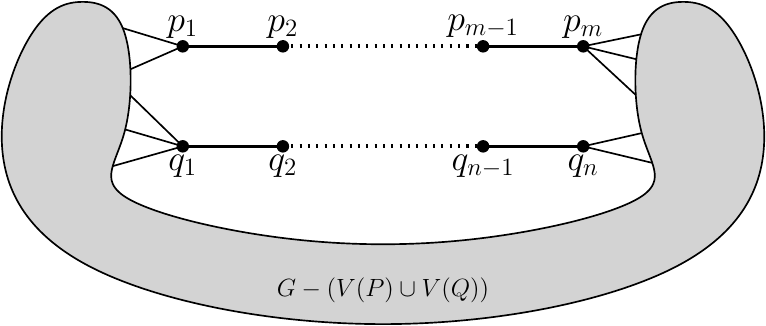}
    \caption{The two vertex-disjoint induced paths $P$ and $Q$ found in $T$ with $U = \emptyset$}
    \label{fig:many-steps-double-path}
\end{figure}

\begin{lemma}\label{lemma:find-ladder-reduction}
    Let $(G,t,X,B,T^*,\chi,T,L,U)$ be a biconnected reduction problem instance with $L = \{\ell_1, \ell_2\}$ and $U = \emptyset$ and $|V(T)| \geq 86$. We can in $\polyG$ time apply either Reduction \ref{red:remove-middle-edge-to-isolated-path} or \ref{red:ladder-reduction} on an edge in $B$.
\end{lemma}

\begin{proof}
    By Lemma \ref{lemma:many-steps-double-path} we can in $\polyG$ time obtain two vertex-disjoint induced paths $P = \langle p_1, p_2, \dots, p_m \rangle$ and $Q = \langle q_1, q_2, \dots, q_n \rangle$ in $G$ with $V(P) \cup V(Q) = \chi(T - L) \wedge p_1,q_1 \in \chi(\ell_1) \wedge p_m,q_n \in \chi(\ell_2)$.
    
    We first prove that $\treewidth{G[V(P) \cup V(Q)] \cup \{p_1q_1, p_mq_n\}} \leq 2$. By Definition \ref{def:smooth-tree-decomposition} we know that there exists a vertex $u \in \chi(\ell_1) \setminus \{p_1,q_1\}$ and a vertex $v \in \chi(\ell_2) \setminus \{p_m,q_n\}$. Let $D_1 \in \mathcal{C}(B - \{p_1,q_1\})$ such that $u \in V(D_1)$ and $D_2 \in \mathcal{C}(B - \{p_m,q_n\})$ such that $v \in V(D_2)$. By Lemma \ref{lemma:biconnected-has-cycle} there exists a simple cycle $C$ in $B$ that contains $u$ and $v$. By Corollary \ref{corollary:treedecomp-edge-separator} we have that $\{p_1,q_1\}$ separates $u$ from $v$ in $B$ (and hence also in $C$), which implies that there exists a path $R_1$ from $p_1$ to $q_1$ in $C[V(D_1) \cup \{p_1,q_1\}]$. Similarly, there exists a path $R_2$ from $p_m$ to $q_n$ in $C[V(D_2) \cup \{p_m,q_n\}]$. Contracting $R_1 - \{q_1\}$ into $p_1$ and $R_2 - \{q_n\}$ in $p_m$ results in $G[V(P) \cup V(Q)] \cup \{p_1q_1, p_mq_n\}$ as a minor of $B$. By Lemma \ref{lemma:minor-treewidth} this yields $\treewidth{G[V(P) \cup V(Q)] \cup \{p_1q_1, p_mq_n\}} \leq \treewidth{B} \leq \treewidth{G - X} \leq 2$.
    
    We prove that $N_G[(V(P) \cup V(Q)) \setminus \{p_1,p_m,q_1,q_n\}] \subseteq V(P) \cup V(Q)$. By Corollary \ref{corollary:treedecomp-edge-separator} we have that $\{p_1,p_m,q_1,q_n\}$ separates $\chi(T - L) \setminus \{p_1, p_m, q_1, q_n\}$ from $V(B) \setminus \chi(T - L)$ in $B$. Because $U = \emptyset$ we can therefore derive $N_G[(V(P) \cup V(Q)) \setminus \{p_1,p_m,q_1,q_n\}] = N_B[(V(P) \cup V(Q)) \setminus \{p_1,p_m,q_1,q_n\}] = N_B[\chi(T - L) \setminus \{p_1, p_m, q_1, q_n\}] \subseteq \chi(T - L) = V(P) \cup V(Q)$.

    Let $H$ be a bipartite graph with $V(H) = V(P) \cup V(Q)$ and $E(H) = (V(P) \times V(Q)) \cap E(B)$, which we can trivially construct in $\polyG$ time. We argue that each $p \in V(P) \setminus \{p_1, p_m\}$ has $N_H(p) \neq \emptyset$. Because $p \in (V(P) \cup V(Q)) \setminus \{p_1,p_m,q_1,q_n\}$ we have $N_G(p) \subseteq V(P) \cup V(Q)$. Because $P$ is an induced path we have $|N_B(p) \cap V(P)| \leq 2$. Definition \ref{def:trivial-problem-instance} and Lemma \ref{lemma:min-degree-3} therefore yields $N_B(p) \cap V(Q) \neq \emptyset$, which implies $N_H(p) \neq \emptyset$. A symmetric argument proves that each $q \in V(Q) \setminus \{q_1, q_n\}$ has $N_H(q) \neq \emptyset$. This then gives $|E(H)| \geq (|V(P)| + |V(Q)| - 4) / 2$. By Lemma \ref{lemma:size=of-smooth-tree-decomposition} we have $|V(T)| = |\chi(T)| - 2  = |\chi(T-L)| = |V(P)| + |V(Q)|$, which yields $|E(H)| \geq (|V(T)| - 4) / 2 \geq 41$. We can in $\polyG$ time obtain a maximal matching $M$ in $H$ (see Observation \ref{obs:get-maximal-matching}).
    
    \begin{case}[$|M| \leq 8$]
        Let $H' = H - M$. Because $M$ is a maximal matching, each $pq \in E(H')$ has $p \in V(M) \vee q \in V(M)$. Because $|V(M)| \leq 16$ there exists some $r \in V(M)$ with $|N_{H'}(r)| \geq \lceil |E(H')| / 16 \rceil \geq \lceil (|E(H)| - 8) / 16 \rceil \geq  \lceil 33 / 16 \rceil \geq 3$, which yields $|N_H(r)| = |N_{H'}(r)| + 1 \geq 4$. W.l.o.g. we assume $r \in V(Q)$. Let $r = q_j \in V(Q)$ and $\{p_a, p_b, p_c, p_d\} \subseteq N_H(q_j)$ with $1 \leq a < b < c < d \leq m$.
    
        Because $V(P(p_a,p_d)) \subseteq (V(P) \cup V(Q)) \setminus \{p_1,p_m,q_1,q_n\}$ we have $N_G[P(p_a,p_d)] \subseteq V(P) \cup V(Q)$. By Lemma \ref{lemma:non-cross-double-path} we have that $\{p_a,p_d,q_j\}$ separates $V(P(p_a,p_d))$ from $V(P[p_1, p_a)) \cup V(Q[q_1,q_j)) \cup V(P(p_d, p_m]) \cup V(Q(q_j, q_n])$ in $G[V(P) \cup V(Q)]$. This therefore yields that $N_G[P(p_a,p_d)] \subseteq (V(P) \cup V(Q)) \setminus (V(P[p_1, p_a)) \cup V(Q[q_1,q_j)) \cup V(P(p_d, p_m]) \cup V(Q(q_j, q_n])) = V(P[p_a,p_d]) \cup \{q_j\}$. I.e., we obtain the situation shown in Figure \ref{fig:find-ladder-reduction}.
    
        We note that $X \cup \{q_j\}$ is a tidy modulator of $G$ and $P(p_a,p_d) \in \mathcal{C}(G - (X \cup \{p_a,p_d,q_j\}))$ with $N_G(P(p_a,p_d)) \cap (X \cup \{q_j\}) = \{q_j\}$. Because $\{q_j\}$ is a limit-1 clique for $(G - p_bq_j,t)$ we can apply Reduction \ref{red:remove-middle-edge-to-isolated-path} on $p_bq_j \in E(B)$.
    \end{case}

    \begin{case}[$|M| \geq 9$]
        Let $1 \leq i, i' \leq m$ such that $p_i$ and $p_{i'}$ are the first and last vertex along $P$ in $V(M)$ respectively. Similarly, let $1 \leq j, j' \leq n$ such that $q_j$ and $q_{j'}$ are the first and last vertex along $Q$ in $V(M)$ respectively. From Lemma \ref{lemma:non-cross-double-path} it follows that all edges in $M$ can be ordered corresponding to the order in which their endpoints appear along both $P$ and $Q$. I.e., the graph $G[V(P) \cup V(Q)]$ has a $L_9$ minor, where the ladder rungs are edges in $M$.
        
        By Lemma \ref{lemma:non-cross-double-path} we have that $\{p_i,p_{i'},q_j,q_{j'}\}$ separates $V(P(p_i,p_{i'})) \cup V(Q(q_j,q_{j'}))$ from $V(P[p_1, p_i)) \cup V(Q[q_1,q_j)) \cup V(P(p_{i'}, p_m]) \cup V(Q(q_{j'}, q_n])$ in $G[V(P) \cup V(Q)]$. Because $U = \emptyset$ we therefore have $N_G[V(P(p_i, p_{i'})) \cup V(Q(q_j, q_{j'}))] \subseteq V(P[p_i, p_{i'}]) \cup V(Q[q_j, q_{j'}])$. Therefore we have $\partial_G(V(P[p_i, p_{i'}]) \cup V(Q[q_j, q_{j'}])) \subseteq \{p_i, p_{i'}, q_j, q_{j'}\}$. From the definitions of $P$, $Q$, $M$, $i$, $i'$, $j$, and $j'$ it then follows that we can apply Reduction \ref{red:ladder-reduction} on $P[p_i,p_{i'}]$ and $Q[q_j,q_{j'}]$, which removes some edge $pq \in M \subseteq E(B)$.
    \end{case}
    
    \noindent
    Because these cases cover all cases, it trivially follows that we can in $\polyG$ time apply either Reduction \ref{red:remove-middle-edge-to-isolated-path} or \ref{red:ladder-reduction} on some edge in $B$. 
\end{proof}

\begin{figure}[ht]
    \centering
    \includegraphics[scale=.8,page=9]{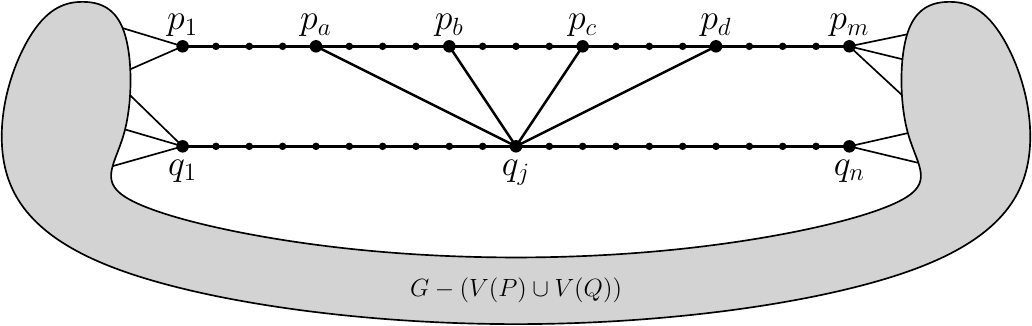}
    \caption{Reducible case where $q_j \in V(Q)$ has four neighbours in $V(P)$}
    \label{fig:find-ladder-reduction}
\end{figure}

\subsection{Reducing Long Paths}
We next combine the situation where a linear tree $T$ has a vertex $u \in \chi(b)$ with $|\bchi(u) \cap V(T)| \geq 9|U|+7$ and the case where $U = \emptyset$ to prove that any linear subtree $T$ of $T^*$ that is large compared to $|U|$ leads to an application of Reduction \ref{red:remove-middle-edge-to-isolated-path} or \ref{red:ladder-reduction}.

We will use a proof similar to the proof of Lemma \ref{lemma:reducing-long-paths-with-frequent-vertex-with-U}. We look at the earliest node $b_u$ along $T$ where a vertex $u \in U$ appears `for the last time'. If the path leading up to $b_u$ is small, then the suffix will be long compared to the number of vertices in $U \setminus \{u\}$, in which case we recurse on this suffix. If the first 86 bags do not contain any vertex from $U$, then on this path we can apply Lemma \ref{lemma:find-ladder-reduction}. If, on the other hand, some of the first 86 bags contain a vertex $w \in U$ and the prefix to $b_u$ is large, then by our choice of $u$ we have found a large linear subtree $T'$ of $T$ with $V(T') \subseteq \bchi(w)$, which allows us to apply Lemma \ref{lemma:reducing-long-paths-with-frequent-vertex-with-U}.

\begin{lemma}\label{lemma:reduce-long-path}
    Let $(G,t,X,B,T^*,\chi,T,L,U)$ be a biconnected reduction problem instance with $|L| = 2$. If $|V(T)| \geq 118|U| + 86$, then we can in $\polyG$ time apply Reduction \ref{red:remove-middle-edge-to-isolated-path} or \ref{red:ladder-reduction} on an edge in $B$.
\end{lemma}

\begin{proof}
    Because $|L| = 2$ we have that $T$ is a linear tree. We let $V(T) = \{b_1, b_2, \dots, b_m\}$ with $m = |V(T)|$ such that each $1 < i < m$ has $N_T(b_i) = \{b_{i-1}, b_{i+1}\}$. We prove that we can apply a reduction using induction on $|U|$.
    
    \begin{basecase}[$U = \emptyset$]
        Because $|V(T)| \geq 86$ correctness follows from Lemma \ref{lemma:find-ladder-reduction}.
    \end{basecase}
    
    \begin{inductivestep}[$U \neq \emptyset$]
        For each $u \in U$ let $b_u$ be the node furthest from $b_1$ with $u \in \chi(b_u)$. Let $u \in U$ be a vertex for which $b_u$ is closest to $b_1$. Let $1 \leq i \leq j \leq m$ such that $b_i$ is the vertex closest to $b_1$ with $u \in \chi(b_i)$ and $b_j = b_u$. Let $W = U \cap V(T[b_1,b_j])$. Due to our choice of $u$ we have $W \subseteq \chi(b_j)$.
        
        \begin{indsubcase}[$j \leq 118$]
            Let $T' = T(b_j,b_m]$, let $L' = \{b_{j+1}, b_m\}$, and let $U' = U \cap \chi(T') \subseteq U \setminus \{u\}$. Then $(G,t,X,B,T^*,\chi,T',L',U')$ is a biconnected reduction problem instance. We have $|V(T')| = |V(T)| - j \geq 118|U|+86-118 \geq 118|U'|+86$. Hence, when recursing on $(G,t,X,B,T^*,\chi,T',L',U')$, the removal of an edge in $B$ follows from the induction hypothesis.
        \end{indsubcase}
        
        \begin{indsubcase}[$U \cap {\chi(T[b_1, b_{86}])} = \emptyset$]
            Let $T' = T[b_1, b_{86}]$, let $L' = \{b_1, b_{86}\}$, and let $U' = \emptyset$. Then $(G,t,X,B,T^*,\chi,T',L',U')$ is a biconnected reduction problem instance. Hence, when recursing on $(G,t,X,B,T^*,\chi,T',L',U')$, the removal of an edge in $B$ follows from Lemma \ref{lemma:find-ladder-reduction}.
        \end{indsubcase}
        
        \begin{indsubcase}[$j \geq 119 \wedge U \cap {\chi(T[b_1, b_{86}])} \neq \emptyset$]
            Let $T' = T[b_{86}, b_{119}]$, let $L' = \{b_{86}, b_{119}\}$, and let $U' = U \cap \chi(T')$. Then $(G,t,X,B,T^*,\chi,T',L',U')$ is a biconnected reduction problem instance. Because $\emptyset \neq U \cap \chi(T[b_1, t_{86}]) \subseteq W \subseteq \chi(b_j)$ we have by Lemma \ref{lemma:bchi-connected-in-subtree} that $V(T') \subseteq \bchi(w)$ for some $w \in W$. Furthermore, we have $|U'| \leq |W| \leq |\chi(b_j)| = 3$, which implies $|V(T')| = 34 \geq 9|U'| + 7$. Hence, when recursing on $(G,t,X,B,T^*,\chi,T',L',U')$, the removal of an edge in $B$ follows from Lemma \ref{lemma:reducing-long-paths-with-frequent-vertex-with-U}.
        \end{indsubcase}
    \end{inductivestep}
    
    \noindent
    Each case describes a straightforward recursive algorithm that applies Reduction \ref{red:remove-middle-edge-to-isolated-path} or \ref{red:ladder-reduction} on an edge in $B$. Because $|U| \leq |V(G)|$ the algorithm described by this recursion recurses at most $\polyG$ times. By Lemma \ref{lemma:size=of-smooth-tree-decomposition} we have $|V(T)| \leq \polyG$ from which it trivially follows that this algorithm takes at most $\polyG$ time.
\end{proof}

\subsection{Reducing Large Trees}
We have now shown that any linear subtree $T$ that is large compared to $|U|$ leads to an application of Reduction \ref{red:remove-middle-edge-to-isolated-path} or \ref{red:ladder-reduction}. It remains to prove that any non-linear tree $T$ that is large compared to $|U|$ will have such a linear subtree $T'$.

We will use the branching rule that was shown in Figure \ref{fig:finding-reducible-structures:recursive-step}. I.e., we choose an arbitrary leaf vertex to be the root of $T$ and we find a vertex $b$ furthest from the root with degree at least three. All subtrees rooted under $b$ will be linear subtrees. This means that when any of them is large that correctness follows from Lemma \ref{lemma:reduce-long-path}. Otherwise, we can remove these subtrees rooted under $b$ and be left with a large remaining subtree with strictly fewer leaves.

\begin{lemma}\label{lemma:reduce-big-biconnected-tree}
    Let $(G,t,X,B,T^*,\chi,T,L,U)$ be a biconnected reduction problem instance with $|L| \geq 2$. If $|V(T)| \geq 876|L| + 118|U|$, then we can in $\polyG$ time apply Reduction \ref{red:remove-middle-edge-to-isolated-path} or \ref{red:ladder-reduction} on an edge in $B$.
\end{lemma}

\begin{proof}
    We prove using induction on $|L|$.

    \begin{basecase}[$|L| = 2$]
        Because $|V(T)| \geq 876|L| + 118|U| > 118|U| + 86$ correctness follows from Lemma \ref{lemma:reduce-long-path}.
    \end{basecase}

    \begin{inductivestep}[$|L| > 2$]
        We pick an arbitrary node from $L$ to be the root of $T$. We let $b \in V(T)$ be a vertex furthest from the root with $|N_T(b)| > 2$. We let $F$ be the subtree of $T$ rooted at $b$. Due to our choice of $b$ every subtree of $F - b$ is linear. We have $|L \cap V(F)| = |N_T(b)| - 1 \geq 2$.
    
        For each $\ell \in L \cap V(F)$ let $T_\ell = T[V(C_\ell) \cup \{b\}]$ where $C_\ell$ is the component in $F - b$ that contains $\ell$. Let $L_\ell = \{\ell, b\}$ and $U_\ell = U \cap \chi(T_\ell) \subseteq U \cap \chi(F)$. Then $(G,t,X,B,T^*,\chi,T_\ell,L_\ell,U_\ell)$ is a biconnected reduction problem instance. For any $\{\ell,\ell'\} \subseteq L \cap V(F)$ it follows from Corollary \ref{corollary:treedecomp-node-separator} that $\chi(b)$ separates $\chi(U_\ell) \setminus \chi(b) \subseteq \chi(T_\ell) \setminus \chi(b)$ from $\chi(U_{\ell'}) \setminus \chi(b) \subseteq \chi(T_{\ell'}) \setminus \chi(b)$ in $B$, which implies $U_\ell \cap U_{\ell'} \subseteq \chi(b)$. By Lemma \ref{lemma:reduce-long-path} we have that any $\ell \in L \cap V(F)$ with $|V(T_\ell)| \geq 118|U_\ell| + 86$ yields that recursing on $(G,t,X,B,T^*,\chi,T_\ell,L_\ell,U_\ell)$ leads to the removal of an edge in $B$. Therefore, for the remaining case we assume no such $\ell$ exists.
        
        Let $T' = T - (F - b)$, let $L' = (L \setminus V(F)) \cup \{b\}$, and let $U' = U \cap \chi(T')$. Then $(G,t,X,B,T^*,\chi,T',L',U')$ is a biconnected reduction problem instance. We note that $|L'| = |L| - |L \cap V(F)| + 1$ and $U' = U \setminus (\chi(F) \setminus \chi(b))$ which implies $|U'| = |U| - |U \cap (\chi(F) \setminus \chi(b))|$. We derive $|V(T')| \geq 876|L'| + 118|U'|$.
\begin{align*}
    ~& |V(T')|                                                                               &\\[1.5ex]
=   ~& |V(T)| - |V(F){{\setminus}}\{b\}|                                                       &~ \text{definition of $T'$}\\[1ex]
=   ~& |V(T)| - \sum_{\ell \in L{{\cap}}V(F)} |V(T_\ell){{\setminus}}\{b\}|                      &~ \text{definition of $T_\ell$}\\[.8ex]
\geq~& |V(T)| - \sum_{\ell \in L{{\cap}}V(F)} (118|U_\ell| + 84)                               &~ |V(T_\ell)| \leq 118|U_\ell| + 85\\[.8ex]
\geq~& |V(T)| - \sum_{\ell \in L{{\cap}}V(F)} (118(|U_\ell{{\setminus}}\chi(b)| + 3) + 84)       &~ |\chi(b)| \leq 3\\[.8ex]
=   ~& |V(T)| - 438|L{{\cap}}V(F)| - 118\sum_{\ell \in L{{\cap}}V(F)} |U_\ell{{\setminus}}\chi(b)| &~ \text{rewriting}\\[.8ex]
\geq~& |V(T)| - 438|L{{\cap}}V(F)| - 118|U{{\cap}}(\chi(F){{\setminus}}\chi(b))|                   &~ U_\ell \subseteq U{{\cap}}\chi(F) \wedge U_\ell{{\cap}}U_{\ell'} \subseteq \chi(b)\\[1.5ex]
\geq~& 876|L| + 118|U| - 438|L{{\cap}}V(F)| - 118|U{{\cap}}(\chi(F){{\setminus}}\chi(b))|          &~ |V(T)| \geq 876|L| + 118|U|\\[1.5ex]
=   ~& 876|L| + 118|U'| - 438|L{{\cap}}V(F)|                                                   &~ |U'| = |U| - |U{{\cap}}(\chi(F){{\setminus}}\chi(b))|\\[1.5ex]
=   ~& 876|L'| + 438|L{{\cap}}V(F)| - 876 + 118|U'|                                            &~ |L'| = |L| - |L{{\cap}}V(F)| + 1\\[1.5ex]
\geq~& 876|L'| + 118|U'|                                                                     &~ |L{{\cap}}V(F)| \geq 2
\end{align*}
        \noindent
        Because $|V(T')| \geq 876|L'| + 118|U'|$ and $|L'| < |L|$, recursing on $(G,t,X,B,T^*,\chi,T',L',U')$ leads to the removal of an edge in $B$, as follows from the induction hypothesis.
    \end{inductivestep}
    
    \noindent
    Each case describes a straightforward recursive algorithm that applies Reduction \ref{red:remove-middle-edge-to-isolated-path} or \ref{red:ladder-reduction} on an edge in $B$. By Lemma \ref{lemma:size=of-smooth-tree-decomposition} we have $|L| \leq |V(T)| \leq \polyG$ which implies that the algorithm described by this recursion recurses at most $\polyG$ times. Furthermore, it trivially follows that each recursive step takes at most $\polyG$ time.
\end{proof}

\begin{theorem}\label{theorem:reduce-big-biconnected-graph}
    Let $(G,t)$ be a non-trivial problem instance, let $X$ be a tidy modulator for $G$, and let $B$ be a biconnected induced subgraph of $G - X$. If $|V(B)| \geq 1988|X| + 2$, then we can in $\polyG$ time apply Reduction \ref{red:remove-middle-edge-to-isolated-path} or \ref{red:ladder-reduction} on an edge in $B$.
\end{theorem}

\begin{proof}
    By Definition \ref{def:trivial-problem-instance} we have $X \neq \emptyset$, which implies $|V(B)|\geq 3$. Therefore, Lemma \ref{lemma:get-smooth-tree-decomposition} yields that we can obtain in $\polyG$ time a smooth tree decomposition $(T,\chi)$ of $B$ of width 2. By Lemma \ref{lemma:size=of-smooth-tree-decomposition} we have $|V(T)| \geq 1988|X|$. We let $L$ and $U$ be such that $(G,t,X,B,T,\chi,T,L,U)$ is a biconnected reduction problem instance. By Lemma \ref{lemma:biconnected-td:number-of-vertices-with-neighbours-outside-of-biconnected} we have $|U| \leq 2|X|$ and by Lemma \ref{lemma:biconnected-td:leaf-has-beta-1} we have $2 \leq |L| \leq |U| \leq 2|X|$. Because $|V(T)| \geq 1988|X| \geq 876|L| + 118|U|$ correctness immediately follows from Lemma \ref{lemma:reduce-big-biconnected-tree}.
\end{proof}

\chapter{Reducing Block-Cut Trees}\label{ch:reducing-block-cut-trees}
In the previous chapter we have given a method that, given a large biconnected induced subgraph $B$ in $G - X$, finds an application of a reduction that removes an edge from $B$. Within this chapter we complement this method by providing reduction rules that reduce the number of maximal biconnected subgraphs in any connected component in $G - (X \cup Y)$, where $X$ is a tidy modulator of $G$ and $Y$ is a component separator for $(G,t,X)$.

To do so we use the notion of a block-cut graph. We base the definition of a block-cut graph on the definition by West \cite[Chapter 4]{Introduction-to-Graph-Theory}. A block-cut graph of a graph $G$ is a bipartite graph consisting of blocks and cuts for vertices. A cut is a vertex $a \in V(G)$ that has $|\mathcal{C}(G - a)| > |\mathcal{C}(G)|$, I.e., its removal strictly increases the number of connected components. Cut vertices are also called articulation vertices. A block $B$ is a maximal biconnected subgraph of $G$. This means that graph $B$ has no articulation vertices, however, articulation vertices of $G$ can be non-articulation vertices in $B$. We furthermore note that any maximal biconnected subgraph of $G$ will be an induced subgraph of $G$.

By Lemma 3.1.4 in \cite{diestel_graph_theory} we have that a block-cut graph of a connected graph $G$ is a tree. Because we will only consider block-cut graphs of connected components in $G - (X \cup Y)$, all our block-cut graphs will be trees.

\begin{definition}[block-cut tree]\label{def:block-cut-tree}
    A block-cut tree $T = (A \cup \mathcal{B}, E)$ of a connected graph $G$ has $A$ being the set of articulation vertices of $G$ and $\mathcal{B}$ being the set of maximal biconnected subgraphs of $G$, and for each $a \in A$ and $B \in \mathcal{B}$ we have $aB \in E(T) \Longleftrightarrow a \in V(B)$.
\end{definition}

\noindent
We display a block-cut tree of an example graph (taken from \cite{Introduction-to-Graph-Theory}) in Figure \ref{fig:block-cut-tree-example}. We will always use black dots to denote (articulation) vertices, and white dots to denote blocks. Furthermore, we use $A(F)$ and $\mathcal{B}(F)$ to respectively refer to the set of articulation vertices and the set of blocks in a subgraph $F$ of a block-cut tree $T$. We define these sets also on subgraphs of $T$, opposed to only $T$ itself, because we will want use these definitions on paths in $T$ and subgraphs of $T$ as well. For a set of blocks $\mathcal{B}$ we let $V(\mathcal{B}) = \bigcup_{B \in \mathcal{B}} V(B)$.

\begin{figure}[ht]
    \centering
    ~
    \begin{subfigure}{.50\textwidth}
      \centering
      \includegraphics[scale=.66, page=1]{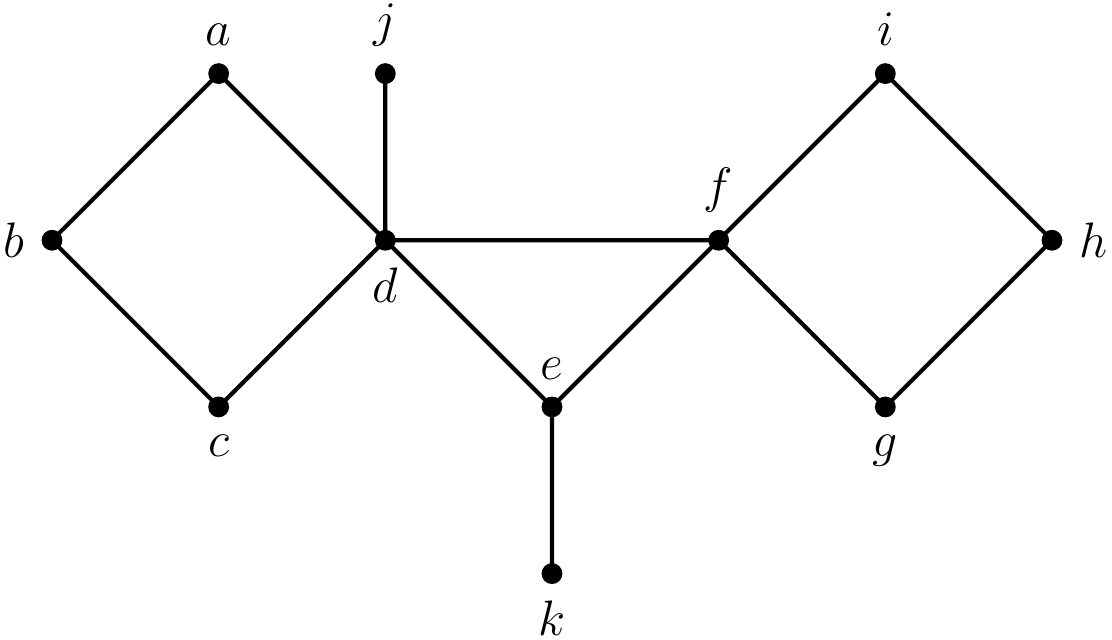}  
      \caption{Graph $G$}
      \label{fig:block-cut-tree-example:G}
    \end{subfigure}
    ~
    \begin{subfigure}{.45\textwidth}
      \centering
      \includegraphics[scale=.66, page=2]{block-cut-tree-example.pdf}
      \caption{Block-cut tree $T$ of $G$}
      \label{fig:block-cut-tree-example:T}
    \end{subfigure}
    \caption{An example graph with a corresponding block-cut tree}
    \label{fig:block-cut-tree-example}
\end{figure}
\clearpage
\noindent
We note that a block-cut tree $T$ of $G$ will have $|V(T)| = O(|A(T)|) \leq O(|V(G)|)$ and each block $B \in \mathcal{B}(T)$ has $|V(B)| \leq |V(G)|$. Hence, we can make the following observation:

\begin{observation}\label{obs:block-cut-tree-polyg-size}
    A block-cut tree $T$ of a connected graph $G$ has $|T| = \polyG$.
\end{observation}

\noindent
We next cite and prove important properties of block-cut trees, which will be used in later subsections.

\begin{lemma}[Algorithm 4.1.23 in \cite{Introduction-to-Graph-Theory}]\label{lemma:get-block-cut-tree}
    A block-cut tree of a connected graph $G$ can be obtained in $\polyG$ time.
\end{lemma}

\begin{lemma}[Proposition 4.1.19 in \cite{Introduction-to-Graph-Theory}]\label{lemma:blocks-share-at-most-one-cut-vertex}
    Two distinct blocks share at most one vertex, which will be an articulation vertex.
\end{lemma}

\begin{lemma}\label{lemma:articulation-vertices-as-separators}
    Let $G$ be a connected graph, let $T$ be a block-cut tree of $G$, let $a \in A(T)$, and let $T' \in \mathcal{C}(T - a)$. Then $a$ separates $V(\mathcal{B}(T')) \setminus \{a\}$ from $V(G) \setminus V(\mathcal{B}(T'))$ in $G$.
\end{lemma}

\begin{proof}
    Assume per contradiction that there exists an edge $uv$ between a vertex $u \in V(\mathcal{B}(T')) \setminus \{a\}$ and vertex $v \in V(G) \setminus V(\mathcal{B}(T'))$. Then $G[\{u,v\}]$ is a (potentially non-maximal) biconnected subgraph of $G$. Hence, there exists some block $B \in \mathcal{B}(T)$ of which $G[\{u,v\}]$ is a subgraph. Because $v \in V(G) \setminus V(\mathcal{B}(T'))$ we have $B \in \mathcal{B}(T) \setminus \mathcal{B}(T')$. Because $u \in V(\mathcal{B}(T')) \setminus \{a\}$ there exists a block $B' \in \mathcal{B}(T')$ with $u \in V(B')$. This yields $u \in V(B) \cap V(B')$, which by Definition \ref{def:block-cut-tree} and Lemma \ref{lemma:blocks-share-at-most-one-cut-vertex} implies $u \in A(T) \setminus \{a\}$ and $\{B, B'\} \subseteq N_T(u)$. This, however, yields the contradiction $B \in \mathcal{B}(T')$. Hence, $a$ separates $V(\mathcal{B}(T')) \setminus \{a\}$ from $V(G) \setminus V(\mathcal{B}(T'))$ in $G$.
\end{proof}

\begin{lemma}\label{lemma:block-cut-tree-only-blocks-as-leaves}
    A block-cut tree $T$ of a connected graph $G$ has only blocks as leaves.
\end{lemma}

\begin{proof}
    Let $a \in A(T)$. We have $|\mathcal{C}(G - a)| > |\mathcal{C}(G)| = 1$, which implies that there exist at least two vertices $\{u,v\} \subseteq N_G(a)$ in different components of $G - a$. Because $G[\{a,u\}]$ and $G[\{a,v\}]$ are (potentially non-maximal) biconnected subgraphs of $G$ there must exist some blocks $B_u, B_v \in \mathcal{B}(T)$ of which $G[\{a,u\}]$ and $G[\{a,v\}]$ are subgraphs respectively. Because $u$ and $v$ are separated in $G$ by $a$ there can not exist a biconnected component in $G$ that contains both $u$ and $v$, which implies $B_u \neq B_v$. By Definition \ref{def:block-cut-tree} this yields $|N_T(a)| \geq 2$. Hence, $a$ can not be a leaf in $T$, which proves that all leaves in $T$ are blocks.
\end{proof}

\begin{lemma}\label{lemma:block-cut-tree-only-blocks-of-size-at-least-2}
    A block-cut tree $T$ of a connected graph $G$ with $|V(G)| > 1$ has only blocks consisting of at least two vertices.
\end{lemma}

\begin{proof}
    Assume per contradiction that $B$ is a block in $G$ with $|V(B)| \leq 1$. Because the empty graph is not a maximal biconnected subgraph of a non-empty graph $G$ we must have $|V(B)| = 1$. Let $u \in V(B)$. Because $G$ is connected and $|V(G)| > 1$ we have $N_G(u) \neq \emptyset$. Let $v \in N_G(u)$. Then $G[\{u,v\}]$ is a biconnected subgraph of $G$ and $B$ is a strict subgraph of $G[\{u,v\}]$, which contradicts $B$ being a maximal biconnected subgraph. Hence, each block consists of at least two vertices.
\end{proof}

\begin{corollary}\label{corollary:block-cut-tree-only-blocks-of-size-at-least-2}
    A block-cut tree $T$ of a connected graph $G$ with $|V(T)| > 1$ has only blocks consisting of at least two vertices.
\end{corollary}

\begin{proof}
    Because $T$ is a connected bipartite graph we have $A(T) \neq \emptyset$. Let $a \in A(T)$. Then $|V(G) \setminus \{a\}| \geq |C(G - a)| > |C(G)| \geq 1$ yields $|V(G)| \geq 3$. Correctness follows from \ref{lemma:block-cut-tree-only-blocks-of-size-at-least-2}.
\end{proof}

\noindent
Before introducing our method to reducing large block-cut trees, we first provide a high-level overview of our approach.

We start by proving that when a block-cut tree $T$ of some component $C \in \mathcal{C}(G - (X \cup Y))$ has many leaves, we are able to contract one of the leaf blocks into the articulation vertex they are adjacent to. We use this to provide a bound on the number of leaves in any block-cut tree we wish to further reduce. As we will prove later, the removal of all branching vertices (vertices with degree at least three) in $T$ results in a linear forest consisting of a bounded number of trees.

We will then show that any vertex $y \in N_G(C) \cap Y$ with many neighbours in $C \in \mathcal{C}(G - (X \cup Y))$ allows us to find a large biconnected induced graph in $G - X$, which can be reduced using Theorem \ref{theorem:reduce-big-biconnected-graph}. In case no such vertex $y$ exists, then we obtain a bound on the number of blocks $B \in \mathcal{B}(T)$ with neighbours in $Y$. By also removing all blocks $B$ with neighbours in $Y$ from the linear forest obtained in previous step we obtain a linear forest consisting of a bounded number of trees, where each block of this forest has $N_G(B) \setminus V(C) \subseteq X$.

In case there are three adjacent blocks in this linear forest that do not have neighbours in $X$, then we can show that we can safely contract one of these blocks into an articulation vertex.

In case there are many blocks that are adjacent to the same vertex in $x \in X$, then we will be able to find a simple path $Q$ that visits many neighbours of $x$. We will then prove that that we can apply Reduction \ref{red:remove-middle-edge-to-isolated-path} on an edge $ux$ with $u \in V(Q)$.

Because our linear forest will either have many blocks adjacent to vertices of $X$, or will have three adjacent blocks that are not adjacent to any vertex in $X$, we will be able to apply either one of these two reduction rules on any large linear subtree.

\section{Bounding Number of Leaves in Block-Cut Trees}
We start by providing a reduction rule that, when a block-cut tree has many leaves, will allow us to contract a leaf block into a single vertex.

\begin{reduction}[contract leaf block]\label{red:contract-block-leaf}
    Let $(G,t)$ be a non-trivial problem instance, let $X$ be a modulator of $G$, let $C \in \mathcal{C}(G - X)$, let $T$ be a block-cut tree of $C$, and let $B \in \mathcal{B}(T)$ with $N_T(B) = \{a\}$. If $N_G(B - a) \setminus \{a\}$ is a limit-1 subset for $(G-B,t)$ and each $x \in N_G(B - a) \setminus \{a\}$ has $\treewidth{G[V(B) \cup \{x\}] \cup \{ax\}} \leq 2$, then contract $B$ into $a$.
\end{reduction}

\noindent
We note that from Corollary \ref{corollary:block-cut-tree-only-blocks-of-size-at-least-2} it follows that applying Reduction \ref{red:contract-block-leaf} strictly decreases the number of vertices in $G$.

\begin{lemma}[safeness]\label{lemma:safeness:contract-block-leaf}
    Let $(G',t)$ be the problem instance obtained by applying Reduction \ref{red:contract-block-leaf} on $(G,t)$. Then $\twtwodeletion{G}{t} \Longleftrightarrow \twtwodeletion{G'}{t}$ holds.
\end{lemma}

\begin{proof}
    Because $G'$ is a minor of $G$ we have by Lemma \ref{lemma:minor-of-G-implication} that $\twtwodeletion{G}{t}$ implies $\twtwodeletion{G'}{t}$. It remains to prove that $\twtwodeletion{G'}{t}$ implies $\twtwodeletion{G}{t}$. We let $S'$ be a solution for $(G',t)$.
    
    We have by Lemma \ref{lemma:limit-m-of-subgraph-to-supergraph} that $N_G(B - a) \setminus \{a\}$ is a limit-1 subset for $(G',t)$. We argue that $N_G(B - a) \setminus \{a\} \neq \emptyset$. By Lemma \ref{lemma:minor-treewidth} we have $\treewidth{B} \leq \treewidth{G - X} \leq 2$, which in case $N_G(B - a) \setminus \{a\} = \emptyset$ would allow us to apply Reduction \ref{red:contract-component-with-small-neighborhood} on $B - a$, which contradicts Definition \ref{def:trivial-problem-instance}. Therefore, there exists some $x \in N_G(B - a) \setminus \{a\}$ for which $N_G(B - a) \setminus S' \subseteq \{a,x\}$. We display this situation when $a,x \notin S'$ in Figure \ref{fig:safeness:contract-block-leaf}.
    
    By the definition of $G'$ we have $G' - S' = ((G \cup \{ax\}) - (B - a)) - S'$. By Lemma \ref{lemma:minor-treewidth} we have $\treewidth{(G[V(B) \cup \{x\}] \cup \{ax\}) - S'} \leq \treewidth{G[V(B) \cup \{x\}] \cup \{ax\}} \leq 2$. Furthermore, $N_G((B - a) - S') \setminus S' \subseteq N_G(B - a) \setminus S' \subseteq \{a,x\}$ holds, which induces a clique in $G \cup \{ax\}$. Hence, we can apply Theorem \ref{theorem:connect-treewidth-graphs} on $((G \cup \{ax\}) - S', G' - S')$ to derive $\treewidth{(G \cup \{ax\}) - S'} \leq 2$. By Lemma \ref{lemma:minor-treewidth} this yields $\treewidth{G - S'} \leq 2$, which implies $\twtwodeletion{G}{t}$.
\end{proof}

\begin{figure}[ht]
    \centering
    \includegraphics[scale=.7,page=2]{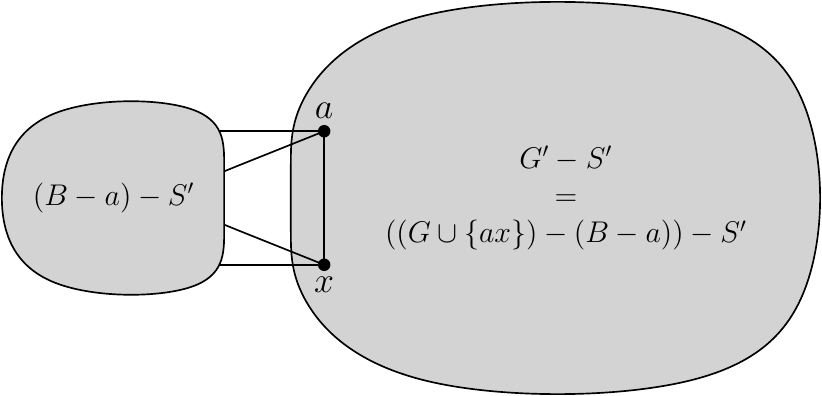}
    \caption{Graph $(G \cup \{ax\}) - S'$ in case $N_G(B - a) \setminus S' = \{a,x\}$}
    \label{fig:safeness:contract-block-leaf}
\end{figure}

\noindent
We next prove that when a block-cut tree $T$ of $C \in \mathcal{C}(G - (X \cup Y))$ has many leaves, that then we can identify one block on which Reduction \ref{red:contract-block-leaf} can be applied. To prove this we first prove that there are a bounded number of blocks $B$ and articulation vertices $a$ with $N_T(B) = \{a\}$ for which $N_G(B-a) \setminus \{a\}$ is not a limit-1 subset for $(G-B,t)$. Then afterwards we show that there are a bounded number of such pairs $(B,a)$ that have some $x \in N_G(B-a) \setminus \{a\}$ for which $\treewidth{G[V(B) \cup \{x\}] \cup \{ax\}} > 2$.

\begin{lemma}\label{lemma:bound-non-limit-1-blocks}
    Let $(G,t)$ be a problem instance, let $X$ be a tidy modulator of $G$, let $Y$ be a component separator for $(G,t,X)$, let $C \in \mathcal{C}(G - (X \cup Y))$, and let $T$ be a block-cut tree of $C$. There exist at most $|N_G(C) \cap Y|(|N_G(C)|-1)(t+2)$ pairs of blocks $B \in \mathcal{B}(T)$ and vertices $a \in A(T)$ for which $N_T(B) = \{a\}$ and $|N_G(B - a) \setminus \{a\}| \geq 2$ and $N_G(B - a) \cap Y \neq \emptyset$, and we can identify these pairs in $\polyG$ time.
\end{lemma}

\begin{proof}
    Let $L = \setdef{(B,a)}{B \in \mathcal{B}(T) \wedge a \in A(T) \wedge N_T(B) = \{a\}}$. I.e., $L$ is the set of all leaf blocks $B$ in $T$ with their respective single articulation vertex $a \in V(B)$. For each $y \in N_G(C) \cap Y \wedge z \in N_G(C) \setminus \{y\}$ let $R_{y,z} = \setdef{(B,a) \in L}{y,z \in N_G(B-a)}$ and let $R = \setdef{(B,a) \in L}{(B,a) \in R_{y,z} \wedge |R_{y,z}| \leq t+2}$. I.e., $R$ is the set of all $(B,a) \in L$ which have some $y \in N_G(B-a) \cap Y$ and some $z \in N_G(B-a) \setminus \{a,y\}$ for which there are at most $t+2$ blocks $(B',a') \in L$ with $y,z \in N_G(B' - a')$. From the definition of $R$ and Observation \ref{obs:block-cut-tree-polyg-size} we immediately derive $|R| \leq |N_G(C) \cap Y|(|N_G(C)|-1)(t+2)$ and $R$ can be constructed in $\polyG$ time. It remains to prove that each $(B,a) \in L \setminus R$ has $|N_G(B-a) \setminus \{a\}| \leq 1$ or $N_G(B - a) \cap Y = \emptyset$.
    
    Assume per contradiction $(B,a) \in L \setminus R$ with $|N_G(B-a) \setminus \{a\}| \geq 2$ and $N_G(B - a) \cap Y \neq \emptyset$. Let $y \in N_G(B - a) \cap Y$ and $z \in N_G(B - a) \setminus \{a,y\}$. Because $B-a$ is a subgraph of $C$ with $V(C) \cap Y = \emptyset$ we have $y \in N_G(C) \cap Y$. Hence, we have $(B,a) \in R_{y,z} \wedge (B,a) \notin R$, which implies that there exist at least $t+3$ pairs $(B',a') \in L$ with $y,z \in N_G(B'-a')$. Graphs $B' - a'$ are connected due to $B'$ being biconnected. From Lemma \ref{lemma:articulation-vertices-as-separators} it follows that these subgraphs $B' - a'$ are pairwise vertex-disjoint. Hence, there exist at least $t+3$ internally vertex-disjoint paths between $y$ and $z$ in $G[V(C) \cup \{y,z\}]$. Lemma \ref{lemma:safeness:add-necessary-edge} states $\twtwodeletion{G[V(C) \cup \{y,z\}]}{0} \Longleftrightarrow \twtwodeletion{G[V(C) \cup \{y,z\}] \cup \{yz\}}{0}$. By Lemma \ref{lemma:minor-treewidth} we have $\treewidth{G[V(C) \cup \{y,z\}]} \leq \treewidth{G - (X \setminus \{z\})} \leq 2$, which therefore implies $\treewidth{G[V(C) \cup \{y,z\}] \cup \{yz\}} \leq 2$.
    
    We note that $C - (B - a)$ must be a connected graph, because $C$ is connected and $\partial_C(B) = \{a\}$. We have that $(B-a, C - (B - a), y, z)$ describes a $K_4$ minor in $G[V(C) \cup \{y,z\}] \cup \{yz\}$, as shown in Figure \ref{fig:bound-non-limit-1-blocks}. Hence, Lemma \ref{lemma:describes-k4-minor} yields $\treewidth{G[V(C) \cup \{y,z\}] \cup \{yz\}} > 2$, which is a contradiction. Therefore, we have that each $(B,a) \in L \setminus R$ has $|N_G(B - a) \setminus \{a\}| \leq 1$ or $N_G(B - a) \cap Y = \emptyset$.
\end{proof}

\begin{figure}[ht]
    \centering
    \includegraphics[scale=.6]{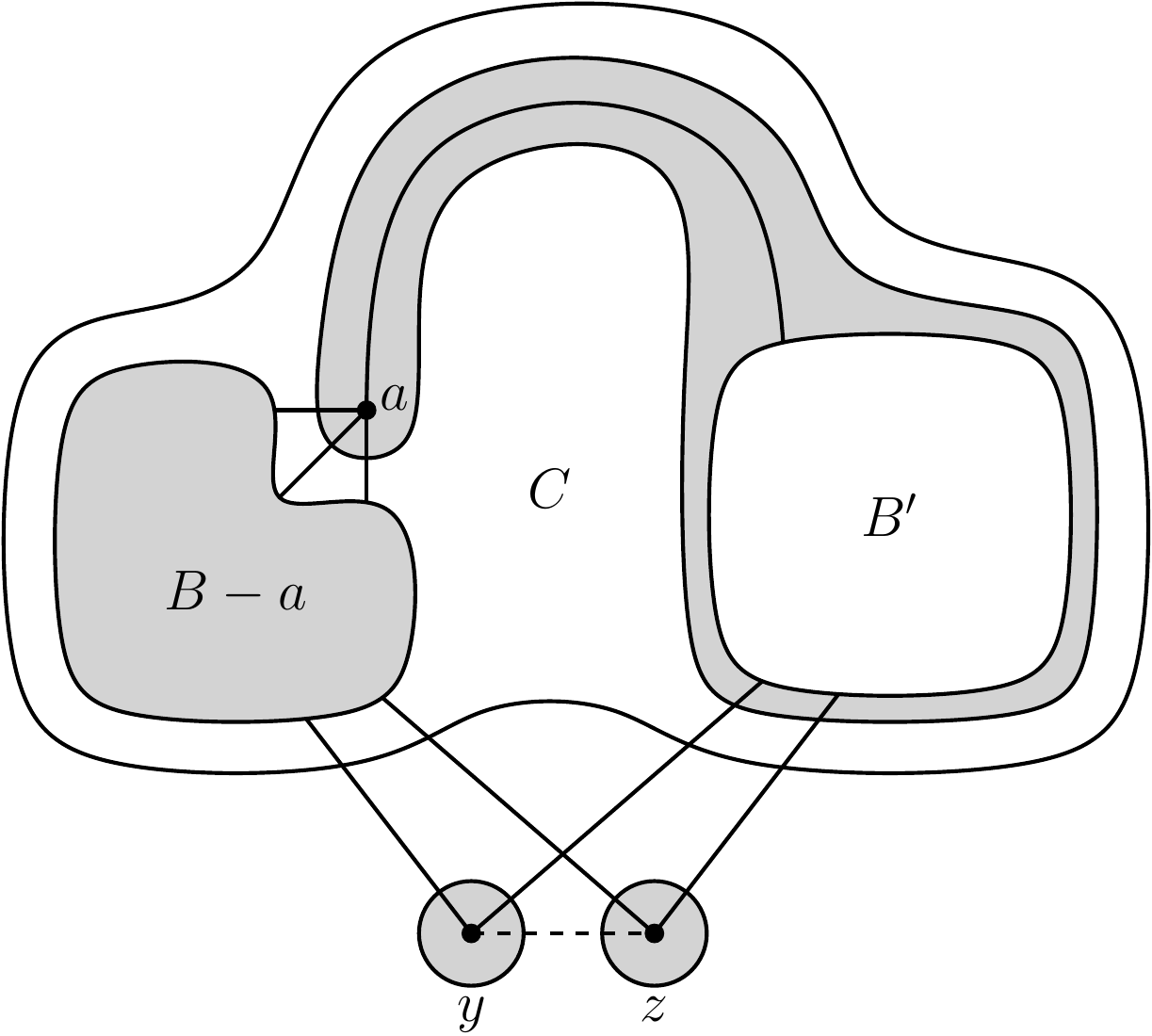}
    \caption{$K_4$ minor in $G[V(C) \cup \{y,z\}] \cup \{yz\}$}
    \label{fig:bound-non-limit-1-blocks}
\end{figure}

\begin{corollary}\label{corollary:bound-non-limit-1-blocks}
    Let $(G,t)$ be a problem instance, let $X$ be a tidy modulator of $G$, let $Y$ be a component separator for $(G,t,X)$, let $C \in \mathcal{C}(G - (X \cup Y))$, and let $T$ be a block-cut tree of $C$. There exist at most $|N_G(C) \cap Y|(|N_G(C)|-1)(t+2)$ pairs of blocks $B \in \mathcal{B}(T)$ and vertices $a \in A(T)$ for which $N_T(B) = \{a\}$ and $N_G(B - a) \setminus \{a\}$ is not a limit-1 subset for $(G-B,t)$, and we can identify these pairs in $\polyG$ time.
\end{corollary}

\begin{proof}
    Let $B \in \mathcal{B}(T)$ and $a \in A(T)$ such that $N_T(B) = \{a\}$. From Lemma \ref{lemma:articulation-vertices-as-separators} and the observation that $B - a$ is a subgraph of $C$ we derive $N_G(B - a) \setminus \{a\} \subseteq N_G(C)$. It therefore follows from Lemma \ref{lemma:limit-m-of-subgraph-to-supergraph} and Definition \ref{def:component-separator} that $N_G(B - a) \setminus \{a\}$ not being a limit-1 subset for $(G - B, t)$ requires $|N_G(B - a)| \geq 2$ and $N_G(B - a) \cap Y \neq \emptyset$. Hence, correctness directly follows from Lemma \ref{lemma:bound-non-limit-1-blocks}.
\end{proof}

\begin{lemma}\label{lemma:bound-yes-limit-1-blocks-non-treewidth-2}
    Let $(G,t)$ be a problem instance, let $X$ be a tidy modulator of $G$, let $C \in \mathcal{C}(G - X)$, and let $T$ be a block-cut tree of $C$. There exist at most $|N_G(C)|$ pairs of block $B \in \mathcal{B}(T)$ and vertices $a \in A(T)$ for which $N_T(B) = \{a\}$ and for which some $x \in N_G(B-a) \setminus \{a\}$ has $\treewidth{G[V(B) \cup \{x\}] \cup \{ax\}} > 2$, and we can identify these pairs in $\polyG$ time.
\end{lemma}

\begin{proof}
    We define $L = \setdef{(B,a)}{B \in \mathcal{B}(T) \wedge N_T(B) = \{a\}}$, we define for each $x \in N_G(C)$ set $R_x = \setdef{(B,a) \in L}{x \in N_G(B-a)}$ and $R = \setdef{(B,a) \in L}{(B,a) \in R_x \wedge |R_x| \leq 1}$. I.e., $R$ is the set of all $(B,a) \in L$ which have some vertex $x \in N_G(B-a) \setminus \{a\}$ such that $(B,a)$ is the only leaf block in $T$ for which this holds. From the definition of $R$ and Observation \ref{obs:block-cut-tree-polyg-size} it follows that $|R| \leq |N_G(C)|$ and $R$ can be constructed in $\polyG$ time. It remains to prove that each $(B,a) \in L \setminus R$ and $x \in N_G(B-a) \setminus \{a\}$ has $\treewidth{G[V(B) \cup \{x\}] \cup \{ax\}} \leq 2$. 
    
    Let $(B,a) \in L \setminus R$ and $x \in N_G(B-a) \setminus \{a\}$. We have $(B,a) \in R_x \wedge (B,a) \notin R$, which implies that there exists some pair $(B',a') \in L \setminus \{(B,a)\}$ with $x \in N_G(B' - a')$. By Lemma \ref{lemma:articulation-vertices-as-separators} we have that $B-a$ and $B'-a'$ are vertex disjoint. We note that $C - (B - a)$ must be a connected graph, because $C$ is connected and $\partial_C(B) = \{a\}$. Furthermore, we have $x \in N_G(B' - a') \cap X \subseteq N_G(C - (B - a))$. Therefore, when in graph $G[V(C) \cup \{x\}]$ we contract $C-(B-a)$ into $a$ we obtain $G[V(B) \cup \{x\}] \cup \{ax\}$. By Lemma \ref{lemma:minor-treewidth} we have $\treewidth{G[V(B) \cup \{x\}] \cup \{ax\}} \leq \treewidth{G[V(C) \cup \{x\}]} \leq \treewidth{G - (X \setminus \{x\})} \leq 2$. Hence, each $(B,a) \in L \setminus R$ and $x \in N_G(B-a) \setminus \{a\}$ has $\treewidth{G[V(B) \cup \{x\}] \cup \{ax\}} \leq 2$.
\end{proof}

\noindent
Similar to the definition of $\nbound{C}$ and $\nbound{Y}$ we define a function $\nbound{L} : \nat^3 \rightarrow \nat$ that states that any block-cut tree $T$ of a component $C \in \mathcal{C}(G - (X \cup Y))$ that has strictly more than $\Lbound{t}{|X|}{|N_G(C) \cap Y|}$ leaves leads to an application of Reduction \ref{red:contract-block-leaf}.

\begin{definition}\label{def:Lbound}
    \[\Lbound{t}{x}{y} = y(x+y-1)(t+2)+x+y\]
\end{definition}

\begin{lemma}\label{lemma:bound-leafs-in-block-cut-tree}
    Let $(G,t)$ be a non-trivial problem instance, let $X$ be a tidy modulator of $G$, let $Y$ be a component separator for $(G,t,X)$, let $C \in \mathcal{C}(G - (X \cup Y))$, and let $T$ be a block-cut tree of $C$. If $T$ has strictly more than $\Lbound{t}{|X|}{|N_G(C) \cap Y|}$ leaves we can in $\polyG$ time apply Reduction \ref{red:contract-block-leaf}.
\end{lemma}

\begin{proof}
    Directly follows from Corollary \ref{corollary:bound-non-limit-1-blocks}, Lemma \ref{lemma:bound-yes-limit-1-blocks-non-treewidth-2}, and the observation that $X \cup Y$ is a tidy modulator of $G$.
\end{proof}

\section{Obtaining a Simple Path without \textit{Y}-Neighbours}
Our goal for this subsection is to obtain a simple path $P$ in $T$ that contains no vertices (being blocks and articulation vertices) with degree three or more and all blocks contained in $P$ do not have neighbours in $Y$. We will first prove that when a component $C \in \mathcal{C}(G - (X \cup Y))$ contains many neighbours of a vertex $y \in Y$, that then we are able to apply Theorem \ref{theorem:reduce-big-biconnected-graph} to reduce component $C$. This yields a bound on the number of vertices in $C$ that have neighbours in $Y$, and indirectly also the number of blocks in $\mathcal{B}(T)$ that have neighbours in $Y$. Using this bound alongside the previously obtained bound on the number of leaves in $T$ allows us to find such a path $P$.

\begin{lemma}\label{lemma:reduce-biconnected-by-adding-y-vertex}
    Let $(G,t)$ be a non-trivial problem instance, let $X$ be a tidy modulator of $G$, let $Y$ be a component separator for $(G,t,X)$, let $C \in \mathcal{C}(G - (X \cup Y))$, and let $y \in N_G(C) \cap Y$. If $|N_G(y) \cap V(C)| > 1988|X|$, we can in $\polyG$ time apply Reduction \ref{red:remove-middle-edge-to-isolated-path} or \ref{red:ladder-reduction} on an edge in $E(G) \setminus E(G[X])$.
\end{lemma}

\begin{proof}
    We let $T$ be a spanning tree of $C$ from which we have exhaustively removed its leaves that were not adjacent to $y$. We have $|V(T)| \geq |N_G(y) \cap V(C)| > 1988|X|$. From Observation \ref{obs:get-spanning-tree} it trivially follows that we can obtain $T$ in $\polyG$ time.
    
    We let $B = G[V(T) \cup \{y\}]$ and argue that $B$ is biconnected. Assume per contradiction that $a$ is an articulation vertex of $B$. Because $T$ is spanning tree of $B - y$ we have $a \neq y$, which implies $a \in V(T)$. Let $a$ be the root of $T$. Every subtree $T' \in \mathcal{C}(T-a)$ will contain at least one leaf of $T$. Because $V(T) \setminus \{a\} = V(B) \setminus \{a,y\}$, we have that each $D \in \mathcal{C}(B - a)$ contains a leaf of $T$. Because all leaves of $T$ are connected to $y$ we have that $B-a$ is connected, which contradicts $a$ being an articulation vertex. Hence, $B$ must be biconnected.
    
    Because $V(B) \cap X = \emptyset$ we have that $B$ is a biconnected induced subgraph of $G - X$ with $|V(B)| = |V(T)| + 1 \geq 1988|X| + 2$. Therefore, by Theorem \ref{theorem:reduce-big-biconnected-graph}, we can in $\polyG$ time apply Reduction \ref{red:remove-middle-edge-to-isolated-path} or \ref{red:ladder-reduction} on an edge in $E(B) \subseteq E(G) \setminus E(G[X])$.
\end{proof}

\noindent
We next focus on obtaining a simple path $P = \langle a_1, B_1, a_2, B_2, \dots, B_k, a_{k+1} \rangle$ in a block-cut tree $T$, such that $N_G(V(\mathcal{B}(P)) \setminus \{a_1, a_{k+1}\}) \subseteq \{a_1, a_{k+1}\} \cup X$. Figure \ref{fig:get-linear-block-cut-tree:example-P} displays an example path with $k = 6$ for which this condition holds. Later on we will show how the graph induced by vertices in blocks along such a path $P$ can be reduced.

\begin{figure}[ht]
\centering
\includegraphics[scale=.7,page=7]{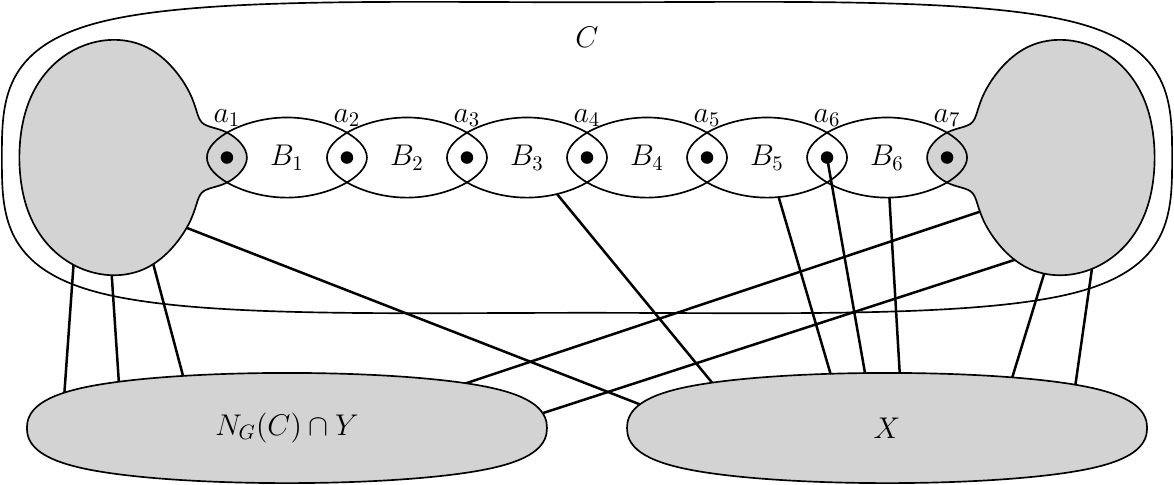}
\caption{A simple path $P$ in a block-cut tree with $N_G(V(\mathcal{B}(P)) \setminus \{a_1,a_7\}) \subseteq \{a_1,a_7\} \cup X$.}
\label{fig:get-linear-block-cut-tree:example-P}
\end{figure}

\noindent
To obtain such a path $P$ we will define a set of blocks and articulation vertices $D \subseteq V(T)$ for which all components in $T - D$ will be linear subtrees corresponding to paths $P$ with $|\mathcal{B}(P)| = k_P$ and $N_G(V(\mathcal{B}(P)) \setminus \{a_1, a_{k_P}\}) \subseteq \{a_1, a_{k_P}\} \cup X$. By bounding the size of $D$ and the number of components in $T - D$ we are able to prove that a large block-cut tree $T$ will contain a large component in $T - D$.

The set $D$ consists out of three sets $D_1$, $D_2$, and $D_3$. We first focus on $D_1$. Sets $D_2$ and $D_3$ will be introduced in the proof of Lemma \ref{lemma:get-linear-block-cut-tree}. Set $D_1$ is the set of all branching vertices (vertices with degree at least three) in $T$. We start by deriving bounds on the size of $D_1$ and the number of components in $T - D_1$.

\begin{lemma}\label{lemma:bound-D1}
    Let $T$ be a tree with $m \geq 2$ leaves and let $D_1 = \setdef{c \in V(T)}{|N_T(c)| \geq 3}$. Then $|D_1| \leq m-2$.
\end{lemma}

\begin{proof}
    Because all vertices in $V(T) \setminus D_1$ have degree at most two, and there are $m$ vertices with degree at most one, we have that $\sum_{c \in V(T) \setminus D_1} |N_T(c)| = 2|V(T) \setminus D_1| - m$. We derive $|D_1| \leq m-2$.
    \allowdisplaybreaks
    \begin{align*}
        |D_1| &\leq \sum_{c \in D_1} (|N_T(c)| - 2) & \text{definition of } D_1\\
        &= \sum_{c \in D_1} |N_T(c)| - 2|D_1| & \text{rewriting}\\
        &= \sum_{c \in D_1} |N_T(c)| + 2|V(T) \setminus D_1| - 2|V(T)| & D_1 \subset V(T)\\
        &= \sum_{c \in D_1} |N_T(c)| + \sum_{c \in V(T) \setminus D_1} |N_T(c)| - 2|V(T)| + m & \hspace{14.1pt}\sum_{c \in V(T) \setminus D_1} |N_T(c)| = 2|V(T){{\setminus}}D_1|{{-}}m\\
        &= \sum_{c \in V(T)} |N_T(c)| - 2|V(T)| + m & D_1 \subset V(T)\\
        &= 2|E(T)| - 2|V(T)| + m & \textrm{degree sum formula}\\
        &= m - 2 & |E(T)| = |V(T)| - 1
    \end{align*}
    \noindent
    We conclude $|D_1| \leq m-2$.
\end{proof}

\begin{lemma}\label{lemma:bound-CD1}
    Let $T$ be a non-empty tree and let $D_1 = \setdef{c \in V(T)}{|N_T(c)| \geq 3}$. Then $|\mathcal{C}(T - D_1)| \leq 2|D_1| + 1$.
\end{lemma}

\begin{proof}
    We prove by induction on $D_1$.
    
    \begin{basecase}[$D_1 = \emptyset$]
        Because $V(T) \neq \emptyset$ we have $|\mathcal{C}(T - D_1)| = |\mathcal{C}(T)| = 1 = 2|D_1| + 1$.
    \end{basecase}
    
    \begin{inductivestep}[$D_1 \neq \emptyset$]
        Let $r \in D_1$. For each $d \in N_T(r)$ let $T_d = T[V(C_d) \cup \{r\}]$ where $C_d \in \mathcal{C}(T - r)$ with $d \in V(C_d)$ and let $D_{1,d} = \setdef{c \in V(T_d)}{|N_{T_d}(c)| \geq 3}$. For each $c \in V(T_d) \setminus \{r\}$ we have $N_{T_d}(c) = N_T(c)$, which implies $D_{1,d} \cup \{r\} = D_1 \cap V(T_d)$.
        
        Each component $C \in \mathcal{C}(T - D_1)$ is a subgraph of $T - r$, which implies that there exists some $T_d$ of which $C$ is a subgraph. Because $D_{1,d} \cup \{r\} = D_1 \cap V(T_d)$ this yields $C \in \mathcal{C}(T_d - (D_{1,d} \cup \{r\}))$, and hence $|\mathcal{C}(T - D_1)| \leq \sum_{d \in N_T(r)} |\mathcal{C}(T_d - (D_{1,d} \cup \{r\}))|$. Because $r$ is a leaf in $T_d$ we have $|\mathcal{C}(T_d - (D_{1,d} \cup \{r\}))| \leq |\mathcal{C}(T_d - D_{1,d})|$. By the induction hypothesis we therefore derive $|\mathcal{C}(T - D_1)| \leq \sum_{d \in N_T(r)} |\mathcal{C}(T_d - D_{1,d})| \leq \sum_{d \in N_T(r)} (2|D_{1,d}| + 1) = 2\sum_{d \in N_T(r)} |D_{1,d}| + |N_T(r)|$. Because all trees $T_d$ and $T_{d'}$ with $d \neq d'$ have $V(T_d) \cap V(T_{d'}) = \{r\}$ we have $D_{1,d} \cap D_{1,d'} = \emptyset$. Because $D_{1,d} \subseteq D_1 \setminus \{r\}$ this yields $|\mathcal{C}(T - D_1)| \leq 2\sum_{d \in N_T(r)} |D_{1,d}| + |N_T(r)| \leq 2|D_1 \setminus \{r\}| + |N_T(r)| = 2|D_1| + |N_T(r)| - 2 \leq 2|D_1| + 1$.
    \end{inductivestep}
    
    \noindent
    We conclude that $|\mathcal{C}(T - D_1)| \leq 2|D_1| + 1$.
\end{proof}

\noindent
In the next proof we introduce sets $D_2$ and $D_3$, and prove that these sets allow us to find a simple path $P = \langle a_1, B_1, a_2, B_2, \dots, B_{k}, a_{k+1} \rangle$ in $T$ with $N_G(V(\mathcal{B}(P)) \setminus \{a_1,a_{k+1}\}) \subseteq \{a_1,a_{k+1}\} \cup X$. Figure \ref{fig:get-linear-block-cut-tree} shows these sets $D_1$, $D_2$, and $D_3$ for an example block-cut tree $T$. In this visualisation we use black dots to denote articulation vertices and white dots to denote blocks. 

\begin{lemma}\label{lemma:get-linear-block-cut-tree}
    Let $(G,t)$ be a problem instance, let $X$ be a tidy modulator of $G$, let $Y$ be a component separator for $(G,t,X)$, let $C \in \mathcal{C}(G - (X \cup Y))$, let $T$ be a block-cut tree of $C$, and let $k > 0$. If $T$ has $m \geq 2$ leaves and each $y \in N_G(C) \cap Y$ has $|N_G(y) \cap V(C)| \leq n$ and $|\mathcal{B}(T)| \geq m(2k+3)+2kn|N_G(C) \cap Y|-3k-4$, then we can in $\polyG$ time obtain a simple path $P = \langle a_1, B_1, a_2, B_2, \dots, B_{k'}, a_{k'+1} \rangle$ in $T$ with $N_G(V(\mathcal{B}(P)) \setminus \{a_1,a_{k'+1}\}) \subseteq \{a_1,a_{k'+1}\} \cup X \wedge k' \geq k$.
\end{lemma}

\begin{proof}
    Let $D_1 = \setdef{c \in V(T)}{|N_T(c)| \geq 3}$. By Lemmas \ref{lemma:bound-D1} and \ref{lemma:bound-CD1} we have $D_1 \leq m-2$ and $|\mathcal{C}(T - D_1)| \leq 2|D_1| + 1 \leq 2m-3$. We note that each component $T' \in \mathcal{C}(T - D_1)$ is a linear tree with $|N_T(T')| \leq 2$. E.g., the situation shown in Figure \ref{fig:get-linear-block-cut-tree:T-D1}.
    
    Let $D_2 = \setdef{B \in \mathcal{B}(T) \setminus D_1}{|N_T(B) \setminus D_1| \leq 1}$. I.e., $D_1$ is the set of all leaf blocks in $T - D_1$. Because each component in $T - D_1$ is a linear tree we have $|D_2| \leq 2|\mathcal{C}(T - D_1)| \leq 4m-6$ and $|\mathcal{C}(T - (D_1 \cup D_2))| \leq |\mathcal{C}(T - D_1)| \leq 2m-3$. We note that each component $T' \in \mathcal{C}(T - (D_1 \cup D_2))$ is a linear tree with $|N_T(T')| = 2$ and each leaf $\ell$ of $T'$ has $\ell \in A(T)$. E.g., the situation shown in Figure \ref{fig:get-linear-block-cut-tree:T-D2}.
    
    Let $D_3 = \setdef{B \in \mathcal{B}(T) \setminus (D_1 \cup D_2)}{N_G(B) \cap Y \neq \emptyset}$. I.e., $D_2$ is the set of all blocks in $T - (D_1 \cup D_2)$ that are adjacent to some vertex in $Y$.
    
    We argue that $|D_3| \leq 2n|N_G(C) \cap Y|$. Assume per contradiction that there exists some vertex $y \in N_G(C) \cap Y$ that is adjacent to $2n+1$ blocks in $\mathcal{B}(T) \setminus (D_1 \cup D_2)$. Because $|N_G(y) \cap V(C)| \leq n$ there must exist a vertex $u \in N_G(y) \cap V(C)$ that is located in at least three blocks in ${B}(T) \setminus (D_1 \cup D_2)$. Lemma \ref{lemma:blocks-share-at-most-one-cut-vertex} implies $u \in A(T)$. By Definition \ref{def:block-cut-tree} we have $|N_T(u)| \geq 3$, which implies $u \in D_1$. Let $B \in N_T(u) \setminus (D_1 \cup D_2)$. By definition of $D_2$ we have $2 \leq |N_T(B) \setminus D_1| \leq |N_T(B) \setminus \{u\}| = |N_T(B)| - 1$, which implies $B \in D_1$, which is a contradiction. Hence, we have $|D_3| \leq 2n|N_G(C) \cap Y|$. 
    
    Because each component in $T - (D_1 \cup D_2)$ is a linear tree we have that the removal of any block $D_3$ from $T$ increases the number of components with at most one. Hence, $|\mathcal{C}(T - (D_1 \cup D_2 \cup D_3))| \leq |\mathcal{C}(T - (D_1 \cup D_2))| + |D_3| \leq 2m + 2n|N_G(C) \cap Y| - 3$. We note that each component $T' \in \mathcal{C}(T - (D_1 \cup D_2 \cup D_3))$ is a linear tree with $|N_T(T')| = 2$, each leaf $\ell$ of $T'$ has $\ell \in A(T)$, and each block $B \in \mathcal{B}(T')$ has $N_G(B) \cap Y = \emptyset$. E.g., the situation shown in Figure \ref{fig:get-linear-block-cut-tree:T-D3-post}.

    Using the previously derived bounds and the pigeonhole principle, we derive that the largest component $T' \in \mathcal{C}(T - (D_1 \cup D_2 \cup D_3))$ must contain at least $k$ blocks.
    
    \begin{align*}
        |\mathcal{B}(T')| &\geq \left\lceil \frac{|\mathcal{B}(T) \setminus (D_1 \cup D_2 \cup D_3)|}{|\mathcal{C}(T - (D_1 \cup D_2 \cup D_3)|} \right\rceil\\
        &\geq \left\lceil \frac{m(2k+3)+2kn|N_G(C) \cap Y|-3k-4 - m + 2 - 4m + 6 - 2n|N_G(C) \cap Y|}{2m + 2n|N_G(C) \cap Y| - 3} \right\rceil\\
        &= \left\lceil \frac{2m(k-1) + 2n(k-1)|N_G(C) \cap Y| - 3(k-1)+1}{2m + 2n|N_G(C) \cap Y| - 3} \right\rceil\\
        &= k-1 + \left\lceil \frac{1}{2m + 2n|N_G(C) \cap Y| - 3} \right\rceil\\
        &= k
    \end{align*}
    
    \noindent
    From Observation \ref{obs:block-cut-tree-polyg-size} it trivially follows that we can construct sets $D_1$, $D_2$, $D_3$ and find the largest component $T' \in \mathcal{C}(T - (D_1 \cup D_2 \cup D_3))$ in $\polyG$ time. We let $T'$ be this largest component, which is shown in Figure \ref{fig:get-linear-block-cut-tree:T'}.
    
    We know that $T'$ is a linear tree with articulation vertices from $A(T)$ as leaves. Let $P = \langle a_1, B_1, a_2, B_2, \dots, B_{k'}, a_{k'+1} \rangle$ with $k' = |\mathcal{B}(T')| \geq k$ be the simple path in $T'$ between its two leaves. Because $T'$ is a linear tree with $|N_T(T')| = 2$ it follows from Lemma \ref{lemma:block-cut-tree-only-blocks-as-leaves} that $\partial_T(T') = \{a_1,a_{k'+1}\}$. From Lemma \ref{lemma:articulation-vertices-as-separators} it therefore follows that $N_C(V(\mathcal{B}(P)) \setminus \{a_1, a_{k'+1}\}) = \{a_1,a_{k'+1}\}$. Because each block $B \in \mathcal{B}(T')$ has $N_G(B) \cap Y = \emptyset$, we conclude that $N_G(V(\mathcal{B}(P)) \setminus \{a_1,a_{k'+1}\}) \subseteq \{a_1,a_{k'+1}\} \cup X$. Hence, path $P$ suffices.
\end{proof}
\clearpage
\begin{figure}[ht]
\newcommand{\figvspace}{.71cm}
\centering
\begin{subfigure}{\textwidth}
    \centering
    \includegraphics[scale=.7,page=1]{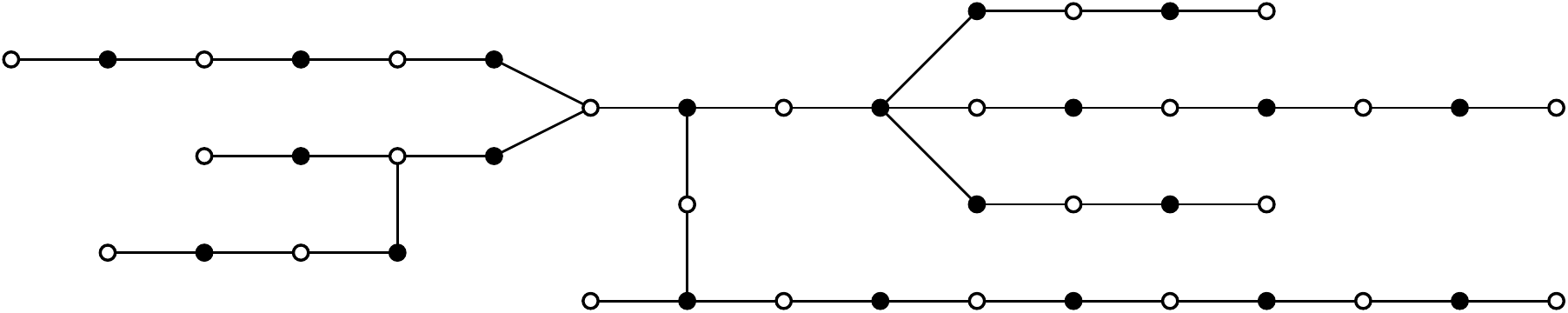}
    \caption{Block-cut tree $T$}
    \label{fig:get-linear-block-cut-tree:T}
\end{subfigure}
\\~\vspace{\figvspace}~\\
\begin{subfigure}{\textwidth}
    \centering
    \includegraphics[scale=.7,page=2]{reduce-block-cut-tree.pdf}
    \caption{Subtree $T - D_1$ with all degree three or more nodes removed}
    \label{fig:get-linear-block-cut-tree:T-D1}
\end{subfigure}
\\~\vspace{\figvspace}~\\
\begin{subfigure}{\textwidth}
    \centering
    \includegraphics[scale=.7,page=3]{reduce-block-cut-tree.pdf}
    \caption{Subtree $T - (D_1 \cup D_2)$ with all block leaves removed}
    \label{fig:get-linear-block-cut-tree:T-D2}
\end{subfigure}
\\~\vspace{\figvspace}~\\
\begin{subfigure}{\textwidth}
    \centering
    \includegraphics[scale=.7,page=4]{reduce-block-cut-tree.pdf}
    \caption{A component in $T - (D_1 \cup D_2)$ and its connections to $Y$}
    \label{fig:get-linear-block-cut-tree:T-D3-pre}
\end{subfigure}
\\~\vspace{\figvspace}~\\
\begin{subfigure}{\textwidth}
    \centering
    \includegraphics[scale=.7,page=5]{reduce-block-cut-tree.pdf}
    \caption{The sub components of previous component present in $T - (D_1 \cup D_2 \cup D_3)$}
    \label{fig:get-linear-block-cut-tree:T-D3-post}
\end{subfigure}
\\~\vspace{\figvspace}~\\
\begin{subfigure}{\textwidth}
    \centering
    \includegraphics[scale=.6,page=6]{reduce-block-cut-tree.pdf}
    \caption{An example component $T'$ in $T - (D_1 \cup D_2 \cup D_3)$}
    \label{fig:get-linear-block-cut-tree:T'}
\end{subfigure}
\caption{Example construction of $D_1$, $D_2$, and $D_3$ and the resulting components}
\label{fig:get-linear-block-cut-tree}
\end{figure}
\clearpage
\noindent
We define a function $\nbound{P} : \nat^4 \rightarrow \nat$ such that for any block-cut tree $T$ of a component $C \in \mathcal{C}(G - (X \cup Y))$ with $|\mathcal{B}(T)| > \Pbound{t}{|X|}{|N_G(C) \cap Y|}{k}$ we can either apply a reduction or find a path in $T$ with at least $k$ blocks.

\begin{definition}\label{def:Pbound}
    \[\Pbound{t}{x}{y}{k} = \max(1,~\Lbound{t}{x}{y}(2k+3)+3976kxy-3k-4)\]
\end{definition}

\begin{lemma}\label{lemma:get-linear-block-cut-tree2}
    Let $(G,t)$ be a non-trivial problem instance, let $X$ be a tidy modulator of $G$, let $Y$ be a component separator for $(G,t,X)$, let $C \in \mathcal{C}(G - (X \cup Y))$, let $T$ be a block-cut tree of $C$, and let $k > 0$. If $|\mathcal{B}(T)| > \Pbound{t}{|X|}{|N_G(C) \cap Y|}{k}$, then we can in $\polyG$ time apply Reduction \ref{red:remove-middle-edge-to-isolated-path} or \ref{red:ladder-reduction} on an edge in $E(G) \setminus E(G[X])$, or we apply Reduction \ref{red:contract-block-leaf}, or we obtain a simple path $P = \langle a_1, B_1, a_2, B_2, \dots, B_{k'}, a_{k'+1} \rangle$ in $T$ with $N_G(V(\mathcal{B}(P)) \setminus \{a_1,a_{k'+1}\}) \subseteq \{a_1,a_{k'+1}\} \cup X \wedge k' \geq k$.
\end{lemma}

\begin{proof}
    If $T$ has strictly more than $\Lbound{t}{|X|}{|N_G(C) \cap Y|}$ leaves correctness follows from Lemma \ref{lemma:bound-leafs-in-block-cut-tree}. We assume that $T$ has at most $\Lbound{t}{|X|}{|N_G(C) \cap Y|}$ leaves.
    
    If there exists a vertex $y \in N_G(C) \cap Y$ with $|N_G(y) \cap V(C)| > 1988|X|$ correctness follows from Lemma \ref{lemma:reduce-biconnected-by-adding-y-vertex}. We assume that each $y \in N_G(C) \cap Y$ has $|N_G(y) \cap V(C)| \leq 1988|X|$.
    
    Because $|\mathcal{B}(T)| > \Pbound{t}{|X|}{|N_G(C) \cap Y|}{k} \geq 1$ we trivially have that $T$ has at least two leaves. Hence, Lemma \ref{lemma:get-linear-block-cut-tree} yields that we can in $\polyG$ time obtain a path $P = \langle a_1, B_1, a_2, B_2, \dots, B_{k'}, a_{k'+1} \rangle$ in $T$ with $N_G(V(\mathcal{B}(P)) \setminus \{a_1,a_{k'+1}\}) \subseteq \{a_1,a_{k'+1}\} \cup X \wedge k' \geq k$.
\end{proof}

\section{Reducing Simple Paths without Neighbours}
Once we have obtained a path $P$ in $T$ where each block in $P$ has no neighbours in $Y$ we will want to find a reduction rule that can be applied on the blocks contained in this path. We will consider two reduction rules. In case many blocks along $P$ share a neighbour $x \in X$ then we will be able to find an application of Reduction \ref{red:remove-middle-edge-to-isolated-path}. Otherwise, we will be able to contract a block without any neighbours in $X \cup Y$. In this subsection we focus on the latter.

We introduce a new reduction rule that can be applied on three adjacent blocks along path $P$ without neighbours in $X \cup Y$. We display this situation in Figure \ref{fig:safeness:reducing-3-induced-blocks-without-neighbours}. This reduction contracts block $B_1$ into $a_1$. The intuitive reason why this reduction is safe is that a solution $S'$ for the resulting problem instance $(G',t)$ that contains a vertex from any of the blocks can be replaced by either $a_1$ or $a_4$. In case $S'$ does not contain a vertex from any of these blocks, then, due to Reduction \ref{red:contract-component-with-small-neighborhood} not being applicable on non-trivial problem instance $(G,t)$, we can derive that $a_1$ and $a_4$ are separated in $G' - S'$, which allows us to apply Theorem \ref{theorem:connect-treewidth-graphs}.
\begin{figure}[ht]
    \centering
    \includegraphics[scale=.66]{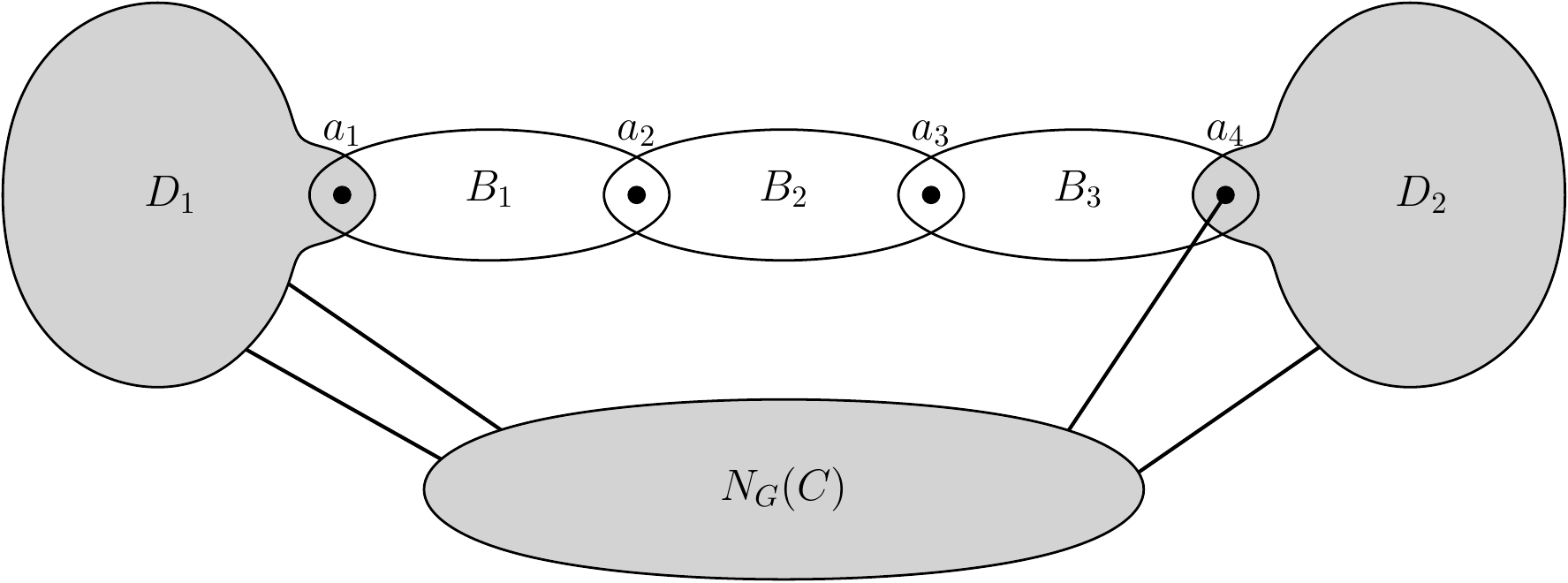}
    \caption{A simple block path with $N_G((V(B_1) \cup V(B_2) \cup V(B_3)) \setminus \{a_1,a_4\}) = \{a_1,a_4\}$}
    \label{fig:safeness:reducing-3-induced-blocks-without-neighbours}
\end{figure}
\begin{reduction}[contract blocks without $X{{\cup}}Y$-neighbours]\label{red:reducing-3-induced-blocks-without-neighbours}
    Let $(G,t)$ be a non-trivial problem instance, let $X$ be a modulator of $G$, let $C \in \mathcal{C}(G - X)$, let $T$ be a block-cut tree of $C$, and let $\langle a_1, B_1, a_2, B_2, a_3, B_3, a_4 \rangle$ be a simple path in $T$ with $N_G((V(B_1) \cup V(B_2) \cup V(B_3)) \setminus \{a_1,a_4\}) = \{a_1,a_4\}$. Then contract $B_1$ into $a_1$.
\end{reduction}
\clearpage

\begin{lemma}[safeness]\label{lemma:safeness:reducing-3-induced-blocks-without-neighbours}
    Let $(G',t)$ be the problem instance obtained by applying Reduction \ref{red:reducing-3-induced-blocks-without-neighbours} on $(G,t)$. Then $\twtwodeletion{G}{t} \Longleftrightarrow \twtwodeletion{G'}{t}$ holds.
\end{lemma}

\begin{proof}
    Because $G'$ is a minor of $G$ we have by Lemma \ref{lemma:minor-of-G-implication} that $\twtwodeletion{G}{t}$ implies $\twtwodeletion{G'}{t}$. It remains to prove that $\twtwodeletion{G'}{t}$ implies $\twtwodeletion{G}{t}$. Let $S'$ be a solution for $(G',t)$, let $H = G[V(B_1) \cup V(B_2) \cup V(B_3)]$, and let $H' = G'[V(B_2) \cup V(B_3) \cup \{a_1\}]$.
    
    \begin{case}[$S' \cap V(H') \neq \emptyset$]
        In case $a_1 \in S'$ let $S = S'$ and otherwise let $S = (S' \setminus V(H')) \cup \{a_4\}$. In either case we have $\{a_1,a_4\} \cap S' \subseteq \{a_1,a_4\} \cap S \neq \emptyset$. By Lemma \ref{lemma:minor-treewidth} we therefore have $\treewidth{(G - S) - (H - \{a_1,a_4\})} \leq \treewidth{(G' - S') - (H' - \{a_1,a_4\})} \leq \treewidth{G' - S'} \leq 2$ and $\treewidth{H - S} \leq \treewidth{G - X} \leq 2$. Because $N_G(H \setminus \{a_1,a_4\}) = \{a_1,a_4\}$ and $\{a_1,a_4\} \cap S \neq \emptyset$ we have $|N_G(H - (S \cup \{a_1,a_4\})) \setminus S| \leq |\{a_1,a_4\} \setminus S| \leq 1$. Hence, $N_G(H - (S \cup \{a_1,a_4\})) \setminus S$ induces a clique in $G - S$. Therefore we can apply Theorem \ref{theorem:connect-treewidth-graphs} on $(G - S, (G - S) - (H - \{a_1,a_4\}))$ to derive $\treewidth{G - S} \leq 2$, which implies $\twtwodeletion{G}{t}$.
    \end{case}
    
    \begin{case}[$S' \cap V(H') = \emptyset$]
        Because $N_G((V(B_2) \cup V(B_3)) \setminus \{a_1,a_4\}) = \{a_2,a_4\}$ we have by Definition \ref{def:trivial-problem-instance} that $\treewidth{G[V(B_2) \cup V(B_3)] \cup \{a_2a_4\}} > 2$. From the definition of $H'$ it therefore follows that $\treewidth{H' \cup \{a_1a_4\}} > 2$, which implies that $a_1$ and $a_4$ must be separated in $G' - S'$ by $H' - \{a_1,a_4\}$. Hence, each component $D \in \mathcal{C}(G' - (V(H') \cup S')) = \mathcal{C}(G - (V(H) \cup S'))$ must have $N_G(D) \setminus S' \subseteq \{a_1\} \vee N_G(D) \setminus S' \subseteq \{a_4\}$, which in either case induces a clique in $G - S'$. By Lemma \ref{lemma:minor-treewidth} we have $\treewidth{(G - S') - (H - \{a_1,a_4\})} = \treewidth{(G' - S') - (H' - \{a_1,a_4\})} \leq \treewidth{G' - S'} \leq 2$ and $\treewidth{H - S'} \leq \treewidth{G - X} \leq 2$. Therefore we can apply Theorem \ref{theorem:connect-treewidth-graphs} on $(G - S', (G - S') - (H - \{a_1,a_4\}))$ to derive $\treewidth{G - S'} \leq 2$, which implies $\twtwodeletion{G}{t}$.
    \end{case}
    
    \noindent
    Because these cases cover all cases, correctness trivially follows.
\end{proof}

\section{Reducing Simple Paths with \textit{X}-Neighbours}
We next focus on reducing paths $P$ with many blocks in $P$ have neighbours in $X$. We show that when some vertex $x \in X$ is a neighbour of many blocks contained in $P$ that then we will be able to apply Reduction \ref{red:remove-middle-edge-to-isolated-path} on an edge between $x$ and a vertex in $V(\mathcal{B}(P))$.

To apply Reduction \ref{red:remove-middle-edge-to-isolated-path} we need to find a simple path $Q$ in $G - X$ that visits multiple neighbours of $x$. We first focus on obtaining this path $Q$, and afterwards show that this path $Q$ allows us to apply Reduction \ref{red:remove-middle-edge-to-isolated-path}. To obtain this path $Q$ we first show that for any block $B \in \mathcal{B}(T)$ with $N_T(B) = \{a_1,a_2\}$ we can find a simple path in $B$ from $a_1$ to a vertex $u \in N_G(x) \cap V(B)$ to $a_2$. Afterwards we concatenate these paths to find a simple path that visits many neighbours of $x$.

\begin{lemma}\label{lemma:block-cut-tree:single-step-single-path}
    Let $G$ be a biconnected graph, let $\{a_1,a_2\} \subseteq V(G)$ and $u \in V(G)$. We can find in $\polyG$ time a simple path in $G$ from $a_1$ to $u$ to $a_2$.
\end{lemma}

\begin{proof}
    In case $u \in \{a_1,a_2\}$ this trivially holds, due to $G$ being connected. Hence, we assume $u \notin \{a_1,a_2\}$. We consider graph $G'$ formed by adding a new vertex $v$ to $G$ that is adjacent to $a_1$ and $a_2$. It trivially holds that $G'$ is biconnected. By Observation \ref{obs:get-max-internally-vertex-disjoint-paths} we can in $\polyG$ time find two internally vertex-disjoint simple paths $Q_1$ and $Q_2$ between $u$ and $v$. Because $\{a_1,a_2\}$ separates $u$ from $v$ in $G'$ we can, w.l.o.g., assume $a_1 \in V(Q_1)$ and $a_2 \in V(Q_2)$. Because $N_{G'}(v) = \{a_1,a_2\}$ we have that $Q_1 - v$ is a simple path from $a_1$ to $u$ and $Q_2 - v$ is a simple path from $u$ to $a_2$. Hence, connecting $Q_1 - v$ to $Q_2 - v$ results in a simple path in $G$ from $a_1$ to $v$ to $a_2$.
\end{proof}

\begin{lemma}\label{lemma:block-cut-tree:many-step-single-path}
    Let $G$ be a connected graph, let $U \subseteq V(G)$, let $T$ be a block-cut tree of $G$, let $P = \langle a_1, B_1, a_2, B_2, \dots B_k, a_{k+1} \rangle$ be a simple path in $T$ that has $m$ blocks $B_i \in \mathcal{B}(P)$ with $V(B_i) \cap U \neq \emptyset$. We can in $\polyG$ time obtain a simple path $Q$ in $G[V(\mathcal{B}(P))]$ with $|V(Q) \cap U| \geq \lceil m / 2 \rceil$.
\end{lemma}

\begin{proof}
    Let $R = \setdef{B_i \in \mathcal{B}(P)}{V(B_i) \cap U \neq \emptyset}$. We remove every second block from $R$ to form set $R'$. As follows from Lemma \ref{lemma:blocks-share-at-most-one-cut-vertex}, each pair of blocks $\{B_i, B_j\} \subseteq R'$ are vertex disjoint. For each $B_i \in R'$ we pick a vertex $u_i \in V(B_i) \cap U$ and for each $B_i \in V(P) \setminus R'$ we pick any vertex $u_i \in V(B_i)$. By Observation \ref{obs:block-cut-tree-polyg-size} and Lemma \ref{lemma:block-cut-tree:single-step-single-path} we can in $\polyG$ time find for each $B_i \in \mathcal{B}(P)$ a simple path $Q_i$ in $B_i$ from $a_i$ to $u_i$ to $a_{i+1}$. We connect all these paths into a simple path $Q$ in $G[V(\mathcal{B}(P))]$ that visits all vertices $\{u_1, u_2, \dots, u_k\}$. Because each $\{B_i,B_j\} \subseteq R'$ are vertex-disjoint we have $|Q \cap U| \geq |R'| = \lceil m / 2 \rceil$. Hence, path $Q$ suffices.
\end{proof}

\noindent
In case we apply Lemma \ref{lemma:block-cut-tree:many-step-single-path} with $U = N_G(x)$ we find a path $Q$ in $G[V(\mathcal{B}(P))]$ with $V(Q) \cap N_G(x) \geq \lceil m / 2 \rceil$, where $m$ is the number of blocks in $\mathcal{B}(P)$ that neighbour $x$. We next show that $m \geq 11$ leads to an application of Reduction \ref{red:remove-middle-edge-to-isolated-path}.

\begin{lemma}\label{lemma:linear-block-cut-tree-reduce-frequent-x-neighbour}
    Let $(G,t)$ be a problem instance, let $X$ be a tidy modulator of $G$, let $x \in X$, let $Y$ be a component separator for $(G,t,X)$, let $C \in \mathcal{C}(G - (X \cup Y))$, let $T$ be a block-cut tree of $C$, and let $P = \langle a_1, B_1, a_2, B_2, \dots B_k, a_{k+1} \rangle$ be a simple path in $T$ with $N_G(V(\mathcal{B}(P)) \setminus \{a_1,a_{k+1}\}) \subseteq \{a_1,a_{k+1}\} \cup X$. If there exist at least 11 blocks $B_i \in \mathcal{B}(P)$ with $x \in N_G(B_i)$, then we can in $\polyG$ time apply Reduction \ref{red:remove-middle-edge-to-isolated-path} on an edge in $E(G) \setminus E(G[X])$.
\end{lemma}

\begin{proof}
    By Lemma \ref{lemma:block-cut-tree:many-step-single-path} we can in $\polyG$ time find a simple path $Q$ in $G[V(\mathcal{B}(P))]$ with $|V(Q) \cap N_G(x)| \geq 6$. Let $\{v_1,v_2,v_3,v_4,v_5,v_6\} \subseteq V(Q) \cap N_G(x)$ be ordered along $Q$. Because $X$ is a tidy modulator we have by Lemma \ref{lemma:x-path-separator} that in $C$ vertex $v_2$ separates $v_1$ from $v_3$ and vertex $v_5$ separates $v_4$ from $v_6$. This implies $v_2, v_5 \in A(P)$. W.l.o.g. let $v_2 = a_i \wedge v_5 = a_j$ with $i < j$.
    
    From Lemma \ref{lemma:articulation-vertices-as-separators} it follows that $\{a_i,a_j\}$ separates $V(\mathcal{B}(P[a_i,a_j])) \setminus \{a_i,a_j\}$ from $V(C) \setminus V(\mathcal{B}(P[a_i,a_j]))$ in $C$. Therefore, because $N_G(V(\mathcal{B}(P)) \setminus \{a_1,a_{k+1}\}) \subseteq \{a_1,a_{k+1}\} \cup X$, we must have $N_G(V(\mathcal{B}(P[a_i,a_j])) \setminus \{a_i,a_j\}) \subseteq \{a_i,a_j\} \cup X$. Because $G[V(\mathcal{B}(P[a_i,a_j])]$ is a subgraph of $C$ we have $N_G(V(\mathcal{B}(P)) \setminus \{a_1,a_{k+1}\}) \setminus \{a_1,a_{k+1}\} \subseteq N_G(C) \cap X$. Because $Y$ is a component separator for $(G,t,X)$ it follows from Lemmas \ref{lemma:limit-m-of-subgraph-to-supergraph} and \ref{lemma:subset-of-limit-m-of-subset} that $N_G(V(\mathcal{B}(P[a_i,a_j])) \setminus \{a_i,a_j\}) \setminus \{a_i,a_j\}$ is a limit-1 subset for $(G-v_3x,t)$. Hence, we can apply Reduction \ref{red:remove-middle-edge-to-isolated-path} on $v_3x \in E(G) \setminus E(G[X])$.
\end{proof}

\section{Reducing Large Block-Cut Trees}
We next combine Lemma \ref{lemma:linear-block-cut-tree-reduce-frequent-x-neighbour} and the definition of Reduction \ref{red:reducing-3-induced-blocks-without-neighbours} to show that a simple path $\langle a_1, B_1, a_2, B_2, \dots, B_k, a_{k+1} \rangle$ in $T$ with $k \geq 30|X|+3$ and $N_G(V(\mathcal{B}(P)) \setminus \{a_1,a_{k+1}\}) \subseteq \{a_1,a_{k+1}\} \cup X$ leads to an application of Reduction \ref{red:remove-middle-edge-to-isolated-path} or \ref{red:reducing-3-induced-blocks-without-neighbours}.

\begin{lemma}\label{lemma:block-cut-tree:reduce-linear-tree}
    Let $(G,t)$ be a non-trivial problem instance, let $X$ be a tidy modulator of $G$, let $Y$ be a component separator for $(G,t,X)$, let $C \in \mathcal{C}(G - (X \cup Y))$, let $T$ be a block-cut tree of $C$, and let $P = \langle a_1, B_1, a_2, B_2, \dots, B_k, a_{k+1} \rangle$ be a simple path in $T$ with $k \geq 30|X|+3$ and $N_G(V(\mathcal{B}(P)) \setminus \{a_1,a_{k+1}\}) \subseteq \{a_1,a_{k+1}\} \cup X$. We can in $\polyG$ time apply Reduction \ref{red:remove-middle-edge-to-isolated-path} on an edge in $(G) \setminus E(G[X])$ or apply Reduction \ref{red:reducing-3-induced-blocks-without-neighbours}.
\end{lemma}

\begin{proof}
    For each $x \in X$ let $R_x = \setdef{B \in \mathcal{B}(P)}{x \in N_G(B)}$. In case some $x \in X$ has $|R_x| \geq 11$ correctness immediately follows from Lemma \ref{lemma:linear-block-cut-tree-reduce-frequent-x-neighbour}. Otherwise we assume that each $x \in X$ has $|R_x| \leq 10$. Let $R = \bigcup_{x \in X} R_x$. We have $|R| \leq 10|X|$ and from Observation \ref{obs:block-cut-tree-polyg-size} it follows that we can obtain $R$ in $\polyG$ time.
    
    Because $|\mathcal{B}(P)| \geq 30|X|+3 \geq 3(|R|+1)$ there must exist three blocks $\{B_i, B_{i+1}, B_{i+2}\} \subseteq \mathcal{B}(P) \setminus R$. From Lemma \ref{lemma:articulation-vertices-as-separators} it follows that $\{a_i,a_{i+3}\}$ separates $V(\mathcal{B}(P[a_i,a_{i+3}])) \setminus \{a_i,a_{i+3}\}$ from $V(C) \setminus V(\mathcal{B}(P[a_i,a_{i+3}]))$ in $C$. Therefore, because $N_G(V(\mathcal{B}(P)) \setminus \{a_1,a_{k+1}\}) \subseteq \{a_1,a_{k+1}\} \cup X$, we must have $N_G(V(\mathcal{B}(P[a_i,a_{i+3}])) \setminus \{a_i,a_{i+3}\}) \subseteq \{a_i,a_{i+3}\} \cup X$. Because $\{B_i, B_{i+1}, B_{i+2}\} \subseteq \mathcal{B}(P) \setminus R$ we have by definition of $R$ that $N_G(V(\mathcal{B}(P[a_i,a_{i+3}])) \setminus \{a_i,a_{i+3}\}) = \{a_i,a_{i+3}\}$. Hence, we can apply Reduction \ref{red:reducing-3-induced-blocks-without-neighbours} on $B_i$.
\end{proof}

\begin{theorem}\label{theorem:reduce-block-cut-tree}
    Let $(G,t)$ be a non-trivial problem instance, let $X$ be a tidy modulator of $G$, let $Y$ be a component separator for $(G,t,X)$, let $C \in \mathcal{C}(G - (X \cup Y))$, let $T$ be a block-cut tree of $C$ with $|\mathcal{B}(T)| > \Pbound{t}{|X|}{|N_G(C) \cap Y|}{30|X|+3}$. We can in $\polyG$ time apply Reduction \ref{red:remove-middle-edge-to-isolated-path} or \ref{red:ladder-reduction} on an edge in $E(G) \setminus E(G[X])$, or apply Reduction \ref{red:contract-block-leaf} or \ref{red:reducing-3-induced-blocks-without-neighbours}.
\end{theorem}

\begin{proof}
    By Lemma \ref{lemma:get-linear-block-cut-tree2} we can in $\polyG$ time apply Reduction \ref{red:remove-middle-edge-to-isolated-path} or \ref{red:ladder-reduction} on an edge in $E(G) \setminus E(G[X])$, or apply Reduction \ref{red:contract-block-leaf}, or we obtain a simple path $P = \langle a_1, B_1, a_2, B_2, \dots, B_k, a_{k+1} \rangle$ in $T$ with $N_G(V(\mathcal{B}(P)) \setminus \{a_1,a_{k+1}\}) \subseteq \{a_1,a_{k+1}\} \cup X \wedge k \geq 30|X|+3$. In case Reduction \ref{red:remove-middle-edge-to-isolated-path}, \ref{red:ladder-reduction}, or \ref{red:contract-block-leaf} is applied correctness is trivial. In case a path $P$ is obtained, then correctness directly follows from Lemma \ref{lemma:block-cut-tree:reduce-linear-tree}.
\end{proof}

\chapter{Combining Results}\label{ch:combining-results}
Finally we will combine the results from Chapters \ref{ch:graph-decompositions}, \ref{ch:reducing-biconnected-subgraphs}, and \ref{ch:reducing-block-cut-trees} to obtain our final kernelization algorithm. We show a description of the algorithm we will deploy below:

\begin{algorithm}
    \SetAlgoLined
	\caption{\textsc{Kernelize}($G, t, \mathcal{A}$)}
	\label{algo:kernelize}
	\smallskip
	\KwIn{A problem instance $(G,t)$ and an $\varepsilon$-approximation algorithm $\mathcal{A}$.}
	\KwOut{A problem instance $(G',t')$ with $\twtwodeletion{G'}{t'} \Longleftrightarrow \twtwodeletion{G}{t}$ and $|V(G')| \leq \text{\textsc{bound}}(t, \varepsilon)$.}
	\medskip
	\everypar={\nl}
	\While{$|V(G)| > \text{\textsc{bound}}(t, \varepsilon)$}{
    	\eIf{\upshape $(G,t)$ is trivial}{
    	    reduce using Lemma \ref{lemma:get-non-trivial-problem-instance}
    	}{
    	    obtain a linked tidy modulator $X$ for $(G,t)$ using Theorem \ref{theorem:get-linked-tidy-modulator}\\
    	    \While{\upshape $(G,t)$ is non-trivial and $X$ is a linked tidy modulator for $(G,t)$}{
    	        obtain a component separator $Y$ for $(G,t,X)$ using Theorem \ref{theorem:get-component-separator}\\
    	        let $C$ be the largest component in $G - (X \cup Y)$\\
    	        let $T$ be a block-cut tree of $C$\\
    	        \eIf{$|\mathcal{B}(T)| > \Pbound{t}{|X|}{4}{30|X|+3}$}{
    	            reduce using Theorem \ref{theorem:reduce-block-cut-tree}
    	        }{
    	            let $B$ be the largest block in $\mathcal{B}(T)$\\
    	            reduce using Theorem \ref{theorem:reduce-big-biconnected-graph}
    	        }
    	    }
    	}
	}
	\Return $(G,t)$
\end{algorithm}

\noindent
We note that function $\textsc{bound}(t, \varepsilon)$ will be defined at the end of this chapter. Furthermore, we note that Algorithm \ref{algo:kernelize} makes the implicit assumption that Theorem \ref{theorem:get-linked-tidy-modulator} and \ref{theorem:get-component-separator} always yield a linked tidy modulator $X$ and a component separator $Y$ respectively. In the remainder of this chapter we will also account for the case where Theorems \ref{theorem:get-linked-tidy-modulator} or \ref{theorem:get-component-separator} return a reduced problem instance.

We will first show that when $X$ is a linked tidy modulator for $(G,t)$ and $|V(G)|$ is large compared to $X$, that then we can apply a reduction that reduces either the number of vertices in $G$ or the number of edges in $G - X$. This step corresponds to lines 7 to 15 in Algorithm \ref{algo:kernelize}. We define a function $\nbound{G} : \nat^3 \rightarrow \nat$, such that for a problem instance $(G,t)$, a linked tidy modulator $X$ of $(G,t)$, and a component separator $Y$ of $(G,t,X)$ with $|V(G)| > \Gbound{t}{|X|}{|Y|}$ we know that we are able to apply a reduction (as will follow from Theorems \ref{theorem:reduce-big-biconnected-graph} and \ref{theorem:reduce-block-cut-tree}).

\begin{definition}\label{def:Gbound}
    \[\Gbound{t}{x}{y} = \Cbound{t}{x+y} \cdot (1988x + 1) \cdot \Pbound{t}{x}{4}{30x+3} + x + y\]
\end{definition}

\begin{lemma}\label{lemma:reduce-large-components:single-step}
    Let $(G,t)$ be a non-trivial problem instance and let $X$ be a linked tidy modulator for $(G,t)$. If $|V(G)| > \Gbound{t}{|X|}{\Ybound{t}{|X|}}$ we can in $\polyG$ time obtain a problem instance $(G',t)$ with $\twtwodeletion{G'}{t} \Longleftrightarrow \twtwodeletion{G}{t}$ for which one of the following holds:
    \begin{enumerate}
        \item $|V(G')| < |V(G)|$
        \item $V(G') = V(G) \wedge |E(G')| < |E(G)|$ and $X$ is a linked tidy modulator for $(G',t)$
    \end{enumerate}
\end{lemma}

\begin{proof}
    By Theorem \ref{theorem:get-component-separator} we can in $\polyG$ time obtain either a problem instance $(G',t)$ with $\twtwodeletion{G}{t} \Longleftrightarrow \twtwodeletion{G'}{t}$ and $|V(G')| < |V(G)|$, or a  component separator $Y$ for $(G,t,X)$ with $|Y| \leq \Ybound{t}{|X|}$, such that $|\mathcal{C}(G - (X \cup Y))| \leq \Cbound{t}{|X| + |Y|}$ and each $C \in \mathcal{C}(G - (X \cup Y))$ has $|N_G(C) \cap Y| \leq 4$. In the former case, correctness trivially follows. Hence we assume the latter.
    
    Let $C \in \mathcal{C}(G - (X \cup Y))$ for which $|V(C)|$ is maximal. We note that from Definition \ref{def:Cbound} it trivially follows that $\nbound{C}$ is a monotonically increasing function. We use this to derive the minimum size of $V(C)$:
    \medskip
    \begin{align*}
        |V(C)| &\geq \left\lceil \frac{|V(G)| - |X| - |Y|}{| \mathcal{C}(G - (X \cup Y))|} \right\rceil\\
        &\geq \left\lceil \frac{\Gbound{t}{|X|}{\Ybound{t}{|X|}} - |X| - \Ybound{t}{|X|} + 1}{\Cbound{t}{|X| + |Y|}} \right\rceil\\
        &\geq \left\lceil \frac{\Cbound{t}{|X|{{+}}\Ybound{t}{|X|}}{{\cdot}}(1988|X|{{+}}1){{\cdot}}\Pbound{t}{|X|}{4}{30|X|{{+}}3}{{+}}1}{\Cbound{t}{|X|{{+}}\Ybound{t}{|X|}}} \right\rceil\\
        &= (1988|X|{{+}}1) \cdot \Pbound{t}{|X|}{4}{30|X|{{+}}3} + \left\lceil \frac{1}{\Cbound{t}{|X|{{+}}\Ybound{t}{|X|}}} \right\rceil\\
        &= (1988|X|{{+}}1) \cdot \Pbound{t}{|X|}{4}{30|X|{{+}}3} + 1
    \end{align*}
    ~\vspace{-.23cm}\\
    By Lemma \ref{lemma:get-block-cut-tree} we can in $\polyG$ time obtain a block-cut tree $T$ of $C$.
    
    \begin{case}[$|\mathcal{B}(T)| > \Pbound{t}{|X|}{4}{30|X|+3}$]
        We have by Theorem \ref{theorem:reduce-block-cut-tree} that we can in $\polyG$ time apply Reduction \ref{red:remove-middle-edge-to-isolated-path} or \ref{red:ladder-reduction} on an edge in $E(G) \setminus E(G[X])$, or Reduction \ref{red:contract-block-leaf} or \ref{red:reducing-3-induced-blocks-without-neighbours}. We let $(G',t)$ be the obtained problem instance.
        
        \begin{subcase}[Reduction \ref{red:remove-middle-edge-to-isolated-path} or \ref{red:ladder-reduction} was applied]\label{proof:reduce-large-components:single-step:1a}
            By Lemmas \ref{lemma:safeness:remove-middle-edge-to-isolated-path} and \ref{lemma:safeness:ladder-reduction} we have $\twtwodeletion{G'}{t} \Longleftrightarrow \twtwodeletion{G}{t}$. By definition of Reductions \ref{red:remove-middle-edge-to-isolated-path} and \ref{red:ladder-reduction} we have $V(G') = V(G) \wedge |E(G')| < |E(G)|$. The removal of an edge in $E(G) \setminus E(G[X])$ never increases the number of internally vertex-disjoint paths between any pair $\{x,y\} \subseteq X$, nor remove an edge between any such a pair. Hence $X$ will be a linked set for $(G',t)$. Because $G'$ is a subgraph of $G$ it also holds that $X$ is a tidy modulator for $G'$. Hence, in this case returning $(G',t)$ suffices.
        \end{subcase}
        
        \begin{subcase}[Reduction \ref{red:contract-block-leaf} or \ref{red:reducing-3-induced-blocks-without-neighbours} was applied]
            By Lemmas \ref{lemma:safeness:contract-block-leaf} and \ref{lemma:safeness:reducing-3-induced-blocks-without-neighbours} we have $\twtwodeletion{G'}{t} \Longleftrightarrow \twtwodeletion{G}{t}$.
            By Lemma \ref{lemma:block-cut-tree-only-blocks-of-size-at-least-2}, the definition of Reductions \ref{red:contract-block-leaf} and \ref{red:reducing-3-induced-blocks-without-neighbours}, and Definition \ref{def:Pbound} we have $|V(G')| < |V(G)|$. Hence, in this case returning $(G',t)$ suffices.
        \end{subcase}
    \end{case}
    
    \begin{case}[$|\mathcal{B}(T)| \leq \Pbound{t}{|X|}{4}{30|X|+3}$]
        Because each vertex $v \in V(C)$ is in at least one block $B \in \mathcal{B}(T)$ there must exist some $B \in \mathcal{B}(T)$ with $|V(B)| \geq \lceil |V(C)| / |\mathcal{B}(T)| \geq 1988|X| + 2$. Then by Theorem \ref{theorem:reduce-big-biconnected-graph} we can in $\polyG$ time apply Reduction \ref{red:remove-middle-edge-to-isolated-path} or \ref{red:ladder-reduction} on an edge in $E(G) \setminus E(G[X])$. Following the same reasoning as in Case 1.1, returning the obtained problem instance $(G',t)$ suffices.
    \end{case}
    
    \noindent
    Because these cases cover all cases, correctness trivially follows.
\end{proof}
\clearpage

\noindent
Lemma \ref{lemma:reduce-large-components:single-step}, in essence, corresponds to a single iteration of the while-loop on lines 6 to 16 in Algorithm \ref{algo:kernelize}. This while-loop will continue iterating as long as Lemma \ref{lemma:reduce-large-components:single-step} yields a problem instance with $|V(G')| = |V(G)|$. We only terminate this while-loop once the number of vertices in the graph have strictly decreased. The following lemma captures this behaviour:

\begin{lemma}\label{lemma:reduce-large-components:many-step}
    Let $(G,t)$ be a non-trivial problem instance and let $X$ be a linked tidy modulator for $(G,t)$. If $|V(G)| > \Gbound{t}{|X|}{\Ybound{t}{|X|}}$ we can in $\polyG$ time obtain a problem instance $(G',t)$ with $\twtwodeletion{G'}{t} \Longleftrightarrow \twtwodeletion{G}{t}$ and $|V(G')| < |V(G)|$.
\end{lemma}

\begin{proof}
    We prove using induction on $|E(G)|$.
    
    \begin{basecase}[$E(G) = \emptyset$]
        We have $\treewidth{G} = 0$, which contradicts $(G,t)$ being a non-trivial problem instance. Hence, this case vacuously holds.
    \end{basecase}
    
    \begin{inductivestep}[$E(G) \neq \emptyset$]
        By Lemma \ref{lemma:reduce-large-components:single-step} we can in $\polyG$ time obtain a problem instance $(G',t)$ with $\twtwodeletion{G'}{t} \Longleftrightarrow \twtwodeletion{G}{t}$ for which either $|V(G')| < |V(G)|$ holds or $V(G') = V(G) \wedge |E(G')| < |E(G)|$ and $X$ is a linked tidy modulator for $(G',t)$.
    
        \begin{indsubcase}[$|V(G')| < |V(G)|$]
            Returning $(G',t)$ trivially suffices.
        \end{indsubcase}
    
        \begin{indsubcase}[$V(G')| \not < |V(G)|$ and $(G',t)$ is a non-trivial problem instance]
            Correctness immediately follows from the induction hypothesis.
        \end{indsubcase}
        
        \begin{indsubcase}[$V(G')| \not < |V(G)|$ and $(G',t)$ is a trivial problem instance]
            Lemma \ref{lemma:get-non-trivial-problem-instance} yields that we can in $\poly{|G'|} \leq \polyG$ time obtain a problem instance $(G'',t)$ with $\twtwodeletion{G''}{t} \Longleftrightarrow \twtwodeletion{G'}{t}$ such that $|V(G'')| \leq 4$ or $(G'', t)$ is a non-trivial problem instance with $|V(G'')| < |V(G')|$. Because $(G,t)$ is non-trivial we have $|V(G)| > 4$, which means that in either case $|V(G'')| < |V(G)|$ holds. Hence, returning $(G'',t)$ suffices.
        \end{indsubcase}
    \end{inductivestep}
    
    \noindent
    Because these cases cover all cases, correctness trivially follows.
\end{proof}

\noindent
We finally define a function $\text{\textsc{bound}} : \nat \times \real \rightarrow \nat$ such that for a given polynomial time $\varepsilon$-approximation algorithm and a problem instance $(G,t)$ with $|V(G)| > \bound{t}{\varepsilon}$ we can apply a reduction. This procedure corresponds to lines 2 to 17 in Algorithm \ref{algo:kernelize}.

\begin{definition}\label{def:bound}
    \[\bound{t}{\varepsilon} = \max(4, \Gbound{t}{\varepsilon t(3t+4)}{\Ybound{t}{\varepsilon t(3t+4)}})\]
\end{definition}

\begin{lemma}\label{lemma:reduce-problem-instance:single-step}
    Let $(G,t)$ be a problem instance and let $\mathcal{A}$ be a polynomial time $\varepsilon$-approximation algorithm for the \textsc{Treewidth-2 Vertex Deletion} problem. If $V(G) > \bound{t}{\varepsilon}$, then we can in $\polyG$ time obtain a problem instance $(G',t')$ with $\twtwodeletion{G'}{t'} \Longleftrightarrow \twtwodeletion{G}{t}$ and $|V(G')| < |V(G)| \wedge t' \leq t$.
\end{lemma}

\begin{proof}
    In case $(G,t)$ is a trivial problem instance, correctness immediately follows from Lemma \ref{lemma:get-non-trivial-problem-instance}. Hence, we assume $(G,t)$ is non-trivial.
    
    From Theorem \ref{theorem:get-linked-tidy-modulator} it follows that we can in $\polyG$ time obtain a problem instance $(G',t')$ with $\twtwodeletion{G}{t} \Longleftrightarrow \twtwodeletion{G'}{t'}$ and $|V(G')| < |V(G)|$ and $t' \leq t$, or we obtain a problem instance $(G',t)$ and a linked tidy modulator $X$ for $(G',t)$ with $\twtwodeletion{G}{t} \Longleftrightarrow \twtwodeletion{G'}{t}$ and $V(G') = V(G)$ and $|X| \leq \varepsilon t (3t+4)$. In case the former holds correctness trivially follows. Hence, we assume the latter.
    
    We note that from Definitions \ref{def:Ybound} and \ref{def:Gbound} it follows that $\nbound{G}$ and $\nbound{Y}$ are monotonically increasing functions. Hence, we have $|V(G')| = |V(G)| > \bound{t}{\varepsilon} \geq \Gbound{t}{\varepsilon t(3t+4)}{\Ybound{t}{\varepsilon t(3t+4)}} \geq \Gbound{t}{|X|}{\Ybound{t}{|X|}}$. Hence, from Lemma \ref{lemma:reduce-large-components:many-step} it follows that we can in $\poly{|G'|} = \polyG$ time obtain a problem instance $(G'',t)$ with $\twtwodeletion{G''}{t} \Longleftrightarrow \twtwodeletion{G'}{t}$ and $|V(G'')| < |V(G')| = |V(G)|$. Returning this problem instance $(G',t)$ trivially suffices.
\end{proof}

\noindent
We can use Lemma \ref{lemma:reduce-problem-instance:single-step} in a straightforward proof by induction to completely describe the behaviour expressed by Algorithm \ref{algo:kernelize}.

\begin{theorem}\label{theorem:reduce-problem-instance:many-step}
    Let $(G,t)$ be a problem instance and let $\mathcal{A}$ be a polynomial time $\varepsilon$-approximation algorithm for the \textsc{Treewidth-2 Vertex Deletion} problem. We can in $\polyG$ time obtain a problem instance $(G',t')$ with $\twtwodeletion{G'}{t'} \Longleftrightarrow \twtwodeletion{G}{t}$ and $|V(G')| \leq \bound{t}{\varepsilon}$.
\end{theorem}

\begin{proof}
    We prove using induction on $|V(G)|$.
    
    \begin{basecase}[$|V(G)| \leq \bound{t}{\varepsilon}$]
        Trivially holds because returning $(G,t)$ suffices.
    \end{basecase}
    
    \begin{inductivestep}[$|V(G)| > \bound{t}{\varepsilon}$]
        By Lemma \ref{lemma:reduce-problem-instance:single-step} we can in $\polyG$ time obtain a problem instance $(G',t')$ with $\twtwodeletion{G'}{t'} \Longleftrightarrow \twtwodeletion{G}{t}$ and $|V(G)'|  < |V(G)| \wedge t' \leq t$. From the induction hypothesis it follows that we can in $\poly{|G'|} \leq \polyG$ time obtain a problem instance $(G'',t'')$ with $\twtwodeletion{G''}{t''} \Longleftrightarrow \twtwodeletion{G'}{t'}$ and $|V(G'')| \leq \bound{t'}{\varepsilon}$. From Definition \ref{def:bound} it follows that \textsc{bound} is a monotonically increasing function, which implies $|V(G'')| \leq \bound{t'}{\varepsilon} \leq \bound{t}{\varepsilon}$. Hence, returning $(G'',t'')$ suffices.
    \end{inductivestep}
    
    \noindent
    Because these cases cover all cases, correctness trivially follows.
\end{proof}

\noindent
Using the proof given in Appendix \ref{appendix:derivation-kernel-size}, we are now able to provide an explicit kernel bound.

\begin{lemma}[Appendix \ref{appendix:derivation-kernel-size}]\label{lemma:kernel-size}
    $\bound{t}{\varepsilon} = O(\varepsilon^{18} t^{41})$.
\end{lemma}

\begin{theorem}
    Let $(G,t)$ be a problem instance. We can in $\polyG$ time obtain a problem instance $(G',t')$ with $\twtwodeletion{G'}{t'} \Longleftrightarrow \twtwodeletion{G}{t}$ and $|V(G')| = O(t^{41})$.
\end{theorem}

\begin{proof}
    Directly follows from Theorem \ref{theorem:reduce-problem-instance:many-step} and Lemmas \ref{lemma:existance-approximation-algorithm} and \ref{lemma:kernel-size}.
\end{proof}

\chapter{Conclusion}\label{ch:conclusion}
We have presented a kernelization algorithm that in $\polyG$ time transforms a problem instance $(G,t)$ into a problem instance $(G',t')$ with $|V(G')| = O(t^{41})$ for which $\twtwodeletion{G}{t} \Longleftrightarrow \twtwodeletion{G'}{t'}$ holds. Hence, our algorithm proves that the \textsc{Treewidth-2 Vertex Deletion} admits a kernel of $O(t^{41})$ vertices, which is the first known explicit kernel size.

Our kernelization algorithm, when combined with a `trivial' $\binom{|V(G)|}{t} \cdot \polyG$ brute force algorithm (i.e. check  every candidate subset of $V(G)$ of size $t$ using Lemma \ref{lemma:find-tree-decompositions}), yields an explicit $O(t^{41t}) \cdot \polyG$ FPT algorithm for the \textsc{Treewidth-2 Vertex Deletion} problem. Although this algorithm is not a single-exponential FPT algorithm, like the algorithms by Kim et al. \cite{A_single-exponential_FPT_algorithm_for_the_K4-minor_cover_problem} and by Fomin et al. \cite{Planar_F-Deletion_Approximation_and_Optimal_FPT_Algorithms}, this algorithm is the first constructive FPT algorithm that solves the \textsc{Treewidth-2 Vertex Deletion} problem.

\paragraph{Kernel Size}
One of the remaining questions is whether, or rather how, a tighter kernel bound can be obtained. As proven by Dell and van Melkebeek \cite{DBLP:journals/jacm/DellM14}, there will not exist a kernel consisting out of $O(t^{2 - \varepsilon})$ edges with $\varepsilon > 0$ unless $\text{coNP} \subseteq \text{NP}/\text{poly}$. Our current kernel of $O(t^{41})$ vertices is far larger than this potential lower bound. Hence, we would prefer bringing these bounds closer together.

One of the major causes of the large degree of the polynomial bound, that can potentially be improved upon, is the large number of components $\nbound{C}$ needed to guarantee that Lemma \ref{lemma:application:reduce-number-of-components} will find an application of Reduction \ref{red:reduce-number-of-components}.

Another effective approach to decrease the degree of the polynomial is to reduce the size $\nbound{Y}$ of component separators. The potentially easiest way to achieve this would be to lower the $\nbound{C}$ term. Other approaches will most likely require either a redefinition of a component separator or a completely new method to obtaining these component separators.

Other foreseeable improvements will likely follow from the introduction of new reduction rules that operate on block-cut trees, or more effective methods to find applications of existing reduction rules. Nevertheless, as follows from the derivation in Appendix \ref{appendix:derivation-kernel-size}, such improvements will have rather limited potential due to $\Pbound{t}{|X|}{|N_G(C) \cap Y|}{30|X| + 3}$ only having quintic dependency on parameter $t$. A more effective approach could involve reducing a block-cut graph (forest) over graph $G - (X \cup Y)$, opposed to individual components in $G - (X \cup Y)$. This will likely requiring changing the definition of a component separator $Y$, due to the currently very limited `interaction' between different components in $G - (X \cup Y)$.

\paragraph{Generalisation to other F-Minor Free Families} An important remaining task is to check whether the method of decomposing graphs outlined in Chapter \ref{ch:graph-decompositions} can be generalised to other instances of the $\mathcal{F}$-\textsc{Minor Cover} problem as well. We note that our method of decomposing graphs is similar to the method of decomposing graphs into near-protrusions introduced by Fomin et al. \cite{Planar_F-Deletion_Approximation_and_Optimal_FPT_Algorithms}. Nevertheless, our method is able to provide the stronger guarantee that components are adjacent to limit-1 cliques in $X$, opposed to the method by Fomin et al., which would for \textsc{Treewidth-2 Vertex Deletion} only guarantee that each component is adjacent to a limit-3 subset in $X$ (see Lemma 25 and 26 in \cite{Planar_F-Deletion_Approximation_and_Optimal_FPT_Algorithms}).

The property that sets $N_G(C) \cap X$ are limit-1 subsets, opposed to limit-3 subsets, is essential to the safety of the constructive reduction rules introduced in Chapters \ref{ch:reducing-biconnected-subgraphs} and \ref{ch:reducing-block-cut-trees}. Although constructive reduction rules that are safe under the guarantee that $N_G(C) \cap X$ is a limit-3 subset exist (e.g., protrusion replacement would be constructive if the set of representatives is known), finding such constructive reductions will be much more difficult.

The question is whether both approaches can be brought together into a single approach that both provides the stronger limit-1 neighbourhood guarantee and is applicable to any instance of the $\mathcal{F}$-\textsc{Minor Cover} problem where $\mathcal{F}$ contains a planar graph. The existence of such an approach would drastically simplify the creation of constructive kernels for this class of problems. A complementing question would be to verify whether the limit-$m$ bound by Fomin et al. \cite{Planar_F-Deletion_Approximation_and_Optimal_FPT_Algorithms} is tight when the component separator $Y$ is of size polynomial in $t$. This, because such a proof would yield that the limit-1 neighbourhood guarantee is not achievable by a generalised kernelization algorithm that operates on all instances of the $\mathcal{F}$-\textsc{Minor Cover} problem where $\mathcal{F}$ contains a planar graph.

\clearpage
\phantomsection
\addcontentsline{toc}{chapter}{Bibliography}
\bibliographystyle{plainurl}
\bibliography{literature}

\clearpage
\phantomsection
\addcontentsline{toc}{chapter}{Appendices}
\appendix 
\chapter{Derivation Kernel Size}\label{appendix:derivation-kernel-size}
In Theorem \ref{theorem:reduce-problem-instance:many-step} it is stated that given an $\varepsilon$-approximation algorithm for the \textsc{Treewidth-2 Vertex Deletion} problem and a problem instance $(G,t)$, that we can obtain a problem instance $(G',t')$ with $\twtwodeletion{G'}{t'} \Longleftrightarrow \twtwodeletion{G}{t}$ and $|V(G')| \leq \bound{t}{\varepsilon}$. Within this appendix we derive an asymptotic upper bound to the size of the obtained problem instance.

We note that in case $t = 0$ holds we can obtain a problem instance $(G',t')$ with $|V(G')| \leq 4 \leq \bound{t}{\varepsilon}$, as follows from Definitions \ref{def:trivial-problem-instance} and \ref{def:bound} and Lemma \ref{lemma:get-non-trivial-problem-instance}. Hence, for this reason it remains to derive an upper bound to the case where $t > 0$.

As discussed in Section \ref{sec:prelim:approximation}, we have that $\varepsilon \geq 1$ will hold for any $\varepsilon$-approximation algorithm for the \textsc{Treewidth-2 Vertex Deletion} problem.

To derive an upper bound we will make use of the Definitions \ref{def:Cbound}, \ref{def:Ybound}, \ref{def:Lbound}, \ref{def:Pbound}, \ref{def:Gbound}, and \ref{def:bound}. We list these definitions below:
\begin{align*}
    \Cbound{t}{x} &= \binom{x}{3}(t+1) + \binom{x}{2}(2t+3)\\[1.5ex]
    \Ybound{t}{x} &= 6\binom{x}{2}(3t+3 + \Cbound{t}{x+3t+3})\\[1ex]
    \Lbound{t}{x}{y} &= y(x+y-1)(t+2)+x+y\\[2ex]
    \Pbound{t}{x}{y}{k} &= \max(1, \Lbound{t}{x}{y}(2k+3)+3976kxy - 3k - 4)\\[2ex]
    \Gbound{t}{x}{y} &= \Cbound{t}{x{{+}}y} \cdot (1988x+1) \cdot \Pbound{t}{x}{4}{30x{{+}}3} + x + y\\[2ex]
    \bound{t}{\varepsilon} &= \max(4, \Gbound{t}{\varepsilon t(3t+4)}{\Ybound{t}{\varepsilon t(3t+4)}})
\intertext{
    We will derive an upper bound to the value of $\bound{t}{\varepsilon}$ using a bottom up approach by first analysing upper-bounds to the values of sub-formulas. For notation convenience we will define $x = \varepsilon t (3t+4)$ and we define $y = \Ybound{t}{x}$.
}
\intertext{
    We first derive an upper bound to the value of $x$
}
    x &= \varepsilon t (3t + 4)\\
    &\leq \varepsilon t (3t + 4t)\\
    &\leq 7 \varepsilon t^2
\intertext{
    We next derive an upper bound to the value of $\Cbound{t}{x + 3t + 3}$
}
    \Cbound{t}{x{{+}}3t{{+}}3} &= \binom{x + 3t + 3}{3}(t+1) + \binom{x + 3t + 3}{2}(2t+3)\\
    &\leq 2t \binom{7 \varepsilon t^2 + 3t + 3}{3} + 5t \binom{7 \varepsilon t^2 + 3t + 3}{2}\\
    &\leq 2t \binom{13 \varepsilon t^2}{3} + 5t \binom{13 \varepsilon t^2}{2}\\
    &< 2t \cdot 2197 \varepsilon^3 t^6 + 5t \cdot 169 \varepsilon^2 t^4\\
    &\leq 5239 \varepsilon^3 t^7
\intertext{
We next derive an upper bound to $y$.
}
    y &= \Ybound{t}{x}\\
    &= 6 \binom{x}{2}(3t+3 + \Cbound{t}{x+3t+3})\\
    &< 6 \binom{7 \varepsilon t^2}{2} (3t + 3 + 5239 \varepsilon^3 t^7)\\
    &< 6 \cdot 49 \varepsilon^2 t^4 \cdot 5245 \varepsilon^3 t^7\\
    &\leq 1542030 \varepsilon^5 t^{11}
\intertext{
    We next derive an upper bound to $\Lbound{t}{x}{4}$.
}
    \Lbound{t}{x}{4} &= 4(x+4-1)(t+2)+x+4\\
    &\leq 4 \cdot 3t (7 \varepsilon t^2 + 3) + 7 \varepsilon t^2 + 4\\
    &< 12t \cdot 10 \varepsilon t^2 + 11 \varepsilon t^2\\
    &\leq 131 \varepsilon t^3
\intertext{
    We next derive an upper bound to $\Pbound{t}{x}{4}{30x+3}$.
}
    \Pbound{t}{x}{4}{30x{{+}}3} &= \max(1, \Lbound{t}{x}{4}(60x+9)+ 15904x(30x+3)-90x-11)\\
    &\leq \max(1, 131 \varepsilon t^3 \cdot (420 \varepsilon t^2 + 9) + 15904 \cdot 7 \varepsilon t^2 (210 \varepsilon t^2 + 3))\\
    &\leq \max(1, 131 \varepsilon t^3 \cdot 429 \varepsilon t^2 + 111328 \varepsilon t^2 \cdot 213 \varepsilon t^2)\\
    &\leq 23769063 \varepsilon^2 t^5
\intertext{
    We next derive an upper bound to $\Cbound{t}{x+y}$.
}
    \Cbound{t}{x+y} &= \binom{x + y}{3}(t+1) + \binom{x + y}{2}(2t+3)\\
    &< 2t \binom{7 \varepsilon t^2 + 1542030 \varepsilon^5 t^{11}}{3} + 5t \binom{7 \varepsilon t^2 + 1542030 \varepsilon^5 t^{11}}{2}\\
    &\leq 2t \binom{1542037 \varepsilon^5 t^{11}}{3} + 5t \binom{1542037 \varepsilon^5 t^{11}}{2}\\
    &< 2t \cdot 3666776026137044653 \varepsilon^{15} t^{33} + 5t \cdot 2377878109369 \varepsilon^{10} t^{22}\\
    &\leq 7333563941664636151 \varepsilon^{15} t^{34}
\intertext{
    Using these previous derivations, we will now derive an upper bound to $\bound{t}{\varepsilon}$. We note that in this derivation we use the observation that $t > 0$ and $\varepsilon \geq 1$ yields $x = \varepsilon t (3t+4) \geq 7 \geq 4$.
}
    \bound{t}{\varepsilon} &= \max(4, \Gbound{t}{\varepsilon t(3t+4)}{\Ybound{t}{\varepsilon t(3t+4)}})\\
    &= \max(4, \Gbound{t}{x}{\Ybound{t}{x}})\\
    &= \max(4, \Gbound{t}{x}{y})\\
    &= \Cbound{t}{x{{+}}y} \cdot (1988x+1) \cdot \Pbound{t}{x}{4}{30x{{+}}3} + x + y\\
    &< 7333563941664636151 \varepsilon^{15} t^{34}{{\cdot}}13917 \varepsilon t^2{{\cdot}}23769063 \varepsilon^2 t^5{{+}}1542037 \varepsilon^5 t^{11}\\
    &\leq 2425899315517822591524501413458 \varepsilon^{18} t^{41}
\intertext{
    When $t > 0$ we can conclude $\bound{t}{\varepsilon} < 2425899315517822591524501413458 \varepsilon^{18} t^{41}$. Hence, we have $\bound{t}{\varepsilon} \leq \max(4, 2425899315517822591524501413458 \varepsilon^{18} t^{41}) = O(\varepsilon^{18} t^{41})$.
}
\end{align*}

\end{document}